\documentclass[12pt]{article}

\usepackage{setspace}
\pdfoutput=1
\usepackage[a4paper,top=1.2in,bottom=1in,left=1.2in,right=1.2in]{geometry}

\usepackage{caption}
\usepackage{titlesec}
\renewcommand{\thesubsection}{\thesection.\arabic{subsection}}
\renewcommand{\thesubsubsection}{\thesubsection.\arabic{subsubsection}}
\titleformat{\section}[hang]{\normalfont\bfseries}{\thesection}{.5em}{}
\titlelabel{\subsubsection}{\empty}
\titleformat{\subsection}{\normalfont\bfseries}{\thesubsection.}{.5em}{}[]
\titleformat{\subsubsection}{\normalfont\bfseries}{\thesubsubsection.}{.5em}{}[]

\newcommand{\sectionbreak}{\vspace{1em}}
\renewcommand{\abstract}[1]{{\gdef\thepoabstract{#1}}}

\usepackage{amsthm}

\makeatletter
\renewcommand\maketitle
{
\thispagestyle{plain}
\rmfamily\selectfont
\ifdefined\@title{\noindent\Large\bfseries\centering \@title \par}\else\fi
\vspace{2em}
\ifdefined\@author{\centering\normalfont \@author \par}
\vspace{.5em}
\ifdefined\thepoaffiliation{\noindent\small\thepoaffiliation\par}\fi\vspace{3em}\else\fi
\ifdefined\thepoabstract{\small\noindent{{\bfseries Abstract}.\;\;}\thepoabstract\par\vspace{2em}}\else\fi
\ifdefined\thepokeywords{{\noindent\bfseries Keywords\;\;}\thepokeywords\par\vspace{5em}}\else\fi
\ifdefined\theporuntitle{\fancyhead[L]{\footnotesize\theporuntitle}}\fi
}
\makeatother

\usepackage[usenames,dvipsnames]{xcolor}
\definecolor{RefColor}{rgb}{0,0,.85}
\definecolor{UrlColor}{rgb}{.5,.5,.5}%
\RequirePackage[colorlinks,linkcolor=RefColor,citecolor=RefColor,urlcolor=UrlColor,linktoc=page]{hyperref}
\hypersetup{pdfinfo={Subject={ }}}
\RequirePackage[OT1]{fontenc}
\usepackage{natbib}

\usepackage[english]{babel}

\usepackage{enumitem}
\setlist[itemize]{leftmargin=1.5em}

\usepackage{amsmath,amssymb,amscd,amsfonts,amsthm,mathtools}
\usepackage{aliascnt}
\usepackage[noabbrev,capitalize]{cleveref}
\usepackage{dsfont}
\usepackage{nccmath}
\usepackage{scalerel}

\allowdisplaybreaks

\usepackage{bibentry}

\usepackage{booktabs}
\usepackage{tikz}
\usepackage[low-sup]{subdepth}
\usetikzlibrary{calc}
\usetikzlibrary{decorations.markings,decorations.pathreplacing}
\tikzstyle{mybraces}=[mirrorbrace/.style={
          decoration={brace, mirror},
          decorate},brace/.style={
          decoration={brace},
          decorate}]
\usepackage{tikz-cd}

\usepackage{tocloft}

\cftsetpnumwidth{5cm}
\cftsetrmarg{6cm}

\usepackage{wrapfig}
\usepackage{subcaption}
\usepackage[nottoc]{tocbibind}
\settocbibname{References}

\newtheoremstyle{basic}{8pt}{10pt}{\itshape}{}{\bfseries}{.}{.4em}{}
\theoremstyle{basic}
\newtheorem{theorem}{Theorem}
\newaliascnt{lemma}{theorem}
\newtheorem{lemma}[lemma]{Lemma}
\aliascntresetthe{lemma}
\crefname{lemma}{Lemma}{Lemmas}
\newaliascnt{fact}{theorem}

\aliascntresetthe{fact}
\crefname{fact}{Fact}{Facts}
\newaliascnt{proposition}{theorem}

\aliascntresetthe{proposition}
\crefname{proposition}{Proposition}{Propositions}
\newaliascnt{corollary}{theorem}
\newtheorem{corollary}[corollary]{Corollary}
\aliascntresetthe{corollary}
\crefname{corollary}{Corollary}{Corollaries}
\theoremstyle{definition}
\newaliascnt{definition}{theorem}
\newtheorem{definition}[definition]{Definition}
\aliascntresetthe{definition}
\crefname{definition}{Definition}{Definitions}
\newtheorem{remark}{Remark}
\newtheorem*{remark*}{Remark}
\newtheorem{example}{Example}
\newtheorem{assumption}{Assumption}
\newaliascnt{choice}{theorem}

\aliascntresetthe{choice}
\crefname{choice}{Choice}{Choices}

\crefname{assumption}{Assumption}{Assumptions}

\usepackage{times}

\usepackage{graphicx}

\usetikzlibrary{shapes.geometric}
\usetikzlibrary{arrows.meta}

\usepackage{nccmath,tikz}
\usepackage{stmaryrd,scalerel} 

\newenvironment{proplist}{\begin{enumerate}[
    leftmargin=2.5em,
    labelwidth=0em,
    label=(\alph*),
    topsep=.1em,
    partopsep=0em,
    itemsep=0em
    ]}
{\end{enumerate}}

\def\P{{\mathbb P}}

\def\Var{\text{\rm Var}}

\def\bd{\boldsymbol{d}}

\def\bw{\boldsymbol{w}}
\def\bx{\boldsymbol{x}}
\def\by{\boldsymbol{y}}
\def\bz{\boldsymbol{z}}

\def\bone{\boldsymbol{1}}

\def\A{\mathbb{A}}

\def\G{\mathbb{G}}

\def\N{\mathbb{N}}
\def\O{\mathbb{O}}
\def\P{{\mathbb P}}

\def\R{\mathbb{R}}
\def\S{\mathbb{S}}
\def\T{\mathbb{T}}

\def\X{\mathbb{X}}

\def\Z{\mathbb{Z}}

\def\x{\mathbb{x}}
\def\y{\mathbb{y}}

\def\cC{\mathcal{C}}

\def\cE{\mathcal{E}}

\def\cH{\mathcal{H}}

\def\cM{\mathcal{M}}
\def\cN{\mathcal{N}}

\def\cP{\mathcal{P}}

\def\cS{\mathcal{S}}
\def\cT{\mathcal{T}}

\def\cX{\mathcal{X}}
\def\cY{\mathcal{Y}}

\DeclareMathOperator{\msum}{\medmath\sum}

\DeclareMathOperator{\mint}{\medmath\int}
\newcommand{\argdot}{{\,\vcenter{\hbox{\tiny$\bullet$}}\,}}
\newcommand{\tagaligneq}{\refstepcounter{equation}\tag{\theequation}}

\newcommand{\msup}{\sup\nolimits}
\newcommand{\minf}{\inf\nolimits}
\newcommand{\mmax}{\max\nolimits}

\newcommand{\ind}{\mathbb{I}}

\newcommand{\mean}{\mathbb{E}}

\newcommand{\Rank}{{\rm Rank}}
\newcommand{\Xobs}{{\X_{\rm obs}}}

\newcommand{\CUalpha}{{\cC_\alpha}}

\usepackage{makecell}

\expandafter\def\expandafter\normalsize\expandafter{%
    \normalsize%
    \setlength\abovedisplayskip{4pt}%
    \setlength\belowdisplayskip{4pt}%
    \setlength\abovedisplayshortskip{-8pt}%
    \setlength\belowdisplayshortskip{2pt}%
}

\begin{document}

\title{
    Data augmented bootstrap: Unifying confidence interval construction by approximate invariance
    \\[-.5em]
}

\author{
    Kevin Han Huang
    \\
    Department of Statistics, \,University of Warwick
    \\[-1.8em]
}

\begin{abstract}
    {
        We propose the data augmented bootstrap (DAB), a framework for constructing confidence intervals from approximately invariant transformations of the data. As special cases, DAB recovers popular methods that rely on exact group symmetries, such as conformal prediction, wild bootstrap for Maximum Mean Discrepancy U-statistics and the recently proposed SymmPI. Meanwhile, DAB also recovers the classical bootstrap method, which exploits the dataset's approximate invariance under uniform sampling of data indices as the dataset size grows. For all DAB methods, we establish theoretical coverage results that interpolate between finite-sample and asymptotic guarantees according to the strength of the invariance, and without assuming a group structure. The approximate invariance is measured in the Kolmogorov distance and, for statistics that satisfy Gaussian universality, reduces to conditional mean and variance matching. This allows us to incorporate data augmentation (DA), a widely used machine learning heuristic based on approximate invariances, into known statistical methods. We empirically test the performance of incorporating DA into bootstrap, wild bootstrap and conformal prediction for simulated settings as well as for image, language and scientific data. 
}
\end{abstract}
\ 
\\[-6em]

\maketitle

\vspace{-1em}

\vspace{-1.5em}

\section{Introduction} \label{sec:intro}

\vspace{-.5em}

Many methods now exist for constructing confidence intervals based on certain group invariance structure of the underlying data. A very popular example is conformal prediction \citep{vovk2005algorithmic}: Roughly speaking, a confidence interval (CI) for the $(n+1)$-th output $Y_{n+1}$ for some input $X_{n+1}$ is produced by ranking
\vspace{-.5em}
\begin{align*}
    R_1, \, R_2, \ldots, \, R_{n+1}\;,
\end{align*}
\\[-2.5em]
where $(R_i)_{i \leq n+1}$ are residuals for the i.i.d.~input-output pairs $(X_i, Y_i)_{i \leq n}$; see \cite{angelopoulos2023conformal} among others for detailed reviews. The validity of conformal prediction, at least in its simplest form, rests on the exchangeability of $(R_i)_{i \leq n+1}$, which implies
\vspace{-.25em}
\begin{align*}
    R_i \;\overset{d}{=}\; R_j
    \qquad 
    \text{ for all } 1 \leq i,j \leq n+1\;.
\end{align*}
\\[-2.25em]
A similar invariance structure has been used in other setups: In permutation tests (see e.g.~\cite{good2005permutation}), one considers test statistics that are distributionally invariant under data index permutations; in wild bootstrap methods for kernel two-sample tests, one can simulate certain degenerate U-statistics by exploiting their invariance under reweighting with Rademacher variables \citep{schrab2023mmd}; a recent work \citep{dobriban2023symmpi} extends conformal prediction to general group symmetries, and shows that adding invariances improves performance. The guarantees for these methods all rest on the \emph{exact distributional invariance} of a test statistic $f(X)$ under some group transformations $\Phi_b$:
\vspace{-.25em}
\begin{align*}
    f(X) \;\overset{d}{=}\; f(\Phi_b(X))
    \;,
    \qquad 
    \text{ dataset }  \;\; X \coloneqq (X_i)_{i \leq n}\;.
    \tagaligneq \label{eq:exact:inv}
\end{align*}
\\[-2.25em]
Under \eqref{eq:exact:inv}, all such methods enjoy distribution-free, finite-sample validity, i.e.~their CIs have the correct coverage for every $n$ (up to an $O(\frac{1}{n})$ error) and do not require users to know the distribution of $X$. However, the requirement of exact group invariance can be too stringent: It excludes many practically relevant notions of approximate invariance, such as small-angle random rotations of images, random shuffling of text data, and other data augmentation\footnote{We follow the convention in machine learning and use \emph{data augmentation} to mean synthetic transformations of a dataset. This differs from a separate meaning in statistics, i.e.~the addition of latent variables.} methods from machine learning \citep{shorten2019survey,shorten2021text,chen2020group,huang2022data}. Moreover, the size of the CI, which measures its efficiency, does depend on the distribution of $X$. For conformal prediction, tremendous efforts have been invested in designing problem-specific score functions (i.e.~design of the residual $R_i$'s) to obtain efficient CIs \citep{lei2018distribution,romano2019conformalized,sadinle2019least,romano2020classification,kato23a}.

To motivate how approximate invariance may be incorporated, we turn to the classical method of bootstrap \citep{efron1992bootstrap}, whose link to exchangeability was already noted by \cite{praestgaard1993exchangeably}. Our key observation is that bootstrap operates under a suitable notion of approximate invariance. At a high level, the bootstrap empirical distribution is formed by ranking
\begin{align*}
    \mfrac{1}{\sqrt{n}} \msum_{i \leq n} X_i 
    \,,\, 
    \quad
    \mfrac{1}{\sqrt{n}} \msum_{i \leq n} X_{\pi_{1i}} 
    \,,\,
    \quad
    \ldots 
    \,,\,
    \quad
    \mfrac{1}{\sqrt{n}} \msum_{i \leq n} X_{\pi_{Bi}} 
    \;,
\end{align*}
where $(\pi_{bi})_{b \leq B, i \leq n}$ are i.i.d.~uniform samples from the index set $[n] \coloneqq \{1,\ldots,n\}$. If we view $f(X) = \frac{1}{\sqrt{n}} \sum_{i \leq n} X_i$ and $f(\Phi_b(X)) = \frac{1}{\sqrt{n}} \sum_{i \leq n} X_{\pi_{bi}}$, the transformation $\Phi_b$ replaces the data indices $\{1,\ldots,n\}$ with a uniformly draw element of the $n$-fold product $[n]^n$, and can be identified as a uniformly random permutation on $[n]^n$. The classical guarantee for bootstrap can be expressed informally as the distributional approximation
\vspace{-.25em}
\begin{align*}
    f(\Phi_b(X)) \,| \, X 
    \;\overset{d}{\approx}\;
    f(X')
    \qquad 
    \text{ as } \;\; n \rightarrow \infty\;,
    \tagaligneq \label{eq:asymp:inv}
\end{align*}
\\[-2.25em]
where $X'$ is an i.i.d.~copy of the dataset $X$. Comparing \eqref{eq:asymp:inv} to \eqref{eq:exact:inv}, we see that the validity of bootstrap comes from a notion of approximate invariance: $\frac{1}{\sqrt{n}} \sum_{i \leq n} X_{\pi_{bi}}$ does not have the same marginal distribution as $\frac{1}{\sqrt{n}} \sum_{i \leq n} X_i$ at finite $n$, but the distributions become close asymptotically. In fact, \eqref{eq:asymp:inv} provides a stronger statement due to the conditioning: While $f(\Phi_b(X))$ and $f(X)$ are dependent at finite $n$, they asymptotically decouple.

\begin{table}[t]
    \caption{\label{table:special:cases} \small Known methods that are special cases of DAB.
    }
    \vspace{-.5em}
    \footnotesize
    \centering
    \begin{tabular}{|c|c|c|c|c|}
        \Xhline{2\arrayrulewidth}
        Examples of DAB & $\Phi$ & Verification & Typical assumption \\
        \Xhline{2\arrayrulewidth}
        Bootstrap  & \makecell{uniform draws or \\
        Gaussian samples} & Examples \ref{example:bootstrap},\ref{example:asymptotic:bootstrap},\ref{example:conditional:bootstrap} & approx.~invariance
        \\
        \hline 
        Wild bootstrap & \makecell{Rademacher or \\
        Gaussian samples} & \Cref{example:marginal:wild:bootstrap,example:asymptotic:wild:bootstrap} &  exact or approx.~invariance
        \\
        \hline 
        Leave-one-out  & cyclic group & \Cref{example:jackknife} in appendix & approx.~invariance
        \\
        \hline
        Permutation test  & permutations & \Cref{example:permutation} in appendix  & exact invariance
        \\
        \hline 
        Conformal prediction & cyclic group & \Cref{sec:validity:exact:inv} & exact invariance
        \\
        \hline 
        SymmPI & general compact group & \Cref{sec:validity:exact:inv} & exact invariance
        \\
        \Xhline{2\arrayrulewidth}
    \end{tabular}
    \\[-.8em]
\end{table}

\vspace{.5em}

Generalising the above observations, this article proposes a framework for constructing CIs that are valid under both exact group invariance and approximate non-group invariance. We coin this framework \emph{data augmented bootstrap} (DAB), due to the theoretical inspiration from bootstrap and the practical motivation of incorporating data augmentations from machine learning. \Cref{table:special:cases} lists known inference methods that fall under DAB. We show that, under this unifying perspective, one may share and extend theoretical insights and practical techniques across the entire DAB family. A key implication is that practical data augmentations, which often only satisfy approximate invariance and may not come from a group structure, can be used in conjunction with existing methods.

\vspace{.5em}

\noindent
\textbf{Connecting \eqref{eq:asymp:inv} to Gaussian universality.} To establish the validity of DAB based on approximate invariance, we require distributional approximation of the form \eqref{eq:asymp:inv}, which may appear difficult to verify at first sight. This article establishes \eqref{eq:asymp:inv} for general classes of estimators and transformations by drawing on the growing body of work on Gaussian universality \citep{mossel2005noise,chatterjee2006generalization,korada2011applications,montanari2022universality,hu2022universality}. Gaussian universality is the observation that, for many functions $f: \R^{nd} \rightarrow \R$ and a suitably behaved $\R^{nd}$-valued random vector $X$,
\vspace{-.5em}
\begin{align*}
    f(X) \;\overset{d}{\approx}\; f(Z) \;,\qquad \text{ where } Z \sim \cN(\mean[X], \Var[X])\;, \tagaligneq \label{eq:universality}
    \\[-2.5em]
\end{align*}
and the distributional approximation becomes exact as $n,d \rightarrow \infty$.  When \eqref{eq:universality} holds, the asymptotic distribution of $f(X)$ is completely characterised by the mean and variance of $X$. For i.i.d.~$X_i$'s, Gaussian universality has been observed for a large class of estimators found in high-dimensional statistics and machine learning (e.g.~\cite{montanari2022universality,hu2022universality,gerace2024gaussian} and other works discussed in \Cref{appendix:related:work}, and in \cite{huang2022data,mallory2025universality} for dependent data under data augmentations). The most relevant work to our setup is \cite{austern2020bootstrap}, who rigorously proves Gaussian universality for the bootstrap method and establishes \emph{conditional} approximations of the form \eqref{eq:asymp:inv} for a general $f$. Our key theoretical tool extends their result to a general class of estimators under the dependence structure induced by DAB; the full universality result is found in  \Cref{appendix:universality} and the implied coverage guarantee is found in \Cref{thm:universality} in the main text.

\vspace{.5em}

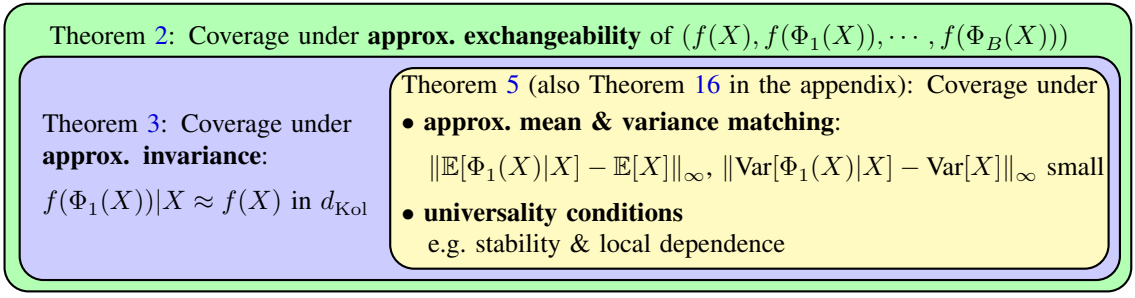
\begin{figure}
    \centering
    \begin{tikzpicture}[
    every node/.style={align=left},
    Blabel/.style={font=\footnotesize, text width=20cm},
    Clabel/.style={font=\footnotesize, text width=10cm},
    Dlabel/.style={font=\footnotesize, text width=10cm},
    ]
    
    \def\Bcorner{9pt}
    \def\Ccorner{9pt}
    \def\Dcorner{9pt}

    \coordinate (B) at (4.8,2);
    \coordinate (C) at (-0.3,0.32);
    \coordinate (D) at (4.45,0.3);

    \coordinate (Bsw) at (-5.8,-1.35);
    \coordinate (Bne) at (9.1,2.45);
    \coordinate (Csw) at (-5.6,-1.2);
    \coordinate (Cne) at (8.95,1.75);
    \coordinate (Dsw) at (-0.7,-1.05);
    \coordinate (Dne) at (8.8,1.6);


    \fill[green!30, rounded corners=\Bcorner] (Bsw) rectangle (Bne);
    \draw[thick, rounded corners=\Bcorner]  (Bsw) rectangle (Bne);

    \fill[blue!20, rounded corners=\Ccorner] (Csw) rectangle (Cne);
    \draw[thick, rounded corners=\Ccorner]  (Csw) rectangle (Cne);

    \fill[yellow!30, rounded corners=\Dcorner] (Dsw) rectangle (Dne);
    \draw[thick, rounded corners=\Dcorner]  (Dsw) rectangle (Dne);


    \node[Blabel] at (B) {
        \Cref{thm:ex:validity}: Coverage under \textbf{approx.~exchangeability} of $(f(X), f(\Phi_1(X)), \cdots, f(\Phi_B(X)))$
    };
    \node[Clabel] at (C) {
        \Cref{thm:Del:Kol:approx:inv}: Coverage under 
        \\
        \textbf{approx. invariance}:
        \\[.5em]
        $f(\Phi_1(X)) | X  \approx f(X)$ in $d_{\rm Kol}$
    };
    \node[Dlabel] at (D) {
        \Cref{thm:universality} (also \Cref{thm:universality:finite} in the appendix): Coverage under 
        \\[.2em]
        $\bullet$ \textbf{approx.~mean \& variance matching}: 
        \\[.5em]
        \quad $\| \mean[\Phi_1(X) | X] - \mean[X]\|_\infty$, $\|\Var[\Phi_1(X) | X] - \Var[X]\|_\infty$ small  
        \\[.5em]
        $\bullet$ \textbf{universality conditions }
        \\
        \quad e.g.~stability \& local dependence
    };
    \end{tikzpicture}
    \caption{Illustration of our theoretical results.}
    \label{fig:coverage}
    \vspace{-.5em}
\end{figure}

The article is organised as follows: \Cref{sec:method} defines DAB. \Cref{sec:validity} establishes theoretical coverage results under exact and approximate invariance, with guarantees interpolating between finite-sample and asymptotic validity; \Cref{fig:coverage} provides an overview. No group assumption is required, and in the case of approximate invariance, we also show how Gaussian universality simplifies the approximation. \Cref{sec:new:DAB:methods} empirically studies the performance of DAB by comparing it against vanilla bootstrap, wild bootstrap for kernel-based U-statistics and conformal prediction, and across simulated and real-life data. \Cref{sec:discussion} concludes by discussing the efficiency of DAB, connections to invariance tests, open problems and comparison to a concurrent work \citep{paul2026probability}. The supplement 
includes additional examples (\Cref{appendix:additional:examples}), further literature comparisons (\Cref{appendix:related:work}), proofs (\Cref{appendix:validity:exact:inv,appendix:thm:Del:Kol:approx:inv,appendix:universality,appendix:proof:compose:exact:with:approx}) and experiment details (\Cref{appendix:experiments}).

\vspace{-.6em}

\subsection{Notation} \label{sec:notation}
\vspace{-.5em}
\textbf{Data.}  $X = (X_{\rm obs}, X_{\rm unobs})$ is a random dataset formed by concatenating a sequence of random observations $X_{\rm obs}$ and a sequence of unknown random variables $X_{\rm unobs}$, which are possibly dependent. $\cX_n$ denotes a generic measurable space that $X$ takes value in, and $\cX_n^{\rm unobs}$ denotes that of $X_{\rm unobs}$. For the universality result only, we require $X$ to live in the Euclidean space $\R^{nd}$, and will assume $X=(X_i)_{i \leq n}$ is a set of $n$ random vectors in $\R^d$ that are possibly dependent. Throughout, $\X=\sigma(X)$ denotes the $\sigma$-algebra generated by $X$, and similarly $\Xobs =\sigma(X_{\rm obs})$. We also write $\A_{\rm trivial}$ as the trivial $\sigma$-algebra.




\noindent
\textbf{Transformations.} Let $\T$ be the set of all $\cX_n \rightarrow \cX_n$ functions, i.e.~its elements operate on the entire dataset. For $1 \leq b \leq B$, $\Phi_b, \Psi_b$ are taken to be random elements of $\T$. $\T$ is not assumed to be a group in general, and its elements can be non-invertible. Deterministic transformations, such as the cycling operation used in conformal prediction, are represented by taking $\Phi_b, \Psi_b$'s as some almost surely deterministic functions.

\noindent
\textbf{Statistic.} We consider a general statistic $f: \cX_n \rightarrow \cY$, where $\cY$ is some space equipped with a total order. $\cY$ is used whenever no assumption other than the total order is made, and we make it clear that $\cY=\R$ when the restriction is necessary.

\noindent
\textbf{Index sets.} We write $[B]=\{1,\ldots,B\}$ and $[B]_0 = \{0\} \cup [B]$.

\noindent
\textbf{Ranking variables.} In conformal prediction, the prediction interval is formed by ranking the residuals. Similarly in DAB, we need a suitable notion of ranking. We allow a general ranking function $\Rank_T: \cY^{B+1} \rightarrow [0,B+1]$, defined based on the total order of $\cY$ as
\begin{align*}
    \Rank_T\big( v ; (v_b)_{b \in [B]} \big) 
    \coloneqq
     \msum_{b=1}^B 
    \ind_{\{ v > v_b \}} 
    + 
    T \big(\msum_{b=1}^B  \ind_{\{ v = v_b\}} + 1\big)
    \;,
    \\[-2.5em]
    \tagaligneq \label{eq:rank}
\end{align*}
where $T$ is a random $\N \rightarrow \R^+ \cup \{0\}$ tie-breaking function. Common choices are:
\begin{itemize}[
    leftmargin=1em,
    labelwidth=0em,
    topsep=0em,
    partopsep=0em,
    itemsep=-0.25em
    ]
    \item \textbf{Uniform tie-breaking. } $T(n) \sim \text{Uniform}(\{ 0, 1, \ldots, n \})$;
    \item \textbf{Smoothed uniform tie-breaking. } $T(n) \sim \text{Uniform}([0, n])$;
\end{itemize}
In conformal prediction, uniform tie-breaking can incur an $O(\frac{1}{B})$ coverage error, whereas smoothed uniform tie-breaking gives exact coverage (Chapter 2, \cite{vovk2005algorithmic}).
\vspace{-1.25em}

\section{Data augmented bootstrap} \label{sec:method}

\vspace{-.5em}

Given a dataset of observed and unobserved random variables $X = (X_{\rm obs}, X_{\rm unobs})$, a central task of statistics is to provide a valid and efficient coverage for properties of $X_{\rm unobs}$. More formally, for $\alpha \in [0,1]$, we seek a random confidence region $\cC_\alpha$ such that
\vspace{-.5em}
\begin{align*}
    \P(  X_{\rm unobs} \in \cC_\alpha(X_{\rm obs}) ) \;\geq\; 1 - \alpha
    \;.
    \\[-2.5em]
    \tagaligneq \label{eq:C:valid}
\end{align*}
Data augmented bootstrap (DAB) is a family of methods of constructing $\cC_\alpha$ by approximating a statistic $f(X)$ by suitably transformed versions of the statistic. Consider $f: \cX_n \rightarrow \cY$ and a set of $B$ possibly random transformations $\Phi \coloneqq (\Phi_b)_{b \leq B}$. For $\alpha \in [0,1]$, let $S_\alpha$ be a size-$(1-\alpha)$ subset of $[0,1]$. The DAB prediction set is given by
\begin{align}
    \CUalpha
    \coloneqq
    \Big\{ 
        \tilde \bx_n   \in \cX_n^{\rm unobs}  
        \;\Big|\; 
        \mfrac{\Rank_T\big( \, f( X_{\rm obs}, \tilde \bx_n) \,;\, (f(\Phi_b(X_{\rm obs}, \tilde  \bx_n)))_{b=1}^B \, \big)}{B+1}
        \,\in\, S_\alpha
    \Big\}
    \;.
    \label{eq:DABU}
\end{align}
Informally, DAB operates under the assumption that 
\vspace{-.5em}
\begin{align*}
    \big( \, f(X)
    \,,\,
    f(\Phi_1(X))
    \,,\,
    \ldots
    \,,\, 
    f(\Phi_B(X)) 
    \,\big)
    \quad \text{ is approximately exchangeable\,,}
\end{align*} 
\\[-2.5em]
since when exchangeability holds, an appropriately chosen and normalised rank function of the $B+1$ variables above will be distributed as $U \sim {\rm Uniform}[0,1]$, in which case 
\vspace{-.5em}
\begin{align*}
    \P( X_{\rm unobs} \in \CUalpha(X_{\rm obs}) )
    \;\approx\;
    \P( U \in S_\alpha )
    \;=\;
    1 - \alpha\;.
\end{align*} 
\\[-2.5em]
Here, ${\rm Uniform}[0,1]$ plays the role of a pivotal distribution in the sense of \cite{hall1992bootstrap}: The rank variable is an appropriate transformation of our statistics such that its distribution becomes independent of data and allows for the construction of a valid CI. 

\vspace{-.5em}

\begin{example}[Bootstrap for 1d mean estimation] \label{example:bootstrap} Let $\tilde X = (\tilde X_i)_{i \leq n} $ be i.i.d.~1d variables with unknown $\mean[\tilde X_1]$ to be estimated (viewed as part of $X_{\rm unobs}$) and unit variance. To see how bootstrap is a form of DAB, we identify the data $X=(\tilde X_1, \ldots, \tilde X_n, \mean[\tilde X_1]) \in \R^{n+1} \equiv \cX_n$, the statistic $f(x_1, \ldots, x_n, x_*) = \frac{1}{\sqrt{n}} \sum_{i \leq n} (x_i - x_*)$, and the transformations
\vspace{-2.5em}
\begin{align*}
    \Phi_b(x_1, \ldots, x_n, x_*) 
    \;=\; 
    ( 
        x_{\pi_{b1}} \,,\, \ldots\,,\, x_{\pi_{bn}} \,,\, \bar x
    )
    \;,
    \quad 
    \bar x = \mfrac{1}{n} \msum_{i \leq n} x_i\;.
\end{align*}
\\[-2.25em]
where $(\pi_{bi})_{b \leq B, i \leq n}$ are i.i.d.~uniform draws from $\{1, \ldots, n\}$. Choosing $S_\alpha = (\frac{\alpha}{2}, 1 -\frac{\alpha}{2})$, the DAB confidence interval $\CUalpha(\tilde X)$ for $\mean[\tilde X_1]$ is given by 
\vspace{-.1em}
\begin{align*}
    \Big[ \mfrac{1}{n} \msum_{i \leq n} \tilde X_i  -  \mfrac{1}{\sqrt{n}} \hat q^{(B)}_{1-\alpha/2} \;,\; \mfrac{1}{n} \msum_{i \leq n} \tilde X_i  - \mfrac{1}{\sqrt{n}} \hat q^{(B)}_{\alpha/2}  \Big]
\end{align*}
\\[-2em]
where $\hat q_{\omega}^{(B)}$ denotes the $\omega$-th quantile of the empirical distribution 
\vspace{-.25em}
\begin{align*}
    \mfrac{1}{B+1} \Big( \msum_{b \leq B} \delta_{f(\Phi_b(X))} + \delta_{f(X)} \Big)
    \;=\;    
    \mfrac{1}{B+1} \Big( 
        \msum_{b \leq B} \delta_{\frac{1}{\sqrt{n}} \sum_{i \leq n} (X_{\pi_{bi}} - \frac{1}{n}\sum_{j \leq n} X_j )}
        +
        \delta_{f(X)} 
    \Big)
        \;.
\end{align*}
\\[-2.25em]
i.e.~the bootstrap distribution plus an extra point mass at $f(X)$. For $B$ large, $\CUalpha(\tilde X)$ gives the nonparametric bootstrap CI for $\mean[\tilde X_1]$. Similarly, setting $\hat s_x^2 = \frac{1}{n} \sum_{i \leq n} (x_i - \bar x)^2$ and $(\eta_{bi})_{b \leq B, i \leq n}$ as i.i.d.~standard normals, the parametric bootstrap CI can be obtained from
\vspace{-.5em}
\begin{align*}
    \Phi_b(x_1, \ldots, x_n, x_*)
    \;=\; 
    (  n^{-1/2} \hat s_x  \, \eta_{b1} \,,\, \ldots \,,\, n^{-1/2} \hat s_x  \, \eta_{bn} \,,\, 0
    )\;.
\end{align*}
\end{example}

\vspace{-1em}

\begin{remark}[Alternative coverage statements] Our theoretical results present coverage statements of the form \eqref{eq:C:valid} for simplicity. In practice, instead of a confidence region for the unobserved data $X_{\rm unobs}$, it may be desirable to ask for the following:
\\
\noindent
(i) \textbf{Confidence region for some downstream property $\rho(X)$ of $X$.}  The coverage guarantee  \eqref{eq:C:valid} for $\cC_\alpha$ is still relevant, as it implies a coverage guarantee for the induced confidence region $\rho(\cC_\alpha, X_{\rm obs} )\coloneqq \{ \rho(X_{\rm obs}, \tilde\bx_n) \,|\, \tilde \bx_n \in \cC_\alpha(X_{\rm obs}) \}$:
\vspace{-.25em}
\begin{align*}
    \P\big(  \rho(X) \in \rho(\cC_\alpha, X_{\rm obs}) \big) 
    \;\geq\;
    \P( X_{\rm unobs} \in \cC_\alpha(X_{\rm obs}) )
    \;\overset{\eqref{eq:C:valid}}{\geq}\; 
    1-\alpha
    \;.
    \\[-2.5em]
    \tagaligneq \label{eq:C:valid:for:a:property}
\end{align*} 
\noindent 
(ii) \textbf{Conditional coverage.} It is often desirable to obtain guarantees that hold conditionally on the observed information, represented via the $\sigma$-algebra $\Xobs = \sigma(X_{\rm obs})$ as
\vspace{-2.25em}
\begin{align*}
    \P(  X_{\rm unobs}  \in \cC_\alpha(X_{\rm obs}) \,|\, \Xobs ) \geq 1 - \alpha 
    \qquad 
    \text{ almost surely\,.}
    \\[-2.5em]
\end{align*}
Our theoretical results can be applied to obtain conditional coverage, provided that their corresponding assumptions are replaced by the conditional versions given $\Xobs$; see \Cref{remark:thm:ex:cond,remark:thm:Del:Kol:approx:inv:cond}, the discussion and examples in \Cref{sec:validity:ex:to:inv} and the proof of all results in the appendix. Note that for conformal prediction, impossibility results on conditional coverage are well-established \citep{vovk2012conditional,lei2014distribution,foygel2021limits}, and DAB does not alleviate this issue since our assumptions do not generally hold conditionally for conformal prediction.
\end{remark}
\vspace{-2em}
\section{Coverage results} \label{sec:validity}

\vspace{-.8em}

This section presents the theoretical results for DAB. We first obtain coverage guarantees under an approximate exchangeability condition on  $f(X)$, $f(\Phi_1(X))$, $\ldots$, $f(\Phi_B(X))$ (\Cref{sec:validity:exact:inv}). When $(\Phi_b)_{b \leq B}$ are i.i.d., we show that this approximate exchangeability can be quantified by a suitable conditional notion of approximate invariance on $f(X)$ under $\Phi_b$ (\Cref{sec:validity:ex:to:inv}). By exploiting a Gaussian universality assumption, this approximate invariance can be further reduced to differences in only the first two moments (\Cref{sec:validity:universality}). Each of the results generalise known guarantees for existing DAB methods, as we shall illustrate in various examples. Some key differences of our results are: (i) Group assumption is not required; (ii) Additional approximate invariances can be incorporated into known methods; (iii) They naturally interpolate between finite-sample guarantees under exact group invariance and asymptotic guarantees under approximate invariance.

\vspace{-.7em}
\subsection{Validity under approximate exchangeability} \label{sec:validity:exact:inv}
\vspace{-.25em}

DAB involves ranking the sequence $(f(X), f(\Phi_1(X)), \ldots, f(\Phi_B(X)))$. To define approximate exchangeability, we compare this sequence to an exactly exchangeable one:

\vspace{-.25em}
\begin{definition}[Exactly exchangeable surrogate] \label{defn:exchangeability} Let $\Psi_0, \ldots, \Psi_B$ be $B+1$ random elements of $\T$ such that $(f(\Psi_0(X)), f(\Psi_1(X)), \ldots, f(\Psi_B(X)))$ is exchangeable.
\end{definition}
\vspace{-.25em}

Consider the rank variable $R_{\Phi}(X) \coloneqq \Rank_T\big(f(X) \,;\, (f(\Phi_b(X)))_{b \in [B]} \big)$  used in DAB and the analogue $R_{\Psi}(X) \coloneqq \Rank_T\big(f(\Psi_0(X)) \,;\, (f(\Psi_b(X)))_{b \in [B]} \big)$ for the exchangeable surrogate. We measure approximate exchangeability by the c.d.f.~difference 
\vspace{-.5em}
\begin{align*}
    \Delta^r_{\rm Kol}
    \;\coloneqq\; 
    \big| \,
        \P\big( R_{\Phi}(X)/ (B+1) < r  \big)   
        \,-\,&\,
        \P\big( R_{\Psi}(X) / (B+1) < r \big)   
    \, \big|
    \qquad 
    \text{ for } r \in [0,1]\;.
\end{align*}
\\[-2.25em]
Observe that with the same  notation, the coverage probability for the DAB CI, $\CUalpha$, reads
\vspace{-.5em}
\begin{align*}
    \P\big( X_{\rm unobs} \in \CUalpha(X_{\rm obs}) \big)    
    \;=\;
    \P\Big( \mfrac{R_\Phi(X)}{B+1} \,\in\, S_\alpha  \Big) 
    \;.
\end{align*}
\\[-2.5em]
We now control this probability under approximate exchangeability and for two common 

\noindent
tie-breaking choices. See \Cref{appendix:coverage:validity} for the result for general tie-breaking functions.

\vspace{-.5em}
\begin{theorem} \label{thm:ex:validity}  Let $r \in [0,1]$. For uniform tie-breaking, 
\vspace{-.25em}
\begin{align*}
    \Big| \P\Big(\,
    \mfrac{R_{\Phi}(X)}{B+1}
    \,<\,
    r
     \Big)
    - 
    r
    \Big| 
    \;\leq\; \Delta^r_{\rm Kol} + \mfrac{1}{B+1}\;,
\end{align*}
\\[-2.25em]
and for smoothed uniform tie-breaking, 
\vspace{-.25em}
\begin{align*}
    \Big|
    \P\Big(\,
    \mfrac{R_{\Phi}(X)}{B+1}
    \,<\,
    r
     \Big)
    - 
    r
    \Big|
    \;\leq\; \Delta^r_{\rm Kol} \;.
\end{align*}
\end{theorem}
\vspace{-.5em}

\begin{remark}[Conditional coverage] \label{remark:thm:ex:cond} $\P$ can be replaced by $\P(\argdot | \Xobs)$ if the exchangeability condition in \Cref{defn:exchangeability} holds conditionally on $\Xobs$.
\end{remark}

\vspace{-1em}

\Cref{thm:ex:validity} recovers known coverage guarantees under group invariance. To see this, let $\G$ be a compact group with $X \overset{d}{=} g(X) $ for all $g \in \G$. If either (i) $\Phi_0, \Phi_1, \ldots, \Phi_B \overset{\rm i.i.d.}{\sim} \textrm{Uniform}(\G)$ or (ii) $\G$ is finite with size $B+1$, $\Phi_0 = {\rm id}$ and $\Phi_1, \ldots, \Phi_B$ is a deterministic enumeration of the non-identity elements of $\G$, then exactly exchangeability holds for
\vspace{-.5em}
\begin{align*}
    \big( X, \Phi_1(X),\ldots, \Phi_B(X)  \big) 
    \overset{d}{=}&\;
    \big( \Phi_0(X), \Phi_1\Phi_0(X),\ldots, \Phi_B\Phi_0(X) \big) 
    \;.
    \tagaligneq \label{eq:group:inv}
\end{align*}
\\[-2.25em]
In this case, $\Delta^r_{\rm Kol}=0$ for all $r \in [0,1]$, and \Cref{thm:ex:validity} implies $(1-\alpha)$ coverage (up to an  $O(\frac{1}{B})$ error) for any size-$(1-\alpha)$ interval $S_\alpha \subseteq [0,1]$. The same group invariance and corresponding guarantees can be found in conformal prediction with $\G$ as the cyclic group, permutation tests with the permutation group, SymmPI with a general compact group \citep{dobriban2023symmpi}, and the wild bootstrap procedures discussed next.

\vspace{-.6em}

\begin{example}[Exact coverage for wild bootstrap for U-statistics] \label{example:marginal:wild:bootstrap} Let $X = (X_i)_{i \leq n}$ be i.i.d.~mean-zero random vectors in $\R^d$ and $f(X) = \frac{1}{n(n-1)} \sum_{i \neq j} X_i^\top X_j$. The wild bootstrap procedure used for simulating degenerate U-statistics \citep{dehling1994random,fromont2013two,leucht2013dependent,chwialkowski2014wild} is a special case of DAB by drawing i.i.d.~Rademacher variables $(\epsilon_{bi})_{b \leq B, i \leq n}$ and identifying 
\vspace{-.25em}
\begin{align*}
    \Phi_b(x_1, \ldots, x_n) \;\coloneqq\; 
    (\epsilon_{b1} x_1 \,,\, \ldots \,,\, \epsilon_{bn} x_n ) 
    \;.
    \tagaligneq \label{eq:wild:bootstrap:rademacher}
\end{align*}
\\[-2.25em]
Assume sign-flipping symmetry $X_1 \overset{d}{=} -X_1$; this arises, for example, in the Maximum Mean Discrepancy (MMD) statistic in two-sample tests (see \Cref{sec:wild:bootstrap:add}). Then
\vspace{-.25em}
\begin{align*}
    f(\Phi_b(X)) 
    \;=\; \mfrac{1}{n(n-1)} \msum_{i \neq j} \epsilon_{bi} \epsilon_{bj} X_i^\top X_j 
    \;\overset{d}{=}\; 
    f(X)
    \;.
\end{align*}
\\[-2.25em]
Exact exchangeability holds by identifying the group $\G = \{-1, +1\}^n$ in the argument in \eqref{eq:group:inv}, and \Cref{thm:ex:validity} holds with $\Delta^r_{\rm Kol} = 0$. Note that this argument also generalises beyond a linear dot product kernel. This gives the same finite-sample coverage for wild bootstrap for MMD by \cite{schrab2023mmd}, who first exploited the sign-flipping invariance. 
\end{example}

\vspace{-.6em}

When $\Delta^r_{\rm Kol} > 0$, \Cref{thm:ex:validity} provides coverage beyond the exact group invariance discussed above. For conformal prediction, prior works \citep{barber2023conformal,xu2025wasserstein} have measured the errors of exchangeability violation in the total variation or Wasserstein distance. We will see that our c.d.f.~difference allows us to connect approximate invariance to universality results typically established in Kolmogorov distance.
    
\vspace{-.5em}

\begin{remark} \label{remark:weighted:ex}
A well-studied violation of \Cref{defn:exchangeability} is weighted exchangeability. For conformal predictions and permutation tests, it has been shown \citep{tibshirani2019conformal,barber2023conformal,ramdas2023permutation} that coverage can still be obtained by using suitable reweightings; see \Cref{appendix:related:work} for a further literature review and discussion. We expect that with a similar reweighted rank function as \cite{ramdas2023permutation}, one may also obtain coverage for DAB by approximating $(f(X), f(\Phi_1(X)), \ldots, f(\Phi_B(X)))$ with a weighted exchangeable sequence. 
\end{remark}
\vspace{-1.2em}

\subsection{Validity under approximate invariance} \label{sec:validity:ex:to:inv}

\vspace{-.5em}

The approximate exchangeability error $\Delta^r_{\rm Kol}$ involves identifying an exchangeable surrogate sequence of statistics, which is compared against the original sequence. We now show that this error can be reduced to a suitable notion of approximate invariance of the original statistic $f(X)$ under the transformation $\Phi_1$. We use a simplifying assumption:

\vspace{-.5em}
\begin{assumption} \label{asst:iid:smooth}  (i) $\Phi_1, \ldots, \Phi_B$ are i.i.d.~random elements of $\T$ and independent of $X$; 
(ii) smoothed uniform tie-breaking is used; 
(iii) the output space of $f$, $\cY$, is taken to be $\R$.
\end{assumption}
\vspace{-.5em}

The idea is to apply \Cref{thm:ex:validity} with some appropriate $\Psi_b$'s, and control the error via
\vspace{-.5em}
\begin{align*}
    \Delta_{\rm inv}(X)
    \;\coloneqq\;
    \msup_{t \in \R}
    \big|  
        \P( f(\Phi_1(X)) \leq t \,|\, X ) 
        \,-\, 
        \P( f(X) \leq t  ) 
    \big|
    \;.
    \\[-2.5em]
\end{align*} 
$\Delta_{\rm inv}(X)$ compares the conditional distribution of $f(\Phi_1(X)) | X$ to that of $f(X)$ through a Kolmogorov distance. The coverage statement is stated in terms of a control of $\Delta_{\rm inv}(X)$ in the $L_\nu$ norm, $\| \argdot \|_{L_\nu} \coloneqq (\mean[ | \argdot |^\nu ])^{1/\nu}$, defined for $\nu \geq 1$.


\begin{theorem} \label{thm:Del:Kol:approx:inv} Let $r \in [0,1]$, $\nu \geq 1$ and $B \geq 2$. Under Assumption \ref{asst:iid:smooth}, almost surely
\begin{align*}
     \Big| \P\Big(  \mfrac{R_\Phi(X)}{B+1}  < r  \Big)  
        \,-\,
        r
    \Big|
    \;\leq&\;
    \mfrac{12}{4^{1/(\nu+1)}}
    \,
    \big( 3 \, \| \Delta_{\rm inv}(X) \|_{L_\nu} \big)^{\frac{\nu}{\nu+1}} 
    + 
    \mfrac{1}{B-1} 
    +  
    \mfrac{\sqrt{3 \log B}}{\sqrt{B-1}} 
    \;.
\end{align*}
\end{theorem}

\begin{remark}[Conditional coverage] \label{remark:thm:Del:Kol:approx:inv:cond} To control $\P(  \frac{R_\Phi(X)}{B+1}  < r \,|\, \Xobs )$, it suffices to replace $\P(f(X) \leq t)$ in the definition of $\Delta_{\rm inv}(X)$ by $\P(f(X) \leq t \,|\, \Xobs)$ and the $L_\nu$ norm by its conditional analogue $(\mean[ | \argdot |^\nu \,|\, \Xobs ])^{1/\nu}$; see \Cref{example:conditional:bootstrap} in the appendix and \Cref{appendix:thm:Del:Kol:approx:inv}.
\end{remark}

\vspace{-.5em}

\begin{remark} \label{remark:thm:Del:Kol:approx:inv}
The first term in \Cref{thm:Del:Kol:approx:inv} adapts to the type of control we have on $\Delta_{\rm inv}(X)$: If we only have $L_1$-control, the first term reads $6\sqrt{3}\| \Delta_{\rm inv}(X) \|^{1/2}_{L_1}$; if $\Delta_{\rm inv}(X) \leq \bar \Delta$ almost surely for some deterministic $\bar \Delta$, taking $\nu \rightarrow \infty$ gives $36 \bar \Delta$; if $X=(X_i)_{i \leq n}$ is a set of $n$ observations and $\Delta_{\rm inv}(X) \xrightarrow{\P} 0$ as $n \rightarrow \infty$, then $\| \Delta_{\rm inv}(X) \|_{L_\nu} \rightarrow 0$ since $\Delta_{\rm inv}(X)$ is almost surely bounded, and we obtain asymptotic coverage. 
\end{remark}

$\Delta_{\rm inv}(X)$ differs from the typical notion of approximate invariance, as it compares the \emph{conditional} distribution $f(\Phi_1(X)) | X$ rather than the marginal distribution $f(\Phi_1(X))$ to the distribution of $f(X)$. The next example shows how this approximate invariance arises in the classical guarantees for bootstrap, and illustrates how $\Delta_{\rm inv}(X)$ can be controlled.

\vspace{-.25em}

\begin{example}[Asymptotic coverage for bootstrap] \label{example:asymptotic:bootstrap} Consider the nonparametric bootstrap setting in \Cref{example:bootstrap}. A classical guarantee of bootstrap \citep{efron1992bootstrap,hall1992bootstrap} can be stated as follows: Let $(\tilde X'_i)_{i \leq n}$ be i.i.d.~copies of $\tilde X_1$. Under mild conditions, as $n \rightarrow \infty$, almost surely
\begin{align*}
    \sup_{t \in \R} \,
    \Big| \P\Big( \mfrac{\sum_{i=1}^n (\tilde X_{\pi_{1i}} - \frac{1}{n} \sum_{j \leq n} \tilde X_j) }{\sqrt{n}} \leq t \,\Big|\, X \Big) - \P\Big( \mfrac{\sum_{i=1}^n (\tilde X_i' - \mean[\tilde X_1]) }{\sqrt{n}} \leq t \Big) \Big| 
    \rightarrow 0
    \;.
\end{align*}
With the notation of \Cref{example:bootstrap}, this is exactly the statement that, as $n \rightarrow \infty$,
\vspace{-.25em}
\begin{align*}
    \Delta_{\rm inv}(X)
    =
    \msup_{t \in \R} 
    \big| \P( f(\Phi_1(X)) \leq t \,|\, X ) - \P( f(X) \leq t ) \big| 
    \rightarrow 0 \; \text{ almost surely},
    \tagaligneq \label{eq:bootstrap:Delta}
    \\[-2.5em]
\end{align*}
which implies asymptotic coverage of bootstrap by \Cref{thm:Del:Kol:approx:inv}. This is a case of \emph{asymptotic invariance}: At finite $n$, even the marginal distribution of $f(\Phi_1(X))$ disagrees with $f(X)$ (consider the ``bad" event when all $(\pi_{bi})_{i \leq n}$ select the same data point), but the large-$n$ asymptotic drives the probability of the ``bad'' events to zero and thus $f(\Phi_1(X)) | X$ and $f(X)$ to the same normal distribution. Meanwhile, the conditioning on $X$ ensures that under \eqref{eq:bootstrap:Delta}, $f(\Phi_1(X)) | X$ and $f(X)$ are asymptotically independent. This makes $(f(X), f(\Phi_1(X)), \ldots, f(\Phi_B(X)))$ asymptotically i.i.d.~and thus exchangeable. 
\end{example}

\Cref{asst:iid:smooth}(i) excludes transformations $\Phi'= (\Phi'_b)_{b \leq B}$ that are a deterministic enumeration of the elements from a finite set, such as the cycling operations in conformal prediction and Jackknife. The next lemma shows that $\Phi'$ is equivalent to uniform random sampling, making \Cref{thm:Del:Kol:approx:inv} still applicable. Again, we do not assume a group structure.

\vspace{-.25em}

\begin{lemma} \label{lem:fixed:to:random} Let $r \in [0,1]$, $\nu \geq 1$, $\cS_B = (\Phi'_1, \ldots, \Phi'_B)$ be a deterministic subset of $\T$ and $(\Phi_b)_{b \leq B}$ be i.i.d.~uniform draws from $\cS_B$. Under \Cref{asst:iid:smooth}(ii) and (iii),
\vspace{-.25em}
\begin{align*}
    \Big| 
    \P\Big( \mfrac{R_{\Phi'}(X)}{B+1} < r \Big)   
    -
    \P\Big( \mfrac{R_{\Phi}(X)}{B+1} < r  \Big)   
    \Big|
    \;\leq\; 
    \mfrac{4 \sqrt{\log (2B)}}{\sqrt{2B}}
    \,+\,
    5 \, \big(3  \big\| \Delta_{\rm inv}(X) \big\|_{L_\nu}\big)^{\nu/(\nu+1)}
    \;,
\end{align*}
\\[-2.25em]
where we have denoted $R_{\Phi'}(X)  \coloneqq \Rank_T\big(f(X) \,;\, (f(\Phi'_b(X)))_{b \in [B]} \big) $.   
\end{lemma}

\vspace{-1em}

\begin{remark} For $\nu=1$, the error in \Cref{lem:fixed:to:random} is $O(\frac{\sqrt{\log B}}{\sqrt{B}} +  \| \Delta_{\rm inv}(X) \|_{L_1}^{\frac{1}{2}})$, which is worse than the error in \Cref{thm:Del:Kol:approx:inv}. This is because our proof uses the concentration of $\frac{R_\Phi(X)}{B+1}$ around the conditional mean $\mean[ \frac{R_\Phi(X)}{B+1} | X, T ] = \frac{R_{\Phi'}(X)}{B+1}$, which incurs an $B^{-\frac{1}{2}}$ error.
\end{remark}

\vspace{-.5em}

\Cref{appendix:additional:examples} contains additional examples of \Cref{thm:Del:Kol:approx:inv}, including asymptotic coverage for wild bootstrap (\Cref{example:asymptotic:wild:bootstrap}) and Jackknife (\Cref{example:jackknife}), and conditional asymptotic coverage for simple cases of bootstrap (\Cref{example:conditional:bootstrap}) and permutation test (\Cref{example:permutation}). \Cref{appendix:thm:Del:Kol:approx:inv:variants} also includes a generalisation of \Cref{thm:Del:Kol:approx:inv} when, instead of \Cref{asst:iid:smooth}(i), $\Phi_b$'s are generated conditionally i.i.d.~given some generic $\sigma$-algebra $\A$.

\color{black}

\vspace{-1em}
\subsection{Quantifying approximate invariance via Gaussian universality} \label{sec:validity:universality}

\vspace{-.5em}
 
\Cref{example:asymptotic:bootstrap} shows how $\Delta_{\rm inv}(X)$ may be controlled for empirical averages, but many practical examples involve highly non-linear $f$'s. One way to control $\Delta_{\rm inv}(X)$ for a more general $f$ is via Gaussian universality: If we can replace $\Phi_1(X)|X$ and $X$ by appropriately chosen Gaussians, $\Delta_{\rm inv}(X)$ may be measured through only moment differences in the data. We now restrict the input space to $\cX_n = (\R^d)^n$. $n$ is the number of data points and $d \equiv d(n)$ is the data dimension, which is allowed to grow with $n$:
\vspace{-.5em}
\begin{align*}
    d\,/\,n \;\rightarrow&\; \gamma
    \qquad 
    \text{ for some } \gamma \in [0,\infty]\;.
    \tagaligneq \label{eq:asymp:regime}
    \\[-2.5em]
\end{align*}
We will use Gaussian universality to control $\Delta_{\rm inv}(X)$ via the variance difference
\vspace{-.25em}
\begin{align*}
    \Delta_{\rm Var}(X) \;\coloneqq\; \|\Var[\Phi_1(X)|X] - \Var[X] \|_\infty\;,
    \\[-2.5em]
\end{align*}
where $\| \argdot \|_\infty$ is the infinity norm on $\R^{nd \times nd}$. The standard ingredients in Gaussian universality are (i) a matching mean condition; (ii) a dependency condition on the data; (iii) a stability condition on $f$; (iv) a moment boundedness condition. We now formulate these for DAB. Since we can always WLOG absorb the mean and variance of $X$ into $f$ by considering the transformed statistic $\bx \mapsto f(  \Var[X]^{1/2}  \bx + \mean[X])$, we first assume:

\vspace{-.5em}

\begin{assumption}[Standardisation] \label{asst:standard} $\mean[X] = 0$ and $\Var[X] = I_{nd}$.
\end{assumption}

\vspace{-.5em}

In the main text, we suppose that $\Phi_1$ is chosen such that the conditional mean of $\Phi_1(X)$ matches that of $X$. This holds, for example, in the case of bootstrap in \Cref{example:bootstrap}.

\vspace{-.5em}

\begin{assumption}[Exactly matching conditional mean] \label{asst:zero:mean} $\mean[\Phi_1(X) | X ]=\mean[X]$ almost surely.
\end{assumption}

\vspace{-.5em}

Since $X$ may consist of dependent coordinates (despite \Cref{asst:standard}) and the transformation $\Phi_1$ can introduce dependence, we need to employ a universality result under local dependence. For an $\R^{nd}$-valued random vector $V$ distributed according to $\mu$, we denote the local dependency neighbourhood of the $(i,j)$-th coordinate $V_{ij}$ as 
\vspace{-.5em}
\begin{align*}
    \cN_{ij}(\mu)
    \coloneqq
    \inf \big\{ \cN \subseteq [n] \times [d] \;\big|\; (i,j) \in \cN\,,\, (V_{ab})_{(a,b) \in \cN} \text{ is independent of } (V_{ab})_{(a,b) \not\in \cN}  \big\}
    .
\end{align*} 
\\[-2.5em]
We denote the maximum size of the local dependency neighbourhood of $V$ as 
\vspace{-.6em}
\begin{align*}
    N(\mu) \;\coloneqq\; \max_{1 \leq i \leq n, \, 1 \leq j \leq d} \, \big| 
        \cN_{ij}(\mu)
    \big|
    \;.
    \tagaligneq \label{eq:defn:local:dependency:measure}
\end{align*}
\\[-2.3em]
Note that $\cN_{ij}(\mu)$ and $N(\mu)$ are deterministic quantities. Let $\mu_{\Phi_1(X)|X}$ be the conditional probability measure for $\Phi_1(X) | X$, whose existence is assumed throughout this paper (see \cite{faden1985existence}), and $\mu_X$ be the probability measure for $X$. We assume the following:

\vspace{-.5em}

\begin{assumption}[Local dependency] \label{asst:local:dependence} $N_{\rm dep} \coloneqq \max\{ N(\mu_{\Phi_1(X) | X }) \,,\,  N(\mu_{X}) \}$ is $o((nd)^{\frac{1}{4}})$.
\end{assumption}

\vspace{-.5em}

Let $X'$ be an i.i.d.~copy of $X$. We seek to use Gaussian universality to approximate $\Phi_1(X)$, conditioning on $X$, by a Gaussian vector $Z' \sim \cN(0, I_{nd})$, and then approximate $Z'$ by $X'$. The approximations are valid under stability conditions on $f$, assumed here to be thrice-differentiable. For $\tau=1,2,3$ and $\nu \geq 1$, we measure the stability of $f$ along a set of suitable interpolation paths between the variables $\Phi_1(X)$, $Z'$ and $X'$ as
\vspace{-.25em}
\begin{align*}
    \varphi_{\tau;\nu} 
    \;\coloneqq\;
    (nd)^{\frac{\tau}{2}} 
    \,
    \sup_{i \leq n, \, j \leq d\,,\, \gamma \in \cC}
    \,
    \Big\|
    \,
    \big\| 
            \partial_{(i,j)}^\tau \, f\big( \gamma( \Phi_1(X), X',  Z' ) \big) 
    \big\|_{L_4 | X}
    \,
    \Big\|_{L_\nu}
    \;.
    \tagaligneq \label{eq:defn:varphi}
\end{align*}
$\partial^\tau_{(i,j)}$ is the $\tau$-th partial derivative with respect to the $j$-th coordinate of the $i$-th observation, and $\| \argdot \|_{L_4 | X} \coloneqq (\mean[ | \argdot |^4 | X] )^{\frac{1}{4}}$. $(nd)^{\frac{\tau}{2}}$ balances typical scaling factors obtained from differentiating the statistic $f$ (consider, for instance, the sample average $\frac{1}{\sqrt{nd}} \sum_{i \leq n, j\leq d} x_{ij}$). $\cC$ is a set of $(\R^{nd})^3 \rightarrow \R^{nd}$ functions that represent the interpolation paths, defined as
\vspace{-.5em}
\begin{align*}
    \cC
    \;\coloneqq\;
    \big\{ 
        \gamma_{s,\theta;i,j} 
        \,,\,
        \gamma'_{s,\theta;i,j} 
    \big\}_{s, \, \theta \in [0,1]\,;\, i \leq n, \, j \leq d}
    \;,
\end{align*}
\\[-2.5em]
where, for $\bx, \bx', \bz' \in \R^{nd}$, the $(a,b)$-th coordinate of each interpolation path is given as
\begin{align*}
    \big( \gamma_{s,\theta;i,j}(\bx,\bx',\bz') \big)_{ab}
    \;\coloneqq&\;
    \begin{cases}
        \theta \big( \sqrt{s} \, x_{ab} \, + \, \sqrt{1-s} \, z'_{ab} \big)\;,
        &\text{ if } \;
        (a,b) \in \cN_{ij}(\mu_{\Phi_1(X)|X})\;,
        \\
        \sqrt{s} \, x_{ab} \, + \, \sqrt{1-s} \, z'_{ab}\;,
        &\text{ if } \;
        (a,b) \not\in \cN_{ij}(\mu_{\Phi_1(X)|X})\;,
    \end{cases}
    \\
    \big( \gamma'_{s,\theta;i,j}(\bx,\bx',\bz') \big)_{ab}
    \;\coloneqq&\;
    \begin{cases}
        \theta \big( \sqrt{s} \, x'_{ab} \, + \, \sqrt{1-s} \, z'_{ab} \big)\;,
        &\text{ if } \;
        (a,b) \in \cN_{ij}(\mu_{X})\;,
        \\
        \sqrt{s} \, x'_{ab} \, + \, \sqrt{1-s} \, z'_{ab}\;,
        &\text{ if } \;
        (a,b) \not\in \cN_{ij}(\mu_{X})\;.
    \end{cases}
\end{align*}
See \Cref{appendix:universality:local:dependence} for more detailed interpretations. The stability terms can depend on $N_{\rm dep}$, which may grow with $n$; the next condition ensures that it does not grow too~fast.

\vspace{-.5em}

\begin{assumption}[Stability] \label{asst:stab} There exists some $\nu \geq 1$ such that,  for $\tau=1,2,3$,  
\vspace{-.25em}
\begin{align*}
   \varphi_{\tau;(7-\tau)\nu}
   \;=\; 
   o\Big(
        \min\Big\{ 
            \Big(\mfrac{nd}{N_{\rm dep}^4}\Big)^{1/(2\tau)}
            \,,\, 
            ( N_{\rm dep} \, \| \Delta_{\rm Var}(X) \|_{L_{3\nu/2}} )^{ - \frac{1}{\tau} }
        \Big\}
    \Big)
   \;.
\end{align*}
\end{assumption}

\vspace{-.3em}

We also require a moment boundedness condition:

\vspace{-.75em}

\begin{assumption}[Moment boundedness] \label{asst:moment:bounded} For $\nu$ given in \Cref{asst:stab}, we have
\vspace{-.5em}
    \begin{align*}
        \mmax_{i \leq n, j \leq d} 
        \,
        \max\big\{ 
            \big\| \, \| (\Phi_1(X))_{ij} \|_{L_4 | X} \,\big\|_{L_{6\nu}}
            \,,\,
            \| X_{ij} \|_{L_4}
        \big\}
        \,
        \;=\;
        O(1)\;.
    \end{align*}
\end{assumption}

\vspace{-.25em}

The next condition is stated for convenience and relaxed in the proof in \Cref{appendix:universality}.

\vspace{-.5em}

\begin{assumption}[Density smoothness] \label{asst:density:smooth} The random variable $f(Z')$ has a Lebesgue density that is uniformly bounded from above by some constant that does not depend on $n$ or $d$.
\end{assumption}

\vspace{-.5em}

Under \Cref{asst:standard,asst:zero:mean,asst:local:dependence,asst:stab,asst:moment:bounded,asst:density:smooth}, we can use Gaussian universality to replace $\Phi_1(X)|X$ and $X$ by the corresponding (conditional) Gaussians, and subsequently control $\Delta_{\rm inv}(X)$ through a variance difference $\Delta_{\rm Var}(X) = \| \Var[\Phi_1(X) | X] - \Var[X] \|_\infty$. We state the asymptotic result here, and refer readers to \Cref{appendix:universality} for a finite-sample bound.

\begin{theorem} \label{thm:universality} Under \Cref{asst:standard,asst:zero:mean,asst:local:dependence,asst:stab,asst:moment:bounded,asst:density:smooth} and the asymptotic \eqref{eq:asymp:regime}, with $\nu \geq 1$ given in \Cref{asst:stab}, if additionally $\| \Delta_{\rm Var}(X) \|_{L_{3\nu/2}} = o( N_{\rm dep}^{-1})$, then 
\vspace{-.5em}
\begin{align*}
    \| \Delta_{\rm inv}(X) \|_{L_\nu} \;\rightarrow\; 0\;. \tagaligneq \label{eq:Del:inv:vanish}
\end{align*}
\\[-2.5em]
If additionally \Cref{asst:iid:smooth} holds, for $B$ allowed to be fixed or growing with $n$, we have
\vspace{-2.2em}
\begin{align*}
    \sup_{r \in [0,1]}
    \Big| \P\Big(  \mfrac{R_\Phi(X)}{B+1}  < r  \Big)  
        \,-\,
        r
    \Big|
    \;\rightarrow\; 0\;.
    \tagaligneq \label{eq:coverage:error:vanish}
\end{align*}
\end{theorem}

 The condition on $\| \Delta_{\rm Var}(X) \|_{L_{3\nu/2}}$ is satisfied for $\| \Delta_{\rm Var}(X) \|_{L_{3\nu/2}}  = o( (nd)^{-\frac{1}{4} })$ under \Cref{asst:local:dependence}. \eqref{eq:Del:inv:vanish} is proved by an interpolation method standard in universality results \citep{montanari2022universality,mallory2025universality}, and \eqref{eq:coverage:error:vanish} follows from \eqref{eq:Del:inv:vanish}. The conditions of \Cref{thm:universality} can be difficult to parse for readers unfamiliar with the universality literature. We make several comments:
\\
(a) \;\; \Cref{asst:local:dependence,asst:stab,asst:moment:bounded} (local dependency, stability and moment boundedness) are standard in the universality literature; see prior works referenced in \Cref{sec:intro} and \Cref{appendix:universality:local:dependence} for a detailed discussion of these conditions. We verify local dependency for different transformations $\Phi_1$ in \Cref{appendix:verify:local}. We also discuss the applicability of stability for examples of $f$ in \Cref{appendix:verify:stability}, such as estimators that can be approximated by a delta method, U-statistics in wild bootstrap and conformal prediction score functions;
\\
(b) \;\; \Cref{asst:standard,asst:density:smooth} appear to be restrictive but can be relaxed, as discussed earlier;
\\
(c) \;\; The strongest condition is the matching mean condition in \Cref{asst:zero:mean}. In \Cref{remark:universality:local:dependence:extension} in \Cref{appendix:universality:local:dependence}, we show that a mean mismatch can be accommodated at the expense of more notation and a first-order stability condition. We also show that, while the $\| \argdot \|_\infty$ norm on $\Delta_{\rm Var}(X)$ can introduce additional dependence on $n$ and $d$, this can be replaced by a quantity that measures the average coordinate of the matrix $\Delta_{\rm Var}(X)$;
\\
(d)\;\; While \Cref{thm:universality} requires conditional moment matching, we discuss in \Cref{appendix:theory:examples} how these can be reduced to marginal moment matching in specific examples.

\Cref{thm:universality} says that, under Gaussian universality, coverage can be established if the transformation $\Phi_1$ ``reproduce the first two moments of the original data". For bootstrap of the 1d empirical average (\Cref{example:bootstrap}),  $\mean[\Phi_1(X)|X] = \mean[X] = 0$, whereas $\Var[\Phi_1(X)|X]$ is a diagonal $\R^{n \times n}$ matrix with equal entries on the diagonal,
\vspace{-.25em}
\begin{align*}
    (\Var[\Phi_1(X)|X])_{11} \;=\; \mfrac{1}{n} \msum_{i \leq n} \Big( X_i - \mfrac{1}{n} \msum_{j \leq n} X_j \Big)^2 
    \;,
\end{align*}
\\[-2.25em]
which is a consistent estimator of $\Var[X_1]$.  In particular, 
\vspace{-.25em}
\begin{align*}
    \|  \Delta_{\rm Var}(X) \|_{L_{3\nu/2}}
    \;=\;
    \big\| \, |  (\Var[\Phi_1(X)|X])_{11}  - \Var[X_1] | \, \big\|_{L_{3\nu/2}}
    \;=\;
    O(n^{-1/2})
    \;\rightarrow\; 0
    \;,
\end{align*}
\\[-2.25em]
so \Cref{thm:universality} holds. Indeed, bootstrap consistency for the empirical average relies on the CLT, which is a special case of Gaussian universality. For bootstrap on a general function $f$, \cite{austern2020bootstrap} use Gaussian universality to obtain finite-sample bounds for a quantity similar to $\Delta_{\rm inv}(X)$. They observe that, without knowledge of the true mean $\mean[X]$, the bootstrapped mean and the data mean are typically mismatched, and the error gets amplified when $f$ is unstable (see their Example 3.4 and Section 3.2), making a first-order stability condition necessary for bootstrap consistency. We make the same observation in \Cref{remark:universality:local:dependence:extension} in \Cref{appendix:universality:local:dependence}. Thus, an alternative way to view \Cref{thm:universality} is that it extends the coverage guarantees for bootstrap to more general DAB confidence intervals, constructed with other approximately invariant transformations. 

\vspace{-1em}

\subsection{Some useful extensions of \Cref{thm:Del:Kol:approx:inv,thm:universality}} \label{sec:extension:universality}

\vspace{-.5em}

\paragraph{Simplified variance-matching criterion.} \Cref{thm:universality} aims to cover a large class of statistics and, at the cost of its generality, requires that $\Var[\Phi_1(X)|X]$ needs to (asymptotically) match $\Var[X]$ uniformly across all observations and coordinates. This can be relaxed for statistics with more structures. For instance, consider a plug-in estimator
\vspace{-.25em}
\begin{align*}
    f(x_1, \ldots, x_n) \;=\; g\Big( \mfrac{1}{n} \msum_{i \leq n} x_i \Big)
    \tagaligneq \label{eq:plug:in}
\end{align*}
\\[-2.25em]
where $g: \R^d \rightarrow \R$ is a smooth function. Provided that $\partial g(\mean[X_1]) \neq 0$ and under regularity conditions, for $d$ fixed, the CLT and the delta method imply that 
\vspace{-.25em}
\begin{align*}
   \mfrac{\sqrt{n}}{ \sqrt{ (\partial g(\mean[X_1]))^\top \Var[X_1] \, \partial g(\mean[X_1])} } \, ( f(X) - g(\mean[X_1])) \;\xrightarrow{d}\; \cN(0,1)
   \qquad 
   \text{ as } n \rightarrow \infty
    \;.
\end{align*}
\\[-2.25em]
For i.i.d.~$\R^d \rightarrow \R^d$ transformations $\phi_{bi}$'s, we consider
\vspace{-.5em}
\begin{align*}
    &\;\hspace{9.5em}\Phi_1(x_1, \ldots, x_n) \;=\; (\phi_{11}(x_1), \ldots, \phi_{1n}(x_n))
    \;,
    \tagaligneq \label{eq:iid:aug}
    \\
    &\;\mu_\phi(X) \;\coloneqq\; \mfrac{1}{n} \msum_{i \leq n} \mean[\phi_{11}(X_i) | X_i] 
    \qquad\text{ and }\qquad
    \Sigma_\phi(X) \;\coloneqq\; \mfrac{1}{n} \msum_{i \leq n} \Var[\phi_{11}(X_i) | X_i] 
    \;.
\end{align*}
\\[-2.5em]
By a similar argument with the central limit theorem for independent but non-identically distributed random variables, and under additional regularity conditions on $\phi_{bi}$'s, the following conditional Kolmogorov distance converges to zero almost surely:
\begin{align*}
    \sup_{t \in \R}
    \Big| 
        \P\Big(
            \mfrac{\sqrt{n} \, ( f(\Phi_1(X)) - g(\mu_\phi(X)) ) }{((\partial g(\mu_\phi(X)))^\top \Sigma_\phi(X) \, \partial g(\mu_\phi(X)) )^{1/2}}  
            \leq t
        \,\Big|\, X
        \Big)
        -
        \P_{Z \sim \cN(0,1)}(Z \leq t)
    \Big|.
\end{align*}
Then, instead of needing to match $(\mean[\Phi_1(X)|X],\Var[\Phi_1(X)|X])$ with $(\mean[X], \Var[X])$ uniformly across coordinates, it suffices to match the empirical averages 
\vspace{-.5em}
\begin{align*}
    (\mu_\phi(X), \Sigma_\phi(X))
    \;\approx\;
    (\mean[X_1], \Var[X_1])
    \;.
    \tagaligneq \label{eq:average:moment:match}
\end{align*}
\\[-2.5em]
This also works for a growing $d$ by the non-classical delta method in \cite{huang2024gaussian}.

\vspace{-1em}

\paragraph{Composing approximately invariant transformations with exactly invariant transformations}

In practice, one may wish to incorporate approximately invariant transformations $(\tilde \Phi_b)_{b \leq B}$ into a DAB method whose original transformations $(\cT_b)_{b \leq B}$ already enjoy exchangeability. Naively applying \Cref{thm:Del:Kol:approx:inv}, the error $\Delta_{\rm inv}(X)$ would be measured with respect to the composite transformation $\cT_1 \circ \tilde \Phi_1$. By adapting the proof of \Cref{thm:Del:Kol:approx:inv}, we show that only the invariance error with respect to $\tilde \Phi_1$ is required:

\vspace{-.5em}

\begin{lemma} \label{lem:compose:exact:with:approx} Suppose $\Phi_1 = \cT_1 \circ \tilde \Phi_1$, where $\cT_1$ and $\tilde \Phi_1$ are independent random elements of $\T$ that are independent of $X$. Assume that $\cT_1(X)$ has the same marginal distribution as $X$. Then, denoting the random function $f_\cT \coloneqq f \circ \cT_1$, we have that almost surely
\vspace{-.25em}
\begin{align*}
    \Delta_{\rm inv}(X)
    \;\leq\;
    \mean\Big[
    \msup_{t \in \R}
    \big|  
        \P( f_\cT ( \tilde \Phi_1(X)) \leq t \,|\, X  , f_\cT ) 
        \,-\, 
        \P( f_\cT(X) \leq t \,|\, f_\cT ) 
    \big|
    \,\Big|\, X
    \Big]
    \;.
\end{align*} 
\\[-2.25em]
If instead $\Phi_1 = \tilde \Phi_1 \circ \cT_1$, for $\nu \geq 1$, almost surely
\vspace{-.25em}
\begin{align*}
    \| \Delta_{\rm inv}(X) \|_{L_\nu}
    \;\leq\;
   \big\|
   \, \msup_{t \in \R}
    \big|  
        \P( f(\tilde \Phi_1(X)) \leq t \,|\, X ) 
        \,-\, 
        \P( f(X) \leq t ) 
    \big|
    \; \big\|_{L_\nu}
    \;.
\end{align*}
\end{lemma}

\vspace{-.8em}

\vspace{-1.5em}

\section{New DAB methods by incorporating approximately invariant transformations} \label{sec:new:DAB:methods}

\vspace{-.5em}

For simplicity, let $X=(X_i)_{i \leq n}$ be i.i.d.~$\R^d$ random vectors and $(\Phi_b)_{b \leq B}$ be i.i.d.~drawn unless otherwise specified. The coverage results in \Cref{sec:validity} enable us to build DAB methods  from approximately invariant transformations. In this section, we empirically investigate observation-wise data augmentations and bootstrap (\Cref{sec:pure:aug:bootstrap}), wild bootstrap (\Cref{sec:wild:bootstrap:add}) and conformal prediction (\Cref{sec:conformal}). In particular, for known DAB methods with transformations  $(\Phi^{\rm old}_b)_{b \leq B}$ --- i.e.~bootstrap, wild bootstrap and conformal predictions --- we introduce new transformations $(\Phi^{\rm new}_b)_{b \leq B}$ by taking the composition 
\vspace{-.5em}
\begin{align*}
    \Phi_b \;=\; \Phi^{\rm old}_b \circ \Phi^{\rm new}_b    
    \;.
\end{align*}
\\[-2.5em]
For each example, we derive the invariance conditions required for coverage in \Cref{appendix:theory:examples}. Broadly speaking, for those in \Cref{sec:pure:aug:bootstrap,sec:wild:bootstrap:add}, Gaussian universality applies and we only need conditional moment matching, which can be relaxed to marginal moment matching for bootstrap. For the ones in \Cref{sec:conformal}, we show that it suffices to have \emph{marginal} c.d.f.~matching between $f(\Phi_b(X))$ and $f(X)$, which is a relaxation of the requirement in \Cref{sec:validity:ex:to:inv}.

\vspace{-1em}

\subsection{Observation-wise transformations and variants of bootstrap} \label{sec:pure:aug:bootstrap}

\vspace{-.5em}

Let $\phi_{bi}$'s be i.i.d.~$\R^d \rightarrow \R^d$ transformations. We consider observation-wise transforms
\vspace{-.25em}
\begin{align*}
    \Phi^{\rm new}_b(x_1,\ldots,x_n) \coloneqq (\phi_{b1}(x_1), \ldots, \phi_{bn}(x_n) )\;.
    \tagaligneq \label{eq:DA}
\end{align*}
\\[-2.25em]
This is exactly \emph{data augmentations} in the machine learning literature \citep{shorten2019survey,shorten2021text,chen2020group,huang2022data}. We consider two  DAB confidence intervals (CIs): (i) DAB with $\Phi^{\rm new}_b$, i.e.~purely observation-wise data augmentations; (ii) DAB with $\Phi^{\rm bootstrap}_b \circ \Phi^{\rm new}_b$, where $\Phi^{\rm bootstrap}_b$ are the bootstrap transformations defined in \Cref{example:bootstrap}.  Two concrete examples are considered:

\noindent
\textbf{Orthogonal transformations for mean-zero and isotropic random vectors.} Consider the toy setting where $X_1, \ldots, X_n$ are i.i.d.~$\R^d$ random vectors with $d$ fixed, and 
\vspace{-.5em}
\begin{align*}
    \mean[X_1] = 0\;,
    \;\;
    \Var[X_1] = I_d
    \;\;\text{ and }\;\;
    f(X) =\mfrac{1}{\sqrt{n}} \msum_{i \leq n}  X_i^\top \bone_d, 
    \;\;\; 
    \text{ where } \bone_d = (1,\ldots, 1)^\top.
\end{align*}
\\[-2.5em]
Let $\O_d$ be the group of $\R^{d \times d}$ orthogonal matrices, and draw the augmentations in \eqref{eq:DA} as
\vspace{-.5em}
\begin{align*}
    \phi_{11}, \ldots, \phi_{Bn} \;\overset{\rm i.i.d.}{\sim}\; \textrm{Uniform}(\O_d)\;.
    \tagaligneq \label{eq:ortho}
\end{align*}
\\[-2.5em]
Neither exchangeability nor group invariance is guaranteed in this case, since $X_i$'s may not be Gaussian. Nevertheless, asymptotic validity holds under Gaussian universality. 

\noindent
\textbf{Statistics on parallel sampling in AI-for-science.} In modern multi-electron simulations, one models the joint distribution of electron positions in a physical system by some large neural network $p_{\hat \theta}$ \citep{carleo2017solving,hermann2020deep,pfau2020ab}. The true distribution $p_*$ typically satisfies invariance with respect to a group $\G$,
\vspace{-.5em}
\begin{align*}
    X^* \overset{d}{=} g(X^*) 
    \quad 
    \text{ for all } g \in \G 
    \text{ and } X^* \sim p_{\theta^*}\;,
    \tagaligneq \label{eq:AISci:inv}
\\[-2.5em]
\end{align*}
but the trained network may not \citep{huang2025diagonal}. To assess the quality of $p_{\hat \theta}$, one runs $n$ parallel, independent MCMC chains targeting $p_{\hat \theta}$. At step $t$, this yields $n$ i.i.d.~samples $X^{(t)} \coloneqq (X^{(t)}_1, \ldots, X^{(t)}_n)$, where $X^{(t)}_i \sim q_t$ and $q_t$ approximates $p_{\hat \theta}$, and one estimates
\vspace{-.5em}
\begin{align*}
    \mean_{X \sim p_{\hat \theta}}[ h(X) ]
    \qquad 
    \text{ by }
    \qquad 
    f(X^{(t)}) \coloneqq
    \mfrac{1}{n} \msum_{i \leq n} h\big( X^{(t)}_i \big)
    \tagaligneq \label{eq:AISci:est}
\\[-2.5em]
\end{align*}
where $h: \R^d \rightarrow \R$ is a function of interest. A confidence interval for \eqref{eq:AISci:est} helps to understand whether sufficiently many chains have been used, and is typically quantified by Gaussian and bootstrap CIs. In view of the invariance \eqref{eq:AISci:inv}, a natural idea is to incorporate transformations from $\G$. Specifically, we draw the augmentations in \eqref{eq:DA} as
\vspace{-.5em}
\begin{align*}
    \phi_{11}, \ldots, \phi_{Bn} \;\overset{\rm i.i.d.}{\sim}\; \textrm{Uniform}(\G)\;.
    \\[-2.75em]
\end{align*}
Note again that invariance is not guaranteed at a finite step $t$.

For both setups, there are now four candidate CIs: the Gaussian CI, the DAB CI with $\Phi^{\rm new}_b$, the bootstrap CI with $\Phi^{\rm bootstrap}$ and the DAB CI with $\Phi^{\rm bootstrap} \circ \Phi^{\rm new}_b$. \Cref{fig:bootstrap:simulate} reports simulations for $X_i \sim \cN(0,I_2)$, where orthogonal transformations $\O_2$ are used for $\Phi^{\rm new}_b$. We find that (i) all CIs have valid coverage and similar sizes at $n=30$, when the CLT approximates an i.i.d.~average of $n$ random variables well; (ii) at a smaller sample size $n=5$, $\O_2$ transformations improve coverage. The same empirical findings hold for non-Gaussian random vectors, in which case only approximate invariance holds; see \Cref{appendix:experiments:pure:aug:bootstrap}. \Cref{fig:bootstrap:ferminet:nk} confirms the findings for samples from AI-for-science neural nets and with a suitably chosen rotational group $\G$; see \Cref{appendix:experiments:pure:aug:bootstrap} for the full setup.

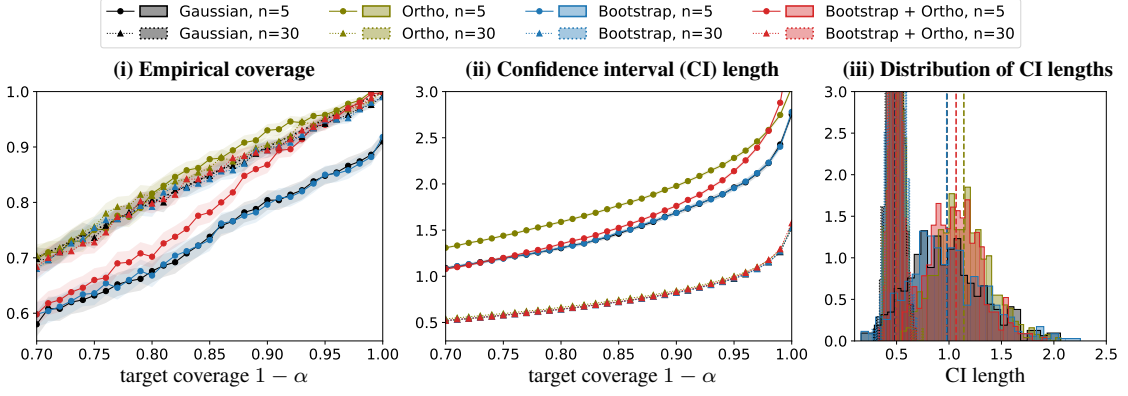
\begin{figure}[t]
  \centering
  \begin{tikzpicture}
        \coordinate (gaussian) at (0,0);
        \begin{scope}[shift={(gaussian)}]
            \node[inner sep=0pt] at (0,0)
                {\includegraphics[width=\textwidth]{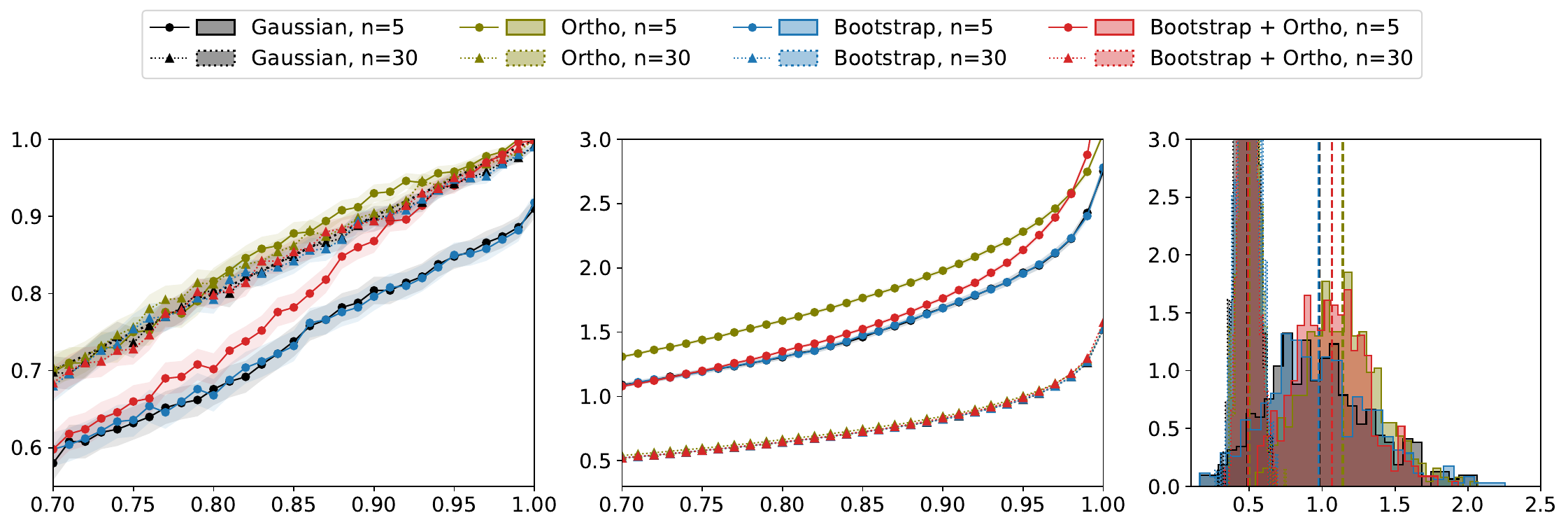}};
            
            \node[inner sep=0pt] at (-4.6,1.45){\scriptsize \textbf{(i) Empirical coverage}};
            \node[inner sep=0pt] at (0.8,1.45){\scriptsize \textbf{(ii) Confidence interval (CI) length}};
            \node[inner sep=0pt] at (5.5,1.45){\scriptsize \textbf{(iii) Distribution of CI lengths}};

            \node[inner sep=0pt] at (-4.6,-2.6){\scriptsize target coverage $1-\alpha$};
            \node[inner sep=0pt] at (0.8,-2.6){\scriptsize target coverage $1-\alpha$};
            \node[inner sep=0pt] at (5.6,-2.6){\scriptsize CI length};
        \end{scope}
  \end{tikzpicture} 
  \caption{Gaussian, bootstrap and DAB CIs for averages of 2d Gaussians over 500 random trials. In (i), the dashed line $y=x$ indicates the desired coverage level.  See \Cref{fig:bootstrap:simulate:Gaussian:nk} for plots with varying $n$ and $B$.
  }
  \vspace{-.5em}
  \label{fig:bootstrap:simulate} 
\end{figure}
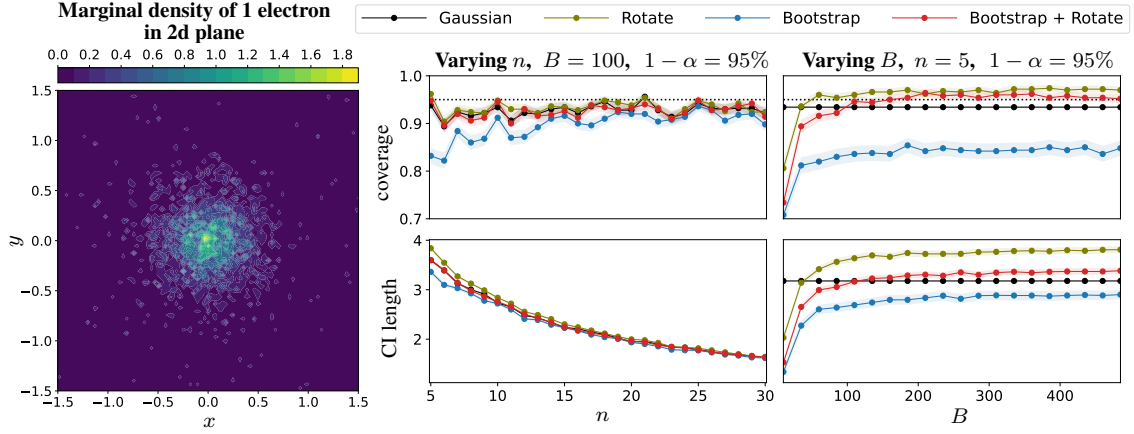
\begin{figure}[t]
  \centering
  \begin{tikzpicture}

        \node[inner sep=0pt] at (0,0)
            {\includegraphics[width=.7\textwidth]{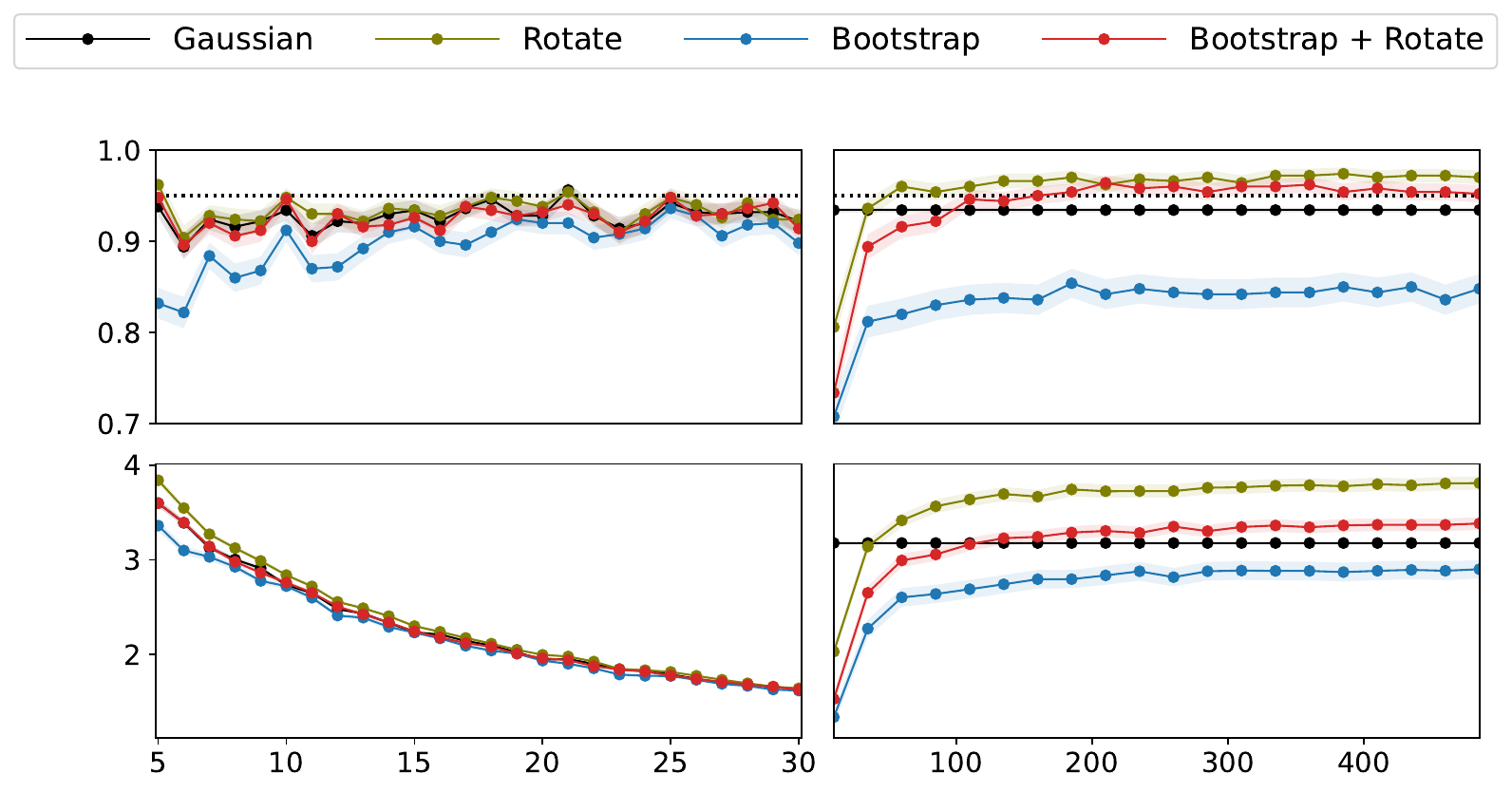}};
        
        \node[inner sep=0pt] at (-1.86,1.9){\scriptsize \textbf{Varying $n$, \;$B=100$, \, $1-\alpha=95\%$}};

        \node[inner sep=0pt] at (2.85,1.9){\scriptsize \textbf{Varying $B$, \;$n=5$, \, $1-\alpha=95\%$}};

        \node[inner sep=0pt] at (-1.85,-2.8){\scriptsize $n$};

        \node[inner sep=0pt] at (2.87,-2.8){\scriptsize $B$};
        
        \node[inner sep=0pt,align=center] at (-7.25,2.4){\scriptsize \textbf{Marginal density of 1 electron}
        \\[-.5em]
        \scriptsize \textbf{in 2d plane}};

        \node at (-7.25,-0.33){\includegraphics[width=.32\textwidth]{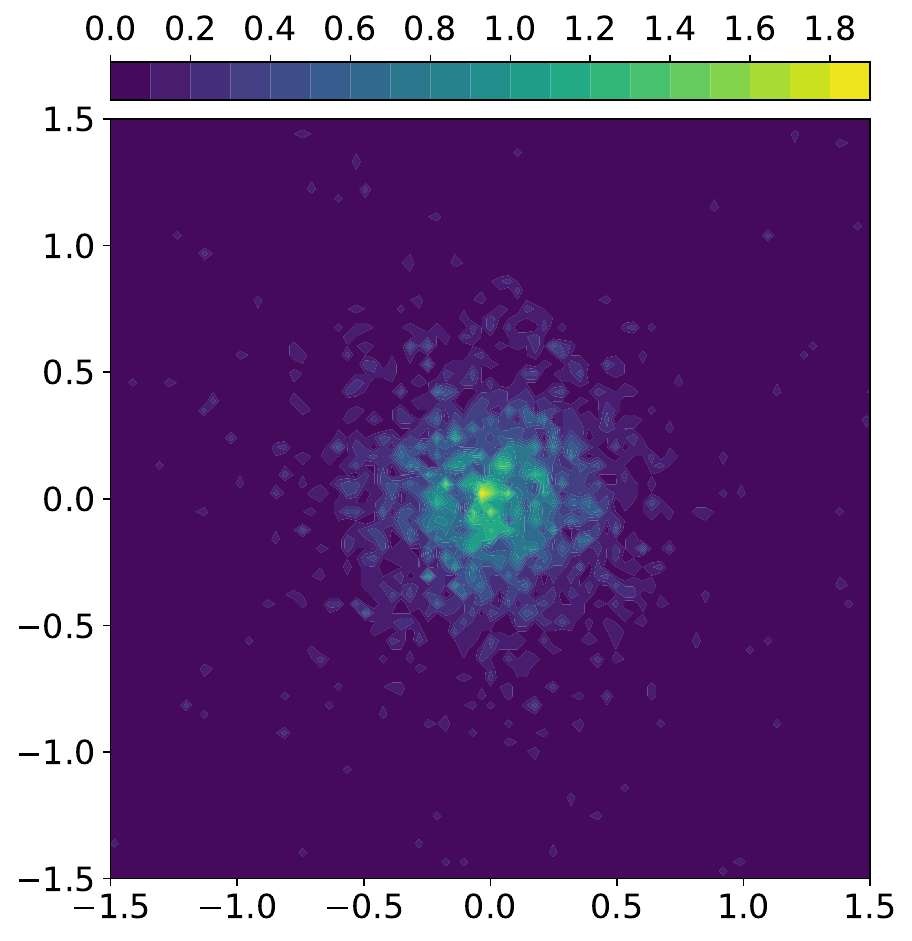}};
        \node[inner sep=0pt,rotate=90] at (-9.6,-0.48){\scriptsize $y$};
        \node[inner sep=0pt] at (-7.05,-2.88){\scriptsize $x$};

        \node[inner sep=0pt,rotate=90] at (-4.72,0.55){\scriptsize coverage};

        \node[inner sep=0pt,rotate=90] at (-4.65,-1.5){\scriptsize CI length};
  \end{tikzpicture} 
  \\[-.5em]
  \caption{CIs for $3$-electron configuration samples from a FermiNet wavefunction trained for the Lithium atom, under varying $n$ (number of Markov chains) and $B$. \emph{Left.} Visualisation of rotational symmetry through the marginal density for $1$-electron in the 2d plane. \emph{Right.} Empirical coverage and length of the CIs for the $x$-component of the electron dipole moment of the samples.
  The error bars are over 500 random trials, with each trial drawing a fresh batch of MCMC samples. See \Cref{fig:bootstrap:ferminet} for plots with varying $\alpha$. 
  } 
  \label{fig:bootstrap:ferminet:nk} 
  \vspace{-1em}
\end{figure}

\vspace{-.5em}

\subsection{Wild bootstrap with additional transformations for degenerate U-statistics} \label{sec:wild:bootstrap:add}

Given i.i.d.~samples $\{Y_i\}_{i=1}^n$ from a distribution $P$ in $\R^{d'}$ and i.i.d.~samples $\{Z_i\}_{i=1}^n$ from a distribution $Q$, we consider testing $H_0: P = Q$ with Maximum Mean Discrepancy (MMD) \citep{gretton2012kernel}. We follow \cite{steinwart2012mercer} to define the kernel function $\kappa(y, y') \coloneqq \langle \varphi(y), \varphi(y') \rangle_{\cH}$ for feature map $\varphi: \R^{d'} \rightarrow \cH$ and a Hilbert space $\cH$. The MMD is defined as
$
    D(Q,P)
    \coloneqq \mean_{Y, Y' \sim P}[\kappa(Y,Y')] - 2 \mean_{Y \sim P, Z \sim Q}[\kappa(Y,Z)] + \mean_{Z, Z' \sim Q}[\kappa(Z,Z')] 
$,
and a standard unbiased estimator is the degree-two U-statistic 
\vspace{-.5em}
\begin{align*}
    D_n
    \;\coloneqq\; \mfrac{1}{n(n-1)} \msum_{1 \leq i \neq j \leq n} 
    \,
    u(Y_i, Y_j, Z_i, Z_j)
    \;,
\end{align*}
\\[-2.3em]
where $u(y,y',z,z') \coloneqq \kappa(y, y') + \kappa(z, z') - \kappa(y, z') - \kappa(y', z)$. Under $H_0$, $D(Q,P)=0$, and $D_n$ is degenerate and asymptotically distributed as an infinite sum of weighted chi-squares, which is intractable to simulate. Wild bootstrap, as illustrated in \Cref{example:marginal:wild:bootstrap}, provides a finite-sample-valid CI for $D_n$ by exploiting sign-flipping symmetry. 

Suppose we know that $P$ is approximately invariant under some random $\R^{d'} \rightarrow \R^{d'}$ transformation $\phi$. Let $(\phi_{ri})_{ r \in \{1,2\}, i \leq n}$ be i.i.d.~copies of $\phi$ and $\epsilon_i$'s be i.i.d.~Rademacher random variables in \eqref{eq:wild:bootstrap:rademacher}. We propose DAB with the transformed statistic
\vspace{-.5em}
\begin{align*}
    f(\Phi_1(X))
    \;=\;
    \mfrac{1}{n(n-1)} \msum_{1 \leq i \neq j \leq n} 
    \epsilon_i \epsilon_j  
    \,
    u\big( 
        \,
        \phi_{1i}(Y_i) 
        ,
        \phi_{1j}(Y_j)
        ,
        \phi_{2i}(Z_i)
        ,
        \phi_{2j}(Z_j)
    \big)
    \;.
    \tagaligneq \label{eq:WB:main:text}
\end{align*}
\\[-2.25em]
\Cref{appendix:condition:wild:bootstrap} shows that for valid coverage under $H_0$, the approximate invariance condition can be formulated via conditional moment matching, but in the space of the feature map $\varphi$ rather than in the original data space. While this conditional moment matching condition is unverifiable for kernels beyond simple linear kernels such as $\kappa(y,y')=y^\top y'$, we empirically compare the DAB method \eqref{eq:WB:main:text} against wild bootstrap under four settings, all with the Gaussian RBF kernel: (a) 2d simulated data under mean shift (\Cref{fig:wb:simulate:rademacher}); (b) Noiseless v.s.~noisy MNIST images (\Cref{fig:wb:mnist:alt}); (c) odd v.s.~even-numbered classes in CIFAR-10 images (\Cref{appendix:experiments:wb}); (d) physical signal v.s.~noise in a HIGGS boson dataset (\Cref{appendix:experiments:wb}). Our DAB variant visibly improves test power in (a), (b) and (c), and has a similar test power as vanilla wild bootstrap in (d).

\begin{figure}[t]
  \centering
  \begin{tikzpicture}
        \coordinate (rademacher) at (0,0);
        \begin{scope}[shift={(rademacher)}]
            \node[inner sep=0pt] at (-3.1,0)
                {\includegraphics[width=.69\textwidth]{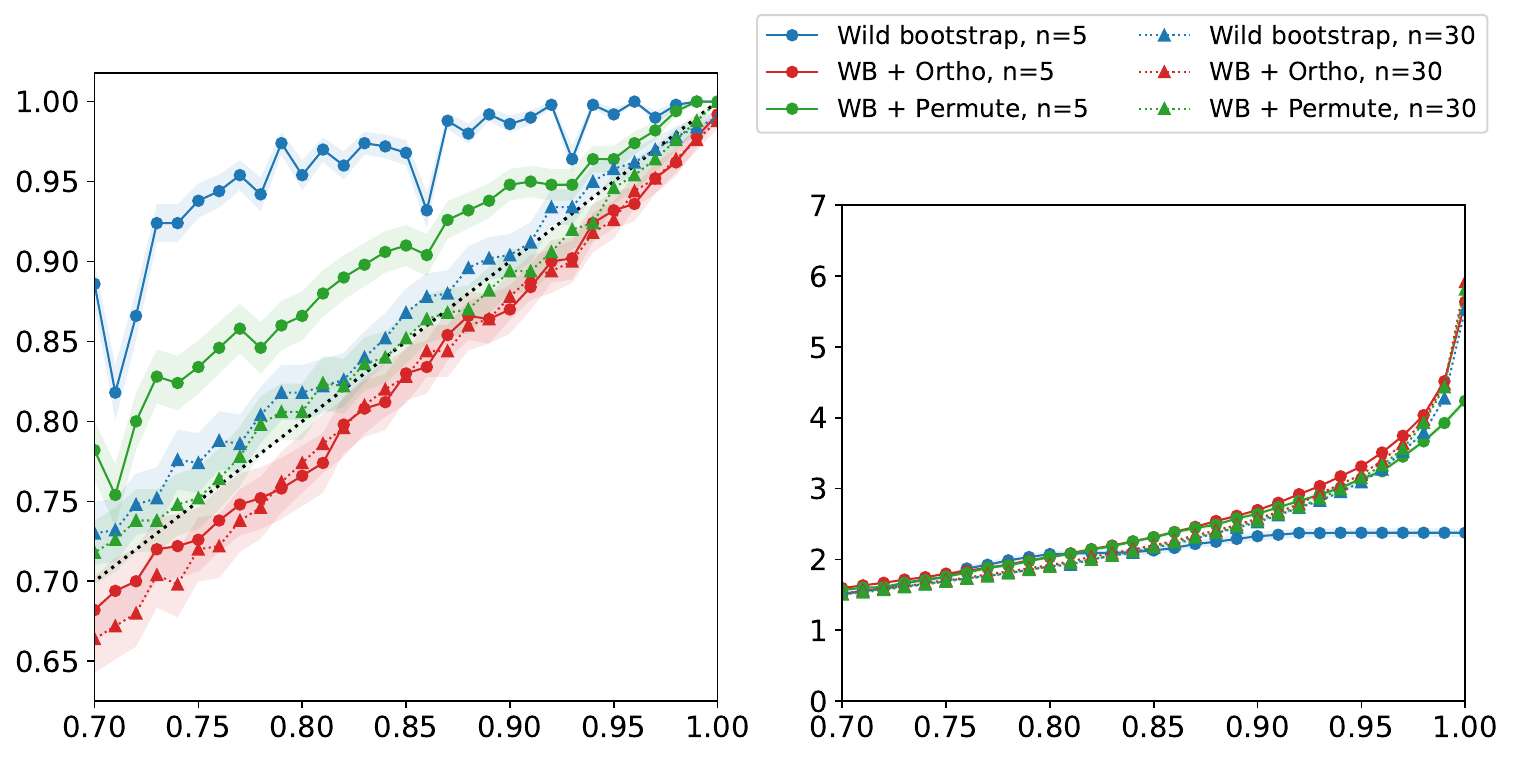}};

            \node[inner sep=0pt] at (4.6,1.25)
                {\includegraphics[width=.32\textwidth]{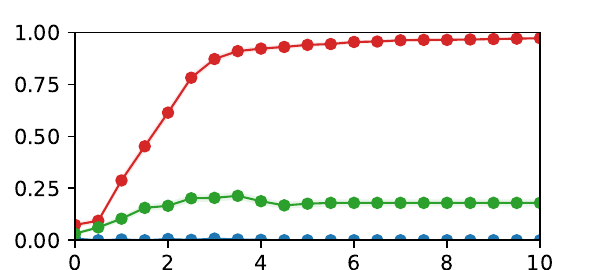}};

            \node[inner sep=0pt] at (4.6,-1.4)
                {\includegraphics[width=.32\textwidth]{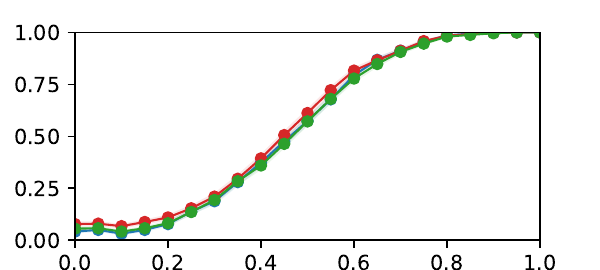}};
            
            \node[inner sep=0pt] at (-5.6,2.3){\scriptsize \textbf{(i) Empirical coverage}};
            \node[inner sep=0pt] at (-0.45,1.4){\scriptsize \textbf{(ii) Confidence interval (CI) length}};
            \node[inner sep=0pt] at (4.7,2.3){\scriptsize \textbf{(iii) Test power, $n=5$}};
            \node[inner sep=0pt] at (4.7,0){\scriptsize mean shift};
            \node[inner sep=0pt] at (4.7,-0.4){\scriptsize \textbf{(iv) Test power, $n=30$}};
            \node[inner sep=0pt] at (4.7,-2.7){\scriptsize mean shift};
 
            \node[inner sep=0pt] at (-5.6,-2.7){\scriptsize target coverage $1-\alpha$};
            \node[inner sep=0pt] at (-0.4,-2.7){\scriptsize target coverage $1-\alpha$};
        \end{scope}
  \end{tikzpicture} 
  \\[-.2em]
  \caption{CIs for MMD statistics with RBF kernel on 2d Rademachers over 500 random trials. \textbf{(i)} Empirical v.s. target coverage. \textbf{(ii)} CI length~v.s. target coverage. \textbf{(iii) and (iv)} Test power against mean shift. 
  }
  \vspace{-.5em}
  \label{fig:wb:simulate:rademacher} 
\end{figure}
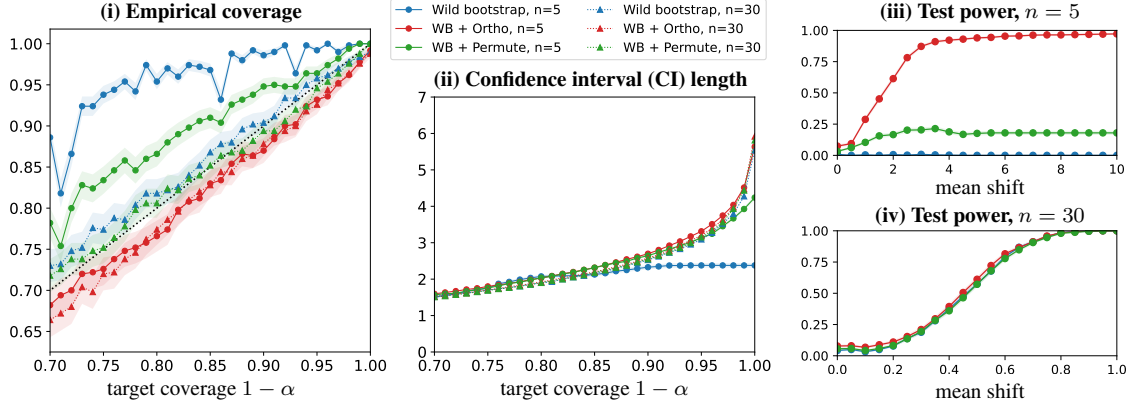

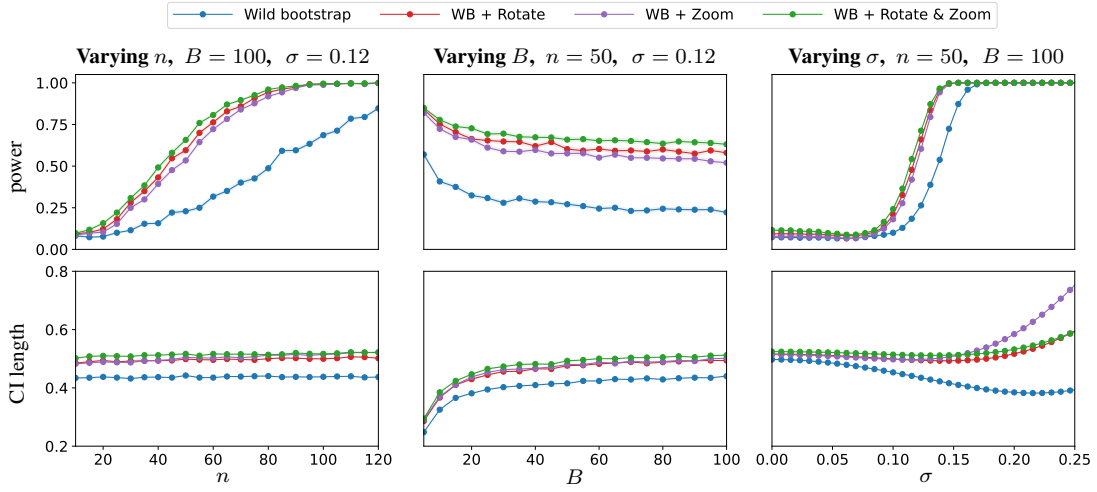
\begin{figure}[t]
  \centering
  \begin{tikzpicture}
            \node[inner sep=0pt] at (0,0)
                {\includegraphics[width=.95\textwidth]{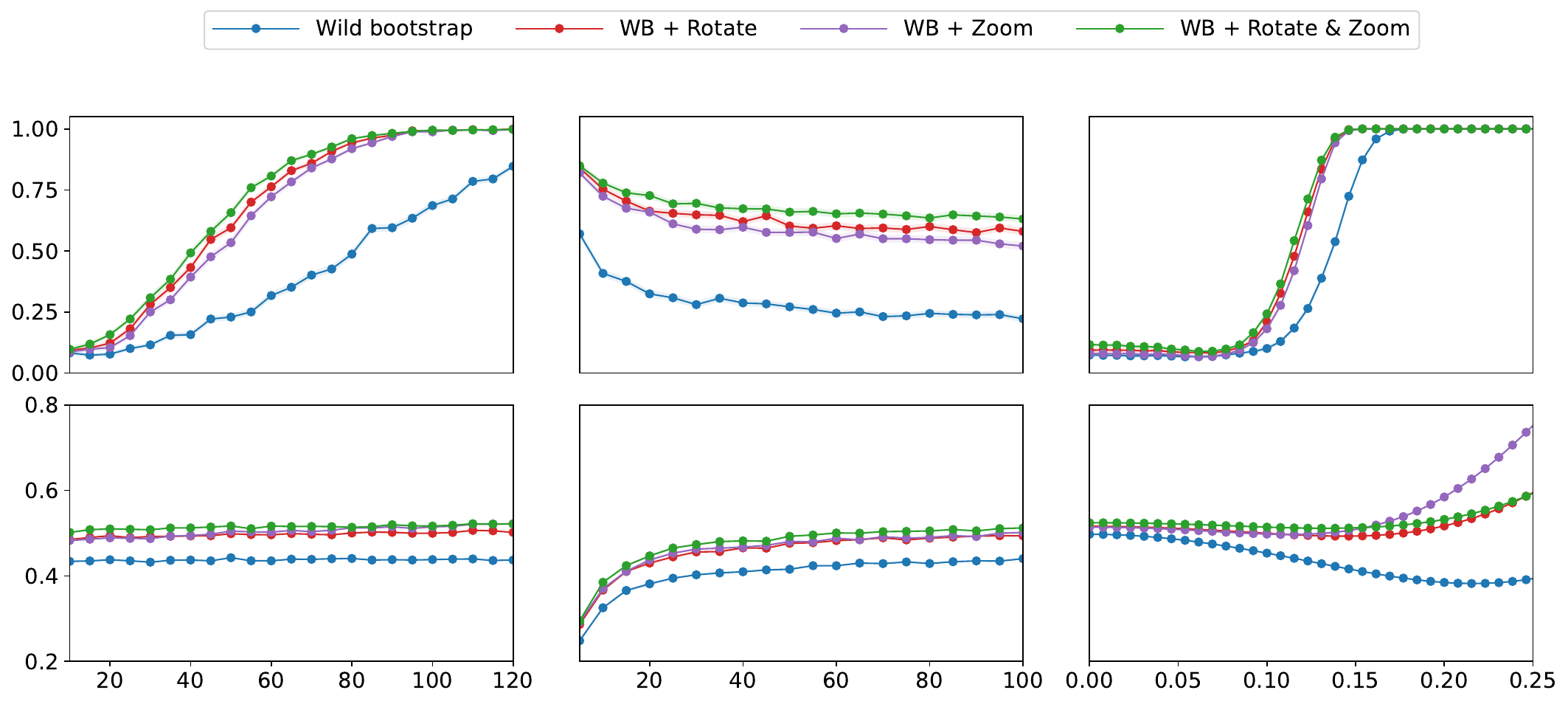}};

            \node[inner sep=0pt] at (-4.5,2.35){\scriptsize \textbf{Varying $n$, \;$B=100$, \, $\sigma=0.12$}};

            \node[inner sep=0pt] at (0.15,2.35){\scriptsize \textbf{Varying $B$, \;$n=50$, \, $\sigma=0.12$}};

            \node[inner sep=0pt] at (4.8,2.35){\scriptsize \textbf{Varying $\sigma$, \;$n=50$, \, $B=100$}};

            \node[inner sep=0pt,rotate=90] at (-7.2,0.9){\scriptsize power};

            \node[inner sep=0pt,rotate=90] at (-7.2,-1.7){\scriptsize CI length};

            \node[inner sep=0pt] at (-4.5,-3.2){\scriptsize $n$};

            \node[inner sep=0pt] at (0.15,-3.2){\scriptsize $B$};

            \node[inner sep=0pt] at (4.8,-3.2){\scriptsize $\sigma$};
  \end{tikzpicture} 
  \\[-.4em]
  \caption{CIs for MMD statistics with RBF kernel for testing noiseless MNIST images against those with pixel-wise additive i.i.d.~$\cN(0,\sigma^2)$ noise. The experiments are over varying $n$, $B$, $\sigma$ and over 1000 random draws of the dataset, with a fixed target Type-I error $\alpha=5\%$. See \Cref{fig:wb:mnist:null} for results under the null.
  }
  \vspace{-1em}
  \label{fig:wb:mnist:alt} 
\end{figure}

\vspace{-1.2em}

\subsection{Conformal prediction with data augmentation} \label{sec:conformal}

\vspace{-.6em}

Consider i.i.d.~random input-output pairs $X_i=(V_i, Y_i)$, such that $(X_i)_{2 \leq i \leq n}$ is the calibration set and $X_1$ is the new data point. To see that split conformal prediction (CP) is an example of DAB, we identify $(\Phi_b)_{b \leq B}$ with $B=n-1$ as an enumeration of all deterministic cycling operations on the $n$ input-output pairs, $(x_i)_{i \leq n} = (v_i, y_i)_{i \leq n}$, i.e.
\vspace{-.5em}
\begin{align*}
    \Phi^{\rm CP}_b(x_1, \ldots, x_n) \;=\; 
    \big( x_{\sigma_b(1)}, \ldots, x_{\sigma_b(n)} )\;,
    \quad 
    \sigma_b(i) \coloneqq (( i + b - 1 ) \;{\rm mod}\; n ) + 1\;.
\\[-2.5em]
\end{align*}
Then, for some score function $h: \R^{d_{\rm out}} \times \R^{d_{\rm out}} \rightarrow \R$ and some estimator $g: \R^{d_{\rm in}} \rightarrow \R^{d_{\rm out}}$ fitted on a hold-out dataset, we can identify the statistic in DAB as the conformity score
\vspace{-.5em}
\begin{align*}
    f(X) \;=\; h( Y_1, \, g(V_1) )\;,
    \qquad 
    \text{ which implies } 
    \qquad 
    f\big(\Phi^{\rm CP}_b(X) \big) \;=\; h(Y_{b+1},\, g(V_{b+1}) )\;.
    \\[-2.5em]
\end{align*}
We again assume approximate invariance under a random transformation $\phi = (\phi^{\rm in}, \phi^{\rm out})$, and form DAB from $(n-1)k$ transformations: Let $\phi_{bji}$ be i.i.d.~copies of $\phi$ and define
\vspace{-.5em}
\begin{align*}
    \Phi^{\rm DAB}_{bj}(x_1, \ldots, x_n) 
    =
    \Phi^{\rm CP}_b
    \big(
        \phi_{bj1}(x_1) 
        \,,\,\ldots\,,\,
        \phi_{bjn}(x_n) 
    \big)
    \;\;
    \text{for } 1 \leq b \leq n-1 \text{ and } 1 \leq j \leq k.
    \\[-2.5em]
\end{align*}
Unlike the earlier examples, universality no longer holds: $f(X)$ is completely determined by $X_1$, and violates the stability condition required for \Cref{thm:universality}. Nonetheless, we still have guarantees from the conditional c.d.f.~matching in \Cref{thm:Del:Kol:approx:inv} which, for $k=1$, can be improved to marginal c.d.f.~matching (\Cref{appendix:condition:conformal}) measured by
\vspace{-.5em}
\begin{align*}
    \Delta^{\rm CP}_{\rm inv} 
    \;\coloneqq\;
    \msup_{t \in \R} \big| 
        \P\big(  h(\phi^{\rm out}(Y_1),\, g(\phi^{\rm in}(V_1))) \leq t\big)
        \,-\, 
        \P\big( h(Y_1, \, g(V_1))  \leq t \big)
    \big| 
    \\[-2.4em]
\end{align*}

\begin{figure}[t]
  \centering
  \begin{tikzpicture}
        \coordinate (mnistgrid) at (0,0);
        \begin{scope}[shift={(mnistgrid)}]
            \node[inner sep=0pt] at (0,0)
                {\includegraphics[width=.9\textwidth]{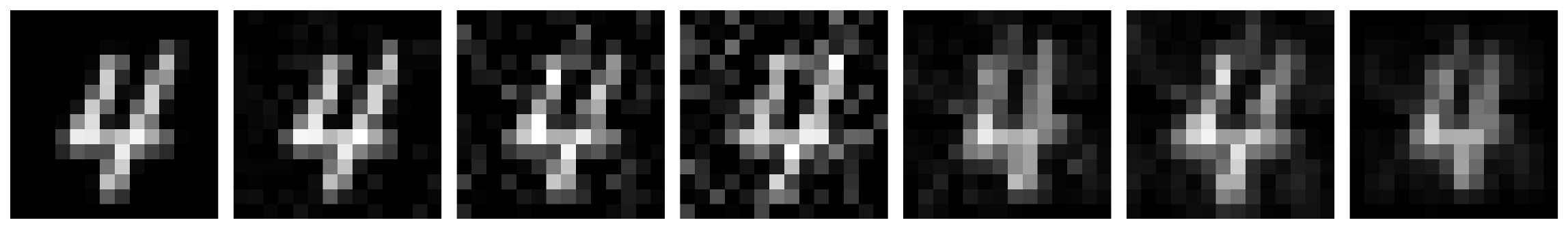}};

            \node[inner sep=0pt] at (-5.7,1.05){\scriptsize \textbf{(i) $\sigma=0$}};
            \node[inner sep=0pt] at (-3.8,1.05){\scriptsize \textbf{(ii) $\sigma=0.05$}};
            \node[inner sep=0pt] at (-1.9,1.05){\scriptsize \textbf{(iii) $\sigma=0.12$}};
            \node[inner sep=0pt] at (0,1.05){\scriptsize \textbf{(iv) $\sigma=0.20$}};
            \node[inner sep=0pt] at (1.9,1.05){\scriptsize \textbf{(v) Rotate}};
            \node[inner sep=0pt] at (3.85,1.05){\scriptsize \textbf{(vi) Zoom}};
            \node[inner sep=0pt] at (5.8,1.05){\scriptsize \textbf{(vii) Rotate \& Zoom}};
        \end{scope}
  \end{tikzpicture}
  \\[-.5em]
  \caption{A noisy MNIST image with varying $\sigma$ in (i) -- (iv) and with different augmentations in (v) -- (vii).  }
  \label{fig:mnist:visual}
  \vspace{-.2em}
\end{figure}
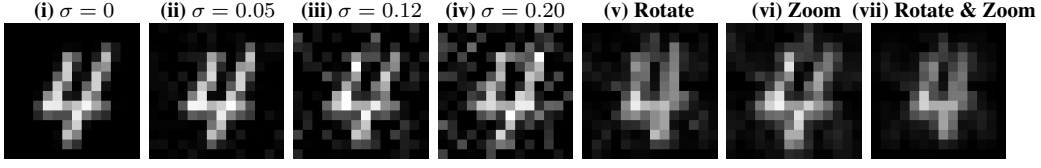

\begin{figure}[t]
  \centering
  \begin{tikzpicture}
        \coordinate (Gaussian) at (-3.8,0);
        \begin{scope}[shift={(Gaussian)}]
            \node[inner sep=0pt] at (0,0)
            {\includegraphics[width=.49\textwidth]{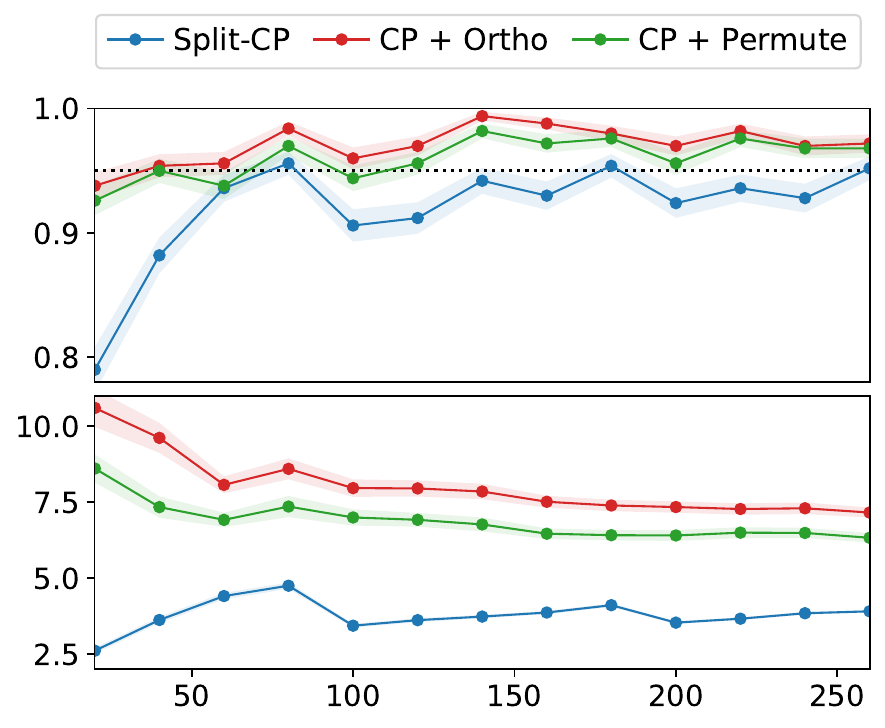}};

        \node[inner sep=0pt, align=left] at (0.2,3.3){
            \scriptsize \textbf{(i) Prediction interval for the outcome of a 2d linear}
            \\[-.5em]
            \scriptsize \quad\  \textbf{model based on linear regression on Gamma vectors}};

            \node[inner sep=0pt] at (0.5,-3.2){\scriptsize $n$};

            \node[inner sep=0pt,rotate=90] at (-3.55,1.05){\scriptsize coverage};

            \node[inner sep=0pt,rotate=90] at (-3.55,-1.45){\scriptsize CI length};
        \end{scope}
        \coordinate (AMP) at (3.8,0); 
        \begin{scope}[shift={(AMP)}] 
        \node[inner sep=0pt] at (0,0) 
            {\includegraphics[width=.49\textwidth]{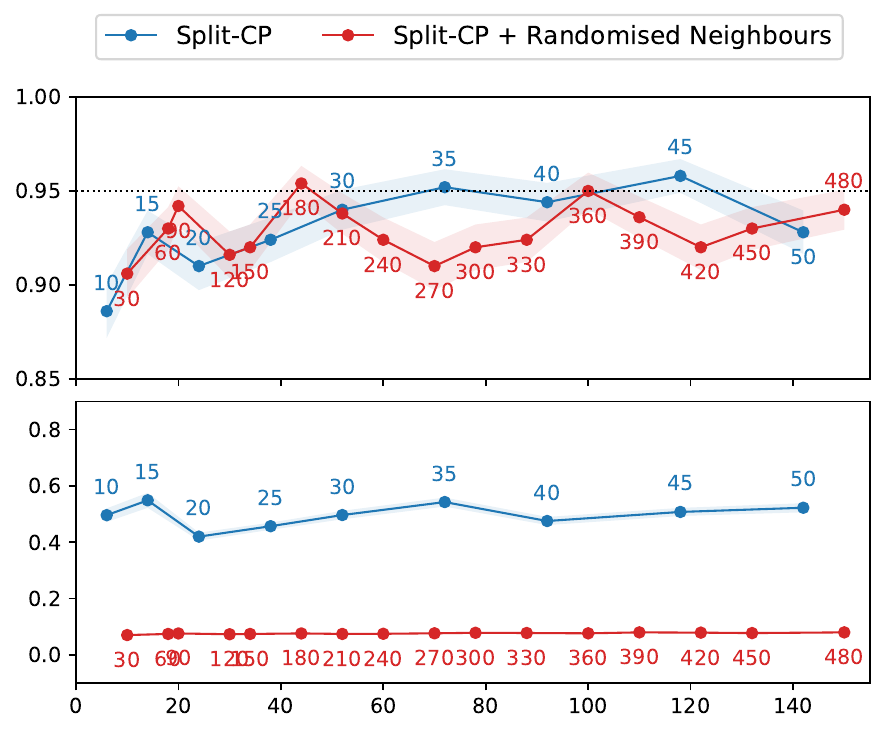}};

        \node[inner sep=0pt, align=left] at (0.2,3.3){
            \scriptsize \textbf{(ii) Prediction interval for the energy of a molecule} 
            \\[-.5em]
            \scriptsize \quad\  \textbf{based on a neural network interatomic potential}};

        \node[inner sep=0pt] at (0.5,-3.2){\scriptsize GPU seconds per random trial};

        \node[inner sep=0pt,rotate=90] at (-3.73,1.05){\scriptsize coverage};

        \node[inner sep=0pt,rotate=90] at (-3.73,-1.45){\scriptsize CI length};
        \end{scope}
  \end{tikzpicture}
  \\[-.5em]
  \caption{DAB-variants of conformal prediction over $500$ random trials and for $1-\alpha=95\%$. (i) In each trial, a random $\beta \in \cN(0,I_2)$ is drawn and a size-$n$ dataset of $X_i = (V_i,Y_i)$ is generated as $Y_i = V_i^\top \beta + \epsilon_i$, where $V_i$ is coordinate-wise i.i.d.~centred Gamma variables with shape and rate $1$, and $\epsilon_i \sim \cN(0,1)$. $0.8 n$ data is used for training and $0.2n$ for calibration. The DAB CI is formed with $k=50$ for predicting the outcome $Y_{\rm new}$ of a fresh i.i.d.~drawn $V_{\rm new}$. (ii) We split the set of C4h-symmetric molecules in the QM-sym dataset of \cite{liang2019qm} into two halves, one for training the GMP+SNN neural network used in \cite{hu2022robust}, and the other for calibration. The DAB CI is formed with $k=1$ for predicting the energy of a randomly chosen molecule. The $x$-axis indicates the GPU seconds required to generate the CI per random trial, whereas the number next to each data point indicates the size of the calibration set used.}
  \label{fig:cp}
  \vspace{-1em}
\end{figure}
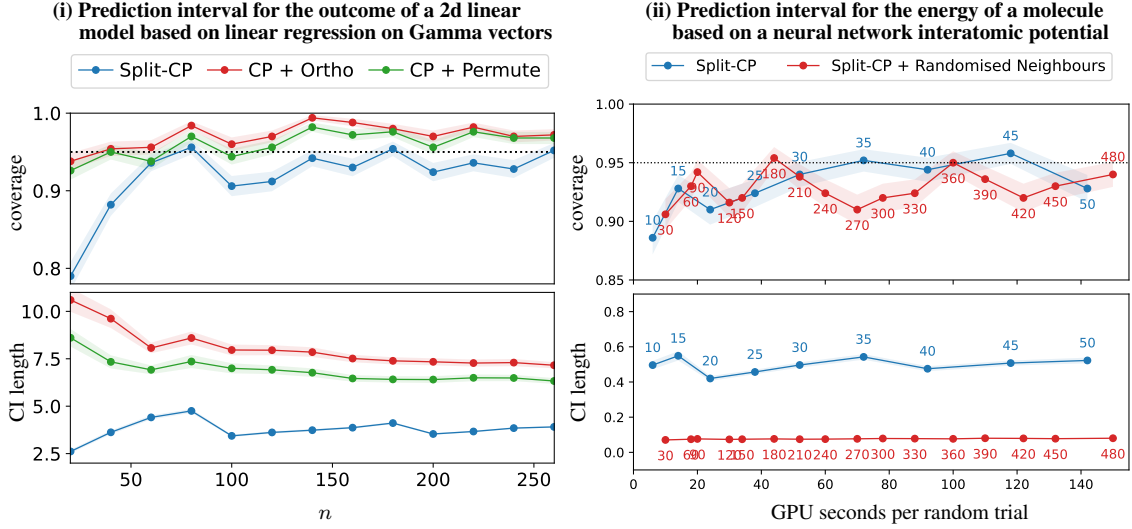

\Cref{fig:cp} compares DAB to CP in two setups: (a) For linear regression on simulated data,  DAB has better coverage at small $n$ but is more conservative than CP at large $n$. (b) For neural networks on molecular data, we use the CP score function proposed by \cite{hu2022robust}, which involves computing the distance of the input molecule to the $K$ closest neighbours in the training data. This computation can be costly, as it involves ordering the entire training set by distance from the input molecule. Our DAB variant uses a randomized selection of $K$ training points instead (denoted \texttt{Randomized Neighbours} in \Cref{fig:cp}(ii)), and gives a smaller CI with a small coverage loss within a given amount of compute. \Cref{appendix:experiments:conformal} includes setup details and further experiments with large language models, where DAB shows similar performance as vanilla CP.

\vspace{-1.5em}

\section{Discussion} \label{sec:discussion}

\vspace{-.5em}

The key idea behind DAB is the following: While group symmetries are mathematically elegant as a definition of invariance, for statistical desiderata such as coverage of a confidence interval (CI), we can afford a looser interpretation of invariance ---
		that a transformed statistic has a similar distribution as that of the original statistic. DAB unifies two prevalent families of methods for confidence interval constructions: One that is based on exact group invariance and provides black-box, finite-sample guarantees, and one that is based on asymptotic theories and provides estimator-specific but asymptotically exact guarantees. This provides the ground for incorporating data augmentations, a technique ubiquitous in machine learning (ML) for incorporating approximate and non-group invariances, into uncertainty quantification procedures. In various ML domains, tremendous empirical efforts have been spent on optimising the augmentation choice to improve the accuracy of estimators \citep{shorten2019survey,shorten2021text}; DAB opens the possibilities for doing the same for improving the quality of CIs. Our empirical results are a first step in this direction: We see visible improvements with DAB in many tasks, but also limited gains when not many symmetries are available (Higgs boson and text examples, \Cref{appendix:experiments}). Many interesting questions are open as to what transformations are optimal for each task, and the answers are likely to be domain-specific.

A crucial aspect of confidence intervals is their efficiency. Under the well-known duality of a CI and a hypothesis test \citep{lehmann2005testing}, the efficiency of the DAB CI $C_\alpha(X_{\rm obs})$ is the power of its induced test statistic $\ind\{ \argdot \not\in C_\alpha(X_{\rm obs}) \}$ for testing whether some given random variable $W$ follows the same distribution as $X_{\rm unobs}$. In fact, the indistinguishable hypotheses for DAB are exactly the distributions that satisfy the same invariances built into DAB. As such, DAB can be viewed as a test statistic for an invariance test \citep{lehmann2005testing,dobriban2022consistency,christie2022testing,chiu2023non}. The efficiency question about DAB becomes a question about the test power of DAB against a broad class of non-invariant alternatives, and our work can be viewed as providing Type-I error control and empirical investigations on the Type-II error. An interesting future direction is to make the connection to testing rigorous and obtain the worst-case test power guarantee of DAB under different alternatives.

One of our key tools is Gaussian universality, which substitutes the role of classical CLT in bootstrap and extends DAB to estimators that are not asymptotically normal. The core observation is that, to obtain coverage guarantees for DAB, it suffices to establish that the distributions of two suitably chosen statistics are close, regardless of whether the intermediate distributional approximation simplifies to a tractable one e.g.~Gaussian. Similarly, there is no reason to restrict ourselves to estimators whose intermediate approximations are described by Gaussian universality. By exploiting modern distributional approximation results for heavy-tailed data and random matrices, we conjecture that DAB can be extended to a much wider range of settings than those considered in this article. In particular, the finite-sample bounds of our \Cref{thm:ex:validity,thm:Del:Kol:approx:inv} provide an entry point for any quantitative estimates for suitable distributional approximation results to be plugged in and yield coverage guarantees. Moreover, the Euclidean structures assumed in \Cref{thm:Del:Kol:approx:inv,thm:universality} are specific to our usage of the Kolmogorov metric and Gaussian universality tools. We conjecture that non-Euclidean variants of DAB can be obtained by adapting our proofs to a different metric with suitable non-Euclidean limit theorems.

After the release of an initial preprint, we were alerted of a concurrent work \citep{paul2026probability} that independently notice the connection between conformal prediction and bootstrap. They obtain results similar to our \Cref{thm:ex:validity,thm:Del:Kol:approx:inv} but also notably accommodating a non-uniform pivot and non-identically distributed variables to be ranked. As applications, they study a range of statistical estimators in detail, and show that asymptotically valid CIs can be obtained without infinitely many resampling steps. In contrast, our emphasis is on the interpretation via invariances and the incorporation of data augmentations, with theoretical groundings from Gaussian universality results on dependent data (\Cref{thm:universality}) and empirical validation with augmentations from ML settings (\Cref{sec:new:DAB:methods}). At a technical level, \cite{paul2026probability} consider CIs without tie-breaking and impose a continuity condition (through a L\'evy concentration function); we do not impose this due to our use of smooth tie-breaking. We expect that their results can be combined with our universality results to expand the scope of DAB. 

\vspace{.5em}


\noindent
\textbf{Acknowledgement.} {This work has been partially supported by the Gatsby Charitable Foundation and by the UK Engineering and Physical Sciences Research Council (EPSRC) (Grant No.~EP/Y028783/1, Prob\_AI Hub). The author thanks Peter Orbanz and Vasco Portilheiro for the helpful discussions and support in this work. The author also thanks Antonin Schrab for helpful comments, especially on wild bootstrap, and for suggesting the name of data augmented bootstrap (DAB). We are grateful to Arun Kumar Kuchibhotla, Woonyoung Chang and Manit Paul for pointing out the connection to \cite{paul2026probability} and an issue with an earlier version of \Cref{thm:Del:Kol:approx:inv}. All codes can be found at \url{github.com/KevHH/DAB_code}. The codes for Higgs boson, QM-sym and language data have been assisted by Codex.
}
\vspace{-1em}

\begingroup
\renewcommand{\sectionbreak}{\vspace{1em}} 
\titlespacing*{\section}{0pt}{0.5em}{0.4em}
\bibliography{references}
\endgroup

\newpage

\appendix
\singlespacing
\crefalias{section}{appendix}
\crefalias{subsection}{subappendix}
\crefalias{subsubsection}{subsubappendix}

\clearpage

\begin{center}
\large Supplementary materials to ``Data augmented bootstrap: Unifying confidence interval construction by approximate invariance"
\end{center}

\noindent
The supplementary materials are organised as follows:
\\[.1em]
\textbf{\cref{appendix:additional:examples}} provides additional examples of data augmented bootstrap (DAB). 
\\[.1em]
\textbf{\cref{appendix:related:work}} provides additional discussions on related works.
\\[.1em]
\textbf{\Cref{appendix:validity:exact:inv}} provides additional results and proofs for the approximate exchangeability results in \Cref{sec:validity:exact:inv}.
\\[.1em]
\textbf{\Cref{appendix:thm:Del:Kol:approx:inv}} provides additional results and proofs for the approximate invariance results in \Cref{sec:validity:ex:to:inv}.
\\[.1em]
\textbf{\Cref{appendix:universality}} states and proves additional Gaussian universality results that complement \Cref{sec:validity:universality}, and verifies universality for specific examples.
\\[.1em]
\textbf{\Cref{appendix:proof:compose:exact:with:approx}} proves \Cref{lem:compose:exact:with:approx}, which concerns the composition of approximately invariant transformations and exactly invariant ones.
\\[.1em]
\textbf{\Cref{appendix:theory:examples}} derives the theoretical conditions for coverage guarantees to hold for each example considered in \Cref{sec:new:DAB:methods}.
\\[.1em]
\textbf{\cref{appendix:experiments}} includes experimental details and additional empirical results that complement those reported in \Cref{sec:new:DAB:methods}.

\vspace{-1em}

\section{Additional examples of DAB} \label{appendix:additional:examples}

\begin{example}[Asymptotic coverage for wild bootstrap] \label{example:asymptotic:wild:bootstrap} Let $X$ and $f(X)$ be defined as in \Cref{example:marginal:wild:bootstrap} with the symmetry $X_1 \overset{d}{=} - X_1$. We consider a more general setup of wild bootstrap, where $\Phi_b$ is defined as in  \Cref{example:marginal:wild:bootstrap} but with generic i.i.d.~variables $(\epsilon_{bi})_{b \leq B, i \leq n}$ with zero mean and unit variance. Note that $\Phi_b$ does not come from a group structure in general, e.g.~in the case of $\epsilon_{bi}$'s being standard Gaussians. A classical result of \cite{dehling1994random}, in the case $d=1$ and the U-statistic $f(X)$ is degenerate, gives that
\begin{align*}
    \| \Delta_{\rm inv}(X) \|_{L_2}
    \;=\;
    \big\| \msup_{t \in \R} 
    \big| \P( f(\Phi_1(X)) \leq t \,|\, X ) - \P( f(X) \leq t ) \big| 
    \big\|_{L_2}
    \rightarrow 0\;.
\end{align*}
By \Cref{thm:Del:Kol:approx:inv}, this again implies valid coverage for wild bootstrap. Similar guarantees are given in \cite{janssen1994weighted} for non-degenerate U-statistics, in which case non-negative weights $(\epsilon_{bi})_{b \leq B, i \leq n}$ are used.
\end{example}

\begin{example}[Conditional coverage for bootstrap] \label{example:conditional:bootstrap} We inherit the notation from \Cref{example:asymptotic:bootstrap}, and consider conditioning on $\Xobs = \sigma(\tilde X_1, \ldots, \tilde X_m)$ for some $m < n$. Recall that $\Var[\tilde X_1] =1$ and assume for simplicity that  $\mean |\tilde X_1|^4 < \infty$. Then the Berry-Ess\'een bound applies: Almost surely,
\begin{align*}
    &\;
    \sup_{t \in \R} 
    \,\Big|
        \,\P( f(X) \leq t \,|\, \Xobs ) 
        -
        \P\Big( \mfrac{\sqrt{n-m} \, \eta + \sum_{i=1}^m (\tilde X_i - \mean[\tilde X_1]) }{\sqrt{n}} \leq t \,\Big|\, \Xobs
        \Big) 
    \Big| 
    \,\leq\,
    \mfrac{ \mean |\tilde X_1|^3 }{\sqrt{n-m}}
    \,,
\end{align*}
where $\eta \sim \cN(0,1)$ is independent of $\Xobs$. Similarly, 
\begin{align*}
    \msup_{t \in \R} \Big|
        \,
        \P( f(\Phi_1(X)) \leq t \,|\, X ) 
        \,-\,
        \P( \hat \sigma \eta  \leq t \,|\, X
        ) 
    \Big| 
    \;\leq\;
    \mfrac{\hat \kappa_3}{\sqrt{n} \, \hat \sigma^3}
    \qquad 
    \text{ almost surely, }
\end{align*}
where $\hat \sigma^2 \coloneqq \frac{1}{n} \sum_{i \leq n} (\tilde X_i - \frac{1}{n} \sum_{j \leq n} \tilde X_j)^2$ and $\hat \kappa_3 \coloneqq \frac{1}{n} \sum_{i \leq n} | \tilde  X_i - \frac{1}{n} \sum_{j \leq n} \tilde X_j|^3$. Meanwhile by the Markov inequality, for any $\epsilon > 0$,
\begin{align*}
    \P\Big( \Big| \mfrac{\sum_{i=1}^m (\tilde X_i - \mean[\tilde X_1])}{\sqrt{n}} \Big| > \epsilon \Big)
    =&\,
    O\Big( \mfrac{m}{n \epsilon^2} \Big)
    \;,
    \\
    \P\Big( \Big| \hat \sigma^2 - \mfrac{(n-m) \Var[\tilde X_1] }{n} \Big| > \epsilon \Big)
    =&\,
    O\Big( \mfrac{1}{n (\epsilon + \frac{n}{m})^2} \Big)
\end{align*}
Therefore, provided that $m = o(n)$, we have that as $n \rightarrow \infty$,
\begin{align*}
    \Delta_{\rm inv}(X)
    \;=\;
    \msup_{t \in \R} 
    \big| \P( f(\Phi_1(X)) \leq t \,|\, X ) - \P( f(X) \leq t \,|\, \Xobs ) \big| 
    \;=\; o_\P(1) \;,
\end{align*}
where $o_\P(1)$ denotes convergence to zero in probability over the distribution of $X$. By \Cref{thm:Del:Kol:approx:inv} and \Cref{remark:thm:Del:Kol:approx:inv}, we obtain valid conditional coverage for bootstrap given $(\tilde X_1, \ldots, \tilde X_m)$ provided that $m=o(n)$. 
\end{example}

\begin{example}[Conditional coverage for a permutation test] \label{example:permutation} Consider a stylised setup of testing the difference of two distributions $P_1$ and $P_2$, given their respective i.i.d.~datasets $(X_{1,i})_{i \leq n}$ and $(X_{2,i})_{i \leq n}$, by a permutation test. This can be identified as DAB in \eqref{eq:DABU} with $X=(X_{1,1}, \ldots, X_{2,n})$, $f(X) = \frac{1}{\sqrt{n}} \msum_{i \leq n} \kappa(X_{1,i}, X_{2,i})$ for some kernel function $\kappa: (\R^d)^2 \rightarrow \R$ and, for $1 \leq b \leq B$, 
\begin{align*}
    \Phi_b(x_{1,1}, \ldots, x_{2,n}) \;=\; ( x_{\pi_b(1,1)}, \ldots, x_{\pi_b(2,n)} )\;,
\end{align*}
where $\pi_b$'s are uniformly drawn permutations on the index set $\{(1,1), \ldots, (2,n)\}$. Under the null hypothesis $P_1=P_2$, $(f(X), f(\Phi_1(X)), \ldots, f(\Phi_B(X)))$ is marginally exchangeable, so \Cref{thm:ex:validity} implies exact marginal coverage. However, this exchangeability is generally violated when we condition on part of the dataset. On the other hand, assuming mild moment conditions on $\kappa(X_{11}, X_{21})$ under $P_1= P_2$, the same argument as \Cref{example:conditional:bootstrap} implies that asymptotic conditional guarantee holds with high probability, if the conditioning is on $\sigma(X_{11}, \ldots, X_{2m})$ for $m = o(n)$. As a consequence, one may choose the kernel function $\kappa$ based on the observations $X_{11}, \ldots, X_{2m}$.
\end{example}

\begin{example}[Coverage for leave-one-out methods] \label{example:jackknife} Consider a set ~$X=(X_i, Y_i)_{i \leq n+1}$ of i.i.d.~covariate-label pairs in $(\R^d \times \R)^{n+1}$. Let $g: (\R^d \times \R)^n \times \R^d \rightarrow \R$ be a regressor such that $g((x_i,y_i)_{i \leq n}, x_0)$ gives the prediction on $x_0$ by the regressor trained on $(x_i,y_i)_{i \leq n}$ and $g$ is invariant under permutations of the $n-1$ training data pairs. We also use $*$ to denote missing data and use $g((*, *),(x_i,y_i)_{i \leq n-1}, x_{n+1})$ to denote a regressor that is only trained on $n-1$ data. We consider the leave-one-out method studied in \cite{steinberger2016leave} (referred to as the Jackknife method by \cite{barber2021predictive}), and identify the method as DAB \eqref{eq:DABU} for predicting $Y_{n+1}$. Let $B=n$ and $(\Phi_b)_{b \leq n}$ be a deterministic enumeration induced by cyclic shifts on $[n]$:
\begin{align*}
    \Phi_b( v_1, \ldots, v_n, v_{n+1}) \;=\; ( (*,*), v_{(b \textrm{ mod } n)+1}, \ldots, v_{(b-2 \textrm{ mod } n)+1}, v_b)\;.
\end{align*}
The estimator $f$ for the leave-one-out method can be identified as 
\begin{align*}
    f(X) \;=\;  | Y_{n+1} - g((X_i,Y_i)_{i \leq n}, X_{n+1} ) |\;.
\end{align*} 
Exact exchangeability does not hold in this case since, for example,
\begin{align*}
    f(\Phi_1(X)) \;=\; \big|\, Y_1 - g((\star,\star), (X_i,Y_i)_{2 \leq i \leq n}, X_1 ) \,\big|
    \;\overset{d}{\neq}\; f(X)\;.
\end{align*} 
Nevertheless, \cite{steinberger2016leave} establishes asymptotic guarantees for stable estimators that, for example, include high-dimensional linear predictors. \cite{steinberger2016leave} shows, in a similar style as \Cref{example:bootstrap}, that $f(\Phi_1(X)) | X$ and $f(X)$ both converge to the same distribution. Under our notation, their result implies 
\begin{align*}
    \msup_{t \in \R} 
    \big| \P( f(\Phi_1(X)) \leq t \,|\, X ) - \P( f(X) \leq t ) \big| 
    \rightarrow 0 \;\; \text{ in probability}\;.
\end{align*}
This statement implies asymptotic exchangeability of $(f(X), f(\Phi_1(X)), \ldots, f(\Phi_B(X)))$ --- for example, in the sense of \Cref{thm:Del:Kol:approx:inv}, with an additional application of \Cref{lem:fixed:to:random} to connect the deterministic transformations to random transformations --- which leads to valid asymptotic coverage. On the other hand,  \cite{barber2021predictive} introduces Jackknife+, a modification of the leave-one-out method above, which is also a special case of DAB and makes use of exact exchangeability to obtain finite-sample guarantees similar to \Cref{thm:ex:validity} and the examples discussed after.
\end{example}
\vspace{-1em}

\section{An extended discussion of related work} \label{appendix:related:work}

The proposed framework of data augmented bootstrap (DAB) shares connections with a diverse range of literature, and we include a further discussion in this section. 

\vspace{.5em}

\noindent
\textbf{Data augmentation.} Data augmentations have been widely employed in machine learning tasks \citep{shorten2019survey,shorten2021text}. In terms of theoretical results in statistics, \cite{chen2020group} is one of the first theoretical works to study the effect of group-based augmentations on a class of statistical estimators, where \cite{huang2022data} is the first to study non-group augmentations without requiring invariance. Notably, \cite{huang2022data} also employs Gaussian universality with dependence, and their theoretical results have inspired this work. Yet, both \cite{chen2020group} and \cite{huang2022data} focus on the effect of data augmentations on estimator quality, and neither has considered its application to uncertainty quantification. SymmPI \citep{dobriban2023symmpi} makes the novel observation that proof techniques for conformal prediction can be extended to general groups, and applies group-based transformations for CI constructions under the condition of exact invariance or equivariance. However, both the group assumption and the invariance assumption exclude many commonly used augmentations. An interesting recent work \citep{wu2024posterior} explores incorporating data augmentations for uncertainty quantification via the Bayesian posteriors. While their Bayesian approach is fundamentally different from DAB, an interesting follow-up direction could be exploring how DAB may be combined with Bayesian approaches for uncertainty quantification, especially in domain-specific problems.

\vspace{.5em}

\noindent
\textbf{Invariance-based inference.} There is already a large body of literature on invariance-based inference and in particular on conformal prediction; a non-exhaustive list of literature has been referenced in \Cref{sec:intro}. Of particular relevance to our work are those that have considered a violation of exchangeability: For conformal prediction, many works (e.g.~\cite{tibshirani2019conformal,barber2023conformal,ramdas2023permutation,guan2023localized,prinster2024conformal}) have considered addressing the violation of exact exchangeability with reweighting techniques, with successful applications to local heterogeneous effects, distributional shifts and many more settings. As discussed in \Cref{remark:weighted:ex}, an interesting avenue of future work is to explore the combination of DAB and weighted ranks for similar setups. When not enough distributional knowledge is known to design exact and computable weights, the incurred coverage error has been controlled in terms of total variation distance \citep{barber2023conformal} or Wasserstein distance \citep{xu2025wasserstein}. In contrast, our \Cref{thm:ex:validity} has controlled the coverage error through a c.d.f.~difference that can be turned into a Kolmogorov distance comparison between $f(\Phi_1(X)) | X$ and $f(X)$. This is what enables Gaussian universality tools to be applied and the connection to bootstrap results to be established. The application of Gaussian universality to compare  $f(\Phi_1(X)) | X$ and $f(X)$ is also heavily inspired by the Gaussian universality toolkits in \cite{austern2020bootstrap}, which establishes bootstrap consistency for estimators beyond empirical averages and can be viewed as a special case of our \Cref{thm:universality}.

\vspace{.5em}

In the context of permutation tests, the application of asymptotic theory to study coverage validity under the violation of exchangeability is not new: For example, \cite{janssen1994weighted} studies this for particular classes of studentised statistics, whereas \cite{canay2017randomization} considers this under approximate group invariance assumptions. In contrast, our results do not assume a group structure, and address a broader class of estimators, i.e.~those for which Gaussian universality applies; see a recent work \citep{huang2024gaussian} for a characterisation of such estimators as well as results on when universality fails for high-dimensional data. Meanwhile, as discussed in \Cref{sec:discussion}, we conjecture that the DAB framework can be extended to settings beyond Gaussian universality.

\vspace{.5em}

\noindent
\textbf{Gaussian universality.} Over the past decade, a wave of results have established Gaussian universality theoretically and empirically for many estimators found in high-dimensional statistics and machine learning, including but not limited to: random feature models \citep{hu2022universality},  regularised regression \citep{han2023universality}, generalised linear models \citep{dandi2023universality}, perceptron models \citep{gerace2024gaussian}, max-margin classifiers \citep{montanari2023universality} and general classes of empirical risk minimisers \citep{montanari2022universality}; see \cite{huang2025universality} for a review of the founding probability works in this area as well as recent results. Our universality result adapts the techniques of \cite{mallory2025universality} to accommodate local dependence as well as the conditioning introduced in our measurement of approximation invariance error. A related work \citep{gibbs2025characterizing} obtains training-conditional coverage of conformal prediction by concentration results related to the training data in the proportional regime. We note that their result addresses a different problem from ours, since our dataset $X$ of interest corresponds to the calibration data and the new data point in the context of conformal prediction. In particular, as discussed in \Cref{sec:conformal}, our universality result does \emph{not} apply to conformal prediction due to a lack of stability, even though our approximate invariance result does.

\section{Additional results and proofs for \Cref{sec:validity:exact:inv}: Approximate exchangeability} \label{appendix:validity:exact:inv}

\subsection{Coverage result under general tie-breaking function} \label{appendix:coverage:validity}

We provide a generalisation of \Cref{thm:ex:validity} that accommodates arbitrary tie-breaking functions. This requires additional notation, which we state next. 

\vspace{.5em}

\noindent
\textbf{Random quantities to be ranked. } Given a subset $\cY' \subseteq  \cY$, we denote its associated multiset by $\lbag \cY' \rbag$, where only the unique values in $\cY'$ and the associated multiplicities matter. As shorthands, we denote
\begin{align*}
    V \;\coloneqq\; (V_0, V_1, \ldots, V_B)\;, 
    \qquad 
    V_b \;\coloneqq\; f(\Psi_b(X)) 
    \qquad 
    \text{ for } b \in [B]_0\;,
\end{align*}
which concern the $\cY$-valued random variables to be ranked in the definition of $R_\Psi(X)$.
Also write $V^*_0 \leq \ldots \leq V^*_B$ as the order statistics of $V_0, \ldots, V_B$. Note that $(V^*_b)_{b \in [B]_0}$ is $\lbag V \rbag$-measurable. 

\vspace{.5em}

\noindent
\textbf{Ranks of $(V^*_b)_{b \in [B]_0}$ to account for ties. } It is convenient to define functions that indicate the range of indices of $V^*_0, \ldots, V^*_B$ that are at a tie. To this end, for any fixed $k \in [B]_0$, we consider the random variables 
\begin{align*}
    R^*_-(k) 
    \;\coloneqq&\;  
    \min\{ b \in [B]_0 \;|\; V^*_b = V^*_k \} 
    \;=\;
    \msum_{b \in [B]_0} \ind\{ V^*_k > V^*_b \} 
    \;,
    \\
    R^*_+(k) \;\coloneqq&\; 
    \max\{ b \in [B]_0 \;|\; V^*_b = V^*_k \} 
    \;=\;
    \msum_{b \in [B]_0} \ind\{ V^*_k \geq  V^*_b \} - 1
    \;.
    \tagaligneq \label{eq:rank:fn:alt}
\end{align*}
$R^*_-(k)$ can be viewed as a zero-indexed rank of $V^*_k$ among $(V^*_b)_{b \in [B]_0}$ that breaks ties always by preferring the first variable that occurs at a tie, whereas $R^*_+(k)$ can be viewed as a similar rank except that ties are always broken by preferring the last variable that occurs at a tie. $R^*_-(k)$ and $R^*_+(k)$ are completely determined by $\lbag V \rbag$. As an example, when $V=(10,10,30,40,40)$, $R^*_-(2) = R^*_+(2) = 2$ --- the index of $30$ which is not at tie with any other numbers --- whereas $R^*_-(0) = 0 < R^*_+(0) = 1$ --- the starting and ending indices of $10$, as there are two numbers $10$ at a tie. Also note the useful property that for $k, k' \in [B]_0$,
\begin{align*}
    R^*_+(k) \;<\; k'
    \quad\Leftrightarrow&\quad 
    &V^*_k \;<\; V^*_{k'}&&
    \quad\Leftrightarrow&\quad 
    k \;<\; R^*_-(k')
    \tagaligneq \label{eq:rank:fn:property:one}
    \;.
\end{align*}
The next result generalises \Cref{thm:ex:validity} to an arbitrary tie-breaking function $T$ in \eqref{eq:rank}.

\begin{theorem} \label{thm:group:guarantee:general:tie} Let $f$ be an $\cX_n \rightarrow \cY$ function for a generic measurable space $\cY$ equipped with total order, and $T$ be a random $\N \cup \{0\} \rightarrow \R^+ \cup \{0\}$ tie-breaking function independent of all other variables such that $T(n) \in [0,n]$ for all $n \in \N \cup \{0\}$. For $r \in (0,1]$, additionally denote $B_r = \lceil r (B+1) \rceil - 1$ and 
\begin{align*}
    C_T(r)
    \;\coloneqq&\; 
    \mfrac{R^*_-(B_r)}{B + 1}
    +
    \P \Big( T(  R^*_+(B_r) - R^*_-(B_r) + 1 ) <  r(B+1) - R^*_-(B_r) \,\big|\,  \textstyle{\lbag V \rbag} \Big) 
    \\
    &\hspace{5em}\,\times\,
    \mfrac{R^*_+(B_r) - R^*_-(B_r) + 1}{B + 1}
    \;,
\end{align*}
and use the convention that $B_0 = C_T(0) = 0$.
If $(f(\Psi_0(X)), \ldots, f(\Psi_B(X)))$ is conditionally exchangeable given $\Xobs$, then for any $r \in [0,1]$, almost surely,
\begin{align*} 
    \Big|
    \,
    \P\Big(
    \mfrac{R_{\Phi}(X)}{B+1}
    \,<\,
    r
    \,\Big|\, \Xobs
    \Big)
    \,-\, 
    \mean[ C_T(r)  \,|\, \Xobs]
    \,
    \Big|
    \;\leq&\;  
    \Delta^r_{\rm Kol}(\Xobs)
    \;.
\end{align*}
\end{theorem}

\subsection{Proof of \cref{thm:ex:validity}} 
We shall apply \Cref{thm:group:guarantee:general:tie} with $\cY=\R$ and calculate $C_T$ under the two choices of tie-breaking. Note that for either tie-breaking function, when $r=0$, we trivially have 
\begin{align*}
    \mean[C_T(0) | \Xobs] \;=\; 0 \;=\; r 
    \quad 
    \text{ almost surely. }
\end{align*}

\vspace{.5em}
For uniform tie-breaking, $T(n) \sim \text{Uniform}\{0, \ldots, n\}$. Then for $r \in (0,1]$,
    \begin{align*}
        C_T(r) 
        \;=&\; 
        \mfrac{R^*_-(B_r)}{B + 1}
        +
        \mfrac{\lceil r(B+1) - R^*_-(B_r) \rceil}{R^*_+(B_r) - R^*_-(B_r) + 2}
        \,
        \mfrac{R^*_+(B_r) - R^*_-(B_r) + 1}{B + 1}
        \\
        \;\overset{(a)}{=}&\;
        \mfrac{\lceil r(B+1) \rceil}{B + 1}
        -
        \mfrac{B_r + 1 - R^*_-(B_r)}{R^*_+(B_r) - R^*_-(B_r) + 2}
        \,
        \mfrac{1}{B + 1}
        \qquad \text{ almost surely }\;.
    \end{align*}
In $(a)$ above, we have noted that since $R^*_-(B_r)$ takes integer values, $R^*_-(B_r) + \lceil r(B+1) - R^*_-(B_r)  \rceil = \lceil r(B+1) \rceil$, as well as the notation that $B_r =\lceil r(B+1) \rceil - 1$. Now since, by construction, $R^*_-(B_r) \leq B_r \leq R^*_+(B_r)$, we get that 
\begin{align*}
    \mfrac{B_r + 1 - R^*_-(B_r)}{R^*_+(B_r) - R^*_-(B_r) + 2} \;\in\; [0,1]\;,
\end{align*}
and also note that $\lceil r(B+1) \rceil - r(B+1) \in [0,1]$. This implies that 
\begin{align*}
    \big| C_T(r) - r \big| \;\leq\; \mfrac{1}{B+1}
    \qquad \text{ almost surely }\;,
\end{align*}
and by  \Cref{thm:group:guarantee:general:tie}, we obtain
\begin{align*}
    &\;
    \Big|
    \,
    \P\Big(
    \mfrac{R_{\Phi}(X)}{B+1}
    \,<\,
    r
    \,\Big|\, \Xobs
    \Big)
    \,-\, 
    r
    \,
    \Big|
    \\
    &\;\leq\; 
    \Big| 
    \,
    \P\Big(
    \mfrac{R_{\Phi}(X)}{B+1}
    \,<\,
    r
    \,\Big|\, \Xobs
    \Big)
    \,-\, 
    \mean[ C_T(r) \,|\, \Xobs]
    \,
    \Big| 
    + \mfrac{1}{B+1}
    \\
    &\;\leq\; \Delta^r_{\rm Kol}(\Xobs) + \mfrac{1}{B+1}\;. 
\end{align*}

\vspace{.5em}

For smoothed uniform tie breaking, $T(n) \sim \text{Uniform}[0,n]$, so 
\begin{align*}
    C_T(r) 
    \;=&\; 
    \mfrac{R^*_-(B_r)}{B + 1}
    +
    \mfrac{
         r(B+1) - R^*_-(B_r)
    }{
        R^*_+(B_r) - R^*_-(B_r) + 1
    }
    \,
    \mfrac{R^*_+(B_r) - R^*_-(B_r) + 1}{B + 1}
    \\
    \;=&\;
    r
    \quad \text{ almost surely.}
\end{align*}
By  \Cref{thm:group:guarantee:general:tie}, we obtain
\begin{align*}
    \Big|
    \,
    \P\Big(
    \mfrac{R_{\Phi}(X)}{B+1}
    \,<\,
    r
    \,\Big|\, \Xobs
    \Big)
    \,-\, 
    r
    \,
    \Big|
    \;=&\; 
    \Big| 
    \,
    \P\Big(
    \mfrac{R_{\Phi}(X)}{B+1}
    \,<\,
    r
    \,\Big|\, \Xobs
    \Big)
    \,-\, 
    \mean[ C_T(r) \,|\, \Xobs]
    \,
    \Big| 
    \\
    \;\leq&\; \Delta^r_{\rm Kol}(\Xobs)\;. 
\end{align*}
This finishes the proof. 
\color{black}
\qed

\subsection{Proof of \Cref{thm:group:guarantee:general:tie}} Recall the shorthands that 
\begin{align*}
    R_{\Phi}(X)  
    \;\coloneqq&\;
    \Rank_T\big(f(X) \,;\, (f(\Phi_b(X)))_{b \in [B]} \big) 
    \;,
    \\
    R_{\Psi}(X) 
    \;\coloneqq&\;
    \Rank_T\big(f(\Psi_0(X)) \,;\, (f(\Psi_b(X)))_{b \in [B]} \big)\;,
    \\ 
    \Delta^r_{\rm Kol}(\Xobs)
    \;\coloneqq&\; 
    \Big| \,
        \P\Big( \mfrac{R_{\Phi}(X)}{B+1} < r \,\Big|\, \Xobs \Big)   
        \,-\,
        \P\Big( \mfrac{R_{\Psi}(X)}{B+1} < r \,\Big|\, \Xobs \Big)   
    \, \Big|\;.
\end{align*}
Therefore by the triangle inequality, the quantity to control can be bounded as 
\begin{align*}
    \Big|
    \,
    \P\Big(
    \mfrac{R_{\Phi}(X)}{B+1}
    \,<\,&
    r
    \,\Big|\, \Xobs
    \Big)
    \,-\, 
    \mean\big[ C_T(r) \,|\, \Xobs \big] 
    \,
    \Big|
    \\
    \;\leq&\;
    \;\Delta^r_{\rm Kol}(\Xobs) 
    \,+\,
    \Big|
    \,
    \P\Big(
    \mfrac{R_{\Psi}(X)}{B+1}
    \,<\,
    r
    \,\Big|\, \Xobs
    \Big)
    \,-\, 
    \mean\big[ C_T(r) \,|\, \Xobs \big] 
    \,
    \Big|
    \\
    \;=&\;
    \Delta^r_{\rm Kol}(\Xobs) 
    \,+\,
    \Big|
    \mean\Big[
    \,
    \P\Big(
    \mfrac{R_{\Psi}(X)}{B+1}
    \,<\,
    r
    \,\Big|\, \Xobs \,,\, \textstyle{\lbag V \rbag}
    \Big)
    \,-\, 
    C_T(r) 
    \,
    \,\Big|\, \Xobs  \Big]
    \,
    \Big| 
    \\
    \;\leq&\;
    \Delta^r_{\rm Kol}(\Xobs) 
    \,+\,
    \mean\Big[
        \;
    \Big|
    \,
    \P\Big(
    \mfrac{R_{\Psi}(X)}{B+1}
    \,<\,
    r
    \,\Big|\, \Xobs \,,\, \textstyle{\lbag V \rbag}
    \Big)
    \,-\, 
    C_T(r) 
    \,
    \Big| 
    \;
    \,\Big|\, \Xobs  \Big]
    \;,
\end{align*}
where we noted that $C_T(r)$ is $\lbag V \rbag$-measurable. The main efforts of this proof focus on showing that, almost surely
\begin{align*}
    \P\Big(
    \mfrac{R_{\Psi}(X)}{B+1}
    \,<\,
    r
    \,\Big|\, \Xobs \,,\, \textstyle{\lbag V \rbag}
    \Big)
    \;=\; 
    C_T(r) 
    \;.
    \tagaligneq 
    \label{eq:group:general:tie:key:result}
\end{align*}
The proof strategy adapts the idea of Lemma 8.7 of \cite{vovk2005algorithmic} and extends it from smoothed uniform tie-breaking to general tie-breaking. To this end, it is convenient to denote the ranks 
\begin{align*}
    R_- \;\coloneqq&\; 
    \msum_{b \in [B]_0} \ind_{\{ V_0 > V_b \}} 
    \;=\;
    \msum_{b \in [B]} \ind_{\{ V_0 > V_b \}} 
    \;,
    \\
    R_+ 
    \;\coloneqq&\; 
    \msum_{b \in [B]_0} \ind_{\{ V_0 \geq V_b \}}  - 1
    \;=\;
    \msum_{b \in [B]} \ind_{\{ V_0 \geq V_b \}} 
    \;,
\end{align*}
where we have used $[B]_0 = [B] \cup \{0\}$ above. With this notation, we can express
\begin{align*}
    R_\Psi(X)
    \;=\; 
    \Rank_T\big(V_0 \,;\, (V_b)_{b \in [B]} \big)
    \;=&\;
    \msum_{b \in [B]}
    \ind_{\{ V_0 > V_{b} \}} 
    + 
    T 
    \Big( \msum_{b \in [B]}  \ind_{\{ V_0 = V_{b}\}}  + 1 \Big)
    \\
    \;=&\; 
    R_- + T ( R_+ - R_- + 1)\;.
\end{align*}
Since $T(n) \leq n$ by assumption, $R_\Psi(X) \in [R_-, R_+ + 1]$. The goal is to compute, for $r \in [0,1]$, the quantity
\begin{align*}
    \P( R_\Psi(&X) < r(B+1) \,|\, \Xobs , \textstyle{\lbag V \rbag}  ) 
    \\
    \;\overset{(a)}{=}&\;
    \P\big(\,
    R_\Psi(X) \leq R_+ + 1 < r(B+1)
    \;\big|\,  \Xobs, \textstyle{\lbag V \rbag} \big)
    \\
    &\;
    +
    \P\big(\,
    R_- \leq R_\Psi(X) < r(B+1) \leq R_+ + 1
    \;\big|\,  \Xobs , \textstyle{\lbag V \rbag} \big) 
    \\
    \;\overset{(b)}{=}&\; 
    \P\big(\,
    R_+ < r(B+1) - 1
    \;\big|\,  \Xobs, \textstyle{\lbag V \rbag} \big)
    \\
    &\;
    +
    \mean\Big[
    \P( R_\Psi(X) < r (B+1) \,|\,  \Xobs, \textstyle{\lbag V \rbag}, R_-, R_+ )
    \, 
    \ind_{\{R_- < r(B+1) \leq R_+ + 1\}}
    \;\Big|\,  \Xobs , \textstyle{\lbag V \rbag} \Big]
    \\
    \;\overset{(c)}{=}&\; 
    \P\big(\,
    R_+ < r(B+1) - 1
    \;\big|\,  \Xobs, \textstyle{\lbag V \rbag} \big)
    \\
    &\; \qquad +
    \mean\Big[
    \; p_T\big(r(B+1) - R_-\,,\, R_+ - R_- + 1 \big)
    \, 
    \ind_{\{R_- < r(B+1) \leq R_+ + 1\}}
    \;\Big|\,  \Xobs , \textstyle{\lbag V \rbag} \Big]\;. \tagaligneq \label{eq:rank:bound:intermediate}
\end{align*}
Note that we have used $R_- \leq R_\Psi(X) \leq R_+ + 1$ almost surely in $(a)$ and the tower rule in $(b)$; in $(c)$, we have taken an expectation over the randomness in $T$ and written 
\begin{align*}
    p_T( x \,,\, n ) \;\coloneqq\; \P( T(n) < x )\;.
\end{align*}
Also note that in the case $r=0$, we have that trivially
\begin{align*}
    \P( R_\Psi(X) < r(B+1) \,|\, \Xobs , \textstyle{\lbag V \rbag}  ) 
    \;=\;
    0
    \;=\;
    \mean[C_T(r) | \Xobs ] 
    \quad 
    \text{ almost surely,}
\end{align*}
so it suffices to focus on the case $r \in (0,1]$.

\vspace{.5em}

To bound \eqref{eq:rank:bound:intermediate}, we shall provide a finer characterisation of the conditional distribution of $R_+$ and $R_-$ given $\Xobs$ and $\lbag V \rbag = \lbag V_0, \ldots, V_B \rbag$.  Denote $V^*=(V^*_0, \ldots, V^*_B)$, the sequence of ordered statistics. Write $V_{-b}$ as the sequence $V$ but with the $b$-th element omitted, and $V^*_{-b}$ as the sequence $V^*$ but with the $b$-th element omitted. Since $\Psi_0, \ldots, \Psi_B$ are exchangeable, $(V_0, \ldots, V_B)$ are exchangeable given $\Xobs$ and thereby exchangeable given $\Xobs$ and $\lbag V \rbag$. In particular, $(V_0, \lbag V_{-0} \rbag) \,|\, (\Xobs, \lbag V \rbag)$ is conditionally uniformly distributed over the size-$(B+1)$ list $\{(V^*_0, \lbag V^*_{-0} \rbag) \,,\, \ldots \,,\, (V^*_B, \lbag V^*_{-B} \rbag) \}$. This implies that
\begin{align*}
    (R_-, R_+)
    \;=\;
    \big( \msum_{v \in \lbag V \rbag} \ind\{ V_0 > v\}  \,,\,  \msum_{v \in \lbag V \rbag} \ind\{ V_0 \geq v\} - 1 \big) 
\end{align*}
can be generated by sampling $U \sim \textrm{Uniform}\{0, \ldots, B\}$ and setting 
\begin{align*}
    R_- 
    \;=&\; 
    \msum_{v \in \lbag V \rbag} \ind_{\{ V^*_U > v \}} 
    \;=\;
    \msum_{b \in [B]_0} \ind_{\{ V^*_U > V^*_b \}} 
    \;\overset{\eqref{eq:rank:fn:alt}}{=}\;
    R^*_-(U)
    \;,
    \\
    R_+ 
    \;=&\; 
    \msum_{v \in \lbag V \rbag} \ind_{\{ V^*_U \geq v \}}  -  1
    \;=\;
    \msum_{b \in [B]_0} \ind_{\{ V^*_U \geq V^*_b \}}  - 1
    \;\overset{\eqref{eq:rank:fn:alt}}{=}\;
    R^*_+(U)
    \;.
\end{align*}
This gives the following distributional controls on $R_+$ and $R_-$: Almost surely,
\begin{align*}
    \P\big(\,
    R_+ < r(B+1) - 1
    \;\big|\,  \Xobs, \textstyle{\lbag V \rbag} \big)
    \;=&\;
    \P\big(\,
    R^*_+(U)  < \lceil r(B+1) \rceil - 1
    \;\big|\,  \Xobs, \textstyle{\lbag V \rbag} \big)
    \\
    \;\overset{\eqref{eq:rank:fn:property:one}}{=}&\; 
     \P\big(\,
    U < R^*_-(B_r)
    \;\big|\,  \Xobs, \textstyle{\lbag V \rbag} \big)
    \\
    \;=&\; 
    \mfrac{R^*_-(B_r)}{B + 1} \;,  \tagaligneq \label{eq:rank:rplus}
    \\
    \P\big(\,
    R_- \geq r(B+1)
    \;\big|\,  \Xobs, \textstyle{\lbag V \rbag} \big)
    \;=&\;
    \P\big(\,
    R^*_-(U) > \lceil r(B+1) \rceil - 1
    \;\big|\,  \Xobs, \textstyle{\lbag V \rbag} \big)
    \\
    \;\overset{\eqref{eq:rank:fn:property:one}}{=}&\; 
    \P\big(\,
    U  > R^*_+(B_r)
    \;\big|\,  \Xobs, \textstyle{\lbag V \rbag} \big)
    \\
    \;=&\; \mfrac{B - R^*_+(B_r)}{B + 1} \;.  \tagaligneq \label{eq:rank:rminus} 
\end{align*}
We have used that $B_r = \lceil r(B+1) \rceil - 1$ and that $R^*_-(B_r)$ and $R^*_+(B_r)$ are both $\lbag V \rbag$-measurable above. Notice that \eqref{eq:rank:rplus} gives the expression for the first term of \eqref{eq:rank:bound:intermediate}. To compute the second term of \eqref{eq:rank:bound:intermediate}, note that on the event $R^*_-(U) < r(B+1) \leq R^*_+(U) + 1$, 
\begin{align*}
    \min\{ b \in [B]_0 \;|\; V^*_b = V^*_U \} - 1
    \;<\;
    r(B+1) - 1
    \;\leq\;
    \max\{ b \in [B]_0 \;|\; V^*_b = V^*_U \}
\end{align*}
and therefore 
\begin{align*}
    \min\{ b \in [B]_0 \;|\; V^*_b = V^*_U \}
    \;\leq\;
    \lceil r(B+1) \rceil - 1 \;=\; B_r 
    \;\leq\;
    \max\{ b \in [B]_0 \;|\; V^*_b = V^*_U \}
    \;,
\end{align*}
which implies 
\begin{align*}
    R^*_-(U) 
    \;=&\;
    R^*_-(B_r)
    &\text{ and }&&
    R^*_+(U) 
    \;=&\; 
    R^*_+(B_r)
    \;,
\end{align*}
both of which are constant almost surely given $\lbag V \rbag$. Therefore, almost surely
\begin{align*}
    \mean\Big[ &
    \, p_T(r(B+1) - R_-\,,\, R_+ - R_- + 1)
    \, 
    \ind_{\{R_- < r(B+1) \leq R_+ + 1\}}
    \;\Big|\,  \Xobs , \textstyle{\lbag V \rbag} \Big]
    \\
    \;=&\; 
    p_T \big(r(B+1) - R^*_-( B_r)\,,\, R^*_+(B_r) - R^*_-(B_r) +1  \big) 
    \\
    &\qquad \qquad  
    \,\times\,
    \P\big( R_- < r (B+1) \leq R_+ + 1 \,\big|\,  \Xobs , \textstyle{\lbag V \rbag} \big)
    \\
    \;\overset{(a)}{=}&\; 
    p_T \big(r(B+1) - R^*_-( B_r)\,,\, R^*_+(B_r) - R^*_-(B_r) + 1 \big) 
    \\
    &\qquad  \times
    \Big( 1 
    - 
    \P\big( R_- \geq r(B+1)  \,\big|\, \Xobs , \textstyle{\lbag V \rbag} \big)
    -
    \P\big( R_+ < r(B+1) - 1 \,\big|\,  \Xobs , \textstyle{\lbag V \rbag} \big)
    \Big)
    \\
    \;\overset{(b)}{=}&\; 
    \P \Big( T\big( R^*_+(B_r) - R^*_-(B_r) + 1\big) <  r(B+1) - R^*_-(B_r) \,\big|\,  \textstyle{\lbag V \rbag} \Big) \,
    \mfrac{R^*_+(B_r) - R^*_-(B_r) + 1}{B + 1}
    \;.
\end{align*}
In $(a)$, we have noted that $R_- \leq R_+$ almost surely, and in $(b)$, we have used \eqref{eq:rank:rplus} and \eqref{eq:rank:rminus} as well as the definition of $p_T$. Combining this with \eqref{eq:rank:rplus} again, almost surely
\begin{align*}
    &\P( R_\Psi(X) < rB \,|\, \Xobs , \textstyle{\lbag V \rbag}  ) 
    \;=\; 
    \eqref{eq:rank:bound:intermediate}
    \\
    &\;=\; 
    \mfrac{R^*_-(B_r)}{B + 1}
    \\
    &\quad
    +
    \P \Big( T(  R^*_+(B_r) - R^*_-(B_r) + 1 ) <  r(B+1) - R^*_-(B_r) \,\big|\,  \textstyle{\lbag V \rbag} \Big) \,
    \mfrac{R^*_+(B_r) - R^*_-(B_r) + 1}{B + 1}
    \\
    &\;=\; C_T(r)
    \;,
\end{align*}
which proves \eqref{eq:group:general:tie:key:result} and therefore finishes the proof. \qed 

\section{Additional results and proofs for \Cref{sec:validity:ex:to:inv}: Approximate invariance} 
\label{appendix:thm:Del:Kol:approx:inv}

\subsection{Variants of \Cref{thm:Del:Kol:approx:inv}} \label{appendix:thm:Del:Kol:approx:inv:variants}

\vspace{.5em}

As mentioned at the end of \Cref{sec:validity:ex:to:inv}, \Cref{thm:Del:Kol:approx:inv} can be extended to allow for obtaining conditional coverage on $\Xobs$, and we can additionally relax the conditioning on $X$ in $\P(f(\Phi_1(X)) \leq t | X )$ to a general $\sigma$-algebra. In this section, we state a result that includes both extensions. We first state a relaxation of \Cref{asst:iid:smooth}:

\begin{assumption} \label{asst:iid:smooth:general} 
(i) $f(\Phi_1(X)), \ldots, f(\Phi_B(X))$ are conditionally i.i.d.~and independent of $f(X)$ given a $\sigma$-algebra $\A$ with $\Xobs \subseteq \A \subseteq \sigma(X)$; (ii) smoothed uniform tie-breaking is used;  (iii) the output space of $f$, $\cY$, is taken to be $\R$.
\end{assumption}

\noindent
The corresponding measure of approximate invariance is then
\begin{align*}
    \Delta_{\rm inv}(\A)
    \;=\;
    \msup_{t \in \R}
    \big|  
        \P( f(\Phi_1(X)) \leq t \,|\, \A ) 
        \,-\, 
        \P( f(X) \leq t \,|\, \Xobs )  
    \big| 
    \;.
\end{align*} 
We also define the conditional $L_\nu$ norm as 
\begin{align*}
    \| \argdot \|_{L_\nu | \Xobs} \;\coloneqq\; (\mean[ | \argdot |^\nu \,|\, \Xobs ])^{1/\nu}\;.
\end{align*}
The corresponding relaxation of \Cref{thm:Del:Kol:approx:inv} reads: 

\begin{theorem} \label{thm:Del:Kol:approx:inv:general} Let $r \in [0,1]$, $\nu \geq 1$ and $B \geq 2$. Under Assumption \ref{asst:iid:smooth:general}, almost surely
\begin{align*}
     &\Big| \P\Big(  \mfrac{R_\Phi(X)}{B+1}  < r \,\Big|\, \Xobs \Big)  
        \,-\,
        r
    \Big|
    \\
    &
    \;\leq\;
    \mfrac{12}{4^{1/(\nu+1)}}
    \,
    \big( 3 \, \| \Delta_{\rm inv}(\A) \|_{L_\nu \,|\, \Xobs} \big)^{\frac{\nu}{\nu+1}} 
    + 
    \mfrac{1}{B-1} 
    +  
    \mfrac{\sqrt{3 \log B}}{\sqrt{B-1}} 
    \;.
\end{align*}
\end{theorem}

\Cref{thm:Del:Kol:approx:inv} follows as an immediate corollary by setting $\A = \sigma(X)$ and $\Xobs = \A_{\rm trivial}$. To prove \Cref{thm:Del:Kol:approx:inv:general}, in view of \Cref{thm:ex:validity}, the starting point is to construct an exchangeable sequence $(f(\Psi_b(X)))_{0 \leq b \leq B}$, which are then compared to the original sequence $(f(X), f(\Phi_1(X)), \ldots, f(\Phi_B(X)))$. Our choice in this proof is 
\begin{align*}
    (f(\Psi_b(X)))_{0 \leq b \leq B} 
    \qquad \text{ are conditionally i.i.d.~distributed as } \qquad  f(X) \,|\, \Xobs\;.
    \tagaligneq \label{eq:choice:Psi}
\end{align*}
Another key element of the proof is the observation that, under \Cref{asst:iid:smooth:general}(ii), we can write the tie-breaking function in \eqref{eq:rank} as 
\begin{align*}
    T(n) \;=\; U n 
    \qquad \text{ for }\; U \sim \textrm{Uniform}[0,1]\;.
\end{align*}
In particular, this allows us to write the normalised rank variable as
\begin{align*}
    \mfrac{R_\Phi(X)}{B+1} 
    \;=&\;
    \mfrac{1}{B+1} \msum_{b=1}^B 
    \Big( 
        \ind_{\{ f(X) > f(\Phi_b(X)) \}} 
        +
        U \,
        \ind_{\{ f(X) = f(\Phi_b(X)) \}} 
    \Big)
    +
    \mfrac{U}{B+1}
    \;.
\end{align*}
Conditioning on the quantities 
\begin{align*}
    U\;,
    \qquad 
    \P(f(X) > f(\Phi_1(X)) \,|\, \A, f(X))\;,
    \qquad 
    \P(f(X) = f(\Phi_1(X)) \,|\, \A, f(X))\;,
\end{align*}
this is an empirical average of conditionally i.i.d.~random variables supported on three points. \Cref{appendix:three:point} provides tools that carefully controls the distribution of an empirical average of three-point random variables, and \Cref{proof:thm:Del:Kol:approx:inv:general} employs those tools to prove \Cref{thm:Del:Kol:approx:inv:general}.

\subsection{Distribution of sums of three-point variables} \label{appendix:three:point}

The key technical ingredient of the proof is the family of 3-point random variables, parameterised by $p, q, u \in [0,1]$ and defined as 
\begin{align*}
    V
    \;\coloneqq\; 
    \begin{cases}
        1 &\text{ with probability }
        p \;,
        \\
        u &\text{ with probability }
        q\;,
        \\
        0 &\text{ with probability } 1-p-q \;.
    \end{cases}
\end{align*}
We write $V_1, V_2, \ldots$ as i.i.d.~copies of $V$, and 
\begin{align*}
    \bar V_B \;\coloneqq\; \mfrac{1}{B} \msum_{b=1}^B  V_b\;.
\end{align*}
The next lemma concerns the derivatives of the c.d.f.~of $\bar V_B$ with respect to $p$ and $q$.  Throughout the proof, we also denote the multinomial coefficient
\begin{align*}
    \mbinom{B}{m,\, k} \;\coloneqq&\; \mfrac{B!}{m! k! (B-m-k)!}
    \;.
\end{align*}

\begin{lemma} \label{lem:derivative:cdf:V} Let $r \in \R$ and $B \geq 2$. Then 
\begin{align*}
    \partial_p \, \P(\bar V_B < r )
    \;=&\;
    - B \, \P\Big( \bar V_{B-1} \in \Big[ \mfrac{Br - 1}{B-1} , \mfrac{Br}{B-1} \Big) \Big)
    \;,
    \\
    \partial_q \, \P(\bar V_B < r )
    \;=&\;
    - B \, \P\Big( \bar V_{B-1} \in \Big[ \mfrac{Br- u }{B-1}, \mfrac{Br}{B-1} \Big) \Big)
    \;.
\end{align*}
    
\end{lemma}

\begin{proof}[Proof of \Cref{lem:derivative:cdf:V}] We first express the c.d.f.~as a combinatorial sum:
\begin{align*}
    \P(\bar V_B < r )
    \;=&\;
    \msum_{m=0}^B 
    \msum_{k=0}^{B-m}
    \mbinom{B}{m,k}
    p^k q^m (1-p-q)^{B-m-k} 
    \, \ind_{\{ k + um < Br \}}
    \;.
\end{align*}
Differentiating with respect to $p$, we obtain 
\begin{align*}
    \partial_p \, &\, \P(\bar V_B < r )
    \\
    \;=&\;
    \sum_{m=0}^{B-1} 
    \sum_{k=1}^{B-m}
    \mfrac{B!}{m! (k-1)! (B-m-k)!}
    \, p^{k-1} q^m (1-p-q)^{B-m-k} 
    \, \ind_{\{ k + um < B r \}}
    \\
    &\;
    -
    \sum_{m=0}^{B-1} 
    \sum_{k=0}^{B-m-1}
    \mfrac{B!}{m! k! (B-m-k-1)!}
    \, p^k q^m (1-p-q)^{B-m-k-1} 
    \, \ind_{\{ k + um < B r \}}
    \\
    \;\overset{(a)}{=}&\;
    \sum_{m=0}^{B-1} 
    \sum_{k=0}^{B-m-1}
    \mfrac{B!}{m! k! (B-m-k-1)!}
    \, p^k q^m (1-p-q)^{B-m-k-1} 
    \, \ind_{\{ k + 1 + um < B r \}}
    \\
    &\;
    -
    \sum_{m=0}^{B-1} 
    \sum_{k=0}^{B-m-1}
    \mfrac{B!}{m! k! (B-m-k-1)!}
    \, p^k q^m (1-p-q)^{B-m-k-1} 
    \, \ind_{\{ k + um < B r \}}
    \\
    \;=&\;
    -
    \sum_{m=0}^{B-1} 
    \sum_{k=0}^{B-m-1}
    \mfrac{B!}{m! k! (B-m-k-1)!}
    \, p^k q^m (1-p-q)^{B-m-k-1} 
    \, \ind_{\{ k + um \in [B r-1, B r)  \}}
    \\
    \;=&\;
    - B \, \P\Big( \bar V_{B-1} \in \Big[ \mfrac{Br - 1}{B-1} , \mfrac{Br}{B-1} \Big) \Big)
    \;.
\end{align*}
In $(a)$ above, we have shifted the index of $k$ by $1$ in the first sum.
For the derivative with respect to $q$, we instead express 
\begin{align*}
    \P(\bar V_B < r )
    \;=&\;
    \msum_{m=0}^B 
    \msum_{k=0}^{B-m}
    \mbinom{B}{m,k}
    q^k p^m (1-p-q)^{B-m-k} 
    \, \ind_{\{ uk + m < Br \}}
    \;.
\end{align*}
Then the exact same calculation applies, except that the shift of index in $(a)$ causes a shift of $u$ in the indicator rather than $1$. This implies 
\begin{align*}
    \partial_q \, \P(\bar V_B < r )
    \;=\;
    - B \, \P\Big( \bar V_{B-1} \in \Big[ \mfrac{Br - u}{B-1}, \mfrac{Br}{B-1} \Big) \Big)
    \;.
\end{align*}
\end{proof}

The next lemma provides two concentration inequalities for $\bar V_B$.

\begin{lemma} \label{lem:V:conc} Let $r \in [0,1]$. If $r \geq p + u q$, then
\begin{align*}
    \P( \bar V_B \geq r )  \;\leq\; 
    e^{ - 2 B \, ( r - (p+u q))^2}
    \;.
\end{align*}
If instead $r \leq p + u q$, then
\begin{align*}
    \P( \bar V_B \leq r ) 
    \;\leq\; 
    e^{- 2 B\, (r-(p+uq))^2}
    \;.
\end{align*}

\end{lemma}

\begin{proof}[Proof of \Cref{lem:V:conc}] First consider the case $r > p+ u q $. By the Chernoff bound and that $u \leq 1$, we have 
\begin{align*}
    \P( \bar V_B \geq r ) 
    \;\leq&\;
    \minf_{t > 0} 
    \, 
    e^{- B rt }
    \,
    \mean\big[ e^{B \bar V_B t} \big]
    \\
    \;=&\;
    \minf_{t > 0} 
    \, 
    e^{-B rt }
    \,
    \big( 1 - p - q + q e^{ut} + p e^t \big)^B
    \;.
\end{align*}
For $t > 0$ and $u \in [0,1]$, by the fundamental theorem of calculus, we get that 
\begin{align*}
    e^{u t} 
    \;=\;
    \big( 1 + (e^t - 1)  \big)^u
    \;=&\; 
    1 + u  (e^t - 1) \, 
    \mint_0^1 \mfrac{1}{(1+\theta(e^t-1))^{1-u}} d\theta
    \\
    \;\leq&\;
    1 -  u + u e^t \;.
\end{align*}
This yields 
\begin{align*}
    \P( \bar V_B \geq r ) 
    \;\leq&\;
    \minf_{t > 0} 
    \, 
    e^{-B rt }
    \,
    \big( 1 - p - u q +  (p+ u q) e^t  \big)^B
    \;.
\end{align*}
Suppose $0 < p + uq < r < 1$. Since the function $x / (1-x)$ is increasing for $x < 1$, we have $\frac{p+ u q}{1-p-uq} < \frac{r}{1-r}$ and therefore
\begin{align*}
    \mfrac{r(1-p-uq)}{(1-r)(p+uq)} > 1 \;.
\end{align*}
Choosing $t = \log \frac{r(1-p-uq)}{(1-r)(p+uq)} \, > 0$, we obtain
\begin{align*}
    \P( \bar V_B \geq r ) 
    \;\leq&\;
    \Big( \mfrac{(1-r)(p+uq)}{r(1-p-uq)} \Big)^{Br}
    \,
    \Big( \mfrac{1-p-uq}{1-r} \Big)^B 
    \\
    \;=&\;
    \exp\Big( 
        B \Big( r \log \mfrac{p+uq}{r} + (1-r) \log \mfrac{1-p-uq}{1-r}   \Big)
    \Big)
    \\
    \;=&\;
    e^{-B\, {\rm D}( r \| p+uq )}
    \;,
\end{align*}
where we have defined the binary relative entropy for $r, p' \in [0,1]$ as 
\begin{align*}
    {\rm D}(r \,\|\,  p' ) \;\coloneqq\; r \log \mfrac{r}{p'} + (1-r) \log \mfrac{1-r}{1-p'} \;\geq\; 0 \;.
\end{align*}
We take $D(r \| 0 ) = +\infty$ for $r > 0$ and $D(1 \| p' ) = \log \frac{1}{p'}$, following the convention in information theory. In the edge case $0 = p+ u q < r \leq 1$, we have 
\begin{align*}
    \P( \bar V_B \geq r ) 
    \;=\; 
    0 
    \;=\;
    e^{-B\, {\rm D}( r \| 0 )}
    \;,
\end{align*}
In the other edge case $0 < p+uq < r = 1$,  we have 
\begin{align*}
    \P( \bar V_B \geq r ) 
    \;\leq&\;
    \minf_{t > 0} 
    \,
    \big( (1 - p - u q)e^{-t} + (p+ u q)  \big)^B
    \\
    \;=&\;
    (p+ u q)^B 
    \\
    \;=&\;
    e^{ - B \log \frac{1}{p+ u q} }
    \;=\;
    e^{ - B \, {\rm D}( 1 \| p+ u q )}
    \;,
\end{align*}
In all cases, this proves that, if $r > p+uq $,
\begin{align*}
    \P( \bar V_B \geq r ) 
    \;\leq\;
    e^{ - B \, {\rm D}( r \| p+u q )}
    \;.
\end{align*}
Now by Pinsker's inequality, we have 
\begin{align*}
    {\rm D}( r \| p+ u q ) 
    \;\geq\; 
    2 ( r - (p+ u q))^2
    \;,
\end{align*}
which gives the desired inequality that 
\begin{align*}
    \P( \bar V_B \geq r ) 
    \;\leq\;
    e^{ - 2 B \, ( r - p - u q)^2}
    \;.
\end{align*}

\vspace{.5em}

For the other case where $r < p + u q$, we use the Chernoff bound again and $u \geq 0$ to obtain 
\begin{align*}
    \P( \bar V_B \leq r ) 
    \;=&\;
    \P( - \bar V_B \geq - r ) 
    \\
    \;\leq&\;
    \minf_{t > 0} 
    \, 
    e^{B rt }
    \,
    \mean\big[ e^{ - B \bar V_B t} \big]
    \\
    \;=&\;
    \minf_{t > 0} 
    \, 
    e^{B rt }
    \,
    \big( 1 - p - q + q e^{-ut} + p e^{-t} \big)^B
    \;.
\end{align*}
For $t > 0$ and $u \in [0,1]$, we can again use the fundamental theorem of calculus to obtain 
\begin{align*}
    e^{-ut} 
    \;=\; 
    \big( 1 - (1 - e^{-t}) \big)^u
    \;=&\;
    1 - u (1-e^{-t}) \mint_0^1 \mfrac{1}{(1-\theta(1-e^{-t}))^{1-u}} d \theta 
    \\
    \;\leq&\;
    1 - u + u e^{-t}\;.
\end{align*}
This implies 
\begin{align*}
    \P( \bar V_B \leq r ) 
    \;\leq&\;
    \minf_{t > 0} 
    \, 
    e^{B rt }
    \,
    \big( 1 - p - u q + u q e^{-t} + p e^{-t} \big)^B
\end{align*}
The rest of the proof proceeds similarly to the case $r < p+uq$: If $0 < r < p + uq < 1$, $\frac{r}{1-r} < \frac{p+uq}{1-p-uq}$, and choosing $t = \log \frac{(1-r)(p+uq)}{r(1-p-uq) } > 0$ gives 
\begin{align*}
    \P( \bar V_B \leq r ) 
    \;\leq\;
    \Big(
        \mfrac{(1-r) (p+uq) }{r(1-p-uq)} 
    \Big)^{Br} 
    \Big( \mfrac{1-p-uq}{1-r} \Big)^B
    \;=\;
    e^{-B\, {\rm D}( r \| p+uq )}
    \;.
\end{align*}
In the edge case where $0 \leq r < p+uq = 1$,
\begin{align*}
    \P( \bar V_B \leq r ) 
    \;=\;
    0
    \;=\;
    e^{-B\, {\rm D}( 1- r \| 0 )}
    \;=\;
    e^{-B\, {\rm D}( r \| 1 )}
    \;,
\end{align*}
and in the other edge case where $0 = r < p < 1$, 
\begin{align*}
    \P( \bar V_B \leq r ) 
    \;\leq&\;
    \minf_{t > 0} 
    \, 
    \big( 1 - p - uq + uq e^{-t} + p e^{-t} \big)^B
    \\
    \;=&\;
    (1-p - u q)^B 
    \;=\; 
    e^{-B \, {\rm D}(1 \| 1-p-uq )}
    \;=\;
    e^{-B \, {\rm D}(0 \| p + u q )}
    \;.
\end{align*}
This proves that, if $r < p + u q$, we have 
\begin{align*}
    \P( \bar V_B \leq r ) 
    \;\leq\; 
    e^{-B\, {\rm D}( r \| p + u q )}
    \;\leq\;
    e^{- 2 B\, (r-p-uq)^2}
    \;,
\end{align*}
where we have again used Pinsker's inequality.

\vspace{.5em}

Finally in the case $r=p+uq$, both bounds hold trivially since 
\begin{align*}
    \max\{ \P( \bar V_B \geq r )  \,,\, \P( \bar V_B \leq r ) \}  \;\leq\; 1 \;=\; e^{- 2 B\, (r-p-uq)^2}\;.
\end{align*}
\end{proof}

Combining \Cref{lem:derivative:cdf:V,lem:V:conc} gives the following control on the derivatives:

\begin{lemma}  \label{lem:derivative:cdf:V:bound} Define 
\begin{align*}
    \rho_r \;\coloneqq&\; \max\Big\{ \min\Big\{ \mfrac{Br - 1}{B-1}, 1 \Big\} \,,\, 0 \Big\} 
    &\text{ and }&&
    \rho'_r \;\coloneqq&\; \max\Big\{ \min\Big\{ \mfrac{Br }{B-1}, 1 \Big\} \,,\, 0 \Big\} 
    \;.
\end{align*}
If $\rho_r > p+ u q$,
\begin{align*}
    \max\big\{ \big| \partial_p \, \P(\bar V_B < r )\big|\,,\, \big| \partial_q \, \P(\bar V_B < r )\big| \big\}
    \;\leq&\;
    B \, e^{ -  2 (B-1) \, ( p+ u q - \rho_r )^2  }
    \;,
\end{align*}
and if $\rho'_r < p + uq$,
\begin{align*}
    \max\big\{ \big| \partial_p \, \P(\bar V_B < r )\big|\,,\, \big| \partial_q \, \P(\bar V_B < r )\big| \big\}
    \;\leq&\; 
    B \, e^{ - 2 (B-1) \, ( p + uq - \rho'_r )^2 }
    \;.
\end{align*}
\end{lemma}

\begin{proof}[Proof of \Cref{lem:derivative:cdf:V:bound}] In the case $\rho_r > p+ u q$, we can apply \Cref{lem:derivative:cdf:V,lem:V:conc} while noting that $\bar V_{B-1} \leq 1$ almost surely to obtain
\begin{align*}
    \big| \partial_p \, \P(\bar V_B < r )\big|
    \;=&\;
    B \, \P\Big( \bar V_{B-1} \in \Big[ \mfrac{Br - 1}{B-1}, \mfrac{Br}{B-1} \Big) \Big)
    \\
    \;\leq&\;
    B \, \P\Big( \bar V_{B-1} \geq \mfrac{Br - 1}{B-1} \Big)
    \\
    \;=&\;
    B \, \P\Big( \bar V_{B-1} \geq \max\Big\{ \min\Big\{ \mfrac{Br - 1}{B-1}, 1 \Big\} \,,\, 0\Big\} \Big)
    \\
    \;=&\;
    B \, \P( \bar V_{B-1} \geq \rho_r )
    \\
    \;\leq&\;
    B \, e^{ - 2 (B-1) \, ( p+ u q - \rho_r )^2  }
    \;.
\end{align*}
Moreover since $u \leq 1$, we also have 
\begin{align*}
    \big| \partial_q \, \P(\bar V_B < r )\big|
    \;=&\;
    B \, \P\Big( \bar V_{B-1} \in \Big[ \mfrac{Br- u }{B-1}, \mfrac{Br}{B-1} \Big) \Big)
    \\
    \;\leq&\; 
    \big| \partial_p \, \P(\bar V_B < r )\big|
    \;\leq\; 
    B \, e^{ - 2 (B-1) \, ( p+ u q - \rho_r )^2  }
    \;.
\end{align*}
This proves the first bound. In the case $\rho'_r < p+uq$, a similar argument gives 
\begin{align*}
    \max\big\{ \big| \partial_p \, \P(\bar V_B < r )\big|\,,\, \big| \partial_q \, \P(\bar V_B < r )\big| \big\}
    \;\leq&\;
    B \, \P\Big( \bar V_{B-1} \leq \mfrac{Br}{B-1} \Big)
    \\
    \;=&\;
    B \, \P\Big( \bar V_{B-1} \leq \max\Big\{ \min\Big\{ \mfrac{Br}{B-1}, 1 \Big\} \,,\, 0\Big\}  \Big)
    \\
    \;=&\;
    B \, \P( \bar V_{B-1} \leq \rho'_r )
    \\
    \;\leq&\; 
    B \, e^{ - 2 (B-1) \, ( p + uq - \rho'_r )^2  }
    \;.
\end{align*}
\end{proof}

\subsection{Proof of \Cref{thm:Del:Kol:approx:inv:general}} \label{proof:thm:Del:Kol:approx:inv:general}

\noindent 
\textbf{Step 1: Setup.}
We first express the quantities of interest in terms of those defined in \Cref{appendix:three:point}. We are concerned with the normalised rank variable 
\begin{align*}
    \mfrac{R_\Phi(X)}{B+1} 
    \;=&\;
    \mfrac{1}{B+1} \msum_{b=1}^B 
    \big( 
        \ind_{\{ f(X) > f(\Phi_b(X)) \}} 
        +
        U \,
        \ind_{\{ f(X) = f(\Phi_b(X)) \}} 
    \big)
    +
    \mfrac{U}{B+1} 
    \\
    \;=&\;
    \mfrac{1}{B+1} \msum_{b=1}^B  V_b(P_1, Q_1, U)  + \mfrac{U}{B+1}   \;,
\end{align*}
where we have defined the random variables $V_b(p,q,u)$ as in \Cref{appendix:three:point}, i.e.~
\begin{align*}
    V_b(p,q,u) \,|\, p,q,u  \;\overset{\rm i.i.d.}{\sim}\; 
    \begin{cases}
        1  &\text{ with probability }
        p \;,
        \\
        u   &\text{ with probability }
        q\;,
        \\
        0 &\text{ with probability } 1-p-q \;,
    \end{cases}
\end{align*}
and also defined 
\begin{align*}
    P_1 \;\coloneqq&\; \P( f(X) > f(\Phi_1(X)) \,|\, \A, \, f(X) )
    &\text{ and }&&
     Q_1 \;\coloneqq&\; \P( f(X) = f(\Phi_1(X)) \,|\, \A, \, f(X) )
     \;.
\end{align*}
$V_b(P_1, Q_1, U)$'s are conditionally i.i.d.~given $P_1$, $Q_1$ and $U$.
Since $\A \supseteq \Xobs$, by the tower rule, 
\begin{align*}
    \P\Big(  \mfrac{R_\Phi(X)}{B+1} < r \,\Big|\, \Xobs \Big)
    \;=\;
    \mean\Big[ \, \P\Big(  \mfrac{1}{B} \msum_{b=1}^B  V_b(P_1, Q_1, U) < 
    r_U \,\Big|\, P_1, Q_1, U \Big) \,\Big|\, \Xobs \Big]
    \;,
\end{align*}
where we have denoted 
\begin{align*}
    r_U \;\coloneqq\; \mfrac{B+1}{B} r - \mfrac{U}{B} \;.
\end{align*}
Conditioning on $P_1$, $Q_1$ and $U$, we can identify the average above with $\bar V_B$ in \Cref{appendix:three:point}, where $p,q,u$ are replaced by $P_1$, $Q_1$ and $U$ respectively.

\vspace{.5em}

On the other hand, recall  from \eqref{eq:choice:Psi} that $(f(\Psi_b(X)))_{0 \leq b \leq B}$ are conditionally i.i.d.~copies of $f(X)$ given $\Xobs$. We can WLOG couple $(f(\Psi_b(X)))_{0 \leq b \leq B}$ and $f(X)$ such that 
\begin{align*}
    f(\Psi_0(X)) = f(X) \quad \text{ almost surely. }
\end{align*}
This allows us to express, as before,
\begin{align*}
    \P\Big(  \mfrac{R_\Psi(X)}{B+1} < r \,\Big|\, \Xobs \Big)
    \;=\;
    \mean\Big[ \, \P\Big(  \mfrac{1}{B} \msum_{b=1}^B  V_b(P_2, Q_2, U) < r_U \,\Big|\, P_2, Q_2, U \Big) \,\Big|\, \Xobs \Big]
    \;,
\end{align*}
where we have defined 
\begin{align*}
    P_2 \;\coloneqq&\; \P( f(X) > f(\Psi_1(X)) \,|\, \Xobs, \, f(X) )
    \;,
    \\
    Q_2 \;\coloneqq&\; \P( f(X) = f(\Psi_1(X)) \,|\, \Xobs, \, f(X) )
     \;.
\end{align*}

\vspace{.5em}

\noindent 
\textbf{Step 2: Taylor expansion and apply derivative bounds from \Cref{lem:derivative:cdf:V:bound}.} Let $\Theta \sim \textrm{Uniform}[0,1]$ be independent of all other random variables, and denote the interpolation 
\begin{align*}
    P_\Theta \;\coloneqq&\; \Theta P_1 + (1-\Theta) P_2 
    &\text{ and }&&
    Q_\Theta \;\coloneqq&\; \Theta Q_1 + (1-\Theta) Q_2 \;.
\end{align*}
We will also be applying \Cref{lem:derivative:cdf:V:bound}. To this end, recall the notation
\begin{align*}
    \rho_{r_U} \;\coloneqq&\; \max\Big\{ \min\Big\{ \mfrac{Br_U - 1}{B-1}, 1 \Big\} \,,\, 0 \Big\} 
    &\text{ and }&&
    \rho'_{r_U} \;\coloneqq&\; \max\Big\{ \min\Big\{ \mfrac{Br_U}{B-1}, 1 \Big\} \,,\, 0 \Big\} 
    \;.
    \tagaligneq \label{eq:rho:r:defn}
\end{align*}
Taking a difference of the c.d.f. above and applying the mean value theorem, we obtain that 
\begin{align*}
        \P\Big(  \mfrac{R_\Phi(X)}{B+1}  \,&\, < r \,\Big|\, \Xobs \Big)  
        \,-\,
        \P\Big(  \mfrac{R_\Psi(X)}{B+1} < r \,\Big|\, \Xobs \Big) 
    \\
    \;=&\;
        \mean\Big[ \, 
        \P\Big(  \mfrac{1}{B} \msum_{b=1}^B  V_b(P_1, Q_1, U) < r_U \,\Big|\, P_1, Q_1, U \Big)
    \\
    &\hspace{2em}
        -
        \P\Big(  \mfrac{1}{B} \msum_{b=1}^B  V_b(P_2, Q_2, U) < r_U  \,\Big|\, P_2, Q_2, U \Big)
     \,\Big|\, \Xobs \Big] 
    \\
    \;=&\;
        \mean\Big[ \, 
        (P_1 - P_2)
        \,
         \partial_p \P\Big(  \mfrac{1}{B} \msum_{b=1}^B  V_b(P_\Theta, Q_\Theta, U) < r_U \,\Big|\, P_\Theta, Q_\Theta, U \Big)
    \\
    &\hspace{2em}
        +
        (Q_1 - Q_2)
        \,
         \partial_q \P\Big(  \mfrac{1}{B} \msum_{b=1}^B  V_b(P_\Theta, Q_\Theta, U) <  r_U  \,\Big|\, P_\Theta, Q_\Theta, U \Big)
     \,\Big|\, \Xobs \Big] 
     \\
     \;\eqqcolon&\;
     \mean[ (\star) \,|\, \Xobs ]
     \;.
\end{align*}
Let $\delta > 0$. We first write 
\begin{align*}
    | \mean[ (\star) \,|\, \Xobs ] |
    \;\leq&\; 
    \mean\big[
        |(\star)| \, \big( 
        \ind_{\{ |P_1 - P_2| \geq \delta  \}} 
        + 
        \ind_{\{ |Q_1 - Q_2| \geq \delta  \}} 
        + 
        \ind_{\{ |P_1 - P_2| < \delta \,,\, |Q_1 - Q_2| < \delta   \}}
    \big)
    \,\big|\, \Xobs \big]
    \;.
\end{align*}
Note that $(\star)$ is a difference between two quantities taking values in $[0,1]$, and therefore $|(\star)| \leq 1$ almost surely. By the Markov inequality, we can bound, for $\nu \geq 1$,
\begin{align*}
    \mean\big[
        |(\star)| \,
        \ind_{\{ |P_1 - P_2| \geq \delta  \}} 
    \,\big|\, 
    \Xobs \big]
    \;\leq&\;
    \P\big(  |P_1 - P_2| \geq \delta  \,\big|\, \Xobs \big)
    \;\leq\;
    \mfrac{\mean[ |P_1 - P_2|^\nu \,|\, \Xobs ]}{\delta^\nu} 
    \;,
    \\
    \mean\big[
        |(\star)| \,
        \ind_{\{ |Q_1 - Q_2| \geq \delta  \}} 
    \,\big|\, 
    \Xobs \big]
    \;\leq&\;
    \P\big(  |Q_1 - Q_2| \geq \delta   \,\big|\, \Xobs \big)
    \;\leq\;
    \mfrac{\mean[ |Q_1 - Q_2|^\nu \,|\, \Xobs ]}{\delta^\nu} 
    \;,
    \tagaligneq \label{eq:approx:inv:Markov:step}
\end{align*}
and note also that, for any $\epsilon > 0$,
\begin{align*}
    &\;
    \mean\Big[ |(\star)| \, \ind_{\big\{\rho_{r_U}  - \epsilon \leq P_\Theta + U Q_\Theta \leq \rho'_{r_U} + \epsilon \,,\, |P_1-P_2| < \delta \,,\, |Q_1 - Q_2| < \delta \big\} } \,\Big|\, \Xobs \Big]
    \\
    &\;\leq\;
    \P\big( \rho_{r_U} - \epsilon - 2\delta \leq  P_2 + U Q_2 \leq \rho'_{r_U} + \epsilon + 2\delta \,\big|\, \Xobs \big)\;.
\end{align*}
This allows us to write 
\begin{align*}
     &\;
     \mean[ (\star) \,|\, \Xobs ] 
     \\
     &\;\leq\;
     \mfrac{\mean[ |P_1 - P_2|^\nu \,|\, \Xobs ]}{\delta^\nu} 
     +
     \mfrac{\mean[ |Q_1 - Q_2|^\nu \,|\, \Xobs ]}{\delta^\nu} 
     \\
     &\;\qquad
     +
     \P\big( \rho_{r_U} - \epsilon - 2\delta \leq  P_2 + U Q_2 \leq \rho'_{r_U} + \epsilon + 2\delta \,\big|\, \Xobs \big)
     \\
     &\;\qquad
     +
     \mean\Big[
        |(\star)| \,
        \ind_{\{ |P_1 - P_2| < \delta \,,\, |Q_1 - Q_2| < \delta    \}} 
        \Big( 
             \ind_{\{ \rho_{r_U} - \epsilon > P_\Theta + U Q_\Theta  \}}  
            +  
            \ind_{\{ \rho'_{r_U} + \epsilon < P_\Theta + U Q_\Theta \}}
        \Big)
    \,\Big|\, \Xobs \Big]
    \;.
\end{align*}
To handle the last term, we write as shorthand
\begin{align*}
    M \;\coloneqq&\;  \max\Big\{ 
            \Big| 
            \partial_p \P\Big(  \mfrac{1}{B} \msum_{b=1}^B  V_b(P_\Theta, Q_\Theta, U) < r_U \,\Big|\, P_\Theta, Q_\Theta, U \Big)
            \Big| 
            \,,\,
    \\
    &\hspace{4em}
            \Big| \partial_q \P\Big(  \mfrac{1}{B} \msum_{b=1}^B  V_b(P_\Theta, Q_\Theta, U) < r_U \,\Big|\, P_\Theta, Q_\Theta, U \Big)\Big|
         \Big\}
         \;.
\end{align*}
This allows us to write
\begin{align*}
     &\;
     \mean\big[
        |(\star)| \,
        \ind_{\{ |P_1 - P_2| < \delta \,,\, |Q_1 - Q_2| < \delta   \}} 
        \Big( 
             \ind_{\{ \rho_{r_U} - \epsilon > P_\Theta + U Q_\Theta  \}}  
            +  
            \ind_{\{ \rho'_{r_U} + \epsilon < P_\Theta + U Q_\Theta \}}
        \Big)
    \,\big|\, \Xobs \big]
    \\
    &\;\leq\;
    \mean\Big[
        (|P_1 - P_2| + |Q_1 - Q_2| ) \, M\,
        \ind_{\{ |P_1 - P_2| < \delta \,,\, |Q_1 - Q_2| < \delta   \}} 
    \\
    &\hspace{5em}
        \Big( 
             \ind_{\{ \rho_{r_U} - \epsilon > P_\Theta + U Q_\Theta  \}}  
            +  
            \ind_{\{ \rho'_{r_U} + \epsilon < P_\Theta + U Q_\Theta \}}
        \Big)
          \,\Big|\,\Xobs \Big]
    \\
    &\;\leq\;
    2\delta  \, \mean\Big[ M 
        \Big( 
             \ind_{\{ \rho_{r_U} - \epsilon > P_\Theta + U Q_\Theta  \}}  
            +  
            \ind_{\{ \rho'_{r_U} + \epsilon < P_\Theta + U Q_\Theta \}}
        \Big)
    \,\Big|\,\Xobs \Big]
    \\
    &\;\leq\; 
    4 \delta B e^{ - 2(B-1) \epsilon^2 } 
     \;,
\end{align*}
where, in the last line, we have recalled that $P_\Theta = \Theta(P_1 - P_2) + P_2$  and  $Q_\Theta = \Theta(Q_1 - Q_2) + Q_2$, and  applied the derivative bounds in \Cref{lem:derivative:cdf:V:bound} on $M$. Combining all the bounds above, we obtain 
\begin{align*}
    &\;
      \Big| \P\Big(  \mfrac{R_\Phi(X)}{B+1}  < r \,\Big|\, \Xobs \Big)  
        \,-\,
        \P\Big(  \mfrac{R_\Psi(X)}{B+1} < r \,\Big|\, \Xobs \Big) 
    \Big|
    \\
    &\;\leq\; 
    \mfrac{\mean[ |P_1 - P_2|^\nu \,|\, \Xobs ]}{\delta^\nu} 
    +
    \mfrac{\mean[ |Q_1 - Q_2|^\nu \,|\, \Xobs ]}{\delta^\nu} 
    \\
    &\hspace{2em} 
    +
    \P\big( \rho_{r_U} - \epsilon - 2\delta \leq  P_2 + U Q_2 \leq \rho'_{r_U} + \epsilon + 2\delta \,\big|\, \Xobs \big)
    +
    4 B \delta 
        e^{ - 2(B-1) \epsilon^2 } 
    \;.
    \tagaligneq \label{eq:approx:inv:after:derivatives}
\end{align*}

\vspace{.5em}

\noindent 
\textbf{Step 3: Use uniformity of $P_2+UQ_2 \,|\, \Xobs$.} To control the remaining terms, we first recall that
\begin{align*}
    P_2 + U Q_2 
    \;=&\;
    \P\big(  f(\Psi_1(X)) < f(X) \,\big|\, \Xobs, f(X) \big)
    +
    U \,
     \P\big(f(\Psi_1(X)) = f(X)  \,\big|\, \Xobs, \, f(X) \big)
     \;,
\end{align*}
where $f(\Psi_1(X))$ is a conditionally i.i.d.~copy of $f(X)$ given $\Xobs$. We seek to prove that $P_2+UQ_2 \,|\, \Xobs$ is uniformly distributed. To see this, for $u \in [0,1]$, define 
\begin{align*}
    y_u 
    \;\coloneqq&\; 
    \inf\big\{ 
        y \in \R 
    \,\big|\;
        \P\big(  f(\Psi_1(X)) \leq y \,\big|\, \Xobs \big)
        \geq u
    \big\}
    \;.
\end{align*}
Then by definition, 
\begin{align*}
    \P\big(  f(\Psi_1(X)) < y_u \,\big|\, \Xobs \big)
    \;<\; 
    u 
    \;\leq\;
    \P\big(  f(\Psi_1(X)) \leq y_u \,\big|\, \Xobs \big)
    \;.
\end{align*}
Consider three cases depending on the value of $f(X)$:
\begin{itemize}
    \item On the event $\{f(X) < y_u\}$, almost surely
    \begin{align*}
        P_2 + Q_2 
        \;=&\;
        \P\big( f(\Psi_1(X)) \leq f(X) \,\big|\, \Xobs, \, f(X) \big)
        \\
        \;\leq&\;
        \P\big( f(\Psi_1(X)) < y_u \,\big|\, \Xobs, \, f(X) \big)
        \;\leq\; u\;;
    \end{align*}
    \item On the event $\{f(X) > y_u\}$, almost surely 
    \begin{align*}
        P_2 
        \;=&\;
        \P\big( f(\Psi_1(X)) < f(X) \,\big|\, \Xobs, \, f(X) \big)
        \\
        \;\geq&\;
        \P\big( f(\Psi_1(X)) \leq y_u \,\big|\, \Xobs, \, f(X) \big)
        \;\geq\; u\;;
    \end{align*}
    \item On the event $\{f(X) = y_u\}$, almost surely 
    \begin{align*}
        P_2 + U Q_2 
        \;\leq\; u 
        \quad 
        \text{ if and only if }
        \quad 
        U 
        \;\leq\; 
        \mfrac{u- \P( f(\Psi_1(X)) < y_u | \Xobs )}{\P( f(\Psi_1(X)) = y_u | \Xobs )}\;.
    \end{align*}
\end{itemize}
This implies that 
\begin{align*}
    \P(P_2 + U Q_2 \leq u \,|\, \Xobs )
    \;=&\;
    \P( f(X) < y_u \,|\, \Xobs)
    \\
    &\;
    +
    \P\Big( f(X) = y_u \,,\, U 
        \leq
        \mfrac{u- \P( f(\Psi_1(X)) < y_u | \Xobs )}{\P( f(\Psi_1(X)) = y_u | \Xobs )}  \,\Big|\, \Xobs\Big)
    \;.
\end{align*}
Since $U \sim \textrm{Uniform}[0,1]$ is independent of $f(X)$, we can compute 
\begin{align*}
    \P(P_2 + U Q_2 \leq u \,|\, \Xobs )
     \;=&\;
    \P( f(X) < y_u \,|\, \Xobs)
    \\
    &\;
    +
    \P( f(X) = y_u \,|\, \Xobs) 
    \;
        \mfrac{u- \P( f(\Psi_1(X)) < y_u | \Xobs )}{\P( f(\Psi_1(X)) = y_u | \Xobs )}  
    \\
    \;=&\; u\;.
\end{align*}
Since this holds for all $u$, we obtain that $P_2 + U Q_2 \,|\, \Xobs$ is distributed as $\textrm{Uniform}[0,1]$. Plugging this into \eqref{eq:approx:inv:after:derivatives}, and noting that $0 \leq \rho'_{r_U} - \rho_{r_U} \leq \frac{1}{B-1}$ by their definitions in \eqref{eq:rho:r:defn}, we obtain 
\begin{align*}
    &\;
      \Big| \P\Big(  \mfrac{R_\Phi(X)}{B+1}  < r \,\Big|\, \Xobs \Big)  
        \,-\,
        \P\Big(  \mfrac{R_\Psi(X)}{B+1} < r \,\Big|\, \Xobs \Big) 
    \Big|
    \\
    &\;\leq\; 
    \mfrac{\mean[ |P_1 - P_2|^\nu \,|\, \Xobs ]}{\delta^\nu} 
    +
    \mfrac{\mean[ |Q_1 - Q_2|^\nu \,|\, \Xobs ]}{\delta^\nu} 
    +
    \big( 
        \mfrac{1}{B-1} + 2\epsilon + 4 \delta
    \big)
    +
    4 B  \delta  e^{ - 2(B-1) \epsilon^2 } 
    \;.
    \tagaligneq \label{eq:approx:inv:almost:there}
\end{align*}

\vspace{.5em}

\noindent 
\textbf{Step 4: Extract approximate invariance term.} We now control $|P_1 - P_2|$ and $|Q_1 - Q_2|$. Recall  
\begin{align*}
    \Delta_{\rm inv}(\A)
    \;=\;
    \msup_{t \in \R}
    \big|  
        \P( f(\Phi_1(X)) \leq t \,|\, \A ) 
        \,-\, 
        \P( f(X) \leq t \,|\, \Xobs )  
    \big|
    \;.
\end{align*} 
This allows us to bound 
\begin{align*}
    &\;
    \mean[ |P_1 - P_2|^\nu \,|\, \Xobs ]
    \;=\;
    \mean\big[ 
    \;
    \big|  
    \,
    \P(  f(\Phi_1(X)) < f(X) \,|\, \A, f(X) ) 
    \\
    &\;\hspace{12em}
    - \;
    \P(  f(\Psi_1(X)) < f(X) \,|\, \Xobs, f(X) )  
    \,
     \big|^\nu \; \,\big|\,\Xobs\big]
     \\
    &\;\leq\;
     \mean\big[ 
    \msup_{t \in \R}
    \big|  
        \P( f(\Phi_1(X)) \leq t \,|\, \A ) 
        \,-\, 
        \P( f(\Psi_1(X)) \leq t \,|\, \Xobs )  
    \big|^\nu \,\big|\,\Xobs\big]
    \\
    &\;=\;
    \big\| \Delta_{\rm inv}(\A) \big\|_{L_\nu \,|\, \Xobs}^\nu
    \;,
    \tagaligneq \label{eq:approx:inv:P:control} 
\end{align*}
where we have used that $f(\Psi_1(X))$ is a conditionally i.i.d.~copy of $f(X)$ given $\Xobs$. Similarly,
\begin{align*}
    &\;\mean[ |Q_1 - Q_2|^\nu \,|\, \Xobs ]
    \;=\;
    \mean\big[ 
    \;
    \big|  
    \,
    \P(  f(\Phi_1(X)) = f(X) \,|\, \A, f(X) ) 
    \\
    &\;\hspace{12em}
    - \;
    \P(  f(\Psi_1(X)) = f(X) \,|\, \Xobs, f(X) )  
    \,
     \big|^\nu \; \,\big|\,\Xobs\big]
     \\
    &\;\leq\;
    2^\nu \,
     \mean\big[ 
    \msup_{t \in \R}
    \big|  
        \P( f(\Phi_1(X)) \leq t \,|\, \A ) 
        \,-\, 
        \P( f(\Psi_1(X)) \leq t \,|\, \Xobs )  
    \big|^\nu \,\big|\,\Xobs\big]
    \\
    &\;=\;
    2^\nu \, \big\| \Delta_{\rm inv}(\A) \big\|_{L_\nu \,|\, \Xobs}^\nu
    \;.
\end{align*}
Plugging these into \eqref{eq:approx:inv:almost:there} gives
\begin{align*}
    &\;
     \Big| \P\Big(  \mfrac{R_\Phi(X)}{B+1}  < r \,\Big|\, \Xobs \Big)  
        \,-\,
        \P\Big(  \mfrac{R_\Psi(X)}{B+1} < r \,\Big|\, \Xobs \Big) 
    \Big|
    \\
    &\;\leq\; 
    \mfrac{3^\nu  \, \| \Delta_{\rm inv}(\A) \|_{L_\nu \,|\, \Xobs}^\nu}{\delta^\nu} 
    +
    \big( 
        \mfrac{1}{B-1} + 2\epsilon + 4 \delta
    \big)
    +
    4 B  \delta  e^{ - 2(B-1) \epsilon^2 } 
    \;.
\end{align*}

\vspace{.5em}

\noindent
\textbf{Step 5: Clean-up.} The bound above holds for any $\epsilon, \delta > 0$. Choosing 
\begin{align*}
    \delta 
    \;=\; 
    \Big( \mfrac{ (3 \, \| \Delta_{\rm inv}(\A) \|_{L_\nu \,|\, \Xobs})^\nu}{4} \Big)^{1/(\nu+1)}
    \;,
    \quad 
    \epsilon 
    \;=\;
    \mfrac{\sqrt{3 \log B}}{ 2 \,\sqrt{B-1}}
    \;,
    \tagaligneq \label{eq:approx:inv:choice:delta:eps}
\end{align*}
we obtain 
\begin{align*}
     &\;\Big| \P\Big(  \mfrac{R_\Phi(X)}{B+1}  < r \,\Big|\, \Xobs \Big)  
        \,-\,
        \P\Big(  \mfrac{R_\Psi(X)}{B+1} < r \,\Big|\, \Xobs \Big) 
    \Big|
    \\
    &\;\leq\;  
    4 \delta + \Big( \mfrac{1}{B-1} +  \mfrac{\sqrt{3 \log B}}{\sqrt{B-1}}  + 4 \delta \Big) + 4 \delta \mfrac{1}{\sqrt{B}}
    \\
    &\;\leq\;
    12 \delta + \mfrac{1}{B-1} +  \mfrac{\sqrt{3 \log B}}{\sqrt{B-1}} 
    \\
    &\;=\; 
    \mfrac{12}{4^{1/(\nu+1)}}
    \,
    \big( 3 \, \| \Delta_{\rm inv}(\A) \|_{L_\nu \,|\, \Xobs} \big)^{\frac{\nu}{\nu+1}} 
    + 
    \mfrac{1}{B-1} 
    +  
    \mfrac{\sqrt{3 \log B}}{\sqrt{B-1}} 
    \;.
\end{align*}
The desired bound is obtained by noting that, by construction, $(f(\Psi_b(X)))_{0 \leq b \leq B}$ is a conditionally exchangeable sequence given $\Xobs$, so by \Cref{thm:ex:validity} with \Cref{remark:thm:ex:cond}, we have $\P\big(  \frac{R_\Psi(X)}{B+1} < r \,\big|\, \Xobs \big) = r$.
\qed

\subsection{Proof of \Cref{lem:fixed:to:random}} 

We prove instead a more general version of \Cref{lem:fixed:to:random}, which allows for conditioning on $\Xobs$:

\begin{lemma} \label{lem:fixed:to:random:conditional}  Let $r \in [0,1]$, $\nu \geq 1$ and $(\Phi'_1, \ldots, \Phi'_B)$ be a deterministic enumeration of a size-$B$ set $\cS_B$. Let $(\Phi_b)_{b \leq B}$ be i.i.d.~uniform draws from the same set $\cS_B$ independent of all other variables. Under \Cref{asst:iid:smooth}(ii) and (iii),
\begin{align*}
    &\;\Big| 
    \P\Big( \mfrac{R_{\Phi'}(X)}{B+1} < r \,\Big|\, \Xobs \Big)   
    -
    \P\Big( \mfrac{R_{\Phi}(X)}{B+1} < r \,\Big|\, \Xobs \Big)   
    \Big|
    \\
    &\hspace{10em}\;\leq\; 
    \mfrac{4 \sqrt{\log (2B)}}{\sqrt{2B}}
    \,+\,
    5 \, \big(3  \big\| \Delta_{\rm inv}(X) \big\|_{L_\nu \,|\, \Xobs}\big)^{\nu/(\nu+1)}
    \;,
\end{align*}
where we have denoted $R_{\Phi'}(X)  \coloneqq \Rank_T\big(f(X) \,;\, (f(\Phi'_b(X)))_{b \in [B]} \big) $.   
\end{lemma}

\begin{proof}[Proof of \Cref{lem:fixed:to:random:conditional}] Since $\Phi' = (\Phi'_1, \ldots, \Phi'_B)$ is a deterministic enumeration of $\cS_B$ and $\Phi_b$'s are i.i.d.~drawn from $\cS_B$, we can express 
\begin{align*}
    \mfrac{R_{\Phi'}(X)}{B+1}
    \;=&\;
    \mfrac{B}{B+1} \Big(
    \mfrac{1}{B} \, \msum_{b=1}^B \ind_{\{ f(X) > f(\Phi'_b(X)) \}} 
    + 
    U \, \mfrac{1}{B} \, \msum_{b=1}^B  \ind_{\{ f(X) = f(\Phi'_b(X)) \}} 
    \Big) 
    + \mfrac{U}{B+1}
    \\
    \;=&\;
    \mfrac{B}{B+1} \Big( \mean\big[ \ind_{\{ f(X) > f(\Phi_b(X)) \}}  \,|\, X\big]
    + 
    U \, \mean\big[ \ind_{\{ f(X) = f(\Phi_b(X)) \}}  \,|\, X\big]\Big)
    + \mfrac{U}{B+1}
    \\
    \;=&\; 
    \mfrac{B}{B+1} ( P_1 + U \, Q_1 )+ \mfrac{U}{B+1}\;,
\end{align*}
where we have inherited the notation $P_1$ and $Q_1$ from the proof of \Cref{thm:Del:Kol:approx:inv:general} in \Cref{proof:thm:Del:Kol:approx:inv:general} with $\A=\sigma(X)$. We also inherit the notation $V_b$, $P_2$ and $Q_2$ from  \Cref{proof:thm:Del:Kol:approx:inv:general}, and denote 
\begin{align*}
    r_U \;=\; \mfrac{B+1}{B} r - \mfrac{U}{B}\;.
\end{align*}
 This allows us to write 
\begin{align*}
    (*)\;\coloneqq&\;
    \Big| 
    \P\Big( \mfrac{R_{\Phi'}(X)}{B+1} < r \,\Big|\, \Xobs \Big)   
    -
    \P\Big( \mfrac{R_{\Phi}(X)}{B+1} < r \,\Big|\, \Xobs \Big)   
    \Big|
    \\
    \;=&\;
    \Big| 
        \mean\Big[ 
            \ind_{\{ P_1 + U \, Q_1 < r_U \}}
            -
            \ind_{\{\frac{1}{B} \sum_{b=1}^B V_b(P_1, Q_1, U) < r_U \} }
            \,\Big|\, \Xobs
        \Big]
    \Big|
    \\
    \;=&\;
    \Big| 
        \mean\Big[ 
            \ind_{\{ P_1 + U \, Q_1 < r_U \,,\, \frac{1}{B} \sum_{b=1}^B V_b(P_1, Q_1, U) \geq r_U \}}
            -
            \ind_{\{ P_1 + U \, Q_1 \geq r_U  \,,\, \frac{1}{B} \sum_{b=1}^B V_b(P_1, Q_1, U) < r_U \} }
            \,\Big|\, \Xobs
        \Big]
    \Big|
    \;.
\end{align*}
As with \Cref{proof:thm:Del:Kol:approx:inv:general}, we let $\delta > 0$, apply the Markov's inequality \eqref{eq:approx:inv:Markov:step} and the bound \eqref{eq:approx:inv:P:control}  with $\nu \geq 1$ to obtain 
\begin{align*}
    &(*)
    \;\leq\;
    \Big| 
        \mean\Big[ 
            \ind_{\{ |P_1 - P_2| < \delta \,,\, |Q_1 - Q_2| < \delta  \}}
            \times 
    \\
    &\hspace{3em}
            \Big(
            \ind_{\{ P_1 + U \, Q_1 < r_U \,,\, \frac{1}{B} \sum_{b=1}^B V_b(P_1, Q_1, U) \geq r_U \}}
            -
            \ind_{\{ P_1 + U \, Q_1 \geq r_U  \,,\, \frac{1}{B} \sum_{b=1}^B V_b(P_1, Q_1, U) < r_U \} }
            \Big)
        \,\Big|\, \Xobs
        \Big]
    \Big|
    \\
    &\hspace{3em}
    + 
    \mfrac{3  \, \| \Delta_{\rm inv}(X) \|_{L_\nu \,|\, \Xobs}^\nu}{\delta^\nu} 
    \\
    &\;\leq\;
    \Big| \mean\Big[ 
        \ind_{\{ |P_1 - P_2| < \delta \,,\, |Q_1 - Q_2| < \delta \,,\, P_1 + U \, Q_1 < r_U   \}}
    \\
    &\hspace{5em}\,\times\,
        \,\P\Big(  \mfrac{1}{B} \msum_{b=1}^B V_b(P_1, Q_1, U) \geq r_U \,\Big|\, P_1, Q_1, U \Big)
    \,\Big|\, \Xobs \Big]  \Big| 
    \\
    &\qquad
    +
    \Big| \mean\Big[ 
        \ind_{\{ |P_1 - P_2| < \delta \,,\, |Q_1 - Q_2| < \delta \,,\, P_1 + U \, Q_1 \geq r_U   \}}
    \\
    &\hspace{5em}\,\times\,
        \,\P\Big(  \mfrac{1}{B} \msum_{b=1}^B V_b(P_1, Q_1, U) < r_U \,\Big|\, P_1, Q_1, U \Big)
    \,\Big|\, \Xobs \Big]  \Big| 
    \\
    &\qquad
    + 
    \mfrac{3^\nu  \, \| \Delta_{\rm inv}(X) \|_{L_\nu \,|\, \Xobs}^\nu}{\delta^\nu} 
    \;.
\end{align*}
By the concentration bounds in \Cref{lem:V:conc}, we have that almost surely 
\begin{align*}
    &\;
    \ind_{\{  P_1 + U \, Q_1 < r_U  \}}
    \,\P\Big(  \mfrac{1}{B} \msum_{b=1}^B V_b(P_1, Q_1, U) 
    \geq r_U \,\Big|\, P_1, Q_1, U \Big)
    \;\leq\;
    e^{ - 2 B \, ( r_U - (P_1 + U Q_1))^2}
    \;,
    \\
    &\;
    \ind_{\{  P_1 + U \, Q_1 \geq r_U   \}}
    \,\P\Big(  \mfrac{1}{B} \msum_{b=1}^B V_b(P_1, Q_1, U) 
    < r_U \,\Big|\, P_1, Q_1, U \Big)
    \;\leq\;
    e^{ - 2 B \, ( r_U - (P_1 + U Q_1))^2}
    \;,
\end{align*}
which implies 
\begin{align*}
    (*) 
    \;\leq&\; 
    2 \,
    \Big| \mean\Big[ 
        \ind_{\{ |P_1 - P_2| < \delta \,,\, |Q_1 - Q_2| < \delta  \}} 
        \, e^{ - 2 B \, ( r_U - (P_1 + U Q_1))^2}
    \,\Big|\, \Xobs \Big]  \Big| 
    +
    \mfrac{3^\nu  \, \| \Delta_{\rm inv}(X) \|_{L_\nu \,|\, \Xobs}^\nu}{\delta^\nu} 
    \;.
\end{align*}
Let $\epsilon > 0$. We can further split the expectation above according to how $| r_U - (P_1+ U Q_1)|$ compares with $\epsilon$, and obtain 
\begin{align*}
    &(*) 
    \;\leq\; 
    2 \, e^{ - 2 B \epsilon^2}
    +
    \mfrac{3^\nu \, \| \Delta_{\rm inv}(X) \|_{L_\nu \,|\, \Xobs}^\nu}{\delta^\nu} 
    \\
    &\;\qquad\quad
    +
    2 \,
    \P\Big( |P_1 - P_2| < \delta \,,\, |Q_1 - Q_2| < \delta  \,,\, | P_1+ U Q_1 - r_U| < \epsilon 
    \,\Big|\, \Xobs \Big) 
    \\
    &\;\leq\;
    2 \, e^{ - 2 B \epsilon^2}
    +
    \mfrac{3^\nu  \, \| \Delta_{\rm inv}(X) \|_{L_\nu \,|\, \Xobs}^\nu}{\delta^\nu} 
    +
    \P\Big( 
         P_2 + U Q_2 \in (r_U-\epsilon - 2\delta , r_U + \epsilon + 2\delta)
    \,\Big|\, \Xobs \Big) 
    \\
    &\;=\;
    2 \, e^{ - 2 B \epsilon^2}
    +
    \mfrac{3^\nu  \, \| \Delta_{\rm inv}(X) \|_{L_\nu \,|\, \Xobs}^\nu}{\delta^\nu} 
    +
    \mfrac{1}{B}
    +
    2 \epsilon + 4 \delta 
    \;.
\end{align*}
In the last line, we have recalled from \eqref{eq:approx:inv:almost:there} that $P_2 + U Q_2 \,|\, \Xobs \sim \textrm{Uniform}[0,1]$ and that $r_U - \frac{B+1}{B} r \in [- \frac{1}{B}, 0]$ almost surely. Choosing 
\begin{align*}
    \epsilon \;=&\; \mfrac{\sqrt{\log (2B)}}{\sqrt{2B}}
    &\text{ and }&&
    \delta \;=&\; 
    \big( 3 \big\| \Delta_{\rm inv}(X) \big\|_{L_\nu \,|\, \Xobs} \big)^{\nu/(\nu+1)}
    \;,
\end{align*}
we obtain the desired bound that
\begin{align*}
    (*)
    \;\leq\; 
    \mfrac{4 \sqrt{\log (2B)}}{\sqrt{2B}}
    \,+\,
    5 \, \big(3  \big\| \Delta_{\rm inv}(X) \big\|_{L_\nu \,|\, \Xobs}\big)^{\nu/(\nu+1)}
    \;.
\end{align*}
\end{proof}

\section{Additional results for \Cref{sec:validity:universality}: Coverage under universality} \label{appendix:universality}

The first goal of this section is to prove a generalisation of \Cref{thm:universality}, which provides an explicit bound on the quantity
\begin{align*}
    \Delta_{\rm inv}(X)
    \;=\;
    \msup_{t \in \R}
    \big|  
        \P( f(\Phi_1(X)) \leq t \,|\, X ) 
        \,-\, 
        \P( f(X) \leq t  ) 
    \big|
    \;.
\end{align*} 
This bound, provided in \Cref{thm:universality:finite}, relies on a conditional universality result for stable functions of locally dependent data (\Cref{thm:universality:local:dependence}). We state, interpret and prove these results in \Cref{appendix:universality:local:dependence,appendix:universality:finite}.

\vspace{.5em}

The second goal of this section is to verify the bound on $\Delta_{\rm inv}(X)$ for the examples considered in \Cref{sec:validity:universality}. This consists of two parts: In \Cref{appendix:verify:local}, we compute the local dependence measure for different transformations $(\Phi_b)_{b \leq B}$; in \Cref{appendix:verify:stability}, we verify the stability conditions for different estimators $f$. We note that for a restricted range of transformations that operate globally on the data set e.g.~those used in permutation tests and conformal prediction, we can still verify local dependence and stability by a suitable reformulation of the estimators, and details are in \Cref{appendix:verify:local,appendix:verify:stability}. For an assessment of the moment approximation quality of selected transformations, we refer readers to the specific examples in \Cref{appendix:theory:examples}.

\subsection{Gaussian universality under local dependence} \label{appendix:universality:local:dependence}

This section presents a generic Gaussian universality result for $\R^m$-valued random vectors under local dependence and for a thrice-differentiable function $g: \R^m \rightarrow \R$. Let 
\begin{align*}
    V \;\coloneqq&\; (V_1, \ldots, V_m)
    &\text{ and }&&
    W \;\coloneqq&\; (W_1, \ldots, W_m)
\end{align*}
be two $\R^m$-valued random vectors, satisfying that 
\begin{align*}
    \mean[V] \;=&\; \mean[W] \;=\; 0\;
    &\text{ and }&&
    \Var[W] \;=&\; I_m\;.
\end{align*}
The Gaussian surrogate $\xi = (\xi_i)_{i \leq m}$, chosen independently of all other variables, is distributed as 
\begin{align*}
    \xi \;\sim&\; \cN(0, I_m)\;.
\end{align*}
For $\Var[V]$ sufficiently close to $I_m$, we seek to establish a universality approximation of the form 
\begin{align*}
    g(V) \;\approx\; g(\xi) \;\approx\; g(W)\;,
\end{align*}
where the approximation of the random variables holds with respect to weak convergence as $m$ becomes large. To formalise this, we need to first formalise the notion of local dependence, stability and anti-concentration.

\vspace{.5em}

\noindent
\textbf{Local dependence. } For $1 \leq i \leq m$, the local dependency neighbourhoods of $V_i$ in $V$ and $W_i$ in $W$ are respectively
\begin{align*}
    \cM_i^V 
    \;\coloneqq&\; \inf \big\{ \cM \subseteq [m] \;\big|\; i \in \cM \text{ and } (V_j)_{j \in \cM} \text{ is independent of } (V_j)_{j \not\in \cM}  \big\}\;,
    \\
    \cM_i^W 
    \;\coloneqq&\; \inf \big\{ \cM \subseteq [m] \;\big|\; i \in \cM \text{ and } (W_j)_{j \in \cM} \text{ is independent of } (W_j)_{j \not\in \cM}  \big\}\;.
\end{align*}
Note that, writing $\mu_V$ and $\mu_W$ as the probability measures of $V$ and $W$, $N(\argdot)$ defined in \eqref{eq:defn:local:dependency:measure} satisfies
\begin{align*}
    N(\mu_V) \;=&\; \mmax_{1 \leq i \leq m} \big|  \cM_i^V\big|
    &\text{ and }&&
    N(\mu_W) \;=&\; \mmax_{1 \leq i \leq m} \big|  \cM_i^W\big|\;.
\end{align*}
We shall measure the maximum local dependency as
\begin{align*}
    N \;\coloneqq&\; \mmax_{1 \leq i \leq m} \max\big\{ \big|  \cM_i^V\big|  \,,\, \big|\cM_i^W\big|  \big\} 
    \;.
\end{align*}

\vspace{.5em}

\begin{figure}[t]
    \centering
    \begin{tikzpicture}
        \node[circle,draw=black,inner sep=4pt] at (-6,0.4) {$V$};

        \draw[<->, thick, dashed, red] (-5,2) -- (-3.8,0.8);
        \node[rectangle,draw=black,above] at (-5,2.1) {$I^{(1)}(s,0)$};
        \node[above] at (-5.3,1) {$I^{(1)}(s,\theta)$};
        \node[above] at (-3.45,2.2) {$\cdots$};
        \draw[<->, thick, dashed, red] (-2.0,2) -- (-3.2,0.8);
        
        \node[above] at (-1.7,1) {$I^{(m)}(s,\theta)$};
        \node[rectangle,draw=black,above] at (-1.9,2.1) {$I^{(m)}(s,0)$};

        \draw[<->, thick, blue] (-5.5,0.4) -- (-4,0.4); 
        \draw[<->, thick, blue] (-3,0.4) -- (-1.5,0.4);

        \node[rectangle,draw=black, above] at (-3.5,0.05) {$I(s) $};

        \node[circle,draw=black,inner sep=2pt] at (-1,0.4) {$\xi$};
    \end{tikzpicture}
    \caption{
        Illustration of the interpolation paths between the two random vectors $V$ and $\xi$. The \textcolor{blue}{blue} solid path varies $s \in [0,1]$ to interpolate between $V_i$ and $\xi_i$. Each \textcolor{red}{red} dashed path, indexed by $I^{(i)}(s,\theta)$, varies $\theta \in [0,1]$ to interpolate between the set of variables inside the dependency neighbourhood of the $i$-th coordinate and a set of zeros. The same interpolation paths are used between $W$ and $\xi$, as well as the corresponding $\R^{nd}$ random variables in \Cref{sec:validity:universality}.
    \label{fig:interpolation:paths}
    }
\end{figure}
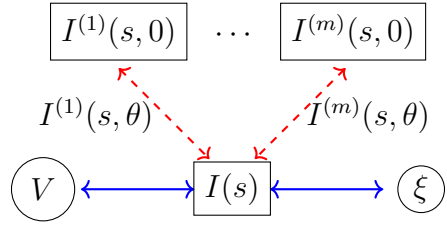

\noindent
\textbf{Stability. } Let $g: \R^m \rightarrow \R$ be a thrice-differentiable statistic of interest. As with the convention in the universality literature, the result will be stated in terms of the stability of $g$ along a suitably chosen interpolation path. Our choice of interpolation path is slightly more complicated due to the local dependence structure, and follows a similar recipe as \cite{mallory2025universality} in the case of high-dimensional logistic regression. For $s \in [0,1]$, we first consider the interpolation path between $V$ and $\xi$,
\begin{align*}
    I(s) \;\coloneqq&\; (I_i(s))_{i \leq m}\;,
    &\,&&
    \text{ where }\quad
    I_i(s) \;\coloneqq&\; \sqrt{s}\;  V_i  \,+\, \sqrt{1-s} \; \xi_i\;.
\end{align*}
For $i \leq m$, we also consider an interpolation path between $I(s)$ and a variant of $I(s)$ with the dependency neighbourhood $\cM^V_i$ removed: For $\theta \in [0,1]$,
\begin{align*}
    &\;
    I^{(i)}_j(s, \theta) 
    \coloneqq
    \begin{cases}
        \theta I_j(s)  &\text{ if } j \in \cM_i^V
        \\
        I_j(s)  &\text{ if } j \not\in \cM_i^V 
    \end{cases}
    \; 
    \text{ for } 1 \leq j \leq m \;,
    \quad
    I^{(i)}(s, \theta)  \coloneqq \Big( I^{(i)}_j(s, \theta)  \Big)_{j \leq m}
    \;.
\end{align*}
We also define the analogous interpolation quantities for $W$:
\begin{align*}
    J(s) \;\coloneqq&\; (J_i(s))_{i \leq m}\;,
    &\,&&
    \text{ where }\quad
    J_i(s) \;\coloneqq&\; \sqrt{s}\;  W_i  \,+\, \sqrt{1-s} \; \xi_i\;,
\end{align*}
and
\begin{align*}
    &\;
    J^{(i)}_j(s,\theta) 
    \coloneqq
    \begin{cases}
        \theta J_j(s)  &\text{ if } j \in \cM_i^W
        \\
        J_j(s)  &\text{ if } j \not\in \cM_i^W
    \end{cases}
    \; 
    \text{ for } 1 \leq j \leq m \;,
    \quad
    J^{(i)}(s,\theta)   \coloneqq \Big( J^{(i)}_j(s,\theta)   \Big)_{j \leq m}
    \;.
\end{align*}
Figure \ref{fig:interpolation:paths} includes a pictorial illustration of our choice of our interpolation paths. We measure the stability of $g$ by its first three partial derivatives along these interpolation paths: For $r=1,2,3$,
\begin{align*}
    \varphi_r(g)
    \;\coloneqq\;
        m^{r/2} 
    \hspace{-.5em}
    \sup_{i,j \leq m; \, \theta, s \in [0,1]}
    \hspace{-.5em}
    \,
    \max\big\{
    \,
    \big\|\, 
        \partial_j^r \, g\big( I^{(i)}(s, \theta) \big) 
        \,
    \big\|_{L_4}
    \,,\,
    \big\|\, 
        \partial_j^r \, g\big( J^{(i)}(s, \theta) \big) 
        \,
    \big\|_{L_4}
    \,
    \big\}
    \;.
\end{align*}
Notice that the stability terms $\varphi_{r;\nu}$ defined in \eqref{eq:defn:varphi} in \Cref{sec:validity:universality} are exactly the analogues of these quantities with an additional $\| \argdot \|_{L_\nu}$. We make several remarks about these notions of stability:
\begin{itemize}
    \item $\varphi_r(g)$ can be replaced by uniform upper bounds on the partial derivatives of $g$. We state them as location-dependent bounds on the interpolation paths, which can often lead to tighter controls in specific examples; see the next bullet point and e.g.~\cite{montanari2022universality}.
    \item The normalisation in $\varphi_r(g)$ is chosen in view of the following toy example. If $g(V) = (\frac{1}{\sqrt{m}} \sum_{i \leq m} V_i )^\kappa$, i.e.~some integer power of the empirical average with $\kappa \geq 3$, and $V_i$'s are i.i.d.~$\R$-valued, zero-mean and bounded, we have that for all $1 \leq i \leq m$,
    \begin{align*}
        \| \partial_i\, g(V)  \|_{L_4}
        \;=&\;
        \Big\| 
        \mfrac{\kappa}{\sqrt{m}} \,  \Big(\mfrac{1}{\sqrt{m}} \msum_{i \leq m} V_i \Big)^{\kappa-1}
        \Big\|_{L_4}
        \;=\;
        O\Big( \mfrac{1}{\sqrt{m}} \Big)
        \;,
        \\
        \| \partial^2_i\, g(V)  \|_{L_4}
        \;=&\;
        \Big\| 
        \mfrac{\kappa(\kappa-1)}{m} \,  \Big(\mfrac{1}{\sqrt{m}} \msum_{i \leq m} V_i \Big)^{\kappa-2}
        \Big\|_{L_4}
        \;=\;
        O\Big( \mfrac{1}{m} \Big)
        \;,
        \\
        \| \partial^3_i\, g(V)  \|_{L_4}
        \;=&\;
        \Big\| 
        \mfrac{\kappa(\kappa-1)(\kappa-2)}{m^{3/2}} \,  \Big(\mfrac{1}{\sqrt{m}} \msum_{i \leq m} V_i \Big)^{\kappa-3}
        \Big\|_{L_4}
        \;=\;
        O\Big( \mfrac{1}{m^{3/2}} \Big)
        \;.
    \end{align*}
    Applying a pre-multiplier of $m^{r/2}$ on each $r$-th derivative therefore ensures that the derivative terms are $O(1)$. This is in general the right scaling for $g$'s that can be well-approximated by a Taylor expansion, have a non-degenerate limit, and are exchangeable in their inputs; see e.g.~\cite{huang2024gaussian} for a formal treatment for symmetric U-statistics, V-statistics and a higher-order delta method.
    \item While not considered in this article, we remark that the stability of an estimator can be improved by techniques such as bootstrap aggregation (bagging) \cite{chen2022debiased,soloff2024bagging}.
\end{itemize}

\noindent
\textbf{Smooth functions and anti-concentration. } As part of the proof, we shall approximate the Kolmogorov distance 
\begin{align*}
    \msup_{t \in \R} |\P(g(V) \leq t) - \P(g(W) \leq t)|
    \tagaligneq \label{eq:Kol}
\end{align*}
by an integral probability distance of the form
\begin{align*}
    \msup_{t \in \R} |\mean[ h_t( g(V) ) - h_t(g(W)) ]|\;,
    \tagaligneq \label{eq:IPM}
\end{align*}
where $(h_t)_{t \in \R}$ is a suitably chosen class of functions with a prescribed level of smoothness; see a standard usage in the literature for distributional approximation e.g. the proof of Theorem 3.3 in \citep{chen2011normal}.  For example, to obtain an $1$-Lipschitz approximation, one may consider

\begin{center}
\begin{tikzpicture}
    \node[inner sep=0pt] at (-7, 0.5) {
        $
        h_{1;t;\delta}(x) 
        \;\coloneqq\; 
        \begin{cases}
            1  & \text{ if }  x < t - \delta\;, \\
            \frac{t - x}{\delta} & \text{ if }  x \in [t - \delta, t)\;, \\
            0  & \text{ if } x \geq t\;.
        \end{cases} 
        $
    };

    \draw[->] (-3,0) -- (1.2,0) node[right] {\scriptsize $x$};
    \draw[-,dashed] (-1,0) -- (-1,1.2);

    \draw[red, thick] (-3,1) -- (-2,1);
    \draw[red, thick] (-2,1) -- (-1,0);
    \draw[red, thick] (-1,0) -- (1,0);

    \draw[blue, thick, dashed] (-2,1) -- (-1,1);
    \draw[blue, thick, dashed] (-1,1) -- (0,0);
    \draw[blue, thick, dashed] (0,0) -- (1,0);

    \node[below, anchor=north] at (-2,0) {\scriptsize $t-\delta$};
    \node[below, anchor=north] at (-1,-0.05) {\scriptsize $t$};

    \node[below, anchor=north] at (0,0) {\scriptsize $t+\delta$};

    \node[above, anchor=south] at (-2,.9) {\scriptsize \color{red} $h_{1;t;\delta}$};

    \node[above, anchor=south] at (-.3,.9) {\scriptsize \color{blue} $h_{1;t + \delta;\delta}$};
\end{tikzpicture}
\end{center}
We will use a thrice-differentiable generalisation of $h_{1;t;\delta}$, obtained from setting $\tilde m =3$ below:

\begin{lemma}[Lemma 34 of \cite{huang2023high}] \label{lem:smooth:approx:ind} Fix any $\tilde m \in \N \cup \{0\}$, $\tau \in \R$ and $\delta > 0$. Then there exists an $\tilde m$-times differentiable $\R \rightarrow \R$ function $h_{\tilde m;\tau;\delta}$ such that $h_{\tilde m;\tau+\delta;\delta}(x) \;\leq\; \ind_{\{x > \tau\}} \leq h_{\tilde m;\tau;\delta}(x)$. For $0 \leq r \leq \tilde m$, the $r$-th derivative $h^{(r)}_{\tilde m;\tau;\delta}$ is continuous and bounded in absolute value by $\delta^{-r}$. Moreover, for every $\epsilon \in [0,1]$, $h^{(\tilde m)}_{\tilde m;\tau;\delta}$ satisfies that
    \begin{align*}
        | h^{(\tilde m)}_{\tilde m;\tau;\delta}(x) - h^{(\tilde m)}_{\tilde m;\tau;\delta}(y) |\;\leq\; C_{\tilde m,\epsilon} \, \delta^{-(\tilde m+\epsilon)}  \, |x - y|^{\epsilon}\;,
    \end{align*}
    with respect to the constant $C_{\tilde m,\epsilon} = \binom{\tilde m}{\lfloor\tilde m/2 \rfloor} (\tilde m+1)^{\tilde m+\epsilon}$.
\end{lemma}

\noindent
The approximation error of \eqref{eq:Kol} by \eqref{eq:IPM}, for $h_t$ of the form above, can be measured by the probability of $g(W)$ lying in a small neighbourhood of the threshold $t$:
\begin{align*}
    \P(  | g(W) - t| \leq \delta )\;.   
    \tagaligneq \label{eq:anti:conc}
\end{align*}
A control on \eqref{eq:anti:conc} is known as an anti-concentration inequality \cite{rogozin1961estimate,chernozhukov2015comparison,cattaneo2025sharp}:
\begin{itemize}
    \item If an anti-concentration inequality for $g(W)$ is available, we can translate an upper bound on the distance of smooth functions \eqref{eq:IPM} into a direct bound on the distance of c.d.f.~\eqref{eq:Kol};
    \item Even if an anti-concentration inequality for $g(W)$ is not available, 
    if $t$ is a point of continuity of $g(W)$, as $\delta \rightarrow 0^+$, \eqref{eq:anti:conc} vanishes. Therefore if  \eqref{eq:IPM} converges to zero, and if $g(W)$ converges in distribution to some limit, we get that $g(V)$ asymptotically converges to the same limit.
\end{itemize}

We are ready to state our general universality result under local dependence.

\begin{theorem}[Universality under local dependence] \label{thm:universality:local:dependence} Let $g: \R^m \rightarrow \R$ be a thrice-differentiable function, and denote 
\begin{align*}
    m_N \;\coloneqq&\; \mfrac{m}{N^4} 
    &\text{ and }&&
    \epsilon_{\rm Var} \;\coloneqq&\; N \, \| \Var[V] - I_m \|_\infty\;.
\end{align*}
Then for any $t \in \R$,
\begin{align*}
    | \P( g(V) > t ) - \P( g(W) > t) |
    &\;\leq\;
    \inf_{\delta > 0}
    \Big\{ 
         \mfrac{\epsilon_{\rm Var}}{2} 
        \,
        \Big( 
            \mfrac{(\varphi_1(g))^2}{\delta^2}
            \,+\,
            \mfrac{\varphi_2(g)}{\delta}
        \Big)
         \,+\,
        \P( | g(\xi) - t | \leq \delta )
    \\
    &\hspace{4em}
        +
        \mfrac{ 4}{3 \sqrt{m_N}}
        \, 
        \Big( 
            \mfrac{(\varphi_1(g))^3}{\delta^3}
            \,+\,
            \mfrac{3 \varphi_1(g) \, \varphi_2(g)}{\delta^2}
            \,+\,
            \mfrac{\varphi_3(g)}{\delta}  \Big)
    \\
    &\hspace{5em}
        \,\times\,
        \max_{i \leq m} \max\big\{ \| V_i \|_{L_4}^3, \| W_i \|_{L_4}^3, 3^{\frac{3}{4}}\big\}
    \Big\}
    \;.
\end{align*}
Suppose in addition that $g$ is a polynomial function with degree at most $\kappa \in \N$. Then there exists a universal constant $C > 0$ such that
\begin{align*}
    &\;
    \msup_{t \in \R} | \P( g(V) > t ) - \P( g(W) > t) |
    \\
    &\;\leq\;
        \mfrac{1}{2} 
        \,
        \Big( 
            (\varphi_1(g))^2 
            \,
            \epsilon_{\rm Var}^{\frac{1}{2 \kappa + 1}}
            \,+\,
            \varphi_2(g) 
            \,
            \epsilon_{\rm Var}^{\frac{\kappa + 1}{2 \kappa + 1}}
        \Big)
        +
        \mfrac{C \kappa}{\| g(\xi) \|_{L_2}^{1/\kappa}}  
        \, 
        \max\Big\{ 
            m_N^{- \frac{1}{2(3\kappa + 1)}}
            \,,\,
            \epsilon_{\rm Var}^{\frac{1}{2 \kappa + 1}}
        \Big\}
        \;.
    \\
    &\qquad
        +
        \mfrac{ 4}{3}
        \, 
        \Big( 
            (\varphi_1(g))^3
            \,
            m_N^{- \frac{1}{2(3\kappa + 1)}}
            \,+\,
            3 \varphi_1(g) \, \varphi_2(g) \,
            m_N^{- \frac{\kappa + 1}{2(3\kappa + 1)}}
            \,+\,
            \varphi_3(g)
            \,
            m_N^{- \frac{2\kappa + 1}{2(3\kappa + 1)}}
        \Big)
    \\
    &\qquad\qquad \;\times\;
        \max_{i \leq m} \max\big\{ \| V_i \|_{L_4}^3, \| W_i \|_{L_4}^3, 3^{\frac{3}{4}}\big\}
    \;.
\end{align*}
\end{theorem}

\begin{remark}[Extensions] \label{remark:universality:local:dependence:extension} We note the following straightforward extensions:
\begin{proplist}
    \item The second bound above also extends to estimators that are approximately polynomial in the $L_2$-norm; see \cite{huang2024gaussian};
    \item The assumption that $\mean[W]=0$ and $\Var[W]=I_m$ is mainly for convenience. To relax them to arbitrary mean and a positive-definite $\Var[W]$, it suffices to replace $g(\argdot)$ by $g( \Var[W]^{-1/2}( \argdot - \mean[W]))$ and measure the local dependence of the random vector 
    \begin{align*}
        \Var[W]^{-1/2}( V - \mean[W])
    \end{align*}
    instead of $V$;
    \item The assumption that $\mean[V] = \mean[W]$ can be relaxed at the cost of an additional error term that captures mean mismatch, analogous to $\epsilon_{\rm Var}$. This is achieved by applying a first-order Taylor expansion for an appropriately chosen smooth function $\gamma_\delta$ (defined in the proof of \Cref{thm:universality:local:dependence} below):
    \begin{align*}
        &\;
        | \mean[ \gamma_\delta(V) - \gamma_\delta(V - \mean[V] + \mean[W] )  ] |
        \\
        &\;\leq\; 
        \msum_{i=1}^m \; \msup_{\theta \in [0,1]}\, \mean | \partial_i \gamma_\delta(V - \theta(\mean[V] - \mean[W])  ) | \,\times\, | \mean[V_i] - \mean[W_i]| 
        \;.
    \end{align*}
    Up to a redefinition of $\varphi_1(g)$ to include additional interpolation paths, the mean mismatch will introduce an additional requirement that
    \begin{align*}
        \sqrt{m} \, \| \mean[V] - \mean[W] \|_\infty \, \varphi_1(g) \;=\; o(1)\;;
    \end{align*}
    \item The $\| \argdot \|_\infty$ norm in the definition of $\epsilon_{\rm Var}$ can lead to additional logarithmic dependence on $m$. By inspecting the proof of \Cref{thm:universality:local:dependence}, it is easy to see that the same result holds with $\epsilon_{\rm Var}$ replaced by
    \begin{align*}
        \epsilon_{\rm Var}
        \;=\;
        \mfrac{1}{m} \msum_{i \leq m} \mfrac{1}{|\cM^V_i|} \msum_{j \in \cM^V_i} \big| \mean[V_i V_j] -  \ind_{\{i = j\}} \big|
        \;.
    \end{align*}
\end{proplist}
We do not state these extensions in full for simplicity of presentation.
\end{remark}

\vspace{.5em}

The proof of \Cref{thm:universality:local:dependence} is by applying the next result to both $V$ and $W$ together with \Cref{lem:smooth:approx:ind} for approximating an indicator by a smooth function. 

\begin{lemma}[Universality with respect to smooth functions] \label{lem:universality:smooth} Let $\delta > 0$ and $h_{t;\delta}$ be given as in \Cref{lem:smooth:approx:ind} with $\tilde m=3$ and $g: \R^m \rightarrow \R$ be a thrice-differentiable function. Then
    \begin{align*}
    &\;
    \sup_{t \in \R} \, | \mean[ h_{t;\delta}(g(V)) - h_{t;\delta}(g(\xi)) ] | 
    \\
    &\qquad\qquad\;\leq\;
    \mfrac{N}{2} 
    \,
    \Big( 
        \mfrac{(\varphi_1(g))^2}{\delta^2}
        \,+\,
        \mfrac{\varphi_2(g)}{\delta}
    \Big)
    \,
    \| \Var[V] - I_m \|_\infty
    \\
    &\qquad\qquad\qquad
    +
    \mfrac{2 N^2}{3 \sqrt{m}}
    \, 
    \Big( 
        \mfrac{(\varphi_1(g))^3}{\delta^3}
        \,+\,
        \mfrac{3 \varphi_1(g) \, \varphi_2(g)}{\delta^2}
        \,+\,
        \mfrac{\varphi_3(g)}{\delta}  \Big)
    \max_{i \leq m} \max\big\{ \| V_i \|_{L_4}^3, 3^{\frac{3}{4}}\big\}
    \;.
\end{align*}
\end{lemma}

The proof of \Cref{lem:universality:smooth} follows a similar recipe as \cite{mallory2025universality}, who prove Gaussian universality of the risk of high-dimensional logistic regression on data that exhibit data-wise block dependent and mixing structure and coordinate-wise local dependence structure. The main difference is that our random variables for comparison do not have exactly matching variance nor local dependency structures.

\begin{proof}[Proof of \Cref{lem:universality:smooth}] Denote the $\R^m \rightarrow \R$ function $\gamma_\delta = h_{t;\delta} \circ g$ and $e_i$ as the $i$-th standard basis vector in $\R^m$. Also denote $\langle \argdot, \argdot \rangle$ as the Euclidean inner product defined in the corresponding Euclidean space of the arguments. We first use the interpolation path $I(s)= \sqrt{s} V + \sqrt{1-s} \xi$ to write
\begin{align*}
    | \mean[ \gamma_\delta(V) - \gamma_\delta(\xi) ] | 
    \;=&\;
    \Big| \mint_0^1  \mean[ \partial_s \gamma_\delta( I(s) ) ] \big| \, ds \Big|
    \\
    \;=&\;
    \Big|
        \mint_0^1   
        \,
        \mean\Big[ 
            \msum_{i=1}^m 
            \big\langle \partial \gamma_\delta( I(s) ) \,,\, e_i \big\rangle 
            \Big( 
                \mfrac{V_i}{2\sqrt{s}} 
                - 
                \mfrac{\xi_i}{2\sqrt{1-s}}
            \Big) 
    \Big] \, ds 
    \Big| 
    \\
    \;\leq&\;
    \mint_0^1   
    \,
    \underbrace{
    \Big|
        \mean\Big[ 
            \msum_{i=1}^m 
            \,
            \big\langle \partial \gamma_\delta( I(s) ) \,,\, e_i \big\rangle 
            \Big( 
                \mfrac{V_i}{2\sqrt{s}} 
                - 
                \mfrac{\xi_i}{2\sqrt{1-s}}
            \Big) 
        \Big] 
    \Big| 
    }_{\eqqcolon Q(s)}
    \, ds 
    \;.
\end{align*}
For convenience we shall define 
\begin{align*}
    I^{\cM^V_i}(s)
    \;\coloneqq\; 
    I(s) 
    -
     I^{(i)}(s, 0)
    \;=\;
    I^{(i)}(s, 1)
    -
    I^{(i)}(s, 0)
     \;.
\end{align*}
By a second-order Taylor expansion with the integral remainder, we obtain that almost surely
\begin{align*}
    \big\langle & \partial \gamma_\delta( I(s) ) \,,\, e_i \big\rangle 
    \;=\;
    \Big\langle  \partial \gamma_\delta\Big(  I^{(i)}(s, 0) + I^{\cM^V_i}(s) \Big)\,,\, e_i \Big\rangle
    \\
    &\;=\;
    \big\langle  \partial \gamma_\delta\big( I^{(i)}(s, 0)  \big) \,,\, e_i \big\rangle
    +
    \Big\langle
    \partial^2 \gamma_\delta\big( I^{(i)}(s, 0)  \big) \,,\, e_i \otimes I^{\cM^V_i}(s)  \Big\rangle
    \\[.5em]
    &\qquad
    +
    \mfrac{1}{2}
    \mint_0^1
    \Big\langle
        \partial^3 \gamma_\delta\big( I^{(i)}(s, \theta) \big) 
        \,,\, 
        e_i \otimes I^{\cM^V_i}(s)  \otimes I^{\cM^V_i}(s) 
    \Big\rangle
    \,(1-\theta)^2
    \, d\theta
    \\
    &\;=\;
    \big\langle  \partial \gamma_\delta\big( I^{(i)}(s, 0)  \big) \,,\, e_i \big\rangle
    \\
    &\qquad 
    +
    \msum_{j \in \cM^V_i}  
    \big\langle
    \partial^2 \gamma_\delta\big( I^{(i)}(s, 0)  \big)  \,,\, e_i \otimes e_j \big\rangle
    \,
    \big( \sqrt{s}\;  V_j  \,+\, \sqrt{1-s} \; \xi_j  \big)
    \\
    &\qquad 
    +
    \mfrac{1}{2}
    \msum_{j,l \in \cM^V_i}  
    \,
    \bigg(
    \mint_0^1
    \Big\langle
    \partial^3 \gamma_\delta\big( I^{(i)}(s, \theta) \big)  \,,\, e_i \otimes e_j \otimes e_l \Big\rangle
    \,(1-\theta)^2 \,
     d\theta
    \\
    &\hspace{8.5em}
        \,\times\, 
        \big( \sqrt{s}\;  V_j  \,+\, \sqrt{1-s} \; \xi_j  \big)
        \big( \sqrt{s}\;  V_l  \,+\, \sqrt{1-s} \; \xi_l  \big)
    \bigg)
    \;. 
\end{align*}
Plugging this into $Q(s)$, while noting that $\mean[V_i] = \mean[\xi_i] = 0$ and $(V_i, \xi_i)$ is independent of $I^{(i)}(s,0)$ by construction, we obtain 
\begin{align*}
    Q(s)
    \;=&\;
    \mfrac{1}{2}
    \msum_{i=1}^m
    \msum_{j \in \cM^V_i}  
    \mint_0^1
    \Big|
    \mean\big[
        \big\langle
        \partial^2 \gamma_\delta\big( I^{(i)}(s, 0)  \big)  \,,\, e_i \otimes e_j \big\rangle
    \big]
        \,
    \mean\big[ V_i V_j  \,-\, \xi_i \xi_j  \big]
    \Big| 
    d s
    \\
    &\;
    +
        \mfrac{1}{2}
        \msum_{i=1}^m
        \msum_{j,l \in \cM^V_i}  
        \,
        \mean
        \bigg[
        \mint_0^1
        \Big\langle
    \partial^3 \gamma_\delta\big( I^{(i)}(s, \theta) \big)  \,,\, e_i \otimes e_j \otimes e_l \Big\rangle
        \,(1-\theta)^2 \,
        d\theta
    \\
    &\hspace{11em}
        \times 
            \big( \sqrt{s}\;  V_j  \,+\, \sqrt{1-s} \; \xi_j  \big)
            \big( \sqrt{s}\;  V_l  \,+\, \sqrt{1-s} \; \xi_l  \big)
    \\
    &\hspace{11em}
        \times 
            \Big( 
                \mfrac{V_i}{2\sqrt{s}} 
                - 
                \mfrac{\xi_i}{2\sqrt{1-s}}
            \Big) 
        \bigg]
    \bigg|
    \\
    \;\overset{(a)}{\leq}&\;
     \mfrac{m N}{2} \, \Delta_2(\gamma_\delta) \, \| \Var[V] - I_m \|_\infty
     \\
     &\;
     +
     \mfrac{m N^2}{6} 
    \, 
    \Delta_3(\gamma_\delta)
    \,
    \max_{i \leq m} \max\{ \| V_i \|_{L_4}^3 \,,\, \| \xi_i \|_{L_4}^3 \}
    \,
    \Big( \mfrac{4}{\sqrt{s}} + \mfrac{4}{\sqrt{1-s}} \Big)
    \\
    \;\overset{(b)}{\leq}&\;
    \mfrac{m N}{2} \, \Delta_2(\gamma_\delta) \, \| \Var[V] - I_m \|_\infty
    \\
    &\;+
    \mfrac{m N^2}{6}
    \, 
    \Delta_3(\gamma_\delta)
    \,
    \max_{i \leq m} \max\{ \| V_i \|_{L_4}^3, 3^{\frac{3}{4}}\}
    \,
    \Big( \mfrac{1}{\sqrt{s}} + \mfrac{1}{\sqrt{1-s}} \Big)
     \;.
\end{align*}
In $(a)$, we have used the triangle inequality and H\"older's inequality, and denoted
\begin{align*}
    \Delta_2(\gamma_\delta)
    \;\coloneqq&\;
    \sup_{i,j,l \leq m;\, \theta, s \in [0,1]}
    \,
    \big\|
            \big\langle
                \partial^2 \gamma_\delta\big( I^{(i)}(s, \theta) \big) 
                \,,\, e_i \otimes e_j 
            \big\rangle
    \big\|_{L_4}
    \;,
    \\
    \Delta_3(\gamma_\delta)
    \;\coloneqq&\;
    \sup_{\substack{i,j,l \leq m \\ \theta, s \in [0,1]}}
    \,
    \big\|
            \big\langle
                \partial^3 \gamma_\delta\big( I^{(i)}(s, \theta) \big) 
                \,,\, e_i \otimes e_j \otimes e_l
            \big\rangle
    \big\|_{L_4}
    \;.
\end{align*}
In $(b)$, we have used the moment formula of $\cN(0,1)$ to compute $\|\xi_i\|_{L_4}^3$. Integrating $Q(s)$ over $s \in [0,1]$, we obtain 
\begin{align*}
    | \mean[ \gamma_\delta(V) - \gamma_\delta(\xi) ] | 
    \;\leq&\;
    \mfrac{m N}{2} \, \Delta_2(\gamma_\delta) \, \| \Var[V] - I_m \|_\infty
    +
     \mfrac{2 m N^2 }{3}
    \, 
    \Delta_3(\gamma_\delta)
    \,
    \max_{i \leq m} \max\{ \| V_i \|_{L_4}^3, 3^{\frac{3}{4}}\}
    \;.
\end{align*}
Recall that $\gamma_\delta =  h_{t;\delta} \circ g$, where $h_{t;\delta}$ is thrice differentiable with $|\partial^r h_{t;\delta}|$ uniformly bounded by $\delta^{-r}$ for $r=1,2,3$. Also recall that by definition,  for $r=1,2,3$,
\begin{align*}
        m^{r/2} 
    \hspace{-.5em} 
    \sup_{i,j \leq m; \, \theta, s \in [0,1]}
    \,
    \big\|\, 
        \big\| \partial_j^r g\big( I^{(i)}(s, \theta) \big)  \big\|_\infty
        \,
    \big\|_{L_4}
    \;\leq\;
    \varphi_r(g)
    \;.
\end{align*}
By the chain rule, we can compute 
\begin{align*}
    \Delta_2(\gamma_\delta)
    \;\leq&\;
    m^{-1}
    \,
    \big(
    \,
    \delta^{-2} \, (\varphi_1(g))^2
    \,+\,
    \delta^{-1} \, \varphi_2(g) 
    \big)
    \;,
    \\
    \Delta_3(\gamma_\delta)
    \;\leq&\;
    m^{-3/2}
    \,
    \big(
    \,
    \delta^{-3} \, (\varphi_1(g))^3
    \,+\,
    3 \delta^{-2} \, \varphi_1(g) \, \varphi_2(g)
    \,+\,
    \delta^{-1} \, \varphi_3(g) 
    \big)
    \;.
\end{align*}
This implies
\begin{align*}
    | \mean[ \gamma_\delta(V) - \gamma_\delta(\xi) ] | 
    \;\leq&\;
    \mfrac{N}{2} 
    \,
    \Big( 
        \mfrac{(\varphi_1(g))^2}{\delta^2}
        \,+\,
        \mfrac{\varphi_2(g)}{\delta}
    \Big)
    \,
    \| \Var[V] - I_m \|_\infty
    \\
    &\;
    +
    \mfrac{2 N^2}{3 \sqrt{m}}
    \, 
    \Big( 
        \mfrac{(\varphi_1(g))^3}{\delta^3}
        \,+\,
        \mfrac{3 \varphi_1(g) \, \varphi_2(g)}{\delta^2}
        \,+\,
        \mfrac{\varphi_3(g)}{\delta}  \Big)
    \\
    &\;\hspace{1em} \,\times\,
    \max_{i \leq m} \max\big\{ \| V_i \|_{L_4}^3, 3^{\frac{3}{4}}\big\}
    \;.
\end{align*}
Recalling that  $\gamma_\delta = h_{t;\delta} \circ g$ and taking a supremum over $t \in \R$ proves the desired bound.
\end{proof}

\vspace{.5em}

We are now ready to prove  \Cref{thm:universality:local:dependence}.

\vspace{.5em}

\begin{proof}[Proof of \Cref{thm:universality:local:dependence}] Let $t \in \R$, $\delta > 0$ and $h_{t;\delta}$ be given as in \Cref{lem:smooth:approx:ind} with $\tilde m=3$. By \Cref{lem:smooth:approx:ind}, we can write 
\begin{align*}
    &\;
    | \P( g(V) > t ) - \P( g(W) > t) |
    \\
    &\;=\;
    \max\{ \P( g(V) > t ) - \P( g(W) > t) \,,\,
    \P( g(W) > t) - \P( g(V) > t ) \}
    \\
    &\;\leq\;
    \max\big\{ 
        \mean[ h_{t;\delta}(g(V)) - h_{t+\delta;\delta}(g(W))]
        \,,\,
        \mean[ h_{t;\delta}(g(W)) - h_{t+\delta;\delta}(g(V))]
    \big\}
    \\
    &\;\leq\;
    \msup_{t \in \R}\big| \mean[ h_{t;\delta}(g(V)) - h_{t;\delta}(g(\xi))] \big|
    +
    \msup_{t \in \R}\big| \mean[ h_{t;\delta}(g(W)) - h_{t;\delta}(g(\xi))] \big|
    \\
    &\qquad
    + 
    \underbrace{
        |\mean[ h_{t;\delta}(g(\xi)) - h_{t+\delta;\delta}(g(\xi))]
        |
    }_{
        \eqqcolon \, (\star)
    }\;.
\end{align*}
The first two quantities are controlled by \Cref{lem:universality:smooth}, whereas by \Cref{lem:smooth:approx:ind} again, 
\begin{align*}
    (\star) 
    \;\leq\;
    \big|
            \P(g(\xi) > t - \delta) 
            - 
            \P( g(\xi) > t + \delta)
    \big|
    \;\leq\;
    \P( | g(\xi) - t | \leq \delta )
    \;.
\end{align*}
Combining the bounds, taking an infimum over $\delta > 0$ and recalling that $\Var[W] = I_m$, we obtain 
\begin{align*}
    &\;
    | \P( g(V) > t ) - \P( g(W) > t) |
    \\
    &\;\leq\;
    \inf_{\delta > 0}
    \Big\{ 
         \mfrac{N}{2} 
        \,
        \Big( 
            \mfrac{(\varphi_1(g))^2}{\delta^2}
            \,+\,
            \mfrac{\varphi_2(g)}{\delta}
        \Big)
        \,
        \| \Var[V] - I_m \|_\infty
         \,+\,
        \P( | g(\xi) - t | \leq \delta )
    \\
    &\hspace{4em}
        +
        \mfrac{ 4 N^2}{3 \sqrt{m}}
        \, 
        \Big( 
            \mfrac{(\varphi_1(g))^3}{\delta^3}
            \,+\,
            \mfrac{3 \varphi_1(g) \, \varphi_2(g)}{\delta^2}
            \,+\,
            \mfrac{\varphi_3(g)}{\delta}  \Big)
    \\
    &\hspace{6em} \,\times\,
        \max_{i \leq m} \max\big\{ \| V_i \|_{L_4}^3, \| W_i \|_{L_4}^3, 3^{\frac{3}{4}}\big\}
    \Big\}
    \;.
\end{align*}
Recalling that $m_N = \frac{m}{N^4}$ and $\epsilon_{\rm Var} = N \, \| \Var[V] - I_m \|_\infty$ gives the first desired bound. 

\vspace{.5em}

Now suppose $g$ is a polynomial function with degree at most $\kappa \in \N$. Since $\xi \sim \cN(0,I_m)$, we can apply the Carbery-Wright inequality (Theorem 8 of \cite{carbery2001distributional}) to obtain that, for some universal constant $C > 0$,
\begin{align*}
     \P( | g(\xi) - t | \leq \delta )
     \;\leq\;
     C \kappa  \, \Big(\mfrac{ \delta}{\| g(\xi) \|_{L_2}} \Big)^{1/\kappa} \;.
\end{align*}
This allows us to choose 
\begin{align*}
    \delta
    \;=\;
    \max\Big\{  
            m_N^{- \frac{1}{2(3\kappa + 1)}}
            \,,\,
            \epsilon_{\rm Var}^{\frac{1}{2 \kappa + 1}}
        \Big\}
    \;=\; 
    \max\Big\{ 
        \Big( \mfrac{N^2}{\sqrt{m}} \Big)^{\frac{\kappa}{3\kappa + 1}}
        \,,\,
        \big( N \, \| \Var[V] - I_m \|_\infty \big)^{\frac{\kappa}{2 \kappa + 1}}
    \Big\}
\end{align*}
to obtain 
\begin{align*}
    &\;
    | \P( g(V) > t ) - \P( g(W) > t) |
    \\
    &\;\leq\;
        \mfrac{1}{2} 
        \,
        \Big( 
            (\varphi_1(g))^2 
            \,
            \epsilon_{\rm Var}^{\frac{1}{2 \kappa + 1}}
            \,+\,
            \varphi_2(g) 
            \,
            \epsilon_{\rm Var}^{\frac{\kappa + 1}{2 \kappa + 1}}
        \Big)
        +
        \mfrac{C \kappa}{\| g(\xi) \|_{L_2}^{1/\kappa}}  
        \, 
        \max\Big\{ 
            m_N^{- \frac{1}{2(3\kappa + 1)}}
            \,,\,
            \epsilon_{\rm Var}^{\frac{1}{2 \kappa + 1}}
        \Big\}
        \;.
    \\
    &\qquad
        +
        \mfrac{ 4}{3}
        \, 
        \Big( 
            (\varphi_1(g))^3
            \,
            m_N^{- \frac{1}{2(3\kappa + 1)}}
            \,+\,
            3 \varphi_1(g) \, \varphi_2(g) \,
            m_N^{- \frac{\kappa + 1}{2(3\kappa + 1)}}
            \,+\,
            \varphi_3(g)
            \,
            m_N^{- \frac{2\kappa + 1}{2(3\kappa + 1)}}
        \Big)
    \\
    &\qquad\qquad \;\times\;
        \max_{i \leq m} \max\big\{ \| V_i \|_{L_4}^3, \| W_i \|_{L_4}^3, 3^{\frac{3}{4}}\big\}
    \;.
\end{align*}
Taking a supremum over $t \in \R$ proves the second bound.
\end{proof}

\subsection{Finite-sample and conditional variant of \Cref{thm:universality}} \label{appendix:universality:finite}

We inherit the notation and assumptions outlined in  \Cref{sec:validity:universality} for \Cref{thm:universality}, with one exception: Instead of imposing \Cref{asst:density:smooth} on density smoothness, we explicitly include the following anti-concentration term in the bound:
\begin{align*}
        \cP_{\rm ac}(\delta)
        \;\coloneqq\;
        \msup_{t \in \R} \P_{Z' \sim \cN(0, I_{nd})}\big( \,\big| f(Z') - t \big| \leq  \delta \, \,\big|\, \Xobs \big)
        \quad 
        \text{ for } \;\; 
        \delta > 0\;.
\end{align*}
We also denote the constant in \Cref{asst:moment:bounded} as 
\begin{align*}
    c_{4;6\nu}
    \;\coloneqq\;
     \mmax_{i \leq n, j \leq d} 
        \,
        \max\big\{ 
            \big\| \, \| (\Phi_1(X))_{ij} \|_{L_4 | X} \,\big\|_{L_{6\nu}}
            \,,\,
            \| X_{ij} \|_{L_4}
        \big\}
\end{align*}

Notice that \Cref{asst:density:smooth} implies the existence of some universal constant $K > 0$ such that
$
    \cP_{\rm ac}(\delta) \leq K \delta
$ 
for all $\delta \geq 0$. In contrast, our bound below allows $\cP_{\rm ac}(\delta)$ to decay sub-linearly in $\delta$ as $\delta \rightarrow 0^+$, which happens e.g.~when $f$ is a bounded-degree polynomial. See \Cref{appendix:universality:local:dependence} for a further discussion on anti-concentration bounds. To state the bound, we also use the notation 
\begin{align*}
    \tilde n \;\coloneqq&\; \mfrac{nd}{(N_{\rm dep})^4} 
    &\text{ and }&&
    \epsilon_{\rm inv} \;\coloneqq&\; N_{\rm dep} \,\times\, \big\| \, \| \Var[\Phi_1(X) | X ] - I_m \|_\infty \,\big\|_{L_{3\nu/2 }}\;.
\end{align*}

\begin{theorem} \label{thm:universality:finite} 
Then under \Cref{asst:standard,asst:zero:mean}, we have
\begin{align*}
    \| \Delta_{\rm inv}(X) \|_{L_\nu}
    \;\leq\;
    \inf_{\delta > 0}
    \big\{ 
        \cE_{\rm inv}(\delta)
        +
        \cE_{\rm uni}(\delta)
        +
        \cP_{\rm ac}(\delta)
    \big\} 
    \;,
\end{align*}
where 
\begin{align*}
    \cE_{\rm inv}(\delta)
    \;\coloneqq&\;
         \mfrac{\epsilon_{\rm inv}}{2} 
        \,
        \Big( 
            \mfrac{(\varphi_{1;6\nu}(g))^2}{\delta^2}
            \,+\,
            \mfrac{\varphi_{2;5\nu}(g)}{\delta}
        \Big)
    \;,
    \\
    \cE_{\rm uni}(\delta)
    \;\coloneqq&\;
        \mfrac{ 4}{3 \sqrt{\tilde n}}
        \, 
        \Big( 
            \mfrac{(\varphi_{1;6\nu}(g))^3}{\delta^3}
            \,+\,
            \mfrac{3 \varphi_{1;5\nu}(g) \, \varphi_{2;5\nu}(g)}{\delta^2}
            \,+\,
            \mfrac{\varphi_{3;4\nu}(g)}{\delta}  \Big)
            \max\big\{ c_{4;6\nu}^3\,,\, 3^{\frac{3}{4}}\big\}
    \;.
\end{align*}
Suppose in addition that $f$ is a polynomial function with degree at most $\kappa \in \N$. Then there exists an universal constant $C > 0$ such that
\begin{align*}
    \| \Delta_{\rm inv}(X) \|_{L_\nu}
    \;\leq\;
    \overline \cE_{\rm inv}
    +
    \overline \cE_{\rm uni}
    +
    C \, \overline \cP_{\rm ac}
    \;,
\end{align*}
where 
\begin{align*}
    \overline \cE_{\rm inv}
    \;\coloneqq&\;
    \mfrac{1}{2} 
        \,
        \Big( 
            (\varphi_{1;6\nu}(g))^2 
            \,
            \epsilon_{\rm inv}^{\frac{1}{2 \kappa + 1}}
            \,+\,
            \varphi_{2;5\nu}(g) 
            \,
            \epsilon_{\rm inv}^{\frac{\kappa + 1}{2 \kappa + 1}}
        \Big)
    \;,
    \\
    \overline \cE_{\rm uni}
    \;\coloneqq&\;
    \mfrac{ 4}{3}
        \, 
        \Big( 
            (\varphi_{1;6\nu}(g))^3
            \,
            \tilde n^{- \frac{1}{2(3\kappa + 1)}}
            \,+\,
            3 \varphi_{1;5\nu}(g) \, \varphi_{2;5\nu}(g) \,
            \tilde n^{- \frac{\kappa + 1}{2(3\kappa + 1)}}
            \,+\,
            \varphi_{3;4\nu}(g)
            \,
            \tilde n^{- \frac{2\kappa + 1}{2(3\kappa + 1)}}
        \Big)
    \\
    &\;\quad\;
    \;\times\;
        \max_{i \leq m} \max\big\{  c^3_{4;6\nu}, 3^{\frac{3}{4}}\big\}
    \;,
    \\
    \overline \cP_{\rm ac}
    \;\coloneqq&\;
    \mfrac{\kappa}{\| f(Z') \|_{L_2}^{1/\kappa}}  
        \, 
        \max\Big\{ 
            \tilde n^{- \frac{1}{2(3\kappa + 1)}}
            \,,\,
            \epsilon_{\rm inv}^{\frac{1}{2 \kappa + 1}}
        \Big\}
    \;.
\end{align*}
\end{theorem} 

\begin{proof}[Proof of \Cref{thm:universality:finite}] This follows directly from \Cref{thm:universality:local:dependence} by replacing $g$ with $f$, $V$ with $\Phi_1(X)|X$ and $W$ with $X$, followed by taking $\| \argdot \|_{L_\nu}$ and H\"older's inequality.
\end{proof}

\begin{proof}[Proof of \Cref{thm:universality}] 
Note that under the conditions of \Cref{thm:universality}, $\tilde n \rightarrow \infty$ and $\epsilon_{\rm inv} \rightarrow 0$. \Cref{asst:stab} and \Cref{asst:moment:bounded} respectively provide the required upper bounds on the stability terms and $c_{4;6\nu}$, and \Cref{asst:density:smooth} implies that $\| f(Z') \|_{L_2} = \Omega(1)$. Therefore \Cref{thm:universality:finite} proves that $\| \Delta_{\rm inv}(X) \|_{L_\nu} \rightarrow 0$. When additionally \Cref{asst:iid:smooth} holds, we can apply this to the bound of \Cref{thm:Del:Kol:approx:inv}, which is independent of $r$, to obtain the desired uniform convergence.
\end{proof}

\vspace{.5em}

 The bound in \Cref{thm:universality:finite} consist of two parts: the stability error $\cE_{\rm uni}(\delta)$ and anti-concentration bound $\cP_{\rm ac}(\delta)$ arise from the universality approximation, whereas $\cE_{\rm inv}(\delta)$ is an approximate invariance error. To interpret the bounds, we observe the following:
\begin{itemize}
    \item \textbf{Interpretation of the universality error. } Suppose that $X$ has i.i.d.~data with i.i.d.~coordinates, and $\Phi_1(X)|X$ also has i.i.d.~data with i.i.d.~coordinates (e.g.~if $\Phi_1$ performs coordinate-wise i.i.d.~noise injection on the entire dataset). In this case, $N_{\rm dep}=1$. Suppose further that $f$ is a degree-$\kappa$ polynomial. Ignoring the approximate invariance error $\epsilon_{\rm inv}$, the bound in \Cref{thm:universality:finite} decays at a rate $(nd)^{-\frac{1}{6\kappa+2}}$, which gets slower for larger values of $\kappa$. The dependence on $n$ may appear sub-optimal since, if e.g.~$f$ is a degree-two polynomial, the rate scales as $O(n^{-\frac{1}{14} })$. \cite{huang2024gaussian} shows that this rate is in fact  nearly optimal: They construct degree-$\kappa$ polynomials of random vectors such that the lower bound for the universality approximation is $\Omega(n^{-\frac{1}{6\kappa} })$.
    A main difference of our bound from \cite{huang2024gaussian} is the local dependency measure $N_{\rm dep}$: This allows for more general dependence across observations and across coordinates, provided that 
    \begin{align*}
        N_{\rm dep} \;\ll\; (nd)^{1/4}\;.
    \end{align*}
    As considered in \cite{mallory2025universality} for high-dimensional logistic regression and \cite{huang2022data} for general estimators, this dependence via $N_{\rm dep}$ will prove to be crucial for analysing the effects of general data augmentations.
    \item  \textbf{Interpretation of the approximate invariance error. } Up to a vanishing universality approximation error, \Cref{thm:universality:finite} controls the approximate invariance error in Kolmogorov distance, $\Delta_{\rm inv}(X)$, via the approximate invariance error in the variance, $\epsilon_{\rm inv}$. This can be much easier to compute, as considered in \Cref{appendix:theory:examples}. We also note that $\epsilon_{\rm inv}$ trades off against the stability terms $\varphi_{1;6}$ and $\varphi_{2;5}$, which is similar to the observation in \cite{austern2020bootstrap} for bootstrap of general functions.
\end{itemize}

\subsection{Computing local dependence measure} \label{appendix:verify:local}

The measure of local dependency in \Cref{thm:universality:finite} is given by
\begin{align*}
    N_{\rm dep} \;\equiv\; N_{\rm dep}(X) \;=\; \max\{ N(\mu_{\Phi_1(X)|X}) \,,\, N(\mu_X) \}\;.
\end{align*}
$N(\mu_X)$ is determined purely by the original dataset, and imposes conditions on observation-wise and coordinate-wise dependence. In this section we focus on computing the value of 
\begin{align*}
     \msup_{\bx \in \R^{nd}} \, N(\mu_{\Phi_1(\bx)}) \;,
\end{align*}
which is a deterministic upper bound to $ N(\mu_{\Phi_1(X)|X})$ and determined entirely by the distribution of $\Phi_1$. Throughout, given a generic vector $\bw \in \R^b$, we denote $(\bw)_l$ as the $l$-th coordinate of $\bw$ and write $(\bw)_A \coloneqq (\bw)_{l \in A}$ for a subset $A \subseteq [b]$.

\vspace{.5em}

\paragraph{``Local" transformations. } We first consider transformations that are independent across many small subsets of coordinates.

\begin{lemma}[Independent transformations in many small subsets] \label{lem:aug:local} Suppose there exists a partition $(A_t)_{t \leq T}$ of the index set $[nd]$ such that for all $\bx \in \R^{nd}$, the collection of variables $\{ (\Phi_1(\bx))_{A_t} \}_{t \leq T}$ are independent. Then $N_* \coloneqq \sup_{\bx \in \R^{nd}} N( \mu_{\Phi_1(\bx)}) \leq \sup_{t \leq T} | A_t| $.
\end{lemma}

\begin{proof} The result follows from the definition of $ N( \mu_{\Phi_1(\bx)})$ by taking the local dependency neighbourhood of each coordinate $l$ to be the $A_t$ that contains $l$.
\end{proof}

\Cref{lem:aug:local} covers a range of ``local'' transformations used in practice. To see this, for simplicity we first assume that transformations are i.i.d.~across the $n$ different observations, i.e.~
\begin{align*}
    \Phi_1(X) \;=\; (\phi_1(X_1), \ldots, \phi_n(X_n) )\;,
\end{align*}
where $X = (X_1, \ldots, X_n)$ are $n$ $\R^d$-valued random vectors and $\phi_1, \ldots, \phi_n$ are i.i.d.~$\R^d \rightarrow \R^d$ transformations. Then \Cref{lem:aug:local} covers the following specifications of $\phi_1$:
\begin{proplist}
    \item \textbf{Random shuffling within small subsets of coordinates.} This is useful, for instance, in genomics data where the data vector can be divided into different groups of homogeneous coordinates \cite{zrimec2020deep,sanchez2021deep}. Formally, this is specified by partitioning $[d]$ into disjoint subsets $(\tilde A_t)_{t \leq \tilde T}$ and defining, for $x \in \R^d$ and $t \leq \tilde T$,
    \begin{align*}
        (\phi_1 (x))_{\tilde A_t} \coloneqq ( x_{\pi_{1t}(i)} )_{i \in \tilde A_t} 
        \tagaligneq \label{eq:random:permute}
    \end{align*}
    where each $\pi_{1t}$ is an independent uniform permutation on $\tilde A_t$ and $(\pi_{1t})_{t \leq \tilde T}$ are independent.
    In this case, the upper bound in \Cref{lem:aug:local} is exactly $N_*  \leq \sup_{t \leq \tilde T} |\tilde A_t|$;
    \item \textbf{Pixel-wise i.i.d.~image augmentations, such as jittering of brightness, hue, contrast and saturation. } An image 
    with 3 colour channels can be modelled as an $\R^{3d'}$ vector $x'$. Denote $(x')_l$ as the $l$-th pixel, which is a 3-dimensional vector. Then the transformations considered can then be formally specified as, for $l \in [d']$,
    \begin{align*}
        (\phi_1(x'))_l \;\coloneqq\; \phi_{1l}((x')_l)\;,
    \end{align*}
    where $(\phi_{1l})_{l \leq d'}$ are i.i.d.~$\R^3 \rightarrow \R^3$ transformations. In this case, $N_* = 3$;
    \item \textbf{Coordinate-wise i.i.d.~noise injection. } This can be formally specified as, for $x \in \R^d$ and $l \leq d$,
    \begin{align*}
        (\phi_1(x) )_l \;\coloneqq\; \psi_{1l}(x_l)\;,
    \end{align*}
    where $(\psi_{1l})_{l \leq d}$ are i.i.d.~$\R \rightarrow \R$ noise injection mechanisms, such as additive noise or multiplicative noise. In this case, $N_* = 1$.
\end{proplist}

\vspace{.5em}

\paragraph{``Global" transformations. } It is also common to consider transformations that introduce dependence across many coordinates, and we state some examples below. For simplicity, we use $x \in \R^d$ to represent an image with 1 colour channel, as an analogous formulation holds for the case with 3 colour channels.
\begin{proplist}
    \item \textbf{Random reflection of images.} Consider $x \in \R^d$ with $d$ even. Upon an appropriate ordering of the $d$ pixels, a random reflection can be formulated as the transformation $\phi$ such that, for $l \leq d$, 
    \begin{align*}
        (\phi_1(x))_l 
        \;\coloneqq\; 
        x_{\pi_1(l)}\;,
        \tagaligneq \label{eq:formulate:reflection}
    \end{align*}
    where $\pi_1: [d] \rightarrow [d]$ is either the identity map or the map $l \mapsto d+1-l$ with equal probabilities. In this case, $N_*=d$;
    \item \textbf{Random rotation of images.} Consider a stylised setup where rotation only cycles the pixels of an image $x \in \R^d$. Pixel rotation can be formulated by partitioning $[d]$ into disjoint equal-sized subsets $(\tilde A_t)_{t \leq \tilde T}$,  and defining $\phi_1$ as in \eqref{eq:formulate:reflection}, where $\pi$ is now a uniform cycling operation that preserves the partition and acts in the same way on all $\tilde A_t$. For instance when $\tilde T=2$, $\pi_1$ is defined by drawing uniformly a permutation $\pi'_1$ on $\{1, \ldots, d/2\}$ and setting 
    \begin{align*}
        \pi_1(l) \;\coloneqq&\; \pi'_1(l) \text{ for } l \leq d/2 
        &\text{ and }&&
        \pi_1(l) \;\coloneqq&\; \pi'_1(d-l) \text{ for } l > d/2 
        \;.
    \end{align*} 
    In this case, $N_*=d$;
    \item \textbf{Random permutations within large subsets of coordinates. } The formulation is identical to \eqref{eq:random:permute} except that $\sup_{t \leq \tilde T} |\tilde A_t|$ is large in this case. For instance, random permutations of $n$ data points lead to $N_* = n$;
    \item \textbf{Random resized cropping, random subset selection and random projection. }  Another common image augmentation technique is \texttt{RandomResizedCrop}, which randomly selects a window from the image for cropping and resizes the obtained pixels to the original image size with interpolation. For simplicity we only consider a toy setup: We partition $[d]$ into disjoint equal-sized subsets $(\tilde A_t)_{t \leq \tilde T}$, draw $\tau_1 \sim \textrm{Uniform}[\tilde T]$, and fill the $d$ pixels by $(x)_{\tilde A_{\tau_1}}$ by repeating the pixels. Formally, writing $k= d / \tilde T$, we define
    \begin{align*}
        (\phi_1(x))_{(l-1)\tilde T + s}
        \;\coloneqq\;
        \text{ $l$-th element of $(x)_{\tilde A_{\tau_1}}$ }
    \end{align*}
    for $1 \leq l \leq k$ and $1 \leq s \leq \tilde T$. Note that this formulation also covers random disjoint subset selection and random projections onto disjoint subsets of data. In this case, $N_*=d$;
    \item \textbf{Mixup augmentation.} Mixup is an augmentation technique that has seen popularity across a large range of data domains, including images, text, audio and video data \citep{chen2020mixtext,meng2021mixspeech,yoon2021ssmix,sahoo2021contrast,jin2024survey}. The exact implementation varies across applications, but the shared core idea is to synthesise artificial data points that interpolate between the original data points. We again consider a simplified formulation, which mirrors the random convex combination adopted in the original mixup paper \citep{zhang2018mixup}: For each $1 \leq i \leq n$, we draw $\omega_i \sim \textrm{Uniform}[0,1]$ and $(\tau_{i1}, \tau_{i2})_{i \leq n}$ be i.i.d.~uniformly drawn two-tuples of distinct numbers from $[n]$. A ``bootstrapped'' version of mixup can then be described as 
    \begin{align*}
        \Phi_1(X)
        \;=\; 
        (\tilde \phi_1(X), \ldots, \tilde \phi_n(X) )
        \;,
        \quad
        \text{ where }
        \qquad
        \tilde \phi_i(X)
            \;=\; 
            \omega_i X_{\tau_{i1}}
            +
            (1-\omega_i) X_{\tau_{i2}}
        \;.
    \end{align*}
    In this case, $N_* = d$, since the mixup weight $\omega_i$ and the data indices $(\tau_{i1}, \tau_{i2})$ are all chosen independently across the $n$ observations but shared across all coordinates.
\end{proplist}

\vspace{1em}
 
For the transformations above, $N_*$ is large, which hurts the bound in \Cref{thm:universality:finite}. There are two remedies:
\begin{proplist}
    \item For $N_*=d$, $N_{\rm dep}^2 / \sqrt{nd} = O(d^{3/2} n^{-1/2})$. This introduces a constraint on the dimensionality $d = o(n^{1/3})$;
    \item For $N_* = n$ e.g.~in (c) above, $N_{\rm dep}^2 / \sqrt{nd} = O(n^{3/2} d^{-1/2})$, which does not converge when $d$ is fixed. Nevertheless, we will only consider this in the set up of a composite transformation 
    \begin{align*}
        \Phi_1 \;=\;  \cT_1 \circ \tilde \Phi_1 \;,
    \end{align*}
    where $\tilde \Phi_1$ and $\cT_1$ are independent random transformations: $\cT_1$ is an exactly exchangeable, inter-observation transformation e.g.~permutations and cycling of $n$ exchangeable data points, whereas 
    \begin{align*}
        \tilde \Phi_1 \;=\; (\tilde \phi_i)_{i \leq n}
    \end{align*}
    consists of $n$ i.i.d.~$\R^d \rightarrow \R^d$ random transformations each acting on a data point. In this case, we will first apply \Cref{lem:compose:exact:with:approx} in \Cref{sec:new:DAB:methods}, which reduces the analysis of $\Delta_{\rm inv}(X)$ to the sole effect of $\tilde \Phi_1$. The local dependency measure corresponding to $\tilde \Phi_1$ can then be bounded by a quantity completely determined by $\tilde \phi_1$, and the corresponding computation of $N_*$ is now independent of $n$. 
\end{proplist}

\vspace{.5em}

\noindent 
\textbf{Composing two approximately invariant transformations when one of them is linear.} Denote $\Phi_1= \cT_1 \circ \tilde \Phi_1$ as above.  When both $\cT_1$ and $\tilde \Phi_1$ satisfy only approximate invariance, \Cref{lem:compose:exact:with:approx} does not apply, and the local dependency measure $N(\mu_{\Phi_1(X)|X})$ can be troublesome to compute. Nevertheless, when $\cT_1$ takes value in the set of linear transformations, it suffices to consider the local dependency measures
\begin{align*}
    \sup_{\bx \in \R^{nd}} N(\mu_{\Var[\cT_1(\bx)]^{-1/2} \cT_1(\bx)})
    \qquad\qquad 
    \text{ and }
    \qquad\qquad 
    \sup_{\bx \in \R^{nd}}  N(\mu_{\Var[\tilde \Phi_1(\bx)]^{-1/2} \tilde \Phi_1(\bx)})\;,
\end{align*}
which are quantities that separately concern the dependencies introduced by $\cT_1$ and $\tilde \Phi_1$. This is achieved by a two-step application of our general universality result (\Cref{thm:universality:local:dependence} with \Cref{remark:universality:local:dependence:extension}):
\begin{itemize}
    \item We first condition on $\cT_1$ and apply universality to the function $f(\cT_1( \argdot ))$, whose stability properties are the same as $f$ by the linearity of $\cT_1$. This allows us to obtain, conditionally on $\cT_1$ and $X$, the distributional approximation
\begin{align*}
    f( \cT_1 ( \tilde \Phi_1(X) ) )  
    \;\overset{d}{\approx}&\;
    f\big( 
        \cT_1 \big( \mean[\tilde \Phi_1(X) | X] + \sqrt{\Var[\tilde \Phi_1(X)| X]} \, \xi   \big) 
    \big)
    \\
    \;=&\;
    f\big( 
        \underbrace{
        \cT_1 \big( \mean[\tilde \Phi_1(X) | X] \big) 
        \,+\, 
        \cT_1 \big( \sqrt{\Var[\tilde \Phi_1(X)| X]} \, \xi \big)   
        }_{\eqqcolon (\star\star)}
    \big)
    \;,
\end{align*}
where $\xi \sim \cN(0, I_{nd})$ is independent of all other variables and the second line above follows by linearity of $\cT_1$. This step incurs an error that can be controlled purely by the local dependency measure $\sup_{\bx \in \R^{nd}} N(\mu_{\Var[\cT_1(\bx)]^{-1/2} \cT_1(\bx)})$;
    \item Next, we apply universality conditionally on $X$ to approximate $(\star\star)$ by a Gaussian, whose conditional mean and variance can be computed by the linearity of $\cT$ as 
\begin{align*}
    \mean[(\star\star) | X ]
    \;=&\;
    \mean\big[ \cT_1 \big( \mean[\tilde \Phi_1(X) | X] \big)  \big]
    \\
    \;=&\;
    \mean[ \cT_1 \circ \tilde \Phi_1(X) \,|\, X ]
    \;=\;
    \mean[ \Phi_1(X) | X]
    \;,
    \\
    \Var[(\star\star) | X ]
    \;=&\;
    \Var\big[ \cT_1 \big( \mean[\tilde \Phi_1(X) | X] \big)  + \cT_1 \big( \sqrt{\Var[\tilde \Phi_1(X)| X]} \, \xi \big)  \,\big|\, X   \big]
    \\
    \;\overset{(a)}{=}&\;
    \Var\big[ \cT_1 \big( \mean[\tilde \Phi_1(X) | X] \big) \,\big|\, X  \big]
    +
    \mean\big[ \cT_1 \Var[\tilde \Phi_1(X) | X] \cT_1^\top  \,\big|\, X \big]
    \\
    \;=&\;
    \mean\big[ \cT_1  \mean[\tilde \Phi_1(X) | X] \, \mean[\tilde \Phi_1(X) | X]^\top \cT_1^\top \big]
    \\
    &\;
    -
    \mean\big[ \cT_1 \circ \tilde \Phi_1(X) |X ] \, \mean\big[ \cT_1 \circ \tilde \Phi_1(X) |X ]^\top 
    \\
    &\;
    +
    \mean\big[ \cT_1 \Var[\tilde \Phi_1(X) | X] \cT_1^\top  \,\big|\, X \big]
    \\
    \;=&\;
    \mean\big[ \cT_1  \mean[\tilde \Phi_1(X) \Phi_1(X)^\top | X] \cT_1^\top \big]
    \\
    &\;-\;
    \mean\big[ \cT_1 \circ \tilde \Phi_1(X) |X ] \, \mean\big[ \cT_1 \circ \tilde \Phi_1(X) |X ]^\top 
    \\
    \;=&\;
    \Var[ \cT_1 \circ \tilde \Phi_1(X) \,|\, X ]
    \;=\;
    \Var[\Phi_1(X) \,|\, X]\;.
\end{align*}
In $(a)$ above, we have used the total law of variance. In other words, the second application of universality gives the correct Gaussian surrogate for $\Phi_1(X) | X$. Moreover, it incurs an error controllable purely by the local dependency measure $\sup_{\bx \in \R^{nd}} N(\mu_{\Var[\tilde \Phi_1(\bx)]^{-1/2} \tilde \Phi_1(\bx)})$.
\end{itemize}
In summary, for composite transformations $\cT_1 \circ \tilde \Phi_1$ where $\cT_1$ is linear, it suffices to consider the local dependency measures related to $\cT_1$ and $\tilde \Phi_1$ individually, and all computations in this section still apply. A formal statement of the universality result requires stating the stability interpolation paths $\cC$ for the two applications of universality, which is notationally cumbersome and hence omitted here.

\subsection{Verification of stability conditions} \label{appendix:verify:stability}

We verify stability conditions for two examples in this section: A polynomial estimator, which covers most estimators considered in this work, and the squared loss statistic used in conformal prediction, which demonstrates the lack of stability for vanilla conformal prediction residual.

\vspace{.5em}

\noindent 
\textbf{Polynomial estimators. } We consider estimators in the class of symmetric and polynomial $\R^{nd} \rightarrow \R$ functions. This can be extended to cover estimators that can be approximated by an $\kappa$-order Taylor expansion and that are symmetric under permutations of data entries. For the next result only, we view $\bx \in \R^{nd}$ as an $nd$-dimensional vector and use the single index $(\bx)_i$ to denote the $i$-th coordinate of $\bx$ for $1 \leq i \leq nd$.

\begin{lemma} \label{lem:stability:poly} Fix $\nu \geq 1$. Suppose $f=p_\kappa$, where  $p_\kappa: \R^{nd} \rightarrow \R$ is a degree-$\kappa$ polynomial that takes the form
\begin{align*}
    p_\kappa(x_1, \ldots, x_{nd})
    \;\coloneqq\; 
    \alpha_0
    + 
    \msum_{l=1}^\kappa  \Big( \mfrac{\alpha_l}{(nd)^{l/2}} \msum_{1 \leq i_1, \ldots, i_l \leq nd} \, (x_{i_1} \times \cdots \times \, x_{i_l})  \Big)
\end{align*}
for $\alpha_0, \ldots, \alpha_\kappa \in \R$ with $\alpha_{\rm max} \coloneqq \max_{1 \leq l \leq \kappa} |\alpha_l|$. Then for $r=1,2,3$,
\begin{align*}
    \varphi_{r;\nu}
    \;\leq\;
    \alpha_{\rm max}\,
    \msum_{l=1}^\kappa
    \Big(
        2^l
        \,
        \max\Big\{ 
        &\;
            \Big\|
            \,
            \big\| 
                    \tilde p_{l-r} \big(\Phi_1(X) \big) 
            \big\|_{L_4 | X}
            \,
            \Big\|_{L_\nu}
            \;,\;
            \big\|
            \,
                    \tilde p_{l-r} (X') 
            \,
            \big\|_{L_\nu}
    \\
    &\; 
            \Big\|
            \,
            \big\| 
                    \tilde p_{l-r} \big(Z^\Phi(X) \big) 
            \big\|_{L_4 | X}
            \,
            \Big\|_{L_\nu}
            \;,\;
            \big\|
            \,
                    \tilde p_{l-r} (Z') 
            \,
            \big\|_{L_\nu}
    \Big\}
    \Big)
    \;,
\end{align*}
where, for $s \in \Z$,  we have denoted 
\begin{align*}
    \tilde p_s(x_1, \ldots, x_{nd}) 
    \coloneqq 
    \begin{cases}
        \big(\frac{1}{\sqrt{nd}} \msum_{i=1}^{nd} x_i\big)^s
        & \text{ if } s \geq 0\;,
        \\
        0 & \text{ otherwise.}
    \end{cases}
\end{align*}
\end{lemma}

\begin{proof}[Proof of \Cref{lem:stability:poly}] This follows from applying the definition \eqref{eq:defn:varphi} and noting that, for $\bx = (x_i)_{i \leq nd}$ and $\by = (y_i)_{i \leq nd}$,
\begin{align*}
    \tilde p_s(\bx + \by) 
    \;=&\;
    \Big(\mfrac{1}{\sqrt{nd}} \msum_{i=1}^{nd} (x_i + y_i) \Big)^s
    \\
    \;\leq&\;
    2^{s-1}
    \,
    \Big(
    \Big(\mfrac{1}{\sqrt{nd}} \msum_{i=1}^{nd} x_i \Big)^s
    +
    \Big(\mfrac{1}{\sqrt{nd}} \msum_{i=1}^{nd} y_i \Big)^s
    \Big)
    \;.
    \tagaligneq \label{eq:crude:contro:polynomial}
\end{align*}
\end{proof}

The $L_\nu$ norm of $\tilde p_s$ can typically be shown to be $O\big( N_{\rm dep}^{s/2} \big)$ under independence, martingale dependence or mixing conditions; see e.g.~\cite{mackey2014matrix,dedecker2015moment,vershynin2018high,durmus2024probability}. The $2^m$-dependence can be removed by controlling the norm of $\tilde p_s$  directly on the interpolation path rather than applying the crude inequality in \eqref{eq:crude:contro:polynomial}; this is treated in \cite{huang2024gaussian}, which provides a stability control in the case $X_1, \ldots, X_n$ are independent and provided that $\kappa = o( \log(n) )$. We also refer readers to \cite{huang2024gaussian} for two detailed discussions:
\begin{itemize}
    \item \cite{huang2024gaussian} includes a generic treatment for asymmetric polynomials and approximately polynomial estimators: There, they also show that those polynomials covers a large class of estimators such as high-dimensional U-statistics and V-statistics, as well as other estimators that can be approximated by a delta method.
    \item While the stability terms $\varphi_{r;\nu}$ can additionally introduce factors of the form $N_{\rm dep}^{s/2}$, \cite{huang2024gaussian} also shows that depending on the precise dependence structure, these additional factors may cancel out with $\| f(Z')\|_{L_2}$, introduced by the anti-concentration term in \eqref{thm:universality:finite}. The precise condition on how much $N_{\rm dep}$ is allowed to grow with $n$ and $d$ then depends on individual estimators.
\end{itemize}

\vspace{.5em}

\noindent 
\textbf{Split conformal prediction in the high-dimensional regime. } We consider split conformal prediction with the squared loss error: For $X=(X_i,Y_i)_{i \leq n}$ consisting of random variables $X_i$ in $\R^d$ and $Y_i \coloneqq g_{\rm oracle}(X_i)$, we consider 
\begin{align*}
    f( X) 
    \;\coloneqq\; 
    | Y_1 - g_{\rm trained}(X_1)|^2 
    \;=\;
    | g_{\rm oracle}(X_1) - g_{\rm trained}(X_1)|^2 
    \;,
\end{align*}
where $g_{\rm trained}: \R^d \rightarrow \R$ is trained on a separate dataset and $g_{\rm oracle}: \R^d \rightarrow \R$ is the oracle output-generating function. This statistic appears to be different from those considered in existing universality literature, because the partial derivatives of $f$ with respect to $(X_i, Y_i)_{i >1}$ are all zero, but the partial derivative with respect to $(X_1,Y_1)$ is non-negligible. To address this case, we note that for the purpose of combining data augmentation with conformal prediction, a typical transformation takes the form
\begin{align*}
    \Phi_1 \;=\; \pi_1 \circ \tilde \Phi_1\;,
\end{align*}
where $\pi_1$ is a random cycling of the data index set $[n]$, whereas the approximately invariant transformation reads 
\begin{align*}
    \tilde \Phi_1( (X_1,Y_1), \ldots, (X_n,Y_n) ) 
    \;=&\;
    ( (\phi_1(X_1), Y_1 ), \ldots, (\phi_n(X_n),Y_n))
    \;, 
    \tagaligneq \label{eq:conformal:input:only}
\end{align*}
where $\phi_i: \R^d \rightarrow \R^d$ are random transformations that act only on individual inputs, and the outputs are left unchanged. Using \Cref{lem:compose:exact:with:approx} to condition on the permutation $\pi_1$, which satisfies exact invariance, the estimators to compare via universality take the form
\begin{align*}
    f(\pi_1 \circ \tilde \Phi_1(X) )
    \;=&\;
    \big| 
        g_{\rm oracle}( X_{\pi_1(1)}) 
        - 
        g_{\rm trained}( \phi_{\pi_1(1)}(X_{\pi_1(1)}))
    \big|^2 
    \;,
    \\
    f(\pi_1(X))
    \;=&\;
    \big| 
        g_{\rm oracle}( X_{\pi_1(1)}) 
        - 
        g_{\rm trained}( X_{\pi_1(1)})
    \big|^2 
    \;.
\end{align*}
Therefore, conditioning on $\pi_1$, we only need to apply universality to control 
\begin{align*}
    \msup_{t \in \R}
    \big|  
        \P( \tilde f( X_{\pi_1(1)}, \phi_{\pi_1(1)}(X_{\pi_1(1)}) )  \leq t \,|\, X_{\pi_1(1)} , \pi_1 ) 
        \,-\, 
        \P( \tilde f(X_{\pi_1(1)}, X_{\pi_1(1)} )  \leq t \,|\, \pi_1 ) 
    \big|
    \;,
\end{align*}
where we have defined the function $\tilde f: \R^d \times \R^d \rightarrow \R$ by
\begin{align*}
    \tilde f(x, x') 
    \;\coloneqq\; 
    | g_{\rm oracle}(x) - g_{\rm trained}(x')|^2 
    \;.
\end{align*}
In the high-dimensional regime where $d \rightarrow \infty$ and where the coordiantes of the data exhibits sufficient independence, we can then apply \Cref{thm:universality:finite} with $n=1$, and exploit local dependence of data coordiantes and of transformations. The next lemma controls $\tilde \varphi_{r;\nu}$, the corresponding stability measure of $\tilde f$, under a similar setup as \Cref{lem:stability:poly}.

\begin{lemma} \label{lem:stability:square:error} Assume that $g_{\rm trained}$ and  $g_{\rm oracle}$ both take the form of $p_m$ in \Cref{lem:stability:poly} with $n$ replaced by $1$, and that the corresponding $\alpha_l$'s are uniformly bounded in norm by $\alpha_{\rm max}$. Also assume the form of $\tilde \Phi_1$ in \eqref{eq:conformal:input:only}, and that $X_i$'s are i.i.d.~generated. Then there exists a universal constant $C > 0$ such that, for $r=1,2,3$,
\begin{align*}
    \tilde \varphi_{r;\nu}
    \;\leq\;
    C\,
    \alpha_{\rm max}
    \,
    \msum_{l=1}^m 
    \Big(
        2^l
        \,
        \max\Big\{ 
        &\;
            \Big\|
            \,
            \big\| 
                    \hat p_{2l-r} \big(\phi_1(X_1) \big) 
            \big\|_{L_4 | X}
            \,
            \Big\|_{L_\nu | \Xobs}
            \;,\;
            \big\|
            \,
                    \hat p_{2l-r} (X_1) 
            \,
            \big\|_{L_\nu }
    \\
    &\; 
            \Big\|
            \,
            \big\| 
                    \hat p_{2l-r} \big(Z^\phi_1(X_1) \big) 
            \big\|_{L_4 | X}
            \,
            \Big\|_{L_\nu | \Xobs}
            \;,\;
            \big\|
            \,
                    \hat p_{2l-r} (Z_1) 
            \,
            \big\|_{L_\nu}
    \Big\}
    \Big)
    \;,
\end{align*}
where, for $s \in \Z$,  we have denoted for $x_1 \in \R^d$,
\begin{align*}
    \tilde p_s(x_1) 
    \coloneqq 
    \begin{cases}
        \big(\frac{1}{\sqrt{d}} \msum_{l=1}^{d} x_{1l} \big)^s
        & \text{ if } s \geq 0\;,
        \\
        0 & \text{ otherwise},
    \end{cases}
\end{align*}
and $Z^\phi_1(X_1) \sim \cN(\mean[\phi_1(X_1) | X_1], \Var[\phi_1(X_1) | X_1] )$, $Z_1 \sim \cN(\mean[X_1|\Xobs], \Var[X_1|\Xobs])$.
\end{lemma}

\begin{proof}[Proof of \Cref{lem:stability:square:error}] Applying the chain rule and the same argument as \Cref{lem:stability:poly} gives the desired bound.
\end{proof}

We have assumed that $\tilde \Phi_1$ does not change the output and that $X_i$'s are i.i.d.~for simplicity of presentation, but the same argument extends to the cases where both conditions are violated. The more crucial condition is that, due to the form of the residual error, universality (\Cref{thm:universality:finite}) can only be applied to control $\Delta_{\rm inv}(X)$ when:
\begin{itemize}
    \item we operate in a high-dimensional regime, where $d \rightarrow \infty$, and
    \item both the data and the transformations satisfy a coordinate-wise local dependence condition. This notably excludes the ``global" transformations in \Cref{appendix:verify:local}.
\end{itemize}
As a consequence, this restricts the class of transformations that can be used for conformal prediction to enjoy Gaussian universality guarantees.

\section{Proof of \Cref{lem:compose:exact:with:approx}} \label{appendix:proof:compose:exact:with:approx}
In the case $\Phi_1 = \cT_1 \circ \tilde \Phi_1$, by the invariance of $X$ under $\cT_1$, we can write 
\begin{align*}
    \Delta_{\rm inv}(X)
    \;=&\;
    \msup_{t \in \R}
    \big|  
        \P( f(\cT_1 \circ \tilde \Phi_1(X)) \leq t \,|\, X ) 
        \,-\, 
        \P( f(\cT_1(X)) \leq t ) 
    \big|
    \\
    \;=&\;
    \msup_{t \in \R}
    \big|  
        \P( f_\cT(\tilde \Phi_1(X)) \leq t \,|\, X ) 
        \,-\, 
        \P( f_\cT(X) \leq t  ) 
    \big|
    \\
    \;\overset{(a)}{=}&\;
    \msup_{t \in \R}
    \big|  
        \mean\big[
            \,
        \P( f_\cT(\tilde \Phi_1(X)) \leq t \,|\, X , f_\cT ) 
        \,-\, 
        \P( f_\cT(X) \leq t \,|\, f_\cT ) 
            \,
        \big| X \big]
    \big|
    \\
    \;\leq&\;
    \mean\Big[
    \msup_{t \in \R}
    \big|  
        \P( f_\cT ( \tilde \Phi_1(X)) \leq t \,|\, X  , f_\cT ) 
        \,-\, 
        \P( f_\cT(X) \leq t \,|\, f_\cT ) 
    \big|
    \,\Big|\, X
    \Big]
    \;.
\end{align*} 
In $(a)$, we have separated out the expectation over $f_\cT$ and coupled $f_\cT$ in the two additive terms. In the case $\Phi_1 = \tilde \Phi_1 \circ \cT_1$, we use conditional invariance to note that
\begin{align*}
    \Delta_{\rm inv}(X)
    \;=&\;
    \msup_{t \in \R}
    \big|  
        \P( f(\tilde \Phi_1 \circ \cT_1(X)) \leq t \,|\, X ) 
        \,-\, 
        \P( f(X) \leq t ) 
    \big|
    \\
    \;=&\;
    \msup_{t \in \R}
    \big|  
        \mean[ 
            \P( f(\tilde \Phi_1 \circ \cT_1(X)) \leq t \,|\, X, \cT_1 ) 
        \,|\, X]
        \,-\, 
        \P( f(X) \leq t  ) 
    \big|
    \\
    \;\overset{(b)}{\leq}&\;
    \mean\Big[
    \msup_{t \in \R}
    \big|  
        \P( f(\tilde \Phi_1 \circ \cT_1(X)) \leq t \,|\, \cT_1(X) ) 
        \,-\, 
        \P( f(X) \leq t ) 
    \big|
    \,\Big| \, X \Big]
    \;.
\end{align*}
In $(b)$, we have used Jensen's inequality to move $\mean[ \argdot | X]$ outside. Taking $\| \argdot \|_{L_\nu}$ on both sides and recalling that $\nu \geq 1$, we can apply Jensen's inequality to obtain
\begin{align*}
    \| \Delta_{\rm inv}(X) \|_{L_\nu}
    \;\leq&\;
   \Big\|
   \sup_{t \in \R}
    \big|  
        \P( f(\tilde \Phi_1 \circ \cT_1(X)) \leq t \,|\, \cT_1(X) ) 
        \,-\, 
        \P( f(X) \leq t ) 
    \big|
    \Big\|_{L_\nu}
    \\
    \;=&\;
     \Big\|
   \sup_{t \in \R}
    \big|  
        \P( f(\tilde \Phi_1(X)) \leq t \,|\, X ) 
        \,-\, 
        \P( f(X) \leq t ) 
    \big|
    \Big\|_{L_\nu}
    \;.
\end{align*}
In the last line, we have replaced $\cT_1(X)$ by $X$ using invariance.
\qed

\section{Invariance conditions for the examples in \Cref{sec:new:DAB:methods}} \label{appendix:theory:examples}

In this section, we characterise and discuss the approximate invariance conditions for the examples in \Cref{sec:new:DAB:methods}.

\subsection{Invariance conditions for observation-wise augmentations and variants of bootstrap in \Cref{sec:pure:aug:bootstrap}} \label{appendix:condition:pure:aug:bootstrap}

For observation-wise augmentations, i.e.~DAB with $\Phi^{\rm new}_b$, we note that the discussion in \Cref{sec:extension:universality} applies to both setups in \Cref{sec:pure:aug:bootstrap}. Specifically:
\begin{itemize}
    \item For orthogonal transformations on mean-zero and isotropic random vectors, the conditional mean and variance in \eqref{eq:average:moment:match} can be computed as
    \begin{align*}
        \mu_\phi(X) \;=&\; 0 
        &\text{ and }&&
        \Sigma_\phi(X) \;=&\; \mfrac{1}{n} \msum_{i \leq n} X_i X_i^\top 
        \;\xrightarrow{\rm a.s.}\; 
        I_d
        \;,
    \end{align*}
    which match $\mean[X_1] = 0$ and $\Var[X_1] = I_d$;
    \item For the parallel sampling setup, \eqref{eq:average:moment:match} suggests that we may still have a valid CI provided that
    \begin{align*}
        &\;
        \Big(
            \mean\Big[
                h\big( \phi_{11}(X^{(t)}_1) 
                \big) \,
                \Big| \,
                X^{(t)}_1
            \Big]
            \,,\, 
            \Var\Big[
                h\big( \phi_{11}(X^{(t)}_1) 
                \big) \,
                \Big| \,
                X^{(t)}_1
            \Big]
        \Big)
        \\
        &\hspace{12em}\;\approx\;
        \big(
            \mean\big[h(X^{(t)}_1)\big]
            \,,\, 
            \Var\big[h(X^{(t)}_1)\big]
        \big)
        \,.
    \end{align*}
    Note that this condition can be unverifiable in general, both due to the intractable conditioning and due to the unknown mean and variance of $h(X^{(t)}_1)$. 
\end{itemize}
In contrast to the above, the bootstrap resampling step in DAB with $\Phi^{\rm bootstrap}_b \circ \Phi^{\rm new}_b$ allows for marginal moment matching. To see this, let $\phi_i$'s be i.i.d.~$\R^d \rightarrow \R^d$ transformations as in \Cref{sec:pure:aug:bootstrap} and $\tau_{bi} \overset{\rm i.i.d.}{\sim}\textrm{Uniform}([n])$ represent the bootstrap resampling of data indices. For the theoretical discussion here, we consider statistics $f$ that satisfy the stability conditions required in \Cref{thm:universality}, and for experiments, we focus on empirical averages for simplicity. We WLOG let $X_i$'s be i.i.d.~mean-zero random vectors; this is automatically satisfied for the empirical average in \Cref{example:bootstrap} where $X=(\tilde X_i - \mean[\tilde X_1])_{i \leq n}$, and the general non-zero-mean case can be handled by an additional stability requirement at the expense of more notation (see \Cref{sec:validity:universality}). In this case, $\Phi_b = \Phi^{\rm bootstrap}_b \circ \Phi^{\rm new}_b$ can be expressed as
\begin{align*}
    \Phi_b(X) 
    \;=&\; 
    \big(
        \phi_{\tau_{b1}}(X_{\tau_{b1}}) - \bar X_\phi
        \,,\, \ldots \,,\, 
        \phi_{\tau_{bn}}(X_{\tau_{bn}}) - \bar X_\phi
    \big)\;,
    \quad
     \bar X_\phi
     \;\coloneqq\;
    \mfrac{1}{n} \msum_{i \leq n} \phi_i(X_i)
    \;.
\end{align*}
When \Cref{thm:universality} holds, to obtain valid coverage, it suffices to compute the conditional mean and variance:
\begin{align*}
    \mean[\Phi_1(X) | X] \;=\; \mean\big[ \, \mean[\Phi_1(X) \,|\, X, (\phi_i)_{i \leq n}] \,\big|\, X \big] \;\overset{(a)}{=}\;  0 \;=\; \mean[X]\;,
\end{align*}
where $(a)$ follows from the empirical centering in bootstrap. $\Var[\Phi_1(X) | X]$ is an $\R^{nd \times nd}$ block diagonal matrix, consisting of identical $\R^{d \times d}$ blocks of entries given by
\begin{align*}
    &\;\Var\big[ \phi_{\tau_1}(X_{\tau_1}) - \bar X_\phi \,\big|\, X  \big]    
    \\
    &\;=\;
    \mfrac{n-1}{n^2} \msum_{i=1}^n
    \mean[ \phi_i(X_i) \phi_i(X_i)^\top | X_i ]
    -
    \mfrac{1}{n^2} \msum_{i \neq j}^n 
    \mean[ \phi_i(X_i) | X_i] \; \mean[ \phi_j(X_j) | X_j]^\top
    \;.
\end{align*}
In particular, the variance difference term in \Cref{thm:universality} can be controlled as
\begin{align*}
    \Delta_{\rm Var}(X) 
    \;\leq&\;
    \Big\| 
        \mfrac{n-1}{n^2} \msum_{i=1}^n \mean[ \phi_i(X_i) \phi_i(X_i)^\top | X_i ] 
        -
        \Var[X_1] 
    \Big\|_\infty
    \tagaligneq \label{eq:bootstrap:augment:var:diff}
    \\
    &\;
    +
    \Big\| 
        \mfrac{1}{n^2} \msum_{i \neq j}^n 
        \mean[ \phi_i(X_i) | X_i] \; \mean[ \phi_j(X_j) | X_j]^\top
    \Big\|_\infty
    \;.
    \tagaligneq \label{eq:bootstrap:augment:cov}
\end{align*}
Under an additional condition on $\phi_1$ that $\mean[\phi_1(X_1)]=\mean[X_1]=0$ and a moment boundedness condition, by Nemirovski's moment inequality (see e.g.~Lemma 14.24 of \cite{buhlmann2011statistics}) we have
\begin{align*}
    \| \eqref{eq:bootstrap:augment:var:diff} \|_{L_{3\nu/2}} 
    \;=&\;
    O\Big( \mfrac{\sqrt{\log d} \, \| \Var[ \phi_1(X_1) ] - \Var[X_1] \|_\infty}{\sqrt{n}} \Big)
    \;,
    \tagaligneq \label{eq:nemirovski}
\end{align*}
and under an additional uniform sub-Gaussian tail condition on the coordinates of $\mean[\phi_i(X_i) \phi_i(X_i)^\top | X_i]$, by the Hanson-Wright inequality (see e.g.~\cite{vershynin2018high}),
\begin{align*}
    \|\eqref{eq:bootstrap:augment:cov}\|_{L_{3\nu/2}}
    \;=&\;
    O\Big( \mfrac{\log d}{n} \Big) \;.
\end{align*}
As such, by the conditions of \Cref{thm:universality}, we have valid coverage for any choice of $\phi_i$'s such that
\begin{align*}
    N_{\rm dep} \;=&\; o\Big( \mfrac{n}{\log d} \Big)
    &\text{ and }&&
    \| \Var[ \phi_1(X_1) ] - \Var[X_1] \|_\infty
    \;=&\;
    o\Big( \mfrac{\sqrt{n}}{N_{\rm dep}\, \sqrt{\log d} } \Big)\;.
\end{align*}    
Importantly, bootstrap resampling allows the conditional moment matching in \Cref{thm:universality} to be improved to \emph{marginal} moment matching, and a sufficient condition for $\Delta_{\rm Var}(X)$ to be small is the marginal distributional invariance
\begin{align*}
    \phi_1(X_1) \;\overset{d}{=}\; X_1\;.
\end{align*}

\subsection{Invariance conditions for wild bootstrap with additional transformations in \Cref{sec:wild:bootstrap:add}} \label{appendix:condition:wild:bootstrap}

Notice that wild bootstrap in \Cref{sec:wild:bootstrap:add} is applied only to degree-two U-statistics. By \Cref{lem:stability:poly} and its subsequent discussion, we expect the universality result of \Cref{thm:universality} to apply under additional coordinate dependence conditions on $X_1$ and $\phi(X_1)$ as well as a condition that the Hilbert space $\cH$ is well-approximated by a subspace of some ambient Euclidean space. In this section, we show that for the DAB variant of wild bootstrap introduced in \Cref{sec:wild:bootstrap:add}, when \Cref{thm:universality} holds, the approximate invariance condition reduces to a conditional moment matching condition in the feature map space $\varphi$.

\vspace{.5em}

Let $(\phi_{ri})_{ r \in \{1,2\}, i \leq n}$ be i.i.d.~$\R^{d'} \rightarrow \R^{d'}$ transformations that are identically distributed as $\phi$, as defined in \Cref{sec:wild:bootstrap:add}. Suppose we have the additional knowledge that $P$ exhibits approximate invariance with respect to $\phi$ in the sense that the probability measure $P_\phi$ of $\phi(X_1)$ satisfies 
\begin{align*}
    \bd(P, P_\phi) \leq \epsilon
\end{align*}
for some probability metric $\bd$ and $\epsilon > 0$ to be specified later. The two-sample test problem is the same as testing
\begin{align*}
    H^\phi_0: \; P = Q\,,\, \bd(P, P_\phi) \leq \epsilon
    \qquad \text{ v.s. } \qquad 
    H^\phi_1: \; P \neq Q\,,\, \bd(P, P_\phi) \leq \epsilon \;.
\end{align*}
Recall that our dataset is $\{Y_i, Z_i\}_{i=1}^n$, where $\{Y_i\}_{i=1}^n$ are samples from $P$ and $\{Z_i\}_{i=1}^n$ are samples from $Q$. Also recall that $\varphi$ are the feature maps such that the kernel function $\kappa(y,y') = \langle \varphi(y), \varphi(y')\rangle_{\cH}$.  We can express DAB in \Cref{sec:wild:bootstrap:add} as the transformation $\Phi_1$ that acts on a dataset $(y_i, z_i)_{i \leq n}$  as
\begin{align*}
    &\;
    \Phi_1\big( 
        \varphi(y_1) - \varphi(z_1), 
        \ldots,  
        \varphi(y_n) - \varphi(z_n)
    \big)
    \\
    &
    \;=\;
    \big(
        \epsilon_1 \,\times\, (
            \varphi( \phi_{11}(y_1)) 
            - 
            \varphi(\phi_{21}(z_1))
        )
         \,,\, \ldots \,,\, 
         \epsilon_n  \,\times\, (\varphi(\phi_{1n}(y_n)) - \varphi(\phi_{2n}(z_n)))
    \big) 
    \;,
\end{align*}
where $\epsilon_i$'s are i.i.d.~Rademacher random variables as in \eqref{eq:wild:bootstrap:rademacher}.  While $\varphi$ is intractable and therefore $\Phi_1$ is not directly computable, its corresponding statistic $f(\Phi_1(X))$ \emph{is} computable, giving the expression
\begin{align*}
    f(\Phi_1(X))
    \;=\;
    \mfrac{1}{n(n-1)} \msum_{1 \leq i \neq j \leq n} 
    \epsilon_i \epsilon_j  
    \,
    u\big( 
        \,
        \phi_{1i}(Y_i) 
        ,
        \phi_{1j}(Y_j)
        ,
        \phi_{2i}(Z_i)
        ,
        \phi_{2j}(Z_j)
    \big)
    \;,
\end{align*}
which recovers the expression \eqref{eq:WB:main:text} in the main text.

\vspace{.5em}

To formulate the approximate invariance condition by applying \Cref{thm:universality}, we take for simplicity that $\cH = \R^q$ with some $q \in \N$, so that the feature maps $\varphi$ have Euclidean outputs. We first observe that since  $\epsilon_i$'s are zero-mean and independent,
\begin{align*}
    \mean[ \Phi_1(X) | X ] \;=\; 0 \;\overset{H^\phi_0}{=}\; \mean[X]\;,
\end{align*} 
which implies both conditional and marginal mean matching. For variance, note that 
\begin{align*}
    &\;
    \Var\big[ \epsilon_1 \times (
            \varphi( \phi_{11}(Y_1)) 
            - 
            \varphi(\phi_{21}(Z_1))
    ) \,\big|\, (Y_i,Z_i)_{i \leq n} \big]
    \\
    &\;=\;
    \Var\big[ (
            \varphi( \phi_{11}(Y_1)) 
            - 
            \varphi(\phi_{21}(Z_1))
    ) \,\big|\, Y_1, Z_1 \big]
    \;,
\end{align*}
so the approximate invariance condition is that the conditional variance matching error
\begin{align*}
    \big\| 
        \big\| 
            \Var\big[ (\varphi( \phi_{11}(Y_1)) - \varphi(\phi_{21}(Z_1))) \,\big|\, Y_1, Z_1 \big] 
            -
            \Var\big[ (\varphi( \phi_{11}(Y_1)) - \varphi(\phi_{21}(Z_1))) \big]
        \big\|_\infty
    \big\|_{L_{3\nu/2}}
\end{align*}
is small. Note that the above reasoning gives coverage validity under $H^\phi_0$ but not necessarily so under $H^\phi_1$. This gives a valid test statistic for testing $H^\phi_0$ against $H^\phi_1$.

\vspace{.5em}

A notable limitation is that the approximate invariance condition above can be difficult to verify, due to (i) the conditioning and (ii) the fact that the moments are in terms of the abstract feature maps $\varphi$, which can be intractable to compute. One may hope to relax (i) to marginal moment matching by incorporating bootstrap and a similar argument as \Cref{appendix:condition:pure:aug:bootstrap}. However, it is known \citep{arcones1992bootstrap} that for a degree-two degenerate U-statistics, a consistent bootstrap typically requires performing a bootstrap on its second-order Hoeffding decomposition with centering. In our notation, this corresponds to bootstrap with empirical centering by quantities involving $\frac{1}{n} \sum_{i \leq n} \varphi( \phi_{1i}(Y_i) )$ and $\frac{1}{n} \sum_{i \leq n} \varphi( \phi_{2i}(Z_i) )$, which can again be computationally intractable.

\vspace{.5em}

\subsection{Invariance conditions for conformal prediction with data augmentation in \Cref{sec:conformal}} \label{appendix:condition:conformal}

Under $\Phi^{\rm DAB}_b$ defined in \Cref{sec:conformal}, the statistics used in forming the DAB CIs can be expressed as
\begin{align*}
    f\big(\Phi^{\rm DAB}_{bj}(X)\big) 
    \;=\; 
    h\Big(
        \phi^{\rm out}_{bj}(Y_{b+1}) 
        \,,\,
        g\Big( \phi^{\rm in}_{bj}(V_{b+1})  \Big)
    \Big)
    \;,
    \tagaligneq \label{eq:DAB:CP:stat}
\end{align*}
where we have taken $(\phi^{\rm in}_{bj} , \phi^{\rm out}_{bj})$ as i.i.d.~copy of $(\phi^{\rm in}, \phi^{\rm out})$. Unlike the previous examples, we can no longer use universality for coverage guarantees: $f(X)$ is completely determined by $X_1$, and therefore violates the stability condition required for universality in \Cref{thm:universality} except for specific high-dimensional cases (see \Cref{appendix:verify:stability} for details). Moreover, one should not attempt to fix this by e.g.~bootstrap resampling, since to test whether $X_1$ behaves differently from the calibration set, the conformal prediction score $f(X)$ ought to be sensitive to changes in $X_1$. On the other hand, the remaining results in \Cref{sec:validity} still apply for characterising the loss in coverage as a result of approximate invariance. In particular for $k=1$, since $(f(\Phi^{\rm DAB}_b(X)))_{b \leq B}$ are i.i.d.~and $f(X), f(\Phi^{\rm DAB}_1(X)), \ldots, f(\Phi^{\rm DAB}_{n-1}(X))$ are independent, we may apply a variant of \Cref{thm:Del:Kol:approx:inv} (\Cref{thm:Del:Kol:approx:inv:general} in the appendix) that relaxes the conditioning on $X$ and quantifies the approximate invariance error via the \emph{marginal} c.d.f.~error
\begin{align*}
      \Delta^{\rm CP}_{\rm inv} 
    \;\coloneqq\;
    \msup_{t \in \R} \big| 
        \P\big(  h(\phi^{\rm out}(Y_1),\, g(\phi^{\rm in}(V_1))) \leq t\big)
        \,-\, 
        \P\big( h(Y_1, \, g(V_1))  \leq t \big)
    \big| 
    \;.
\end{align*}

\vspace{.5em}

\begin{corollary} \label{cor:conformal} Assume the use of smoothed uniform tie-breaking, and write $R_{\rm DAB}(X)$ as the rank variable corresponding to $(\Phi^{\rm DAB}_{bj})_{b \leq n-1, j \leq k}$ defined in \Cref{appendix:condition:conformal}. Let $r \in [0,1]$, $\nu \geq 1$ and $n \geq 3$. Then for $k=1$, 
\begin{align*}
     &\Big| \P\Big(  \mfrac{R_{\rm DAB}(X)}{n}  < r \Big)  
        \,-\,
        r
    \Big|
    \\
    &
    \;\leq\;
    16 \Big( \mfrac{3}{8} \Big)^{\frac{2}{\nu+2}} \, \big\| \Delta^{\rm CP}_{\rm inv} \big\|_{L_\nu }^{\frac{2\nu}{\nu+2}}
    \,+\,
    2   \Big( \mfrac{3}{8} \Big)^{\frac{1}{\nu+2}} \, \big\| \Delta^{\rm CP}_{\rm inv} \big\|_{L_\nu}^{\frac{\nu}{\nu+2}}
    \Big( 
        \mfrac{1}{n-2} 
        + 
        \mfrac{2}{\sqrt{n-1}}
        +
        \mfrac{\sqrt{3 \log (n-1)}}{\sqrt{n-2}}
    \Big)
    \;.
\end{align*}
\end{corollary}

\begin{proof}[Proof of \Cref{cor:conformal}] In the case $k=1$, we have that $\Phi^{\rm DAB}_{b1}(X)$ are independent across $1 \leq b \leq n-1$. Therefore taking $\Xobs = \A = \A_{\rm trivial}$ in \Cref{thm:Del:Kol:approx:inv:general} gives the desired bound.
\end{proof}

As with \Cref{thm:Del:Kol:approx:inv}, \Cref{cor:conformal} controls the coverage loss via the approximate invariance error $\Delta^{\rm CP}_{\rm inv}$. Since the error measures the difference between marginal c.d.f.s,
\begin{align*}
    \Delta^{\rm CP}_{\rm inv} \;=&\; 0
    &\text{ if and only if }&&
    h(\phi^{\rm out}_{11}(Y_1),\, g(\phi^{\rm in}_{11}(V_1)))
    \;\text{ equals }&\;
    h(Y_1, g(V_1))
    \;\text{ in distribution}
    \;.
\end{align*}
Observe that, when restricted to group transformations, DAB with \eqref{eq:DAB:CP:stat} is a special case of SymmPI \citep{dobriban2023symmpi}, which generalises conformal prediction to incorporate invariant group transformations. Our DAB variant of conformal prediction can therefore be interpreted as an extension of SymmPI in two ways:
\begin{proplist}
    \item When exact invariance holds in the sense that $\Delta^{\rm CP}_{\rm inv}=0$,  DAB extends SymmPI by allowing for non-group transformations;
    \item In view of \Cref{cor:conformal}, DAB also extends SymmPI to allow for only approximately invariant transformations, as long as $\Delta^{\rm CP}_{\rm inv}$ is small.
\end{proplist}

\section{Experimental details and additional empirical results} \label{appendix:experiments}

This section includes experiment details and additional experiments that complement those reported in \Cref{sec:new:DAB:methods}.

\vspace{.5em}

\subsection{Experiments for observation-wise augmentations and variants of bootstrap in \Cref{sec:pure:aug:bootstrap}} \label{appendix:experiments:pure:aug:bootstrap}

\begin{figure}[t]
  \centering
  \begin{tikzpicture}
        \coordinate (gaussian) at (0,0);
        \begin{scope}[shift={(gaussian)}]
            \node[inner sep=0pt] at (0,0)
                {\includegraphics[width=.7\textwidth]{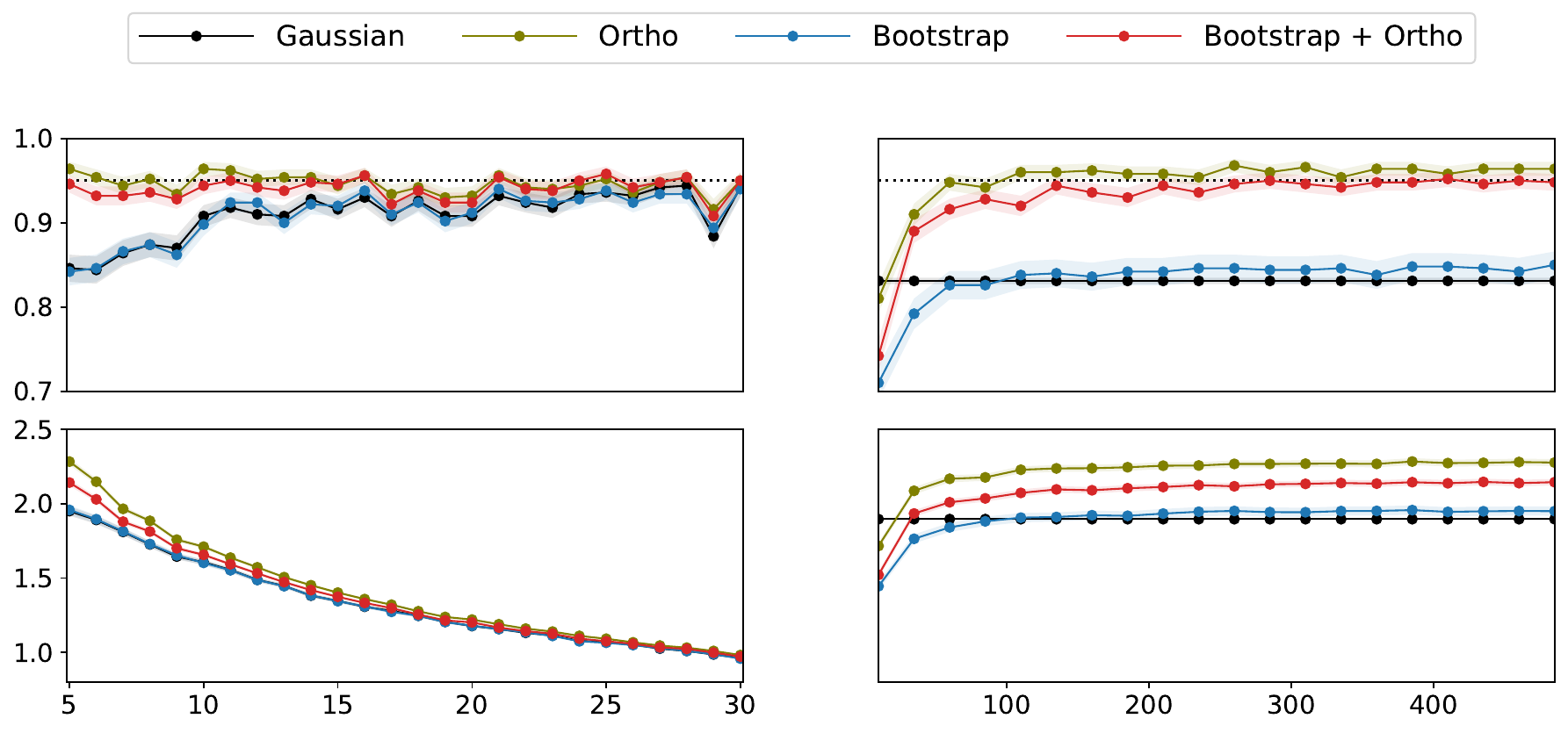}};
            
            \node[inner sep=0pt] at (-2.5,1.7){\scriptsize \textbf{Varying $n$, \;$B=100$, \, $1-\alpha=95\%$}};

            \node[inner sep=0pt] at (2.8,1.7){\scriptsize \textbf{Varying $B$, \;$n=5$, \, $1-\alpha=95\%$}};

            \node[inner sep=0pt,rotate=90] at (-5.4,0.65){\scriptsize coverage};

            \node[inner sep=0pt,rotate=90] at (-5.4,-1.2){\scriptsize CI length};

            \node[inner sep=0pt] at (-2.5,-2.6){\scriptsize $n$};

            \node[inner sep=0pt] at (2.8,-2.6){\scriptsize $B$};
        \end{scope}
  \end{tikzpicture} 
  \\[-.5em]
  \caption{Simulations with empirical averages of 2d Gaussians over 500 random trials, with varying $n$ or $B$. The setup is the same as \Cref{fig:bootstrap:simulate}(i), (ii) and (iii).
  }
  \label{fig:bootstrap:simulate:Gaussian:nk} 
\end{figure}
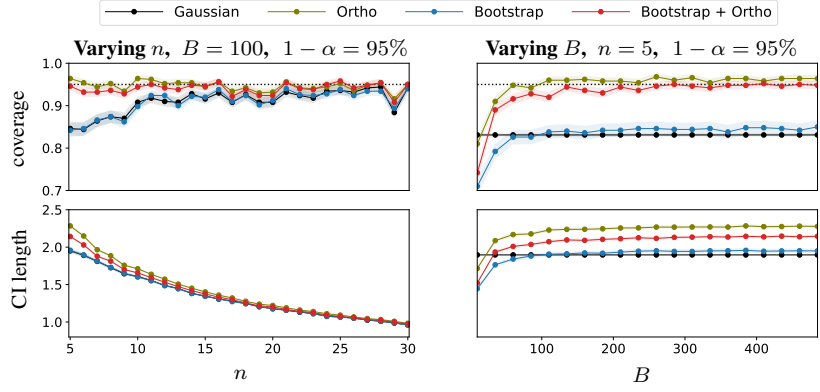
\begin{figure}[t]
  \centering
  \begin{tikzpicture}
        \coordinate (gaussian) at (0,0);
        \begin{scope}[shift={(gaussian)}]
            \node[inner sep=0pt] at (0,0)
                {\includegraphics[width=.7\textwidth]{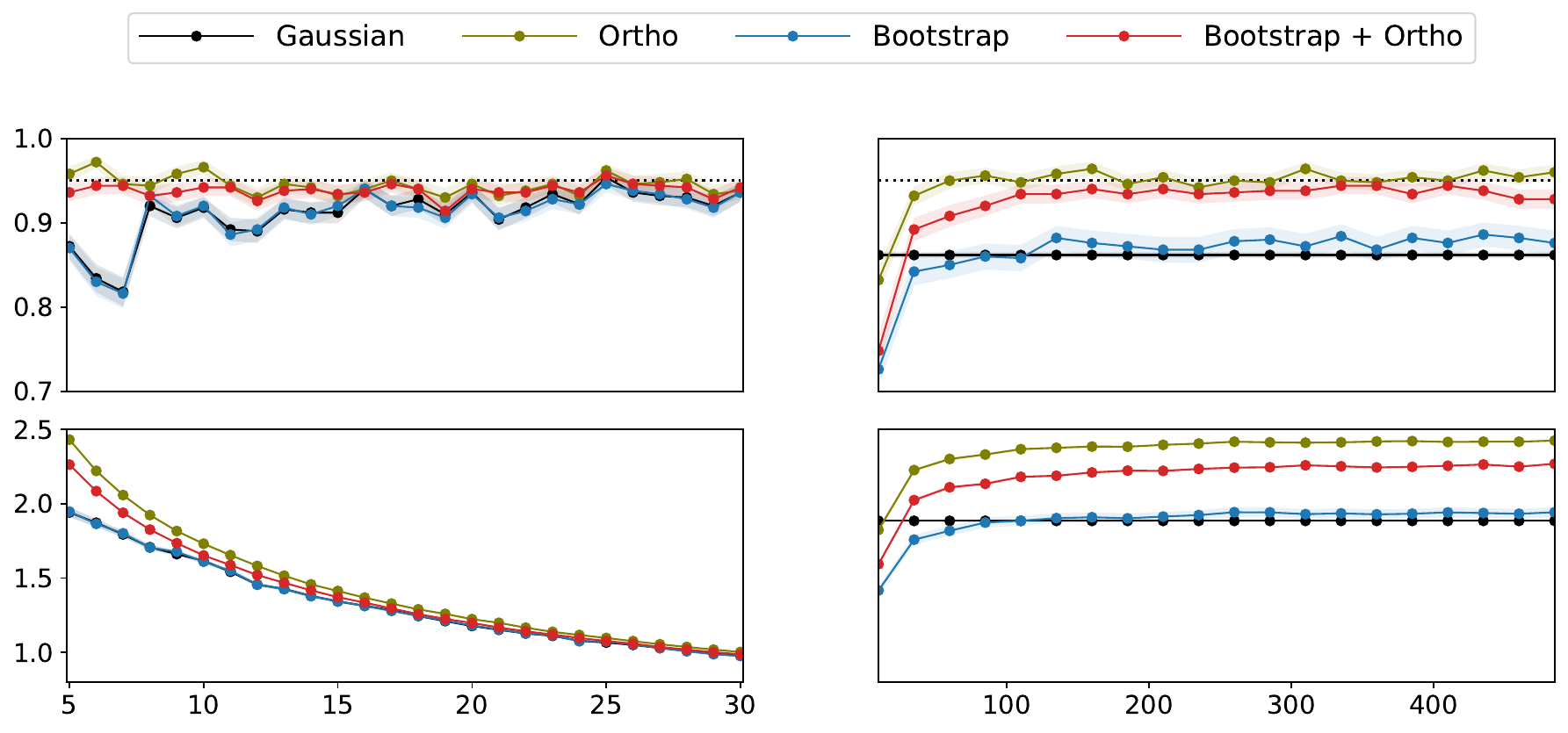}};
            
            \node[inner sep=0pt] at (-2.5,1.7){\scriptsize \textbf{Varying $n$, \;$B=100$, \, $1-\alpha=95\%$}};

            \node[inner sep=0pt] at (2.8,1.7){\scriptsize \textbf{Varying $B$, \;$n=5$, \, $1-\alpha=95\%$}};

            \node[inner sep=0pt,rotate=90] at (-5.4,0.65){\scriptsize coverage};

            \node[inner sep=0pt,rotate=90] at (-5.4,-1.2){\scriptsize CI length};

            \node[inner sep=0pt] at (-2.5,-2.6){\scriptsize $n$};

            \node[inner sep=0pt] at (2.8,-2.6){\scriptsize $B$};
        \end{scope}
  \end{tikzpicture} 
  \\[-.5em]
  \caption{Simulations with empirical averages of 2d Rademachers over 500 random trials, with varying $n$ or $B$. The setup is the same as \Cref{fig:bootstrap:simulate}(iv), (v) and (vi).
  }
  \label{fig:bootstrap:simulate:Rademacher:nk} 
  \vspace{-1em}
\end{figure}
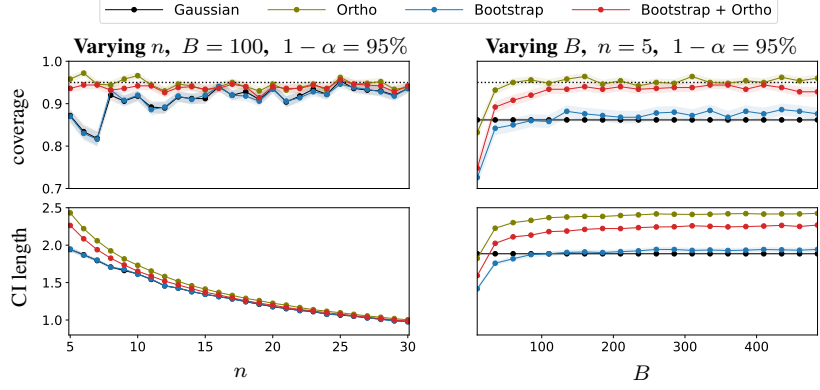
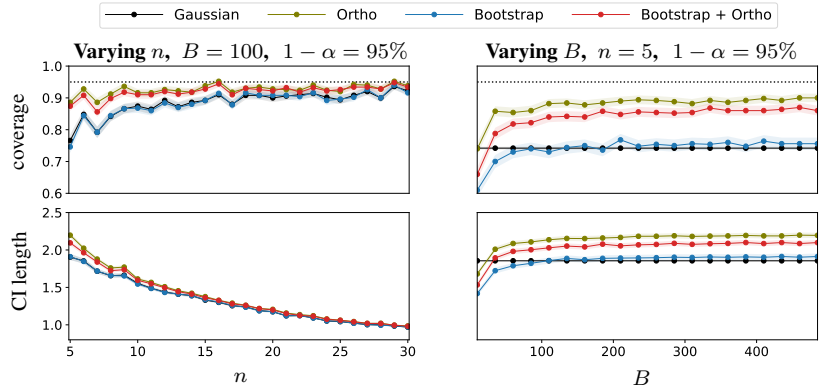
\begin{figure}[t]
  \centering
  \begin{tikzpicture}
        \coordinate (gaussian) at (0,0);
        \begin{scope}[shift={(gaussian)}]
            \node[inner sep=0pt] at (0,0)
                {\includegraphics[width=.7\textwidth]{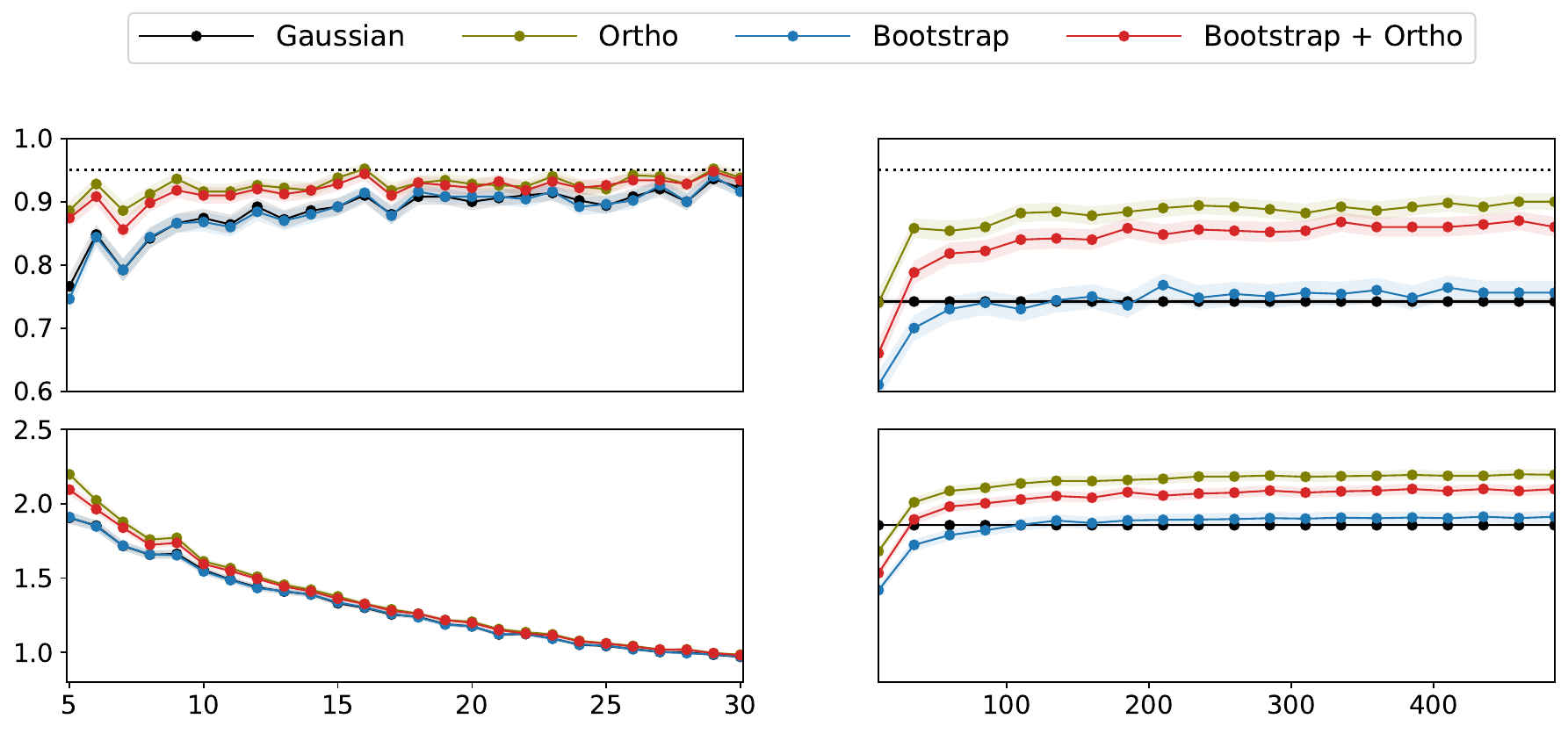}};
            
            \node[inner sep=0pt] at (-2.5,1.7){\scriptsize \textbf{Varying $n$, \;$B=100$, \, $1-\alpha=95\%$}};

            \node[inner sep=0pt] at (2.8,1.7){\scriptsize \textbf{Varying $B$, \;$n=5$, \, $1-\alpha=95\%$}};

            \node[inner sep=0pt,rotate=90] at (-5.4,0.65){\scriptsize coverage};

            \node[inner sep=0pt,rotate=90] at (-5.4,-1.2){\scriptsize CI length};

            \node[inner sep=0pt] at (-2.5,-2.6){\scriptsize $n$};

            \node[inner sep=0pt] at (2.8,-2.6){\scriptsize $B$};
        \end{scope}
  \end{tikzpicture} 
  \\[-.5em]
  \caption{Simulations with empirical averages of 2d centred Gammas over 500 random trials, with varying $n$ or $B$. The setup is the same as \Cref{fig:bootstrap:simulate}(iv), (v) and (vi).
  }
  \label{fig:bootstrap:simulate:Gamma:nk} 
\end{figure}
\begin{figure}[t]
  \centering
  \begin{tikzpicture}
        \coordinate (gaussian) at (0,0);
        \begin{scope}[shift={(gaussian)}]
            \node[inner sep=0pt] at (0,0)
                {\includegraphics[width=\textwidth]{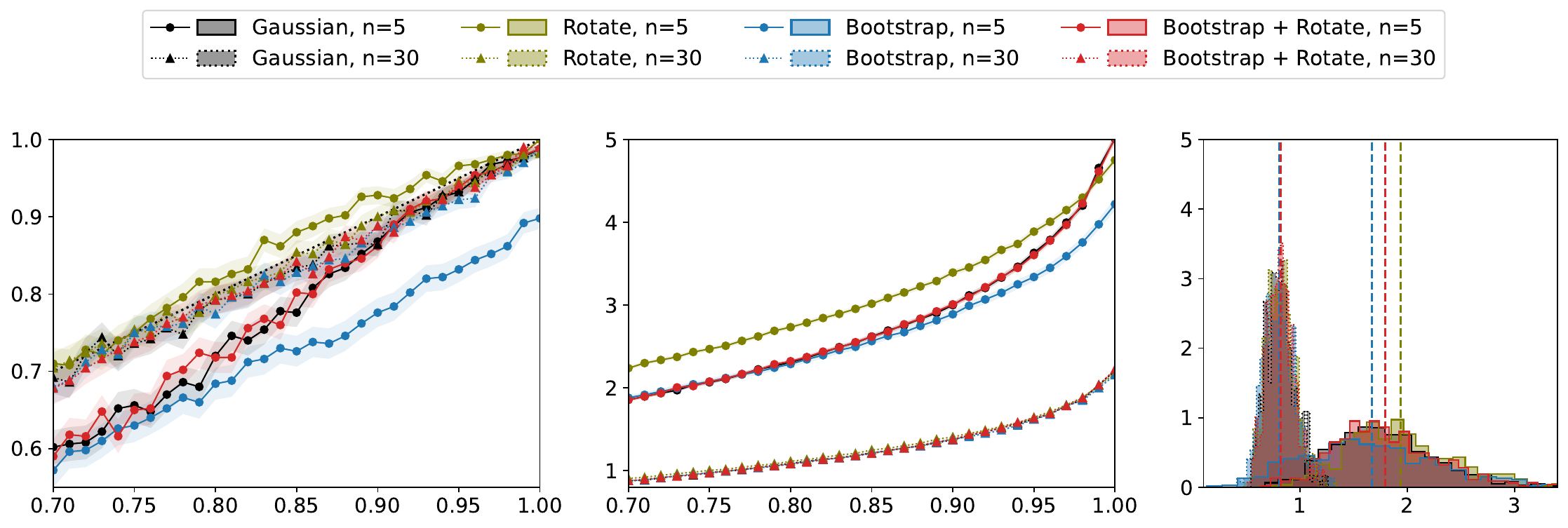}};
            
            \node[inner sep=0pt] at (-4.6,1.45){\scriptsize \textbf{(i) Empirical coverage}};
            \node[inner sep=0pt] at (0.8,1.45){\scriptsize \textbf{(ii) Confidence interval (CI) length}};
            \node[inner sep=0pt] at (5.5,1.45){\scriptsize \textbf{(iii) Distribution of CI lengths}};

            \node[inner sep=0pt] at (-4.6,-2.6){\scriptsize target coverage $1-\alpha$};
            \node[inner sep=0pt] at (0.8,-2.6){\scriptsize target coverage $1-\alpha$};
            \node[inner sep=0pt] at (5.6,-2.6){\scriptsize CI length};
        \end{scope}
  \end{tikzpicture} 
  \caption{Quality of the CIs for the $x$-component of the electron dipole moment of a FermiNet wavefunction trained for the Lithium atom, over 500 random trials. \textbf{(i)} Empirical v.s.~target coverage, where the dashed line $y=x$ indicates the desired coverage level. \textbf{(ii)} CI length~v.s. target coverage. \textbf{(iii)} Empirical distribution of CI lengths over the random trials. This follows the same setup as \Cref{fig:bootstrap:ferminet:nk}. 
  }
  \label{fig:bootstrap:ferminet} 
  \vspace{-1em}
\end{figure}
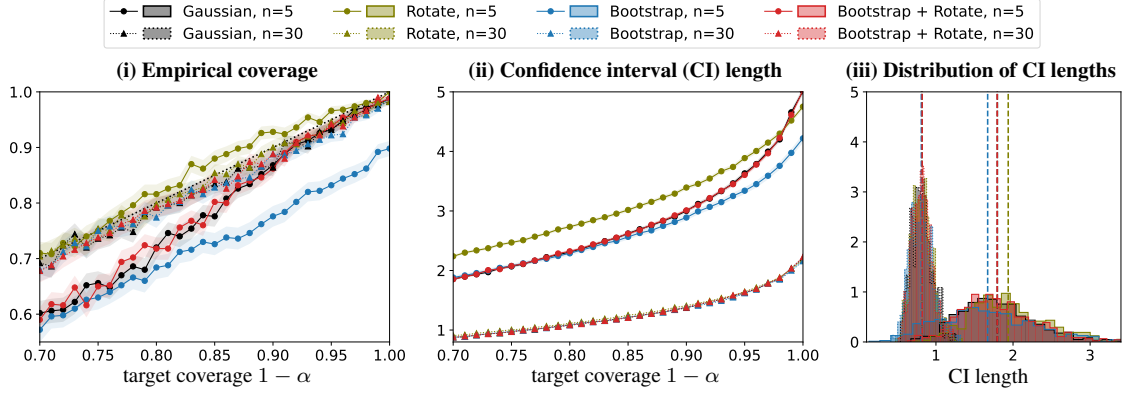

\noindent
\textbf{Orthogonal transformations for mean-zero and isotropic random vectors.} The setup in \Cref{fig:bootstrap:simulate} and \Cref{fig:bootstrap:simulate:Gaussian:nk,fig:bootstrap:simulate:Rademacher:nk,fig:bootstrap:simulate:Gamma:nk} is as described in \Cref{sec:pure:aug:bootstrap} with $d=2$ and three different data distributions:
\begin{itemize}[
    leftmargin=1em,
    labelwidth=0em,
    topsep=0.5em,
    partopsep=0em,
    itemsep=0em
]
    \item \emph{2d Gaussian. } $X_i \overset{\rm i.i.d.}{\sim} \cN(0, I_2)$;
    \item \emph{2d Rademacher. } $X_i$ are i.i.d.~2d random vectors and the two coordinates $X_{11}$ and $X_{12}$ are i.i.d.~Rademacher random variables;
    \item \emph{2d centred Gamma vectors. } $X_i$ are i.i.d.~2d random vectors and the two coordinates $X_{11}$ and $X_{12}$ are i.i.d.~distributed as 
    \begin{align*}
        W - \mean[W] \;, \qquad \text{ where } \qquad W \sim \Gamma(1,1)\;.
    \end{align*}
\end{itemize}

\noindent
\textbf{Statistics on parallel sampling from invariant distributions in AI-for-science.} In \Cref{fig:bootstrap:ferminet:nk,fig:bootstrap:ferminet}, we draw MCMC samples from a trained FermiNet \citep{pfau2020ab} that models the distribution of the $3$ electrons in a ground-state Lithium atom. The sampling algorithm is chosen as the default Metropolis-Hastings algorithm in FermiNet. For augmentations, we take the group of simultaneous rotations on all three electron positions,
\begin{align*}
    \G \;=\; \{ (g,g,g) \,|\, g \in \S\O_2 \} \;,
\end{align*}
where $\S\O_2$ denotes the group of simultaneous rotations in the $x$-$y$ plane about the $z$-axis. The test function of interest is the $x$-component of the electron dipole moment of the system: Given a three electron configuration $x=(x_a^\top, x_b^\top, x_c^\top)$ with $x_a,x_b,x_c \in \R^3$,
\begin{align*}
    h(x) \;=\; (x_a)_1 + (x_b)_1 + (x_c)_1\;.
\end{align*}
For a true sample of the three-electron configuration $X^* = ( (X^*_a)^\top, (X^*_b)^\top, (X^*_c)^\top)^\top$ from the ground-state Lithium atom, it is known that $\mean[(X^*_a)_1 + (X^*_b)_1 + (X^*_c)_1] = 0$ and that the joint three-electron distribution is invariant under $\G$,  although neither is guaranteed for the $t$-th step MCMC distribution $X^{(t)}$.

\subsection{Experiments for wild bootstrap with additional transformations in \Cref{sec:wild:bootstrap:add}}  \label{appendix:experiments:wb}

In all experiments in this section, the statistic used is the MMD U-statistic defined in \Cref{sec:wild:bootstrap:add} with the Radial Basis Function (RBF) kernel $\kappa(y,z) = \exp( - \frac{1}{2\gamma} \| y-z\|^2 )$, where $\gamma$ is chosen by the median heuristic (see \cite{garreau2017large} for an overview and the theoretical analysis of median heuristic).

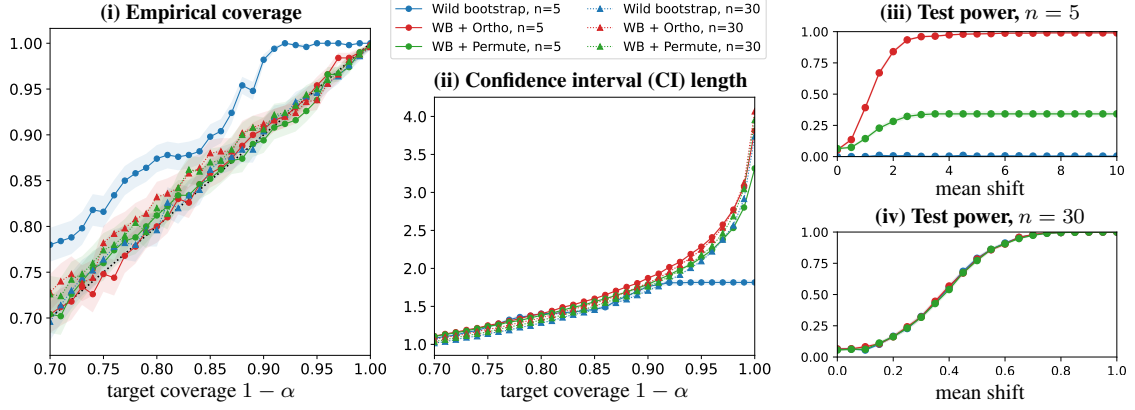
\begin{figure}[t]
  \centering
  \begin{tikzpicture}
        \coordinate (rademacher) at (0,0);
        \begin{scope}[shift={(rademacher)}]
            \node[inner sep=0pt] at (-3.1,0)
                {\includegraphics[width=.69\textwidth]{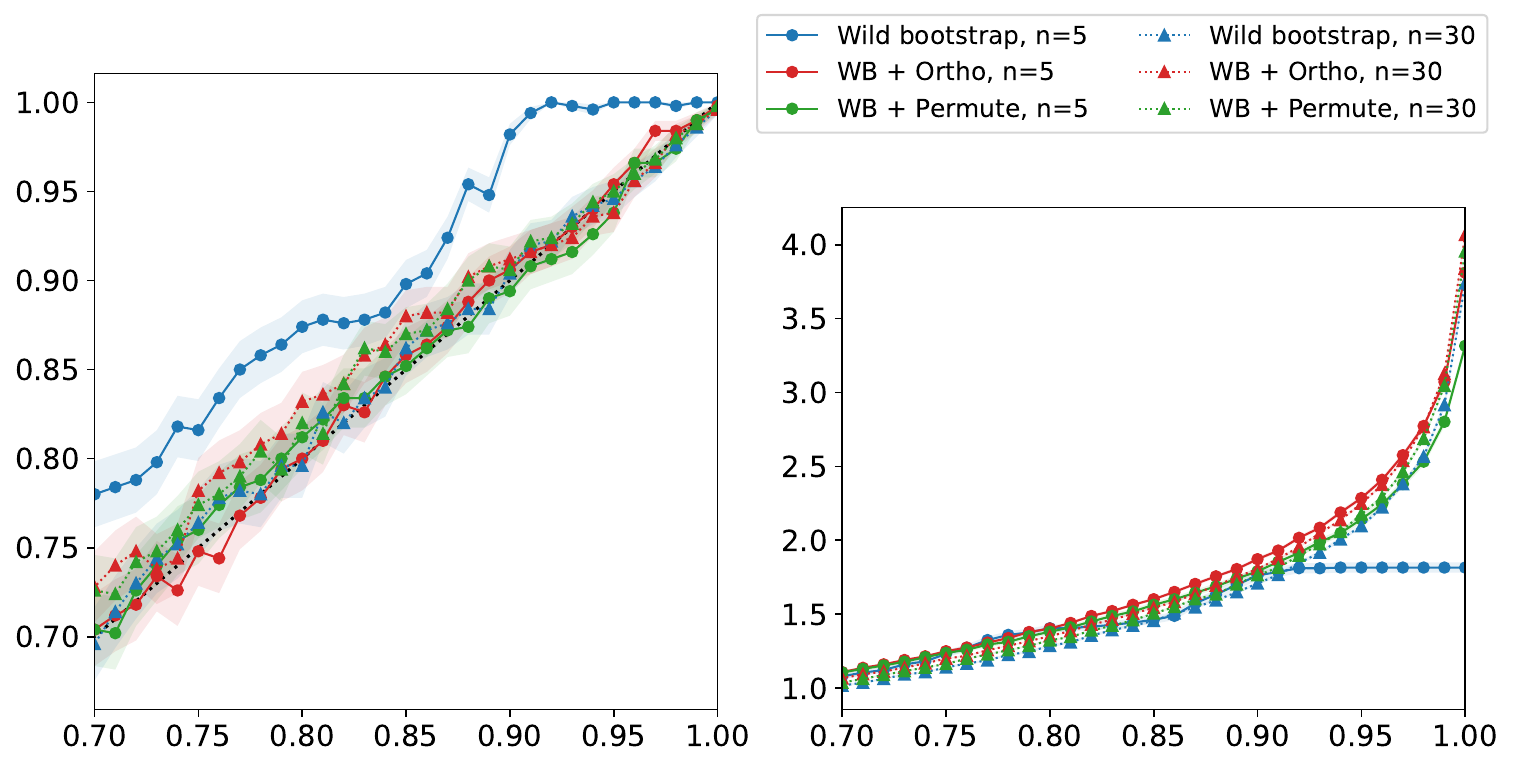}};

            \node[inner sep=0pt] at (4.6,1.25)
                {\includegraphics[width=.32\textwidth]{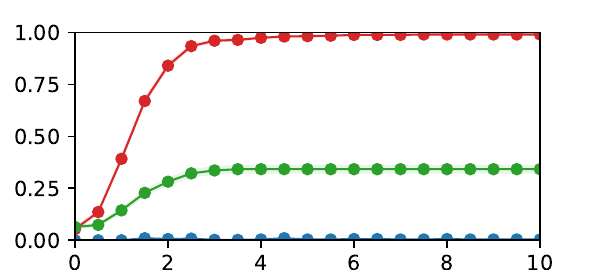}};

            \node[inner sep=0pt] at (4.6,-1.4)
                {\includegraphics[width=.32\textwidth]{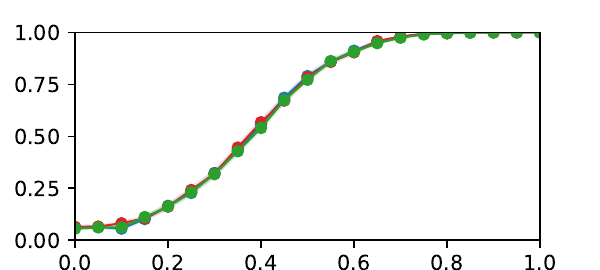}};
            
            \node[inner sep=0pt] at (-5.6,2.3){\scriptsize \textbf{(i) Empirical coverage}};
            \node[inner sep=0pt] at (-0.45,1.4){\scriptsize \textbf{(ii) Confidence interval (CI) length}};
            \node[inner sep=0pt] at (4.7,2.3){\scriptsize \textbf{(iii) Test power, $n=5$}};
            \node[inner sep=0pt] at (4.7,0){\scriptsize mean shift};
            \node[inner sep=0pt] at (4.7,-0.4){\scriptsize \textbf{(iv) Test power, $n=30$}};
            \node[inner sep=0pt] at (4.7,-2.7){\scriptsize mean shift};
 
            \node[inner sep=0pt] at (-5.6,-2.7){\scriptsize target coverage $1-\alpha$};
            \node[inner sep=0pt] at (-0.4,-2.7){\scriptsize target coverage $1-\alpha$};
        \end{scope}
  \end{tikzpicture} 
  \\[-.2em]
  \caption{MMD statistics with RBF kernel on 2d Gaussians over 500 random trials. \textbf{(i):} Empirical v.s. target coverage. \textbf{(ii):} CI length~v.s. target coverage. \textbf{(iii) and (iv):} Test power against mean shift. 
  }
  \label{fig:wb:simulate:gaussian} 
\end{figure}
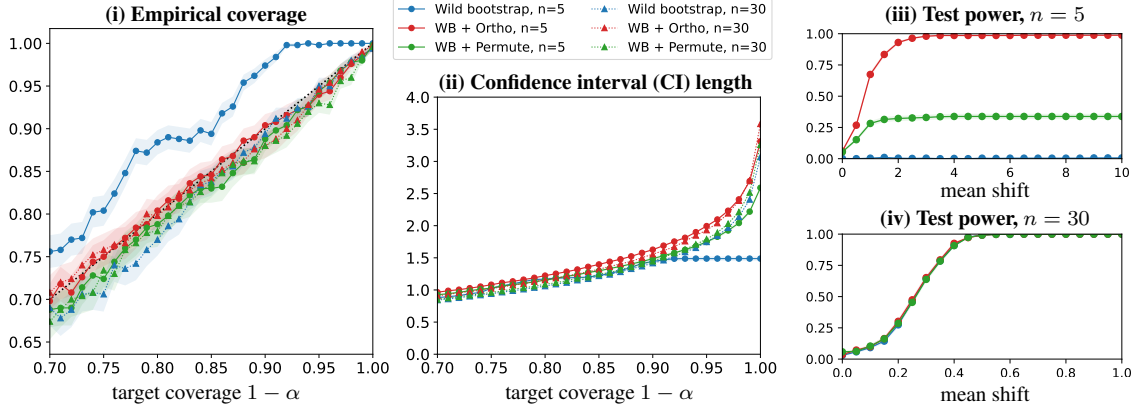
\begin{figure}[t]
  \centering
  \begin{tikzpicture}
        \coordinate (rademacher) at (0,0);
        \begin{scope}[shift={(rademacher)}]
            \node[inner sep=0pt] at (-3.1,0)
                {\includegraphics[width=.69\textwidth]{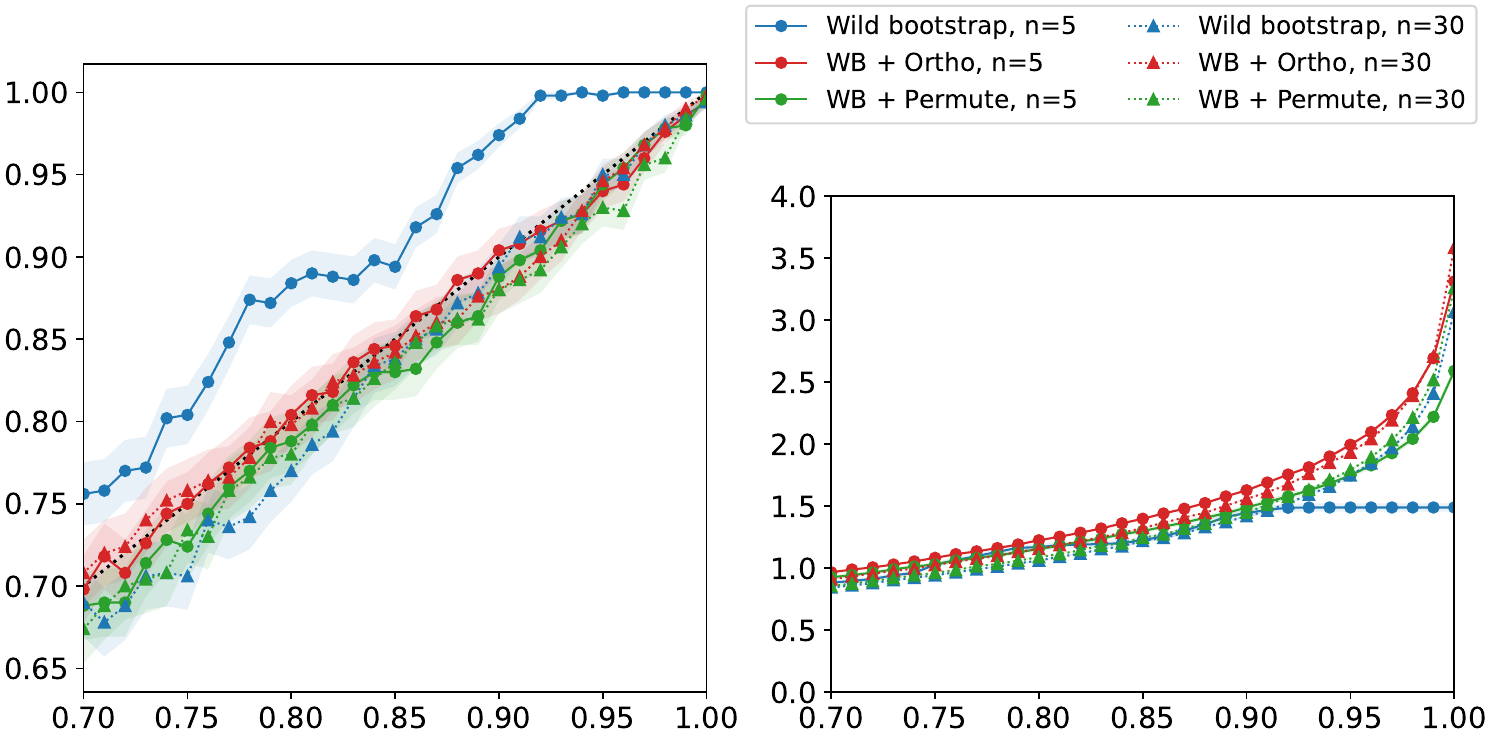}};

            \node[inner sep=0pt] at (4.6,1.25)
                {\includegraphics[width=.32\textwidth]{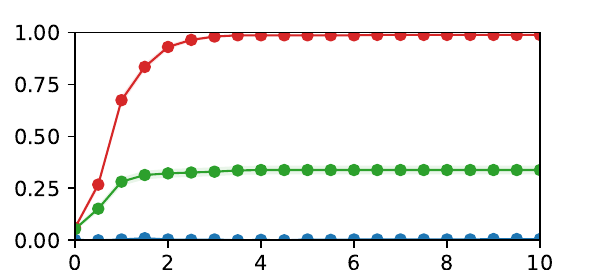}};

            \node[inner sep=0pt] at (4.6,-1.4)
                {\includegraphics[width=.32\textwidth]{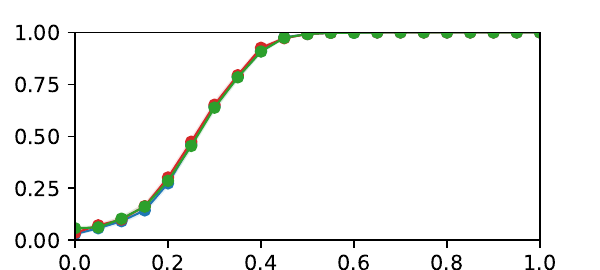}};
            
            \node[inner sep=0pt] at (-5.6,2.3){\scriptsize \textbf{(i) Empirical coverage}};
            \node[inner sep=0pt] at (-0.45,1.4){\scriptsize \textbf{(ii) Confidence interval (CI) length}};
            \node[inner sep=0pt] at (4.7,2.3){\scriptsize \textbf{(iii) Test power, $n=5$}};
            \node[inner sep=0pt] at (4.7,0){\scriptsize mean shift};
            \node[inner sep=0pt] at (4.7,-0.4){\scriptsize \textbf{(iv) Test power, $n=30$}};
            \node[inner sep=0pt] at (4.7,-2.7){\scriptsize mean shift};
 
            \node[inner sep=0pt] at (-5.6,-2.7){\scriptsize target coverage $1-\alpha$};
            \node[inner sep=0pt] at (-0.4,-2.7){\scriptsize target coverage $1-\alpha$};
        \end{scope}
  \end{tikzpicture} 
  \\[-.2em]
  \caption{MMD statistics with RBF kernel on 2d centred Gammas over 500 random trials. \textbf{(i):} Empirical v.s. target coverage. \textbf{(ii):} CI length~v.s. target coverage. \textbf{(iii) and (iv):} Test power against mean shift. 
  }
  \label{fig:wb:simulate:gamma} 
\end{figure}

\vspace{.5em}

\noindent 
\textbf{(a) 2d simulated data under mean shift.} In \Cref{fig:wb:simulate:rademacher,fig:wb:simulate:gaussian,fig:wb:simulate:gamma}, $P$ is one of the three 2d distributions defined in the first setup of \Cref{appendix:experiments:pure:aug:bootstrap}, i.e.~Gaussian, Rademacher and centred Gamma bivariate distributions. Plots \textbf{(i)} and \textbf{(ii)} of each figure show the quality of the DAB CI under the null. Plots \textbf{(iii)} and \textbf{(iv)} of each figure examines the quality of the DAB CI under mean shifts: $Q$ is the same distribution as $P$ except that each random vector is shifted by a deterministic vector $(\theta, \theta)$, and the x-axis of plot \textbf{(iii)} and \textbf{(iv)} gives the value of $\theta$. The augmentations used are (a) random orthogonal transformations, and (b) random permutations of the two coordinates.

\begin{figure}[t]
  \centering
  \begin{tikzpicture}
        \coordinate (gaussian) at (0,0);
        \begin{scope}[shift={(gaussian)}]
            \node[inner sep=0pt] at (0,0)
                {\includegraphics[width=.7\textwidth]{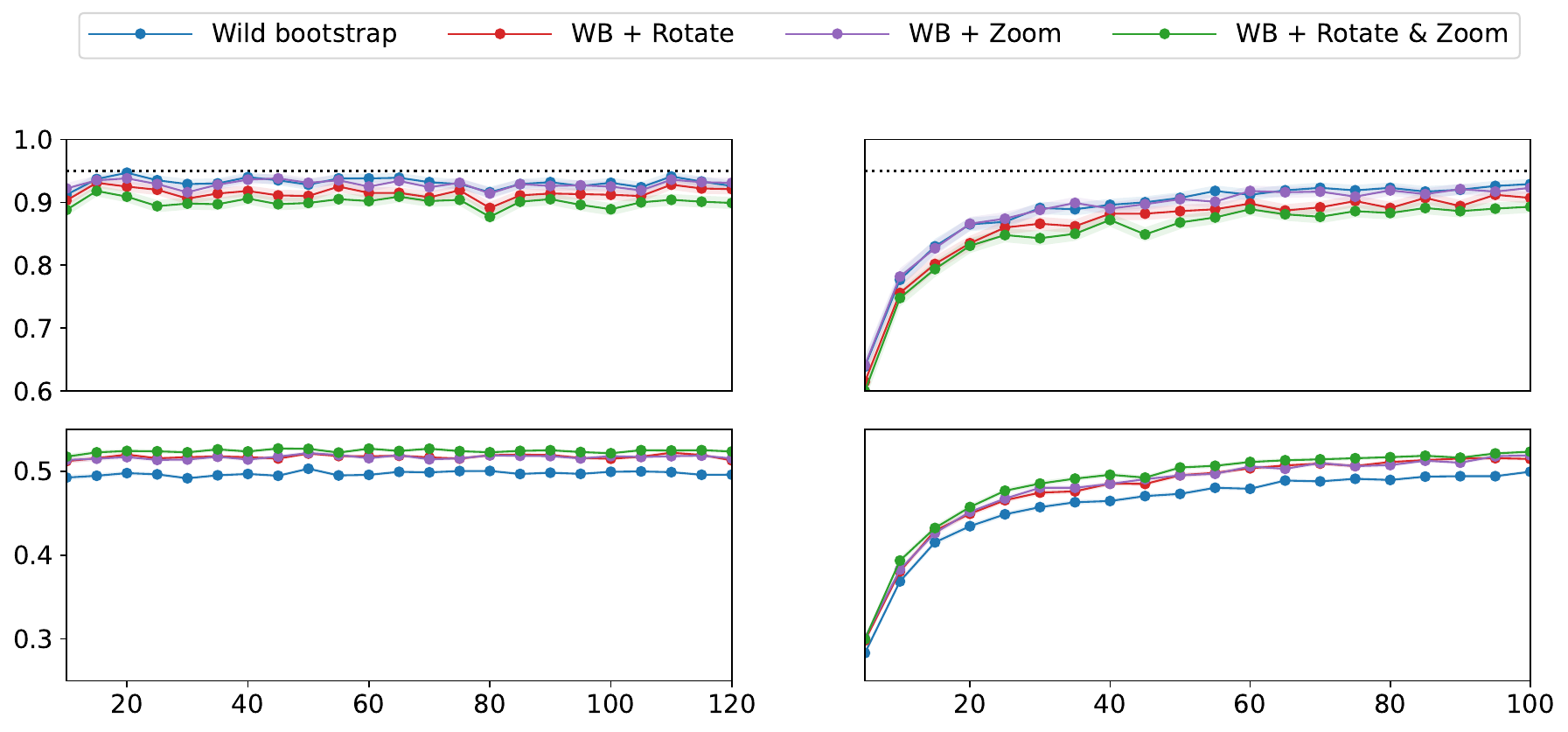}};
            
            \node[inner sep=0pt] at (-2.5,1.7){\scriptsize \textbf{Varying $n$, \;$B=100$, \, $1-\alpha=95\%$}};

            \node[inner sep=0pt] at (2.8,1.7){\scriptsize \textbf{Varying $B$, \;$n=50$, \, $1-\alpha=95\%$}};

            \node[inner sep=0pt,rotate=90] at (-5.4,0.65){\scriptsize coverage};

            \node[inner sep=0pt,rotate=90] at (-5.4,-1.2){\scriptsize CI length};

            \node[inner sep=0pt] at (-2.5,-2.6){\scriptsize $n$};

            \node[inner sep=0pt] at (2.8,-2.6){\scriptsize $B$};
        \end{scope}
  \end{tikzpicture} 
  \\[-.5em]
  \caption{Empirical coverage and CI length of DAB for MMD statistics with RBF kernel under the null of noiseless MNIST images, i.e.~$\sigma=0$. The experiments are over varying $n$, $B$ and over 1000 random draws of the dataset, with a fixed target coverage $1-\alpha=95\%$. See \Cref{fig:wb:mnist:alt} for results for testing against the alternative of noisy images.
  }
  \label{fig:wb:mnist:null} 
\end{figure}
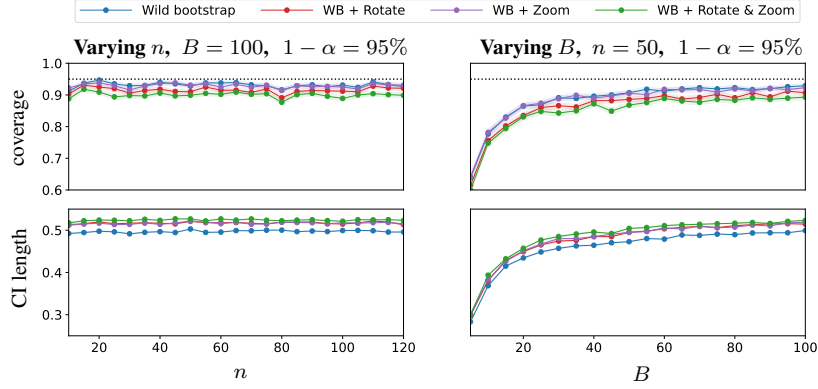

\vspace{.5em}

\noindent
 \textbf{(b) Noiseless v.s. noisy MNIST images \citep{lecun1998gradient}.} In \Cref{fig:wb:mnist:alt,fig:wb:mnist:null}, $P$ is the empirical distribution of all MNIST images, whereas $Q$ is the empirical distribution of MNIST images corrupted with additive Gaussian noise that is pixel-wise i.i.d.~$\cN(0,\sigma^2)$. $n$ represents the number of random samples from each empirical distribution. All MNIST images have been downsized to $14 \times 14$. \Cref{fig:mnist:visual} illustrates a downsampled MNIST image, its noisy versions and its differently augmented versions. The augmentations used are (a) random rotation of images in pixel space with angles uniformly chosen from $[-15, 15]$ in degrees, (b) random zooming of images with ratio uniformly chosen from $[0.9,1.1]$, and (c) a composition of both random rotation and zooming of images.

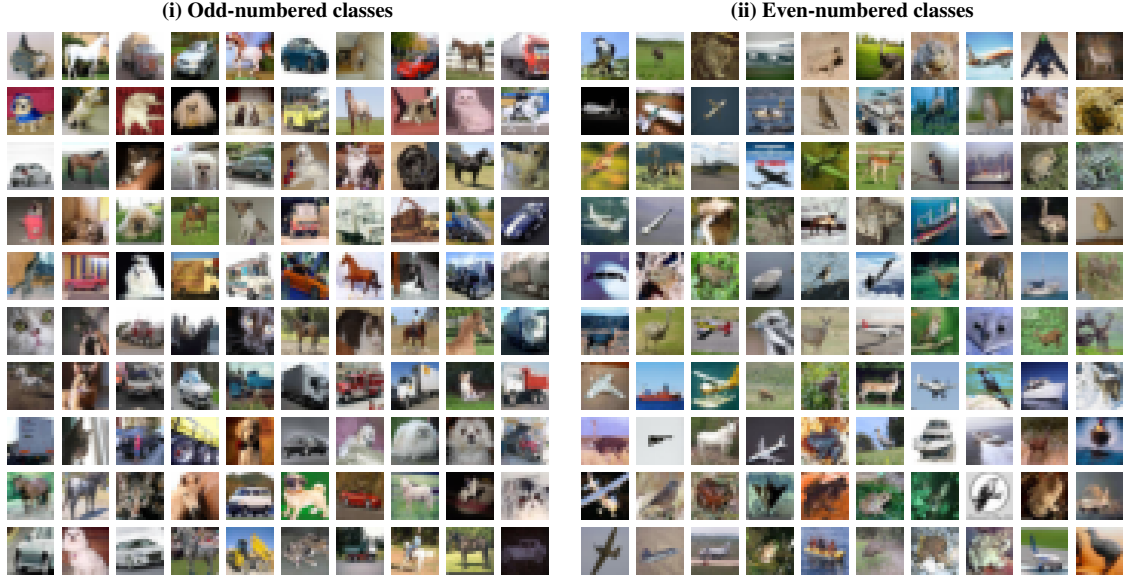
\begin{figure}[t]
  \centering
  \begin{tikzpicture}
        \coordinate (cifargrid) at (0,0);
        \begin{scope}[shift={(cifargrid)}]
            \node[inner sep=0pt] (cifarodd) at (-3.8,0)
                {\includegraphics[width=.49\textwidth]{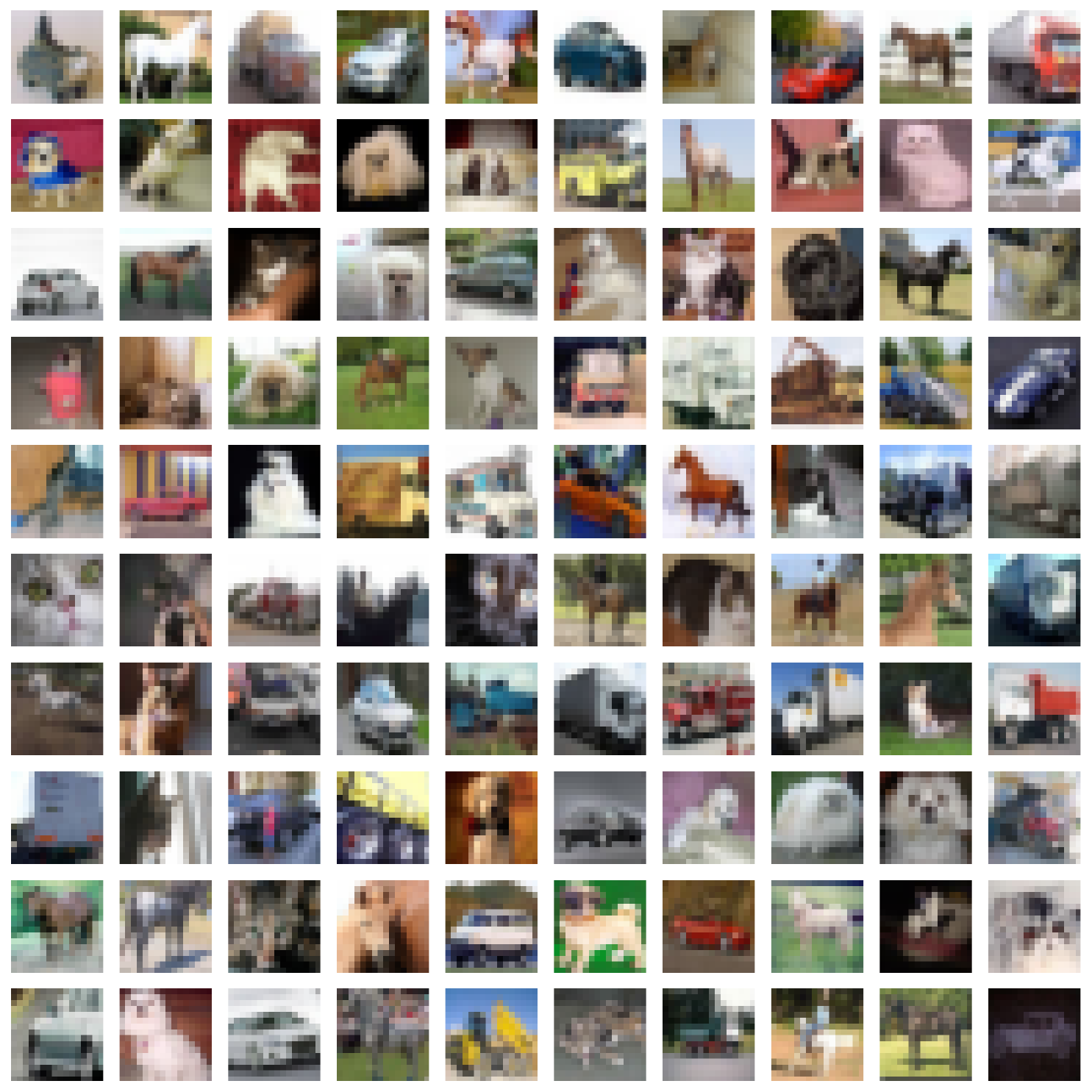}};
            \node[inner sep=0pt] (cifareven) at (3.8,0)
                {\includegraphics[width=.49\textwidth]{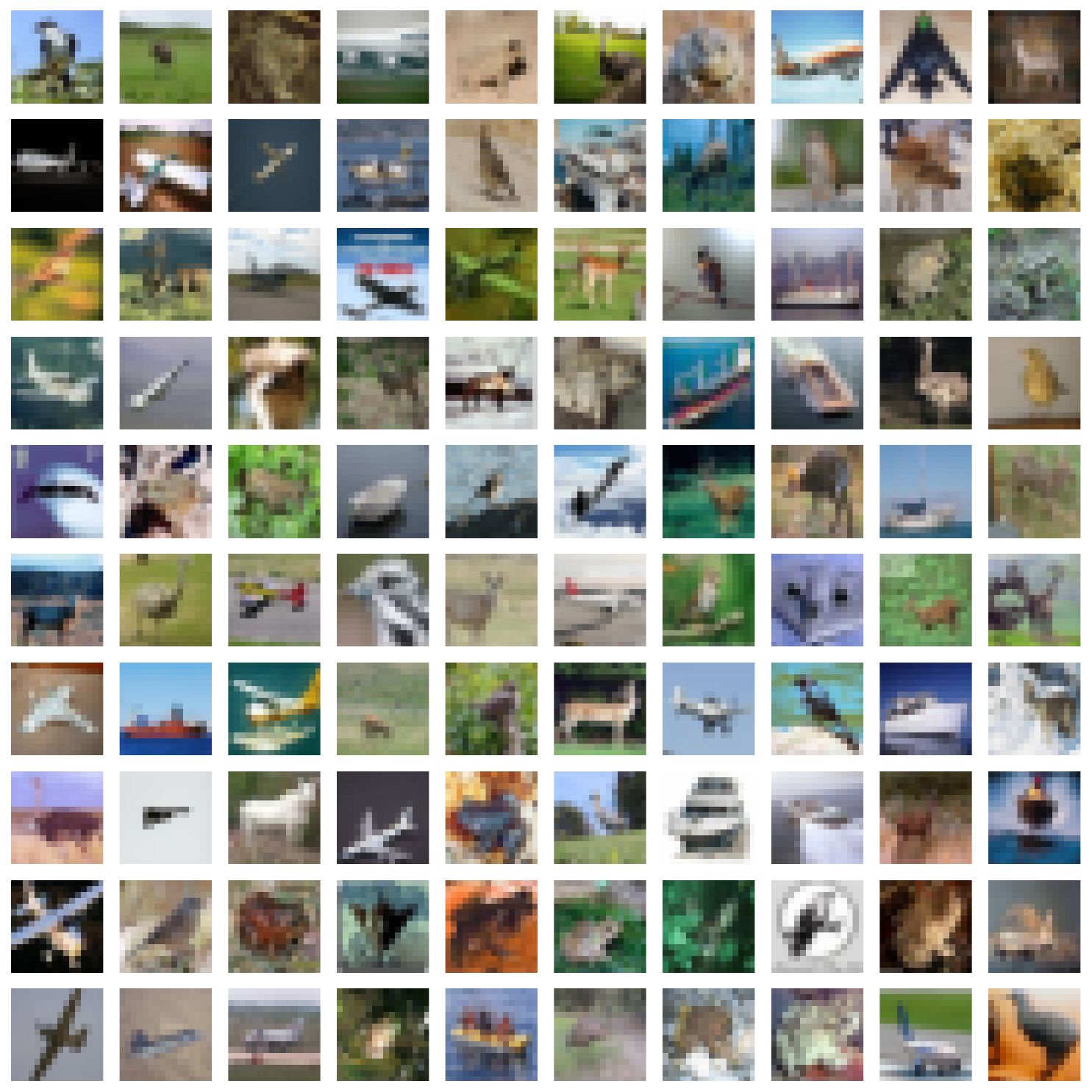}};

            \path (cifarodd.north) ++(0,0.2)
                node[inner sep=0pt]{\scriptsize \textbf{(i) Odd-numbered classes }};
            \path (cifareven.north) ++(0,0.2)
                node[inner sep=0pt]{\scriptsize \textbf{(ii) Even-numbered classes }};
        \end{scope}
  \end{tikzpicture}
  \\[-.2em]
  \caption{Randomly sampled CIFAR-10 images of odd-numbered classes and of even-numbered classes.  The odd-numbered classes are \texttt{'Automobile', 'Cat', 'Dog','Horse','Truck'}, and the even-numbered classes are 
  \texttt{'Airplane', 'Bird', 'Deer','Frog','Ship'}.
  } 
  \label{fig:cifar}
\end{figure}

\vspace{.5em}
\noindent 
\textbf{(c) Odd v.s.~even-numbered classes in CIFAR-10 images \cite{krizhevsky2009learning}.} In \Cref{fig:wb:cifar:null,fig:wb:cifar:alt}, $P$ is the empirical distribution over all CIFAR-10 images of odd-numbered classes and $Q$ is the empirical distribution over all CIFAR-10 images of even-numbered classes.

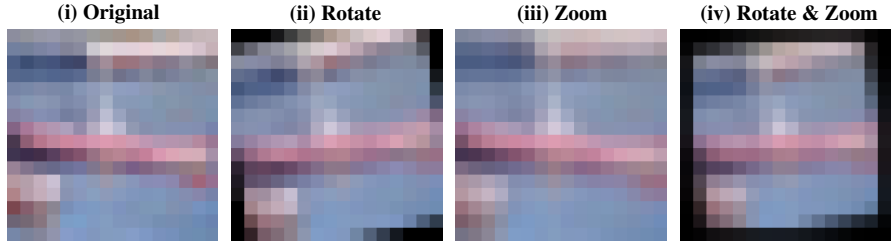
\begin{figure}[h!]
  \centering
  \begin{tikzpicture}
            \node[inner sep=0pt] at (0,0)
                {\includegraphics[width=.8\textwidth]{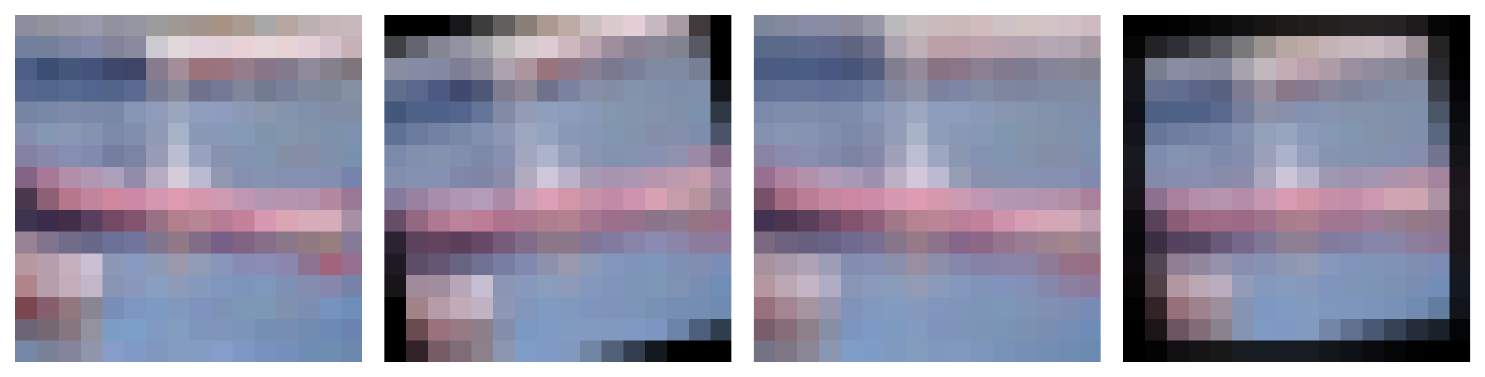}};

            \node[inner sep=0pt] at (-4.5,1.6){\scriptsize \textbf{(i) Original}};
            \node[inner sep=0pt] at (-1.5,1.6){\scriptsize \textbf{(ii) Rotate}};
            \node[inner sep=0pt] at (1.5,1.6){\scriptsize \textbf{(iii) Zoom}};
            \node[inner sep=0pt] at (4.5,1.6){\scriptsize \textbf{(iv) Rotate \& Zoom}};
  \end{tikzpicture}
  \\[-.2em]
  \caption{A randomly sampled CIFAR-10 image with the augmentations considered in our experiments. }
  \label{fig:cifar:visual}
\end{figure}
\begin{figure}[h!]
  \centering
  \begin{tikzpicture}
        \coordinate (gaussian) at (0,0);
        \begin{scope}[shift={(gaussian)}]
            \node[inner sep=0pt] at (0,0)
                {\includegraphics[width=.7\textwidth]{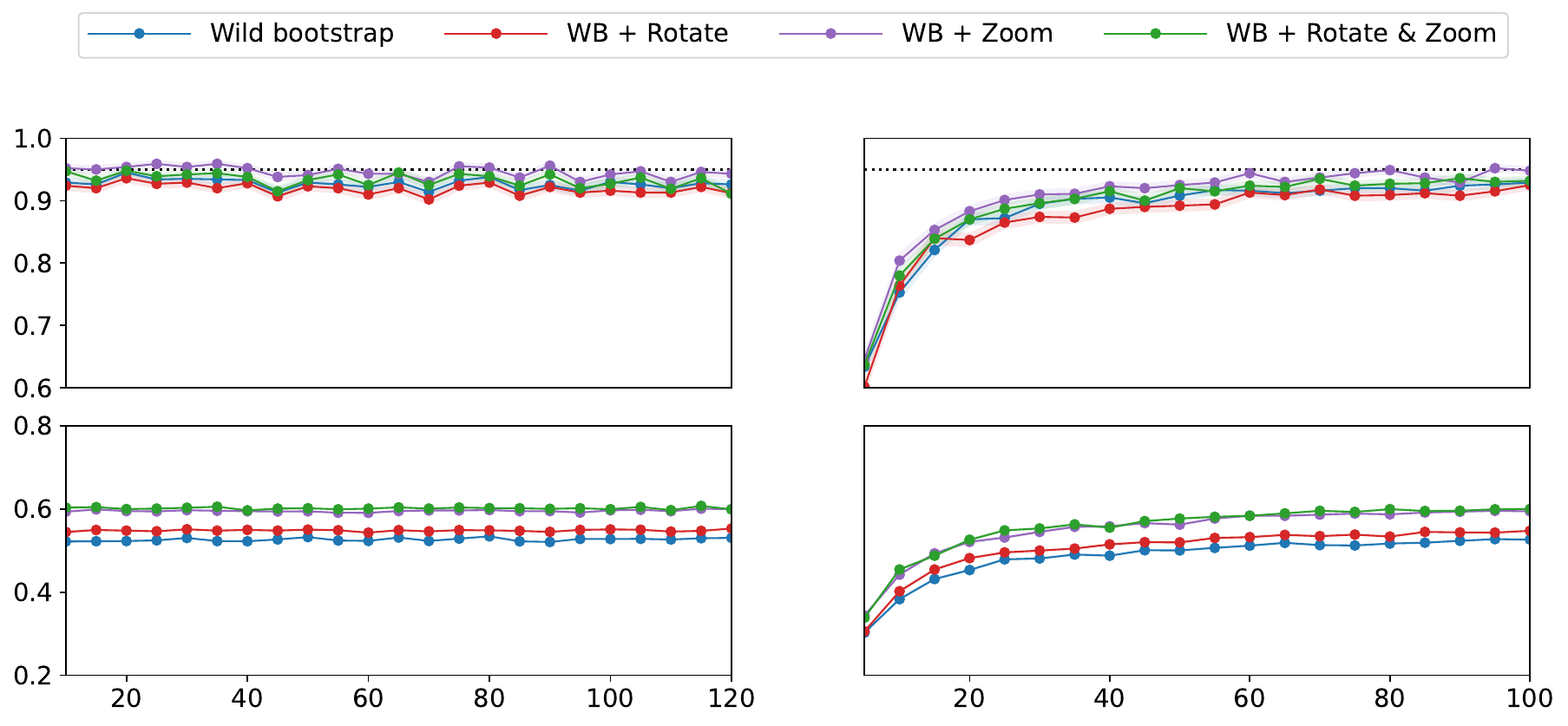}};
            
            \node[inner sep=0pt] at (-2.5,1.7){\scriptsize \textbf{Varying $n$, \;$B=100$, \, $1-\alpha=95\%$}};

            \node[inner sep=0pt] at (2.8,1.7){\scriptsize \textbf{Varying $B$, \;$n=30$, \, $1-\alpha=95\%$}};

            \node[inner sep=0pt,rotate=90] at (-5.4,0.65){\scriptsize coverage};

            \node[inner sep=0pt,rotate=90] at (-5.4,-1.2){\scriptsize CI length};

            \node[inner sep=0pt] at (-2.5,-2.5){\scriptsize $n$};

            \node[inner sep=0pt] at (2.8,-2.5){\scriptsize $B$};
        \end{scope}
  \end{tikzpicture} 
  \\[-.5em]
  \caption{Empirical coverage and CI length of DAB for MMD statistics with RBF kernel under the null of CIFAR images with odd-numbered classes, i.e.~(i) in \Cref{fig:cifar}. The experiments are over varying $n$, $B$ and over 1000 random draws of the dataset, with a fixed target coverage $1-\alpha=95\%$.
  }
  \label{fig:wb:cifar:null} 
\end{figure}
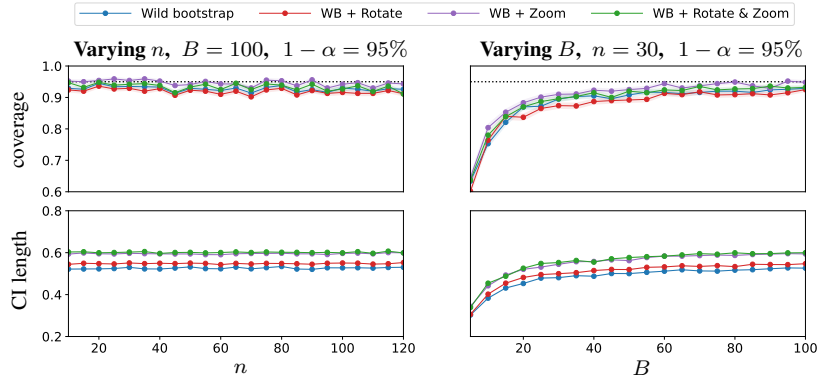
\begin{figure}[h!]
  \centering
  \begin{tikzpicture}
        \coordinate (gaussian) at (0,0);
        \begin{scope}[shift={(gaussian)}]
            \node[inner sep=0pt] at (0,0)
                {\includegraphics[width=.7\textwidth]{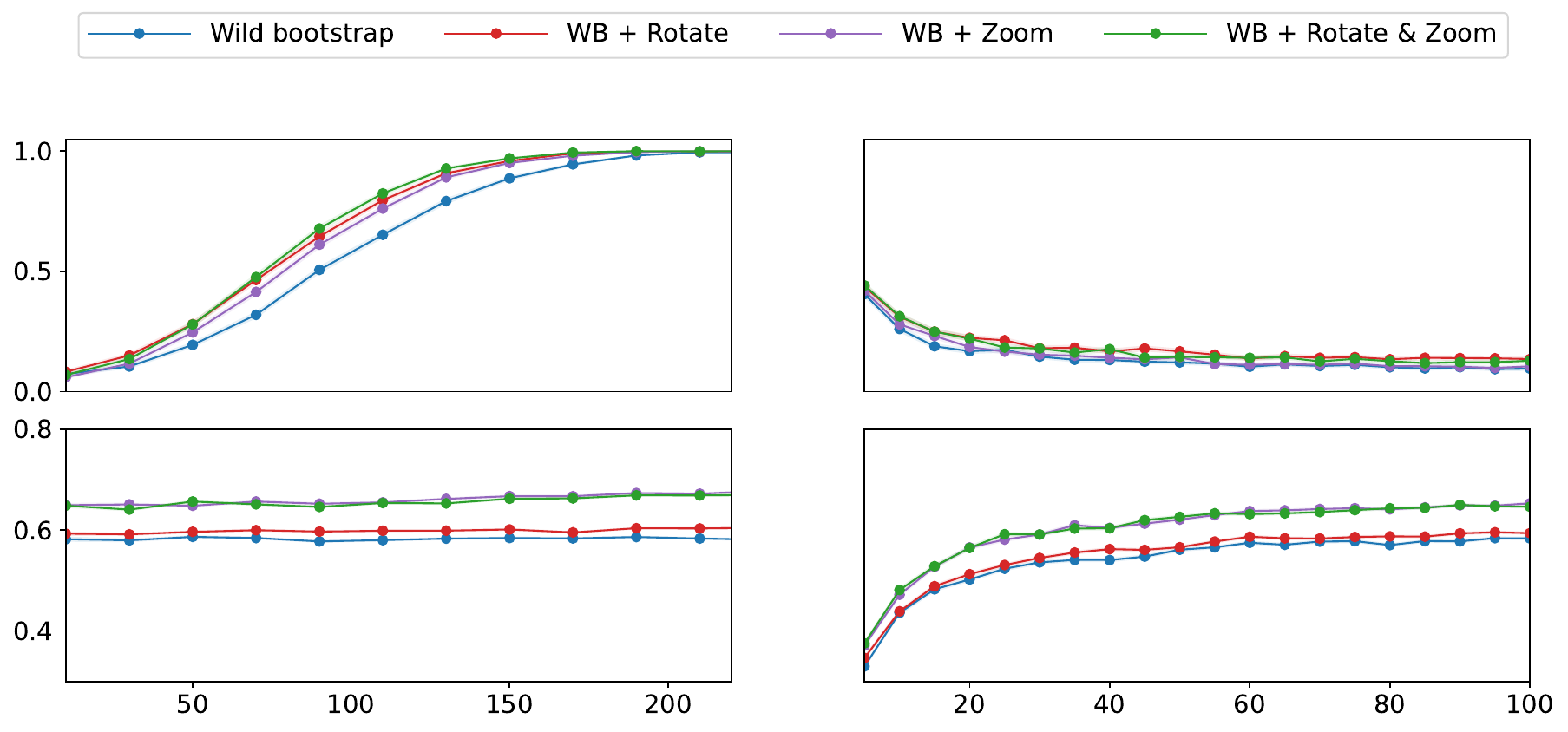}};
            
            \node[inner sep=0pt] at (-2.5,1.7){\scriptsize \textbf{Varying $n$, \;$B=100$}};

            \node[inner sep=0pt] at (2.8,1.7){\scriptsize \textbf{Varying $B$, \;$n=30$}};

            \node[inner sep=0pt,rotate=90] at (-5.4,0.65){\scriptsize power};

            \node[inner sep=0pt,rotate=90] at (-5.4,-1.2){\scriptsize CI length};

            \node[inner sep=0pt] at (-2.5,-2.5){\scriptsize $n$};

            \node[inner sep=0pt] at (2.8,-2.5){\scriptsize $B$};
        \end{scope}
  \end{tikzpicture} 
  \\[-.5em]
  \caption{Empirical power and CI length of DAB for MMD statistics with RBF kernel for testing CIFAR images of odd-numbered classes against those of even-numbered classes, i.e.~(i) against (ii) in \Cref{fig:cifar}. The experiments are over varying $n$, $B$ and over 1000 random draws of the dataset, with a fixed target Type-I error $\alpha=5\%$. 
  }
  \label{fig:wb:cifar:alt} 
\end{figure}
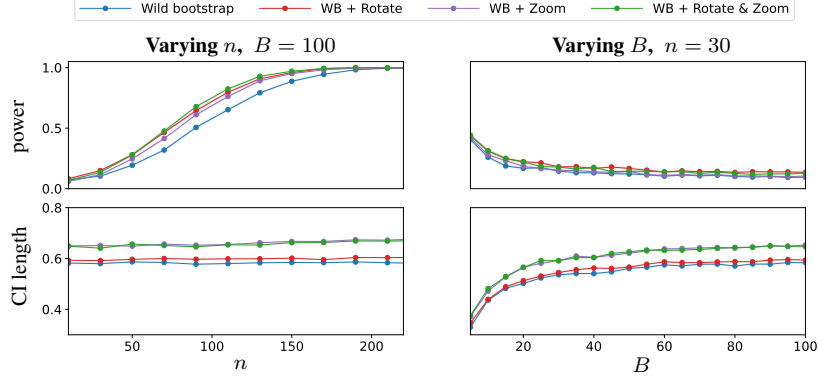

\noindent
\Cref{fig:cifar} illustrates some examples of the CIFAR-10 images of odd- and even-numbered classes. All CIFAR images are downsized to $16 \times 16$. \Cref{fig:cifar:visual} illustrates a downsampled CIFAR image and its differently augmented versions; the augmentations used are the same as those used for MNIST above. Note that DAB slightly outperforms wild bootstrap in terms of test power, but at the cost of giving a slightly wider confidence interval under the null.

\begin{figure}[t]
  \centering
  \begin{tikzpicture}
        \coordinate (gaussian) at (0,0);
        \begin{scope}[shift={(gaussian)}]
            \node[inner sep=0pt] at (0,0)
                {\includegraphics[width=.7\textwidth]{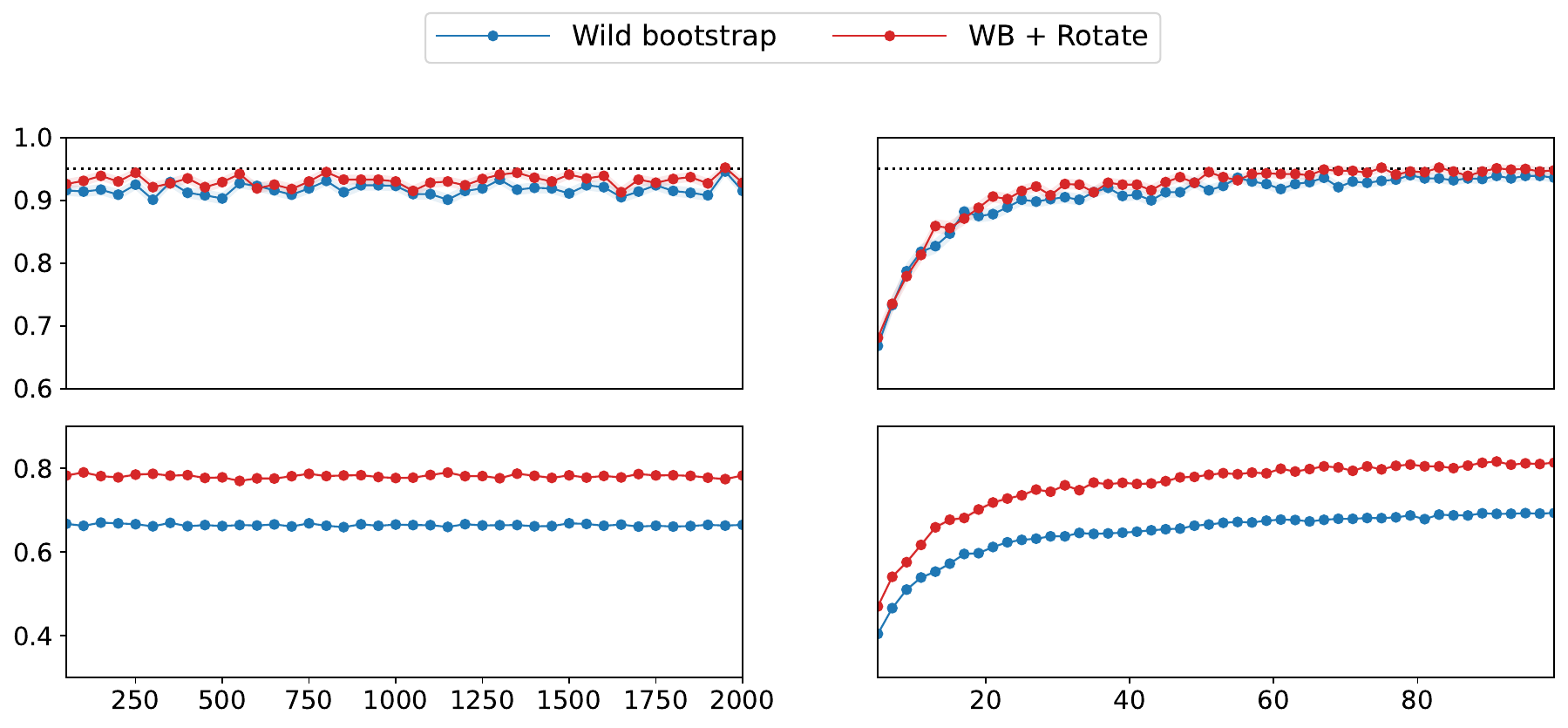}};
            
            \node[inner sep=0pt] at (-2.5,1.7){\scriptsize \textbf{Varying $n$, \;$B=50$, \, $1-\alpha=95\%$}};

            \node[inner sep=0pt] at (2.8,1.7){\scriptsize \textbf{Varying $B$, \;$n=2000$, \, $1-\alpha=95\%$}};

            \node[inner sep=0pt,rotate=90] at (-5.4,0.65){\scriptsize coverage};

            \node[inner sep=0pt,rotate=90] at (-5.4,-1.2){\scriptsize CI length};

            \node[inner sep=0pt] at (-2.5,-2.5){\scriptsize $n$};

            \node[inner sep=0pt] at (2.8,-2.5){\scriptsize $B$};
        \end{scope}
  \end{tikzpicture} 
  \\[-.5em]
  \caption{Empirical coverage and CI length of DAB for MMD statistics with RBF kernel under the null of background noise in the HIGGS dataset. The experiments are over varying $n$, $B$ and over 1000 random draws of the dataset, with a fixed target coverage $1-\alpha=95\%$. The augmentations used in DAB are simultaneous rotations of all azimuthal angles by an angle uniformly drawn from $[-90,90]$ in degrees.
  }
  \label{fig:wb:higgs:null} 
\end{figure}
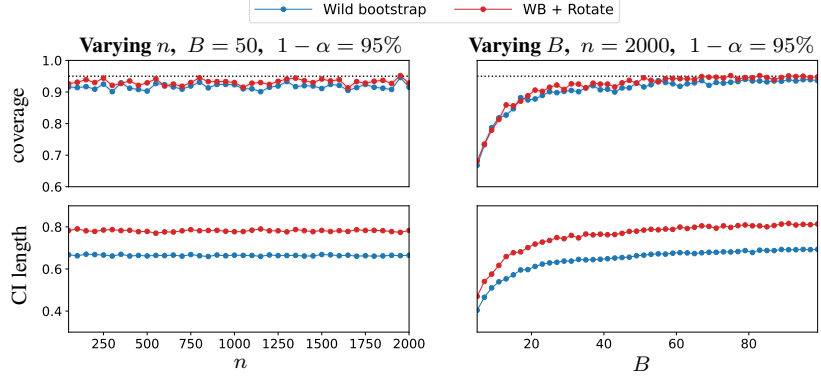
\begin{figure}[t]
  \centering
  \begin{tikzpicture}
        \coordinate (gaussian) at (0,0);
        \begin{scope}[shift={(gaussian)}]
            \node[inner sep=0pt] at (0,0)
                {\includegraphics[width=.7\textwidth]{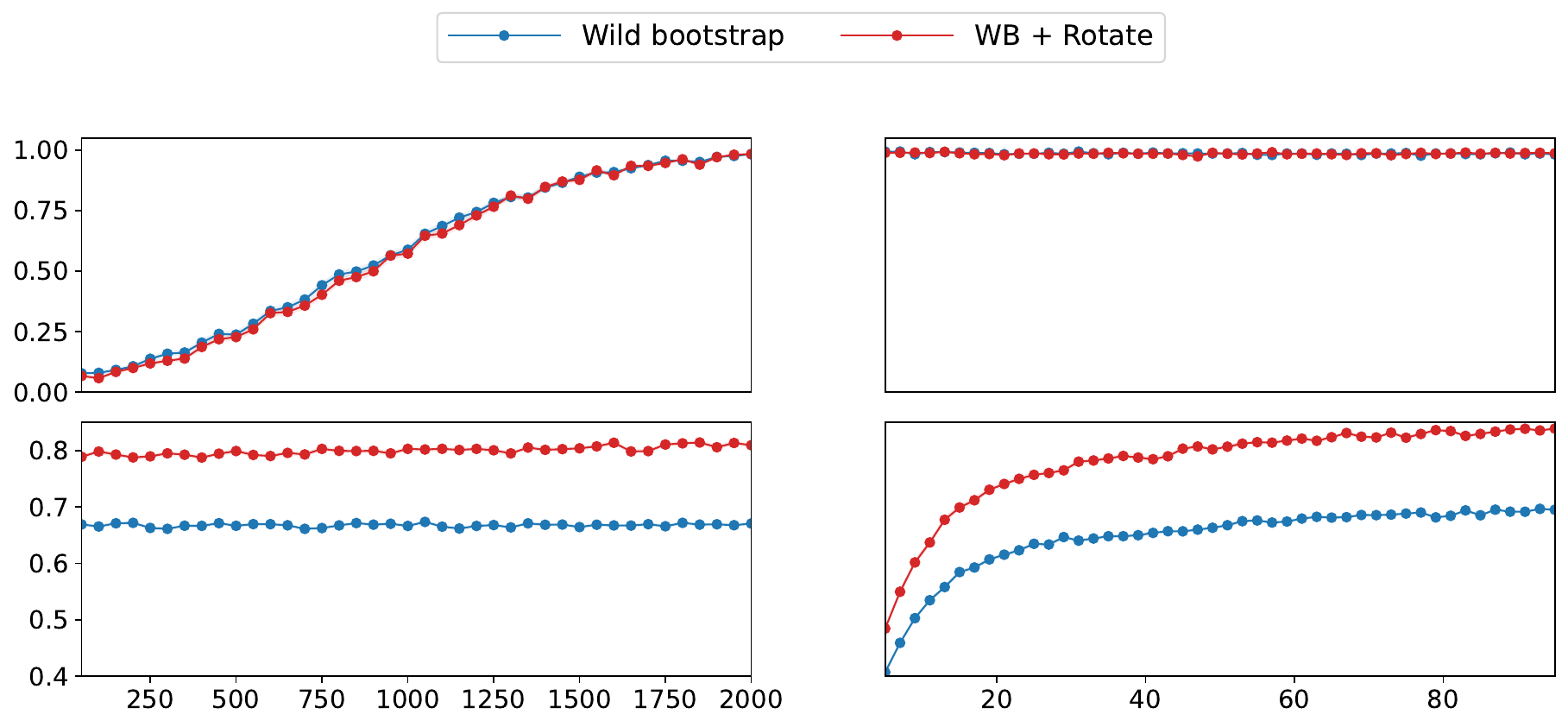}};
            
            \node[inner sep=0pt] at (-2.5,1.7){\scriptsize \textbf{Varying $n$, \;$B=50$, \, $1-\alpha=95\%$}};

            \node[inner sep=0pt] at (2.8,1.7){\scriptsize \textbf{Varying $B$, \;$n=2000$, \, $1-\alpha=95\%$}};

            \node[inner sep=0pt,rotate=90] at (-5.4,0.65){\scriptsize power};

            \node[inner sep=0pt,rotate=90] at (-5.4,-1.2){\scriptsize CI length};

            \node[inner sep=0pt] at (-2.5,-2.5){\scriptsize $n$};

            \node[inner sep=0pt] at (2.8,-2.5){\scriptsize $B$};
        \end{scope}
  \end{tikzpicture} 
  \\[-.5em]
  \caption{Empirical power and CI length of DAB for MMD statistics with RBF kernel for testing background noise against signal in the HIGGS dataset. The experiments are over varying $n$, $B$ and over 1000 random draws of the dataset, with a fixed target Type-I error $\alpha=5\%$.  The augmentations used in DAB are simultaneous rotations of all azimuthal angles by an angle uniformly drawn from $[-90,90]$ in degrees.
  }
  \label{fig:wb:higgs:alt} 
\end{figure}
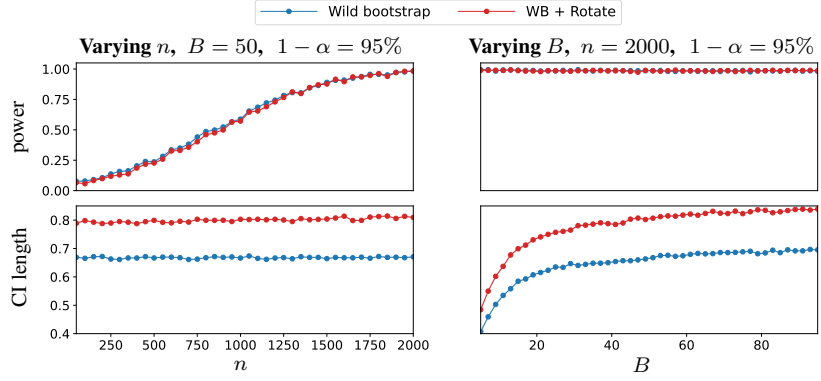
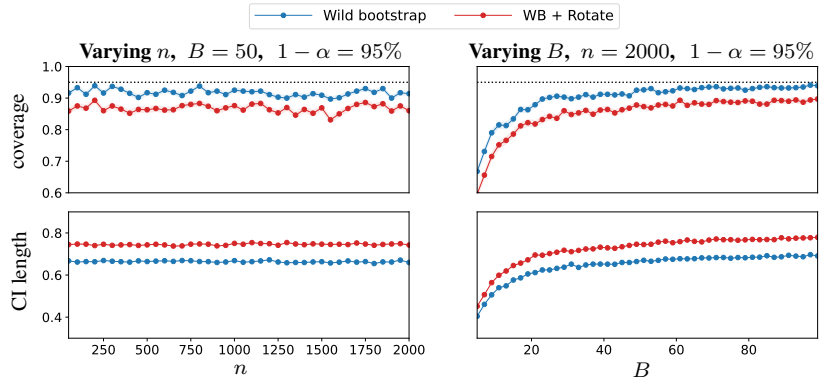
\begin{figure}[t]
  \centering
  \begin{tikzpicture}
        \coordinate (gaussian) at (0,0);
        \begin{scope}[shift={(gaussian)}]
            \node[inner sep=0pt] at (0,0)
                {\includegraphics[width=.7\textwidth]{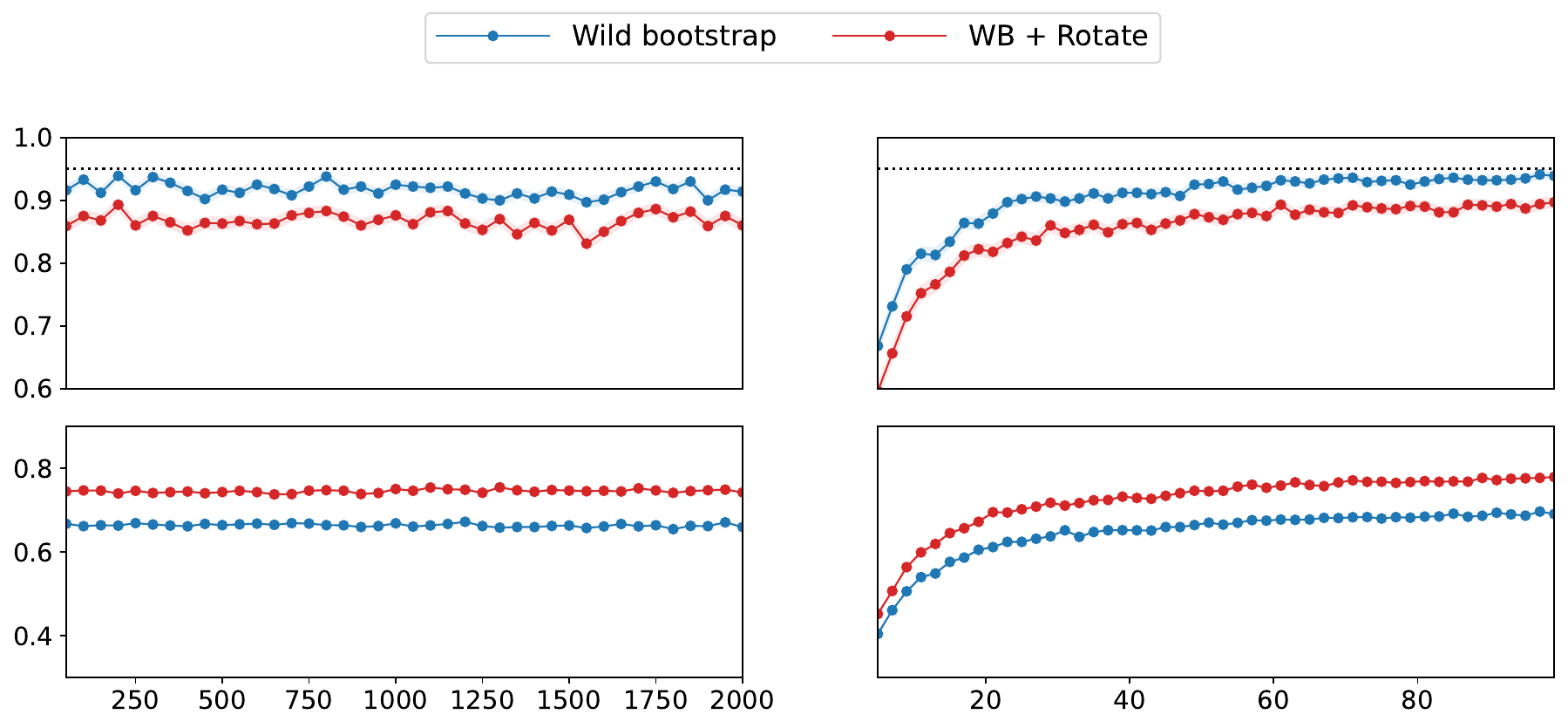}};
            
            \node[inner sep=0pt] at (-2.5,1.7){\scriptsize \textbf{Varying $n$, \;$B=50$, \, $1-\alpha=95\%$}};

            \node[inner sep=0pt] at (2.8,1.7){\scriptsize \textbf{Varying $B$, \;$n=2000$, \, $1-\alpha=95\%$}};

            \node[inner sep=0pt,rotate=90] at (-5.4,0.65){\scriptsize coverage};

            \node[inner sep=0pt,rotate=90] at (-5.4,-1.2){\scriptsize CI length};

            \node[inner sep=0pt] at (-2.5,-2.5){\scriptsize $n$};

            \node[inner sep=0pt] at (2.8,-2.5){\scriptsize $B$};
        \end{scope}
  \end{tikzpicture} 
  \\[-.5em]
  \caption{Empirical coverage and CI length of DAB for MMD statistics with RBF kernel under the null of background noise in the HIGGS dataset. Same setup as \Cref{fig:wb:higgs:null}, except that rotation angles are uniformly drawn from $[-180,180]$ in degrees.
  }
  \label{fig:wb:higgs:null:180} 
\end{figure}
\begin{figure}[t]
  \centering
  \begin{tikzpicture}
        \coordinate (gaussian) at (0,0);
        \begin{scope}[shift={(gaussian)}]
            \node[inner sep=0pt] at (0,0)
                {\includegraphics[width=.7\textwidth]{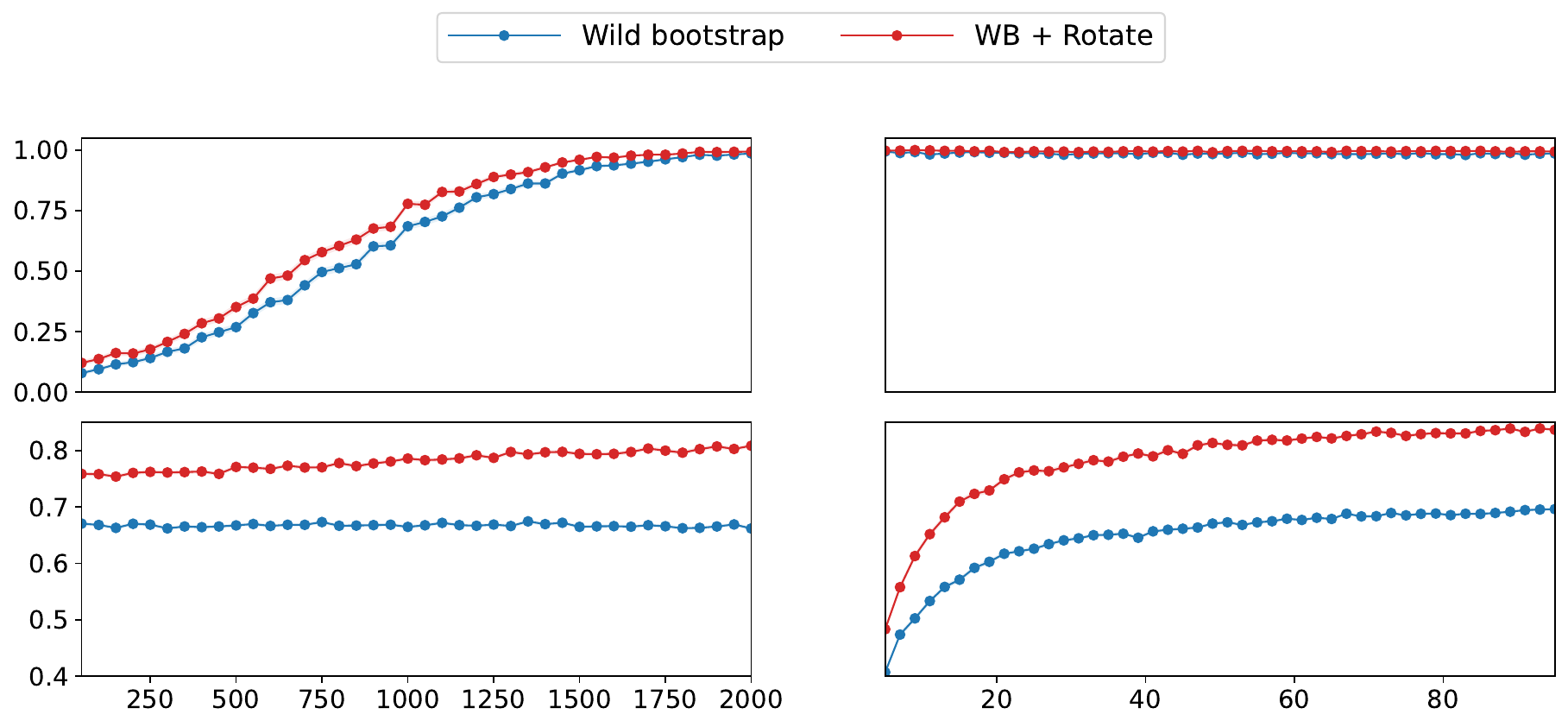}};
            
            \node[inner sep=0pt] at (-2.5,1.7){\scriptsize \textbf{Varying $n$, \;$B=50$, \, $1-\alpha=95\%$}};

            \node[inner sep=0pt] at (2.8,1.7){\scriptsize \textbf{Varying $B$, \;$n=2000$, \, $1-\alpha=95\%$}};

            \node[inner sep=0pt,rotate=90] at (-5.4,0.65){\scriptsize power};

            \node[inner sep=0pt,rotate=90] at (-5.4,-1.2){\scriptsize CI length};

            \node[inner sep=0pt] at (-2.5,-2.5){\scriptsize $n$};

            \node[inner sep=0pt] at (2.8,-2.5){\scriptsize $B$};
        \end{scope}
  \end{tikzpicture} 
  \\[-.5em]
  \caption{Empirical power and CI length of DAB for MMD statistics with RBF kernel for testing background noise against signal in the HIGGS dataset. Same setup as \Cref{fig:wb:higgs:alt}, except that rotation angles are uniformly drawn from $[-180,180]$ in degrees.
  }
  \label{fig:wb:higgs:alt:180} 
\end{figure}
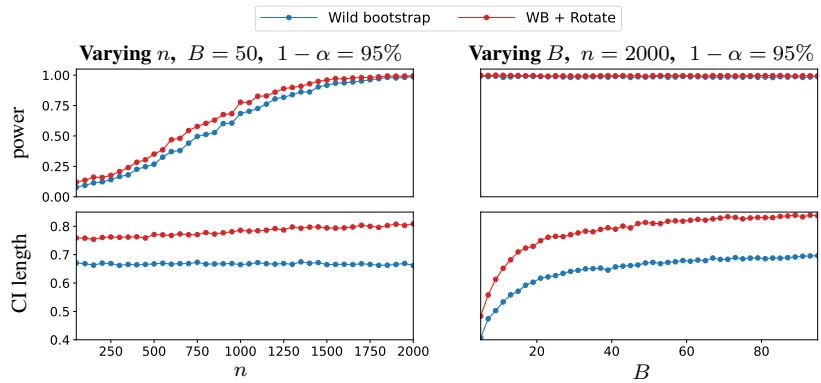
 
\vspace{.5em}
\noindent 
\textbf{(d) Physical signal v.s.~noise in a HIGGS boson dataset \cite{higgs2014}.} In \Cref{fig:wb:higgs:null,fig:wb:higgs:alt}, $P$ is the empirical distribution of the 21-dimensional physical signal vectors in the HIGGS boson dataset, whereas $Q$ is the empirical distribution of the 21-dimensional background noise vectors in the same dataset. Note that we have only kept the 21 low-level features in the HIGGS boson dataset and left out the 7 high-level features that are synthesised from the 21 features. We observe that out of the 21 features, there are 6 features that are azimuthal angles, and we consider augmentations that simultaneously rotate all of them by the same angle. \Cref{fig:wb:higgs:null,fig:wb:higgs:alt} report results for DAB with rotations i.i.d.~uniformly drawn from $[-90,90]$ in degrees, where DAB has no visible difference from vanilla wild bootstrap except for yielding a larger confidence interval. \Cref{fig:wb:higgs:null:180,fig:wb:higgs:alt:180} report those for rotations uniformly drawn from $[-180,180]$ in degrees, where DAB suffers from worse coverage under the null but gains in test power compared to vanilla wild bootstrap. We conjecture that in both cases, there are not enough symmetry information to be exploited in the data for DAB to gain improvements in terms of coverage or power, and the increase in CI length in both cases is due to the additional randomisation injected by DAB.

\clearpage

\begin{figure}[h!]
  \centering
  \begin{tikzpicture}
        \node[inner sep=0pt] at (0,0)
            {\includegraphics[width=.9\textwidth]{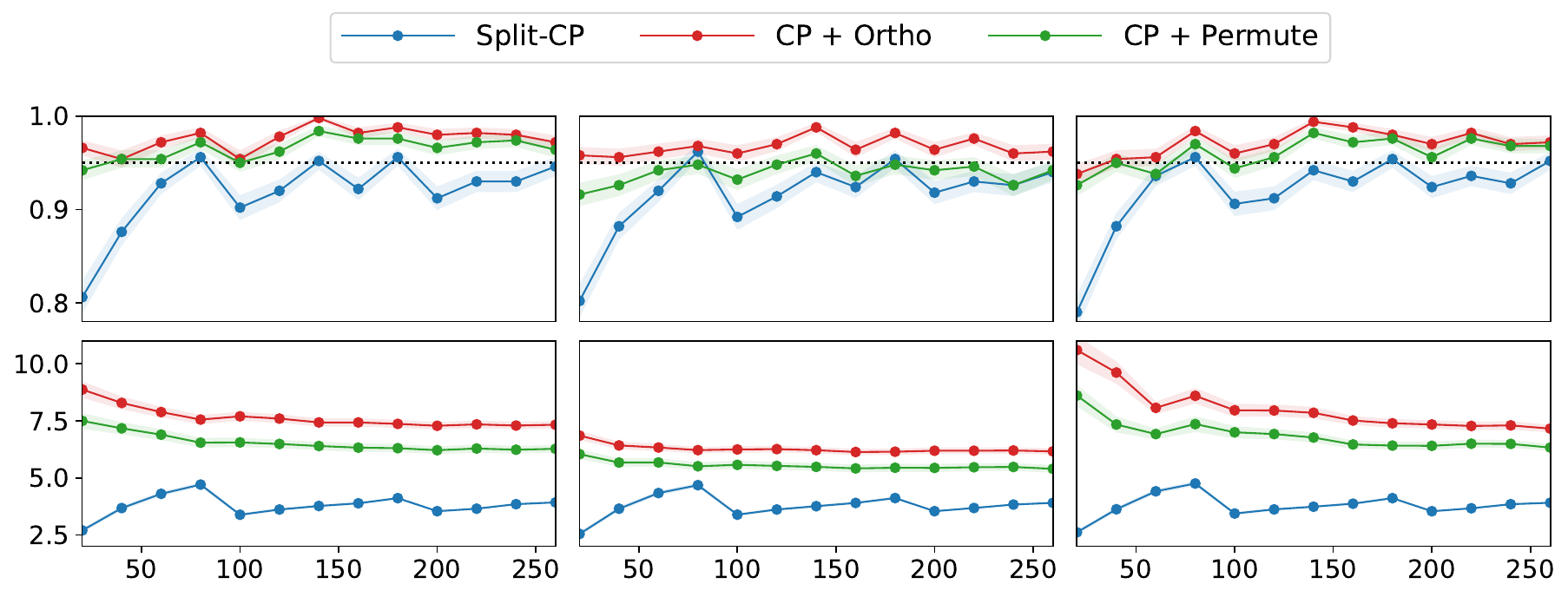}};
        \node[inner sep=0pt] at (-4,1.79){\scriptsize \textbf{Gaussian}};
        \node[inner sep=0pt] at (0.35,1.79){\scriptsize \textbf{Rademacher}};
        \node[inner sep=0pt] at (4.56,1.79){\scriptsize \textbf{Centred Gamma}};

        \node[inner sep=0pt] at (-4,-2.6){\scriptsize $n$};
        \node[inner sep=0pt] at (0.35,-2.6){\scriptsize $n$};
        \node[inner sep=0pt] at (4.5,-2.6){\scriptsize $n$};

        \node[inner sep=0pt,rotate=90] at (-6.75,0.7){\scriptsize coverage};

        \node[inner sep=0pt,rotate=90] at (-6.75,-1.3){\scriptsize CI length};

  \end{tikzpicture}
  \\[-.2em]
  \caption{DAB-variants of conformal prediction for predicting the outcome of a 2d linear model with linear regression, with $500$ random trials, $k=50$ and $1-\alpha=95\%$. Same setup as \Cref{fig:cp}(i) except that $V_i$'s can be 2d Gaussian, Rademacher or centred Gamma vectors.}
  \label{fig:cp:simulate:n}
\end{figure}
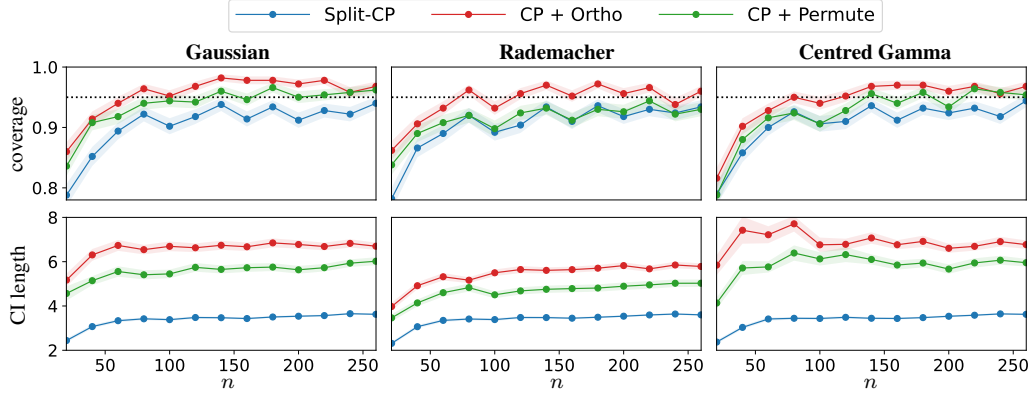
\begin{figure}[h!]
  \centering
  \begin{tikzpicture}
        \node[inner sep=0pt] at (0,0)
            {\includegraphics[width=.9\textwidth]{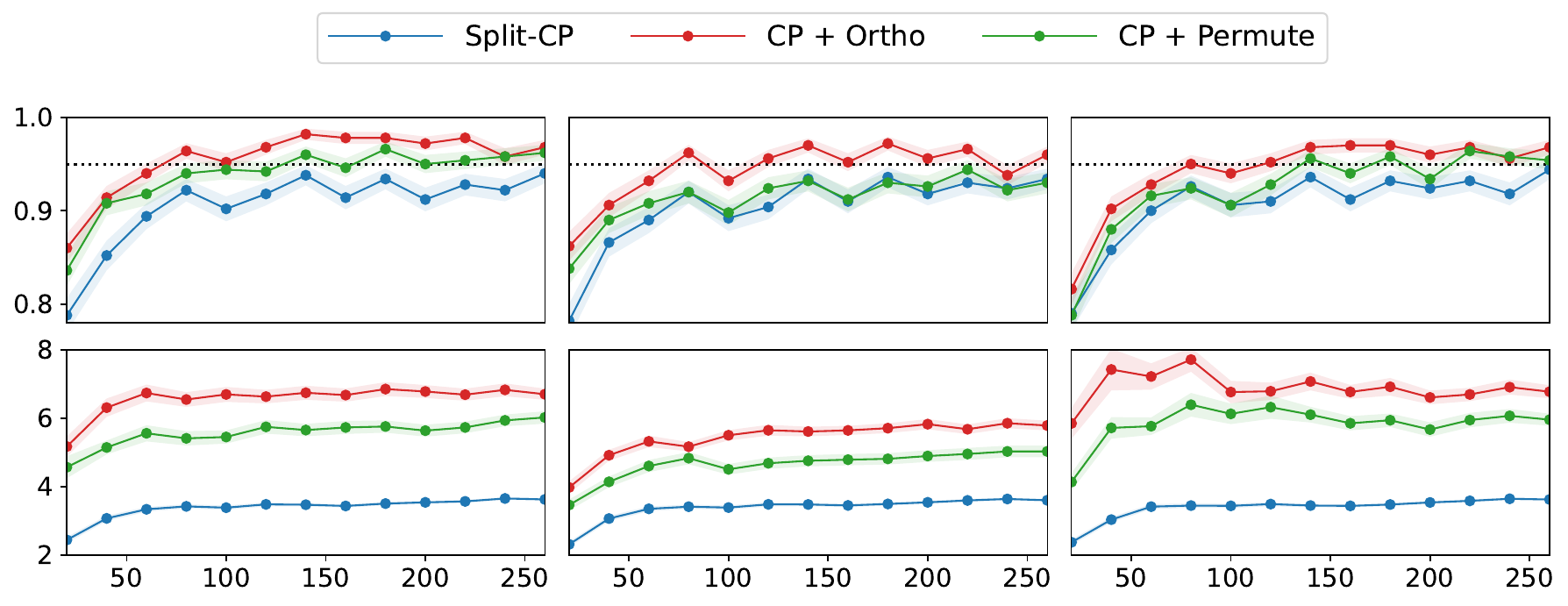}};
        \node[inner sep=0pt] at (-4,1.79){\scriptsize \textbf{Gaussian}};
        \node[inner sep=0pt] at (0.35,1.79){\scriptsize \textbf{Rademacher}};
        \node[inner sep=0pt] at (4.56,1.79){\scriptsize \textbf{Centred Gamma}};

        \node[inner sep=0pt] at (-4,-2.6){\scriptsize $n$};
        \node[inner sep=0pt] at (0.35,-2.6){\scriptsize $n$};
        \node[inner sep=0pt] at (4.5,-2.6){\scriptsize $n$};

        \node[inner sep=0pt,rotate=90] at (-6.75,0.7){\scriptsize coverage};

        \node[inner sep=0pt,rotate=90] at (-6.75,-1.3){\scriptsize CI length};

  \end{tikzpicture}
  \\[-.2em]
  \caption{Same setup as \Cref{fig:cp:simulate:n} except that $k=1$ was used.}
  \label{fig:cp:simulate:n:k1}
\end{figure}
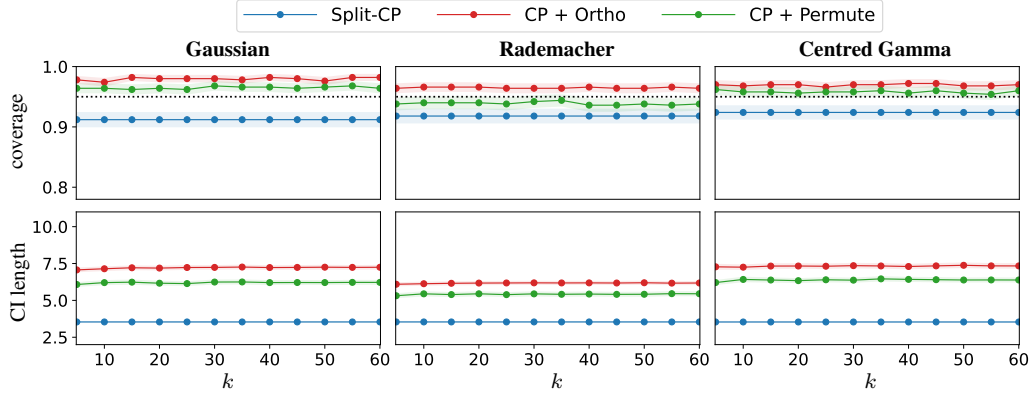
\begin{figure}[t]
  \centering
  \begin{tikzpicture}
        \node[inner sep=0pt] at (0,0)
            {\includegraphics[width=.9\textwidth]{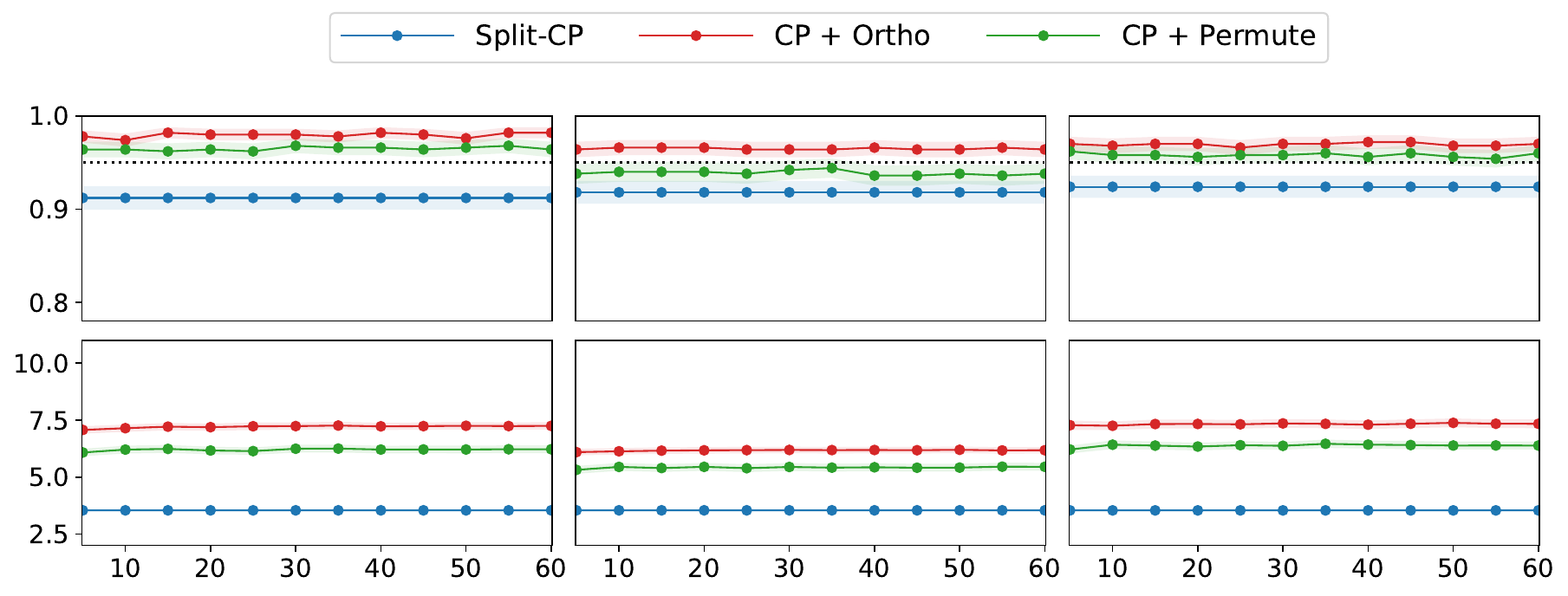}};
        \node[inner sep=0pt] at (-4,1.79){\scriptsize \textbf{Gaussian}};
        \node[inner sep=0pt] at (0.35,1.79){\scriptsize \textbf{Rademacher}};
        \node[inner sep=0pt] at (4.56,1.79){\scriptsize \textbf{Centred Gamma}};

        \node[inner sep=0pt] at (-4,-2.6){\scriptsize $k$};
        \node[inner sep=0pt] at (0.35,-2.6){\scriptsize $k$};
        \node[inner sep=0pt] at (4.5,-2.6){\scriptsize $k$};

        \node[inner sep=0pt,rotate=90] at (-6.75,0.7){\scriptsize coverage};

        \node[inner sep=0pt,rotate=90] at (-6.75,-1.3){\scriptsize CI length};

  \end{tikzpicture}
  \\[-.2em]
  \caption{DAB-variants of conformal prediction for predicting the outcome of a 2d linear model with linear regression, with $500$ random trials, $n=200$ and $1-\alpha=95\%$. Same setting as \Cref{fig:cp:simulate:n}, except that instead of $n$, we vary $k$, the number of additional augmentations on top of the cycling operations in conformal prediction; note that varying $k$ does not have an effect on \texttt{Split-CP}. }
  \label{fig:cp:simulate:k}
\end{figure}
\subsection{Experiments for conformal prediction with data augmentation in \Cref{sec:conformal}} \label{appendix:experiments:conformal}

For both split conformal prediction (CP) and its DAB variant, we always set the number of conformal prediction transformations $B= n-1$.

\vspace{.5em}

\noindent 
\textbf{(a) Linear regression on simulated data.} \Cref{fig:cp}(i) and \Cref{fig:cp:simulate:n,fig:cp:simulate:n:k1,fig:cp:simulate:k} performs conformal prediction (CP) and its DAB variants under the linear model with additive Gaussian noise and with one of the three 2d data distributions defined in \Cref{appendix:experiments:pure:aug:bootstrap}. The augmentations considered are the same as the wild bootstrap setup in (a) of \Cref{appendix:experiments:conformal}. The behaviours of CP and DAB are similar across the three data distributions: DAB has better coverage at small $n$ but is more conservative than CP at large $n$. The coverage improvement in \Cref{fig:cp}(i) and  \Cref{fig:cp:simulate:n} is visible with $k=50$; for comparison, we also report results with $k=1$ in \Cref{fig:cp:simulate:n:k1} and varying $k$ in \Cref{fig:cp:simulate:k} to show that the improvement mainly comes from using a larger $k$ to incorporate more augmentations.

\begin{figure}[t]
  \centering
  \begin{tikzpicture}
        \node[inner sep=0pt] at (0,0)
            {\includegraphics[width=\textwidth]{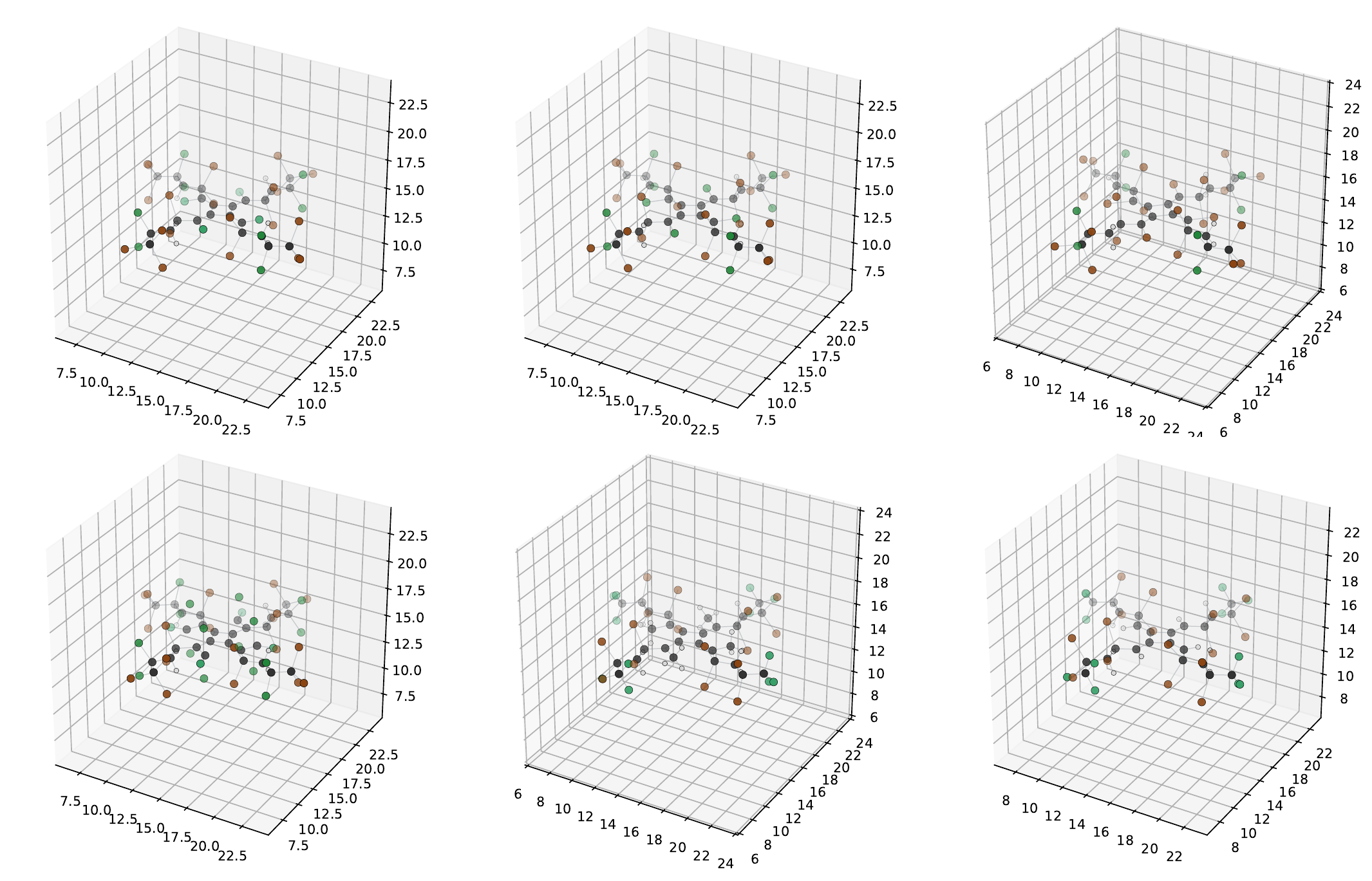}};
  \end{tikzpicture}
  \\[-.2em]
  \caption{Visualisation of selected molecules from the QM-symm database \citep{liang2019qm} with C4h point group symmetry. }
  \label{fig:amp:visual}
\end{figure}

\vspace{.5em}

\noindent 
\textbf{(b) Neural networks fitted on molecular data.} We consider the CP score function proposed by \cite{hu2022robust} applied to the GMP+SNN neural network (Gaussian multi-pole featurisation \cite{lei2022universal} followed by the SingleNN architecture \citep{liu2020singlenn}) and to the data set of C4h-symmetric molecules in the QM-sym database \citep{liang2019qm} (visualised in \Cref{fig:amp:visual}). Specifically:
\begin{itemize}[
    leftmargin=1em,
    labelwidth=0em,
    topsep=0.5em,
    partopsep=0em,
    itemsep=0em
]
    \item We first split the data set into the training set $X^{\rm train}$ and holdout set $X^{\rm holdout}$;  
    \item We train the GMP+SNN neural network on $X^{\rm train}$. Denote $\cX$ as the molecular input space. The trained network gives us a fitted predictor function $\hat E: \cX \rightarrow \R$ for the energy, as well as a fitted latent feature map $\hat g: \cX \rightarrow \R^{64}$ that represents the latent embedding in the last hidden layer of the network (see Figure 1 of \cite{hu2022robust});
    \item We uniformly sample $n-1$ molecule-energy pairs from  $X^{\rm holdout}$ as calibration data, denoted as $(X_2, E_2), \ldots, (X_n,E_n)$, as well as an additional data point $(X_1, E_1)$ as the test data point. The goal is to form a confidence interval for $E_1$;
    \item Let $n_{\rm train} = |X^{\rm train}|$. In the CP score function proposed by \cite{hu2022robust}, for any given $x \in \cX$, we write $X^{\rm train}_1(x), \ldots, X^{\rm train}_{n_{\rm train}}(x)$ as an enumeration of $X^{\rm train}$ ordered according to their Euclidean distance in the neural network feature space to $x$, i.e.~according to $d(x') = \| \hat g(x') - \hat g(x) \|$. This yields the extended dataset 
    \begin{align*}
        X \;=\; \big( X_i, E_i, X^{\rm train}_1(X_i), \ldots, X^{\rm train}_{n_{\rm train}}(X_i)  \big)_{1 \leq i \leq n}\;.
    \end{align*}
    The conformal score function then gives
    \begin{align*}
        f(X) \;=\; \mfrac{|  E_1 - \hat E( X_1 )|}{\frac{1}{M} \sum_{m=1}^M \| \hat g(X_1) - \hat g( X^{\rm train}_m(X_1) ) \| }\;,
    \end{align*}
    i.e.~the residual error in energy prediction is reweighted by the average Euclidean distance from $X_1$ to its $M$ closest neighbours. With $\Phi^{\rm CP}_b$ defined in \Cref{sec:conformal}, the transformed statistics to be ranked are
    \begin{align*}
        f(\Phi_b^{\rm CP}(X)) 
        \;=\; 
        \mfrac{|  E_{b+1} - \hat E( X_{b+1} )|}{\frac{1}{M} \sum_{m=1}^M \| \hat g(X_{b+1}) - \hat g( X^{\rm train}_m(X_{b+1}) ) \|^2 }
        \;\qquad\;
        \text{ for }
        1 \leq b \leq n-1\;.
    \end{align*}
\end{itemize}
The ranking of $X^{\rm train}$ for every input $X_i$ to find its $M$ closest neighbours is the main computational bottleneck in the above approach.

\vspace{.5em}

We propose a computationally more efficient DAB variant that, instead of finding $M$ closest neighbours to each $X_i$, just randomly samples $M$ points from the training dataset. Specifically, we consider $\Phi^{\rm DAB}_b$ defined in \Cref{sec:conformal} with the observation-wise augmentation $\phi$ given by 
\begin{align*}
    \phi(x_i, e_i, x'_1, \ldots, x'_{n_{\rm train}} )
    \;=\;
    (x_i, e_i, x'_{\pi(1)}, \ldots, x'_{\pi(n_{\rm train})} )\;,
\end{align*}
where $\pi$ is a uniform permutation of the index set $\{1, \ldots, n_{\rm train}\}$. In particular, $\pi$ removes any sorting in $(x'_1, \ldots, x'_{n_{\rm train}} )$. The corresponding statistic can be expressed as
\begin{align*}
    f(\Phi_{bj}^{\rm DAB}(X)) 
    \;=\; 
    \mfrac{|  E_{b+1} - \hat E( X_{b+1} )|}{\frac{1}{M} \sum_{m=1}^M \| \hat g(X_{b+1}) - \hat g( \tilde X^{\rm train}_{\pi_{bj}(m)} ) \|^2 }
    \;\qquad\;
    \text{ for }
    1 \leq b \leq n-1\;,
    1 \leq j \leq k\;,
\end{align*}
where we have enumerated $X^{\rm train} = (\tilde X_1, \ldots, \tilde X_{n_{\rm train}})$ and taken $\pi_{bj}$'s to be i.i.d.~copies of $\pi$. Note that we use $k=1$ in this setup to ensure computational efficiency.

\vspace{.5em}

\Cref{fig:cp}(ii) and \Cref{fig:amp:vary:n} report the experiments for comparing the quality of conformal prediction against DAB with $k=1$ as a function of the time taken for computation. For a similar coverage level and compute, we see that DAB achieves consistently smaller CI than conformal prediction. \Cref{fig:amp:vary:k} includes an ablation test with varying $k$, which again shows that DAB has similar coverage as conformal prediction but smaller CIs. Note however that for $k > 1$, the computational cost is multiplied by $k$.

\vspace{.5em}

We also remark that, while this is a molecular dataset with inherent geometric symmetries, the symmetries have already been incorporated in the neural network structure and cannot be used for DAB. We therefore exploit invariance structures in the conformal score function and use DAB to justify the incorporation of additional random transformations on top of conformal prediction.

\begin{figure}[h!]
  \centering
  \begin{tikzpicture}
        \node[inner sep=0pt] at (0,0)
            {\includegraphics[width=\textwidth]{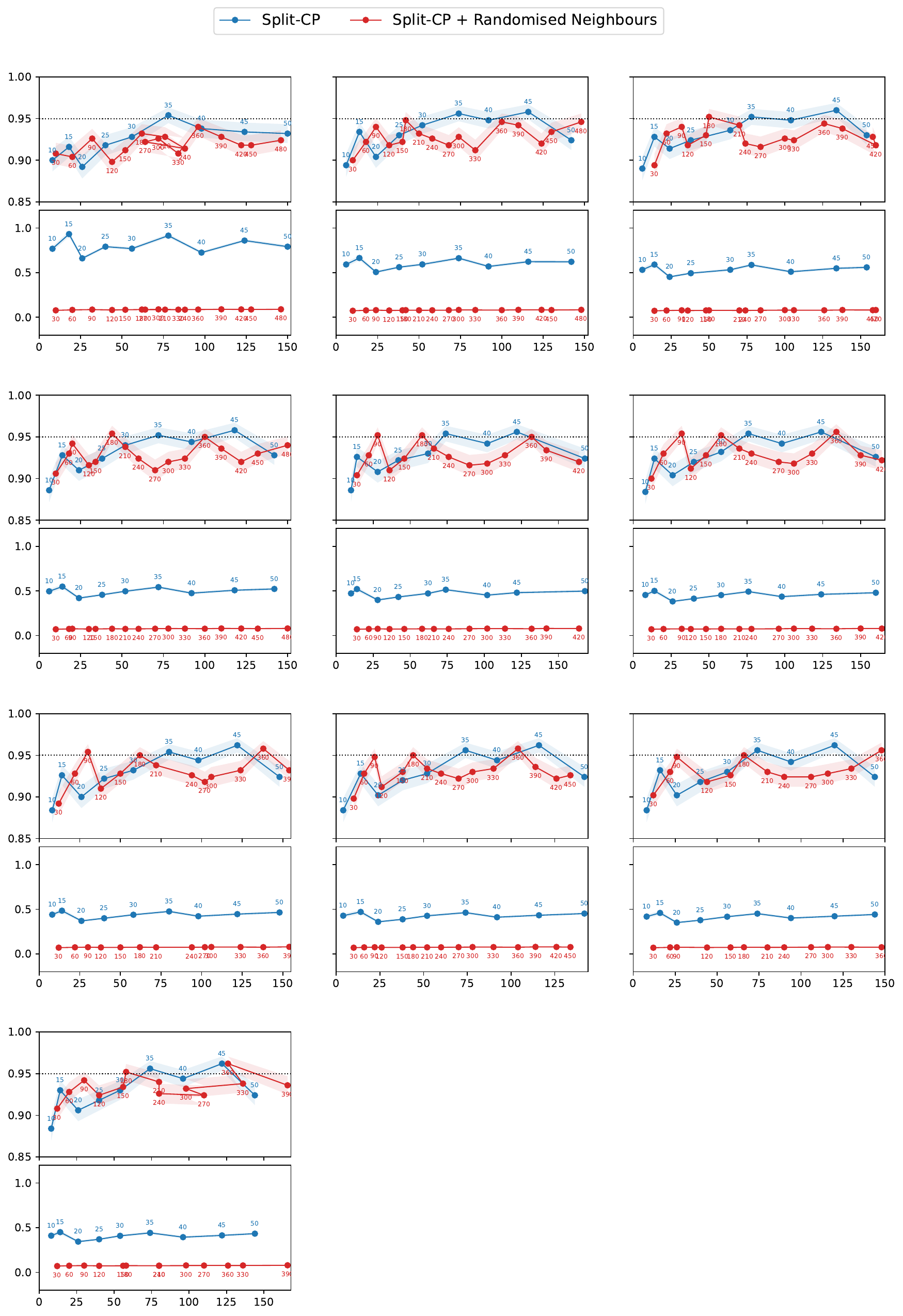}};
        \foreach \metric/\y in {coverage/8.65,{CI length}/6.40,coverage/3.35,{CI length}/1.10,coverage/-1.95,{CI length}/-4.20,coverage/-7.25,{CI length}/-9.50}{
            \node[inner sep=0pt,rotate=90] at (-7.55,\y){\scriptsize \metric};
        }
        \foreach \nn/\x/\y in {2/-4.60/10.,4/0.25/10.,6/5.05/10.,8/-4.60/4.65,10/0.25/4.65,12/5.05/4.65,14/-4.60/-0.65,16/0.25/-0.65,18/5.05/-0.65,20/-4.60/-5.95}{
            \node[inner sep=0pt] at (\x,\y){\scriptsize \textbf{\# neighbours = \nn}};
        }
  \end{tikzpicture}
  \\[-.2em]
  \caption{DAB-variants of conformal prediction for predicting the energy of a randomly chosen C4h-symmetric molecule from the QM-symm database \citep{liang2019qm} and with the conformal score function of \citep{hu2022robust}. Same setup as \cref{fig:cp}(ii), except that we also vary the number of neighbours used in both vanilla conformal prediction and conformal prediction with randomized neighbours. The $x$-axis is again GPU seconds per random trial, and the number next to each data point indicates the size of the calibration set used.}
  \label{fig:amp:vary:n}
\end{figure}
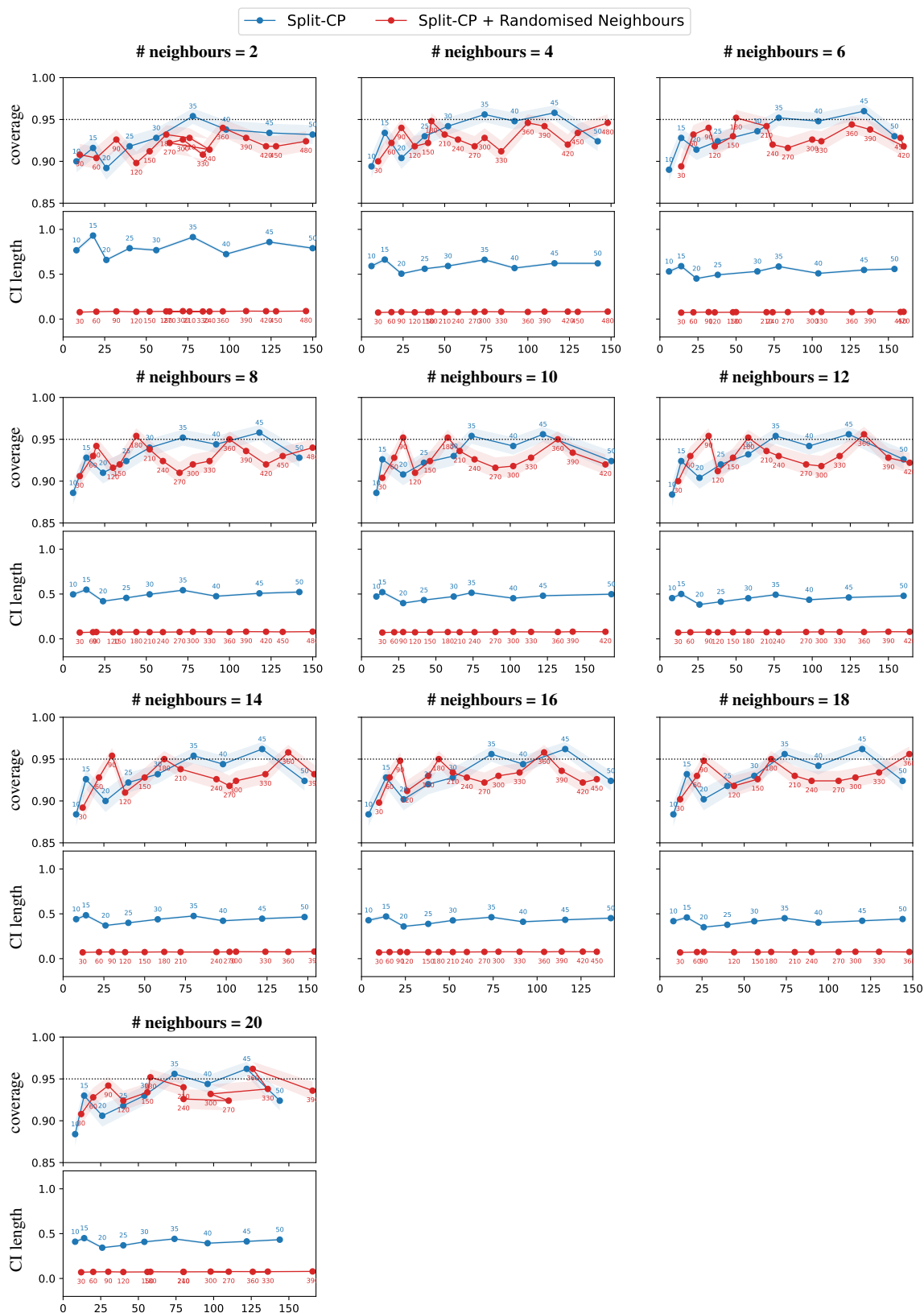
\begin{figure}[h!]
  \centering
  \begin{tikzpicture}
        \node[inner sep=0pt] at (0,0)
            {\includegraphics[width=\textwidth]{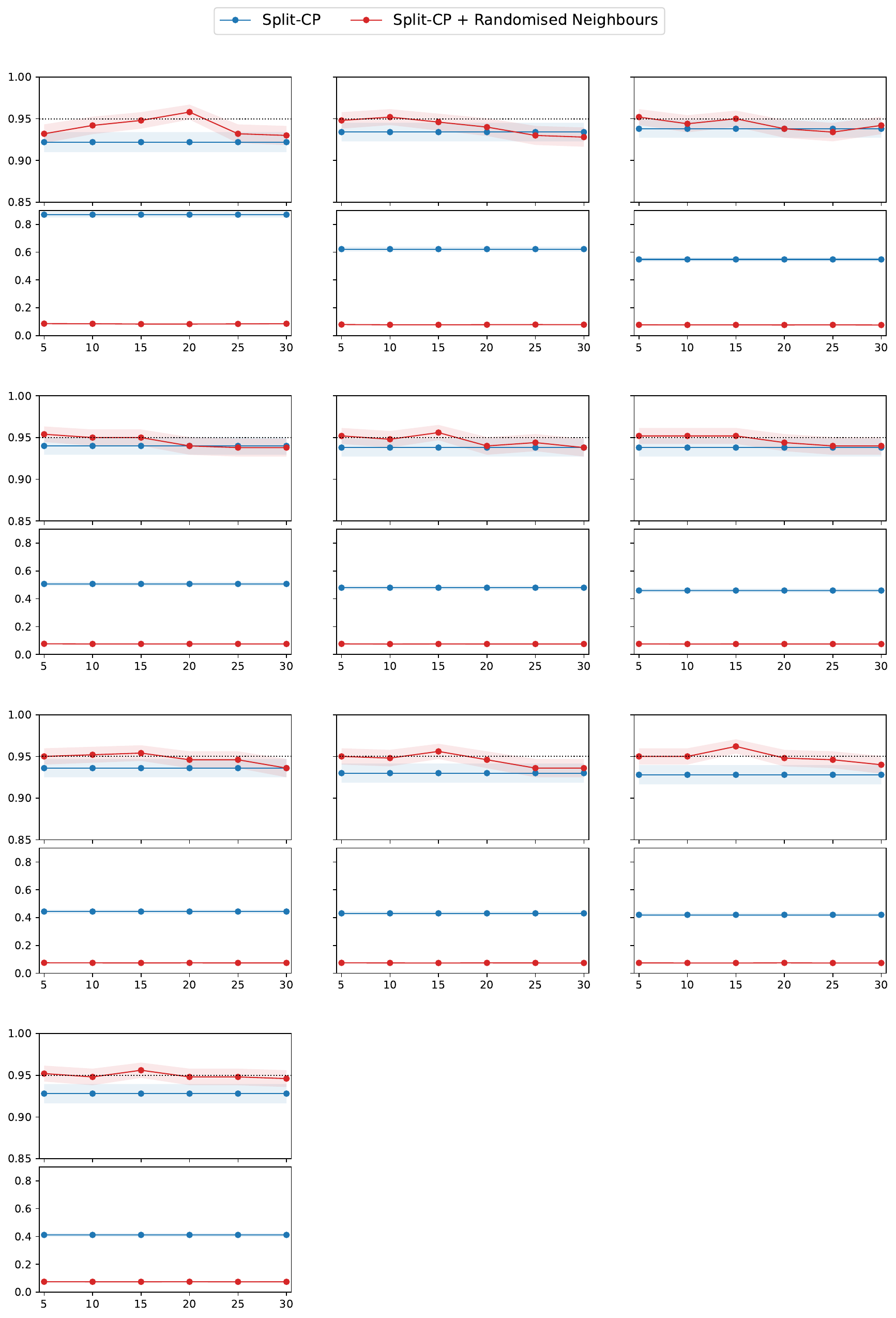}};
        \foreach \metric/\y in {coverage/8.65,{CI length}/6.40,coverage/3.35,{CI length}/1.10,coverage/-1.95,{CI length}/-4.20,coverage/-7.25,{CI length}/-9.50}{
            \node[inner sep=0pt,rotate=90] at (-7.55,\y){\scriptsize \metric};
        }
        \foreach \nn/\x/\y in {2/-4.60/10.15,4/0.25/10.15,6/5.05/10.15,8/-4.60/4.80,10/0.25/4.80,12/5.05/4.80,14/-4.60/-0.50,16/0.25/-0.50,18/5.05/-0.50,20/-4.60/-5.80}{
            \node[inner sep=0pt] at (\x,\y){\scriptsize \textbf{\# neighbours = \nn}};
        }
  \end{tikzpicture}
  \\[-.2em]
  \caption{DAB-variants of conformal prediction for predicting the energy of a randomly chosen C4h-symmetric molecule from the QM-symm database \citep{liang2019qm} and with the conformal score function of \citep{hu2022robust}. Same setup as \cref{fig:cp}(ii), except that we fix the calibration size of \texttt{Split-CP} as $n=30$ and that of \texttt{Split-CP+Randomised Neighbours} as $n=180$. The $x$-axis is now $k$, the number of additional augmentations on top of the cycling operations in conformal prediction; note that varying $k$ does not have an effect on \texttt{Split-CP}.
  }
  \label{fig:amp:vary:k}
\end{figure}

\clearpage

\vspace{.5em}

\noindent 
\textbf{(c) Language model conformal prediction \citep{kumar2023conformal} for the MMLU dataset \citep{hendrycks2020measuring} with a pretrained Qwen3-14B language model \citep{yang2025qwen3}.} We consider the CP method proposed by \cite{kumar2023conformal} applied to the multi-choice question answering dataset MMLU. Let $\cX$ be the text input space, $\cY$ be the set of finite choices, and $\Delta^{|\cY|}$ denote the $|\cY|$-dimensional probability simplex. Given a pre-trained large language model $\hat g: \cX \rightarrow \Delta^{|\cY|}$ and a dataset of question-answer pairs $X=(X_i, Y_i)_{i \leq n}$, the CP method in \cite{kumar2023conformal} amounts to ranking the statistics 
\begin{align*}
    f(X) \;=&\; 1 - \big( \hat g( X_1) \big)_{Y_1}
    &\text{ and }&&
    f(\Phi_b(X)) \;=&\; 1 - \big( \hat g( X_{b+1}) \big)_{Y_{b+1}}
    \;\text{ for $1 \leq b \leq n-1$\;,}
\end{align*}
where $( \hat g( X_1) )_{Y_1}$ reads out the predicted probability of $\hat g$ on $X_1$ at the location of the true label $Y_1$.

\vspace{.5em}

Our DAB variant of CP is motivated by the observation that the true question-label pair $(X_i, Y_i)$ is invariant under the permutation of choice numbering, as illustrated in \Cref{fig:mmlu:visualise} by renumbering choice $C$ as $D$ and choice $D$ as $C$. We therefore use $\Phi^{\rm DAB}_b$ defined in \Cref{sec:conformal} with the observation-wise augmentation $\phi$ given as a uniformly random permutation of the choice numbering, which is a random $\cX \times \cY \rightarrow \cX \times \cY$ that affects both the question text and the answer. 

\vspace{.5em}

\Cref{fig:mmlu:k1,fig:mmlu} report results with a subset of the MMLU (Massive Multitask Language Understanding) dataset by \cite{hendrycks2020measuring}, which consists of subject-specific multiple choice questions.  In more detail, we generate prediction probability vectors by a pretrained Qwen3-14B language model \citep{yang2025qwen3} for MMLU data under the first 6 subjects, namely (i) college computer science, (ii) formal logic, (iii) high school computer science, (iv) computer security, (v) machine learning and (vi) clinical knowledge. For each random trial, we then make $n$ random draws of the generated latents and perform conformal prediction for one of the data with the other $n-1$ calibration data points. \Cref{fig:mmlu:k1} is for $k=1$, whereas \Cref{fig:mmlu} performs an ablation test with varying $k$; in both setups, there is no visible improvement from DAB compared to conformal prediction. We conjecture that this is because the Qwen3-14B model is already close to being exactly invariant under question choice permutations, and there are not enough additional symmetries to be exploited by DAB.

\begin{figure}[t]
  \centering
  \begingroup
  \setlength{\fboxsep}{1.5pt}
  \newcommand{\recordlabel}[1]{{\setlength{\fboxsep}{1.5pt}\colorbox{blue!10}{\strut\textbf{\texttt{#1:}}}}}
  \newcommand{\recordhighlight}[2]{%
      {\setlength{\fboxsep}{0pt}\colorbox{yellow!25}{\recordlabel{#1} #2}}}
  \begin{tikzpicture}
        \node[
            draw,
            rounded corners=2pt,
            fill=gray!3,
            text width=6.9cm,
            inner sep=5pt,
            anchor=north
        ] (original) at (-3.75,0) {
            \raggedright\scriptsize
            \recordlabel{input}
            Any set of Boolean operators that is sufficient to represent all Boolean expressions is said to be complete. Which of the following is NOT complete?
            \\[.45em]
            \recordlabel{A} \{AND, NOT\}
            \\[.25em]
            \recordlabel{B} \{NOT, OR\}
            \\[.25em]
            \recordlabel{C} \{AND, OR\}
            \\[.25em]
            \recordlabel{D} \{NAND\}
            \\[.45em]
            \recordlabel{target} \textbf{C}
        };

        \node[
            draw,
            rounded corners=2pt,
            fill=gray!3,
            text width=6.9cm,
            inner sep=5pt,
            anchor=north
        ] (permuted) at (3.75,0) {
            \raggedright\scriptsize
            \recordlabel{input}
            Any set of Boolean operators that is sufficient to represent all Boolean expressions is said to be complete. Which of the following is NOT complete?
            \\[.45em]
            \recordlabel{A} \{AND, NOT\}
            \\[.25em]
            \recordlabel{B} \{NOT, OR\}
            \\[.25em]
            \recordhighlight{C}{\{NAND\}}
            \\[.25em]
            \recordhighlight{D}{\{AND, OR\}}
            \\[.45em]
            \recordhighlight{target}{\textbf{D}}
        };

        \node[inner sep=0pt, align=center, text width=6.9cm, anchor=north]
            at ([yshift=-.25cm]original.south)
            {\scriptsize \textbf{(i) An MMLU multiple-choice data record}};
        \node[inner sep=0pt, align=center, text width=6.9cm, anchor=north]
            at ([yshift=-.25cm]permuted.south)
            {\scriptsize \textbf{(ii) An option-permuted version of the same record}};
  \end{tikzpicture}
  \endgroup
  \\[-.2em]
  \caption{Example of a multiple-choice record from the MMLU dataset \citep{hendrycks2020measuring} and its option-permuted counterpart. The highlighted rows in (ii) show the exchange of options C and D and the corresponding update of the target label. }
  \label{fig:mmlu:visualise}
\end{figure}

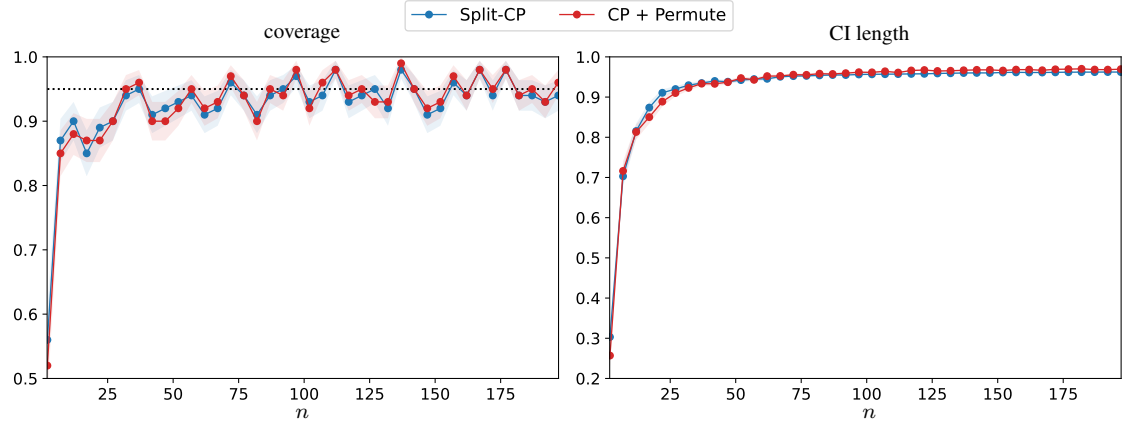
\begin{figure}[t]
    \centering 
    \begin{tikzpicture}
        \node[inner sep=0pt] at (0,0)
            {\includegraphics[width=\textwidth]{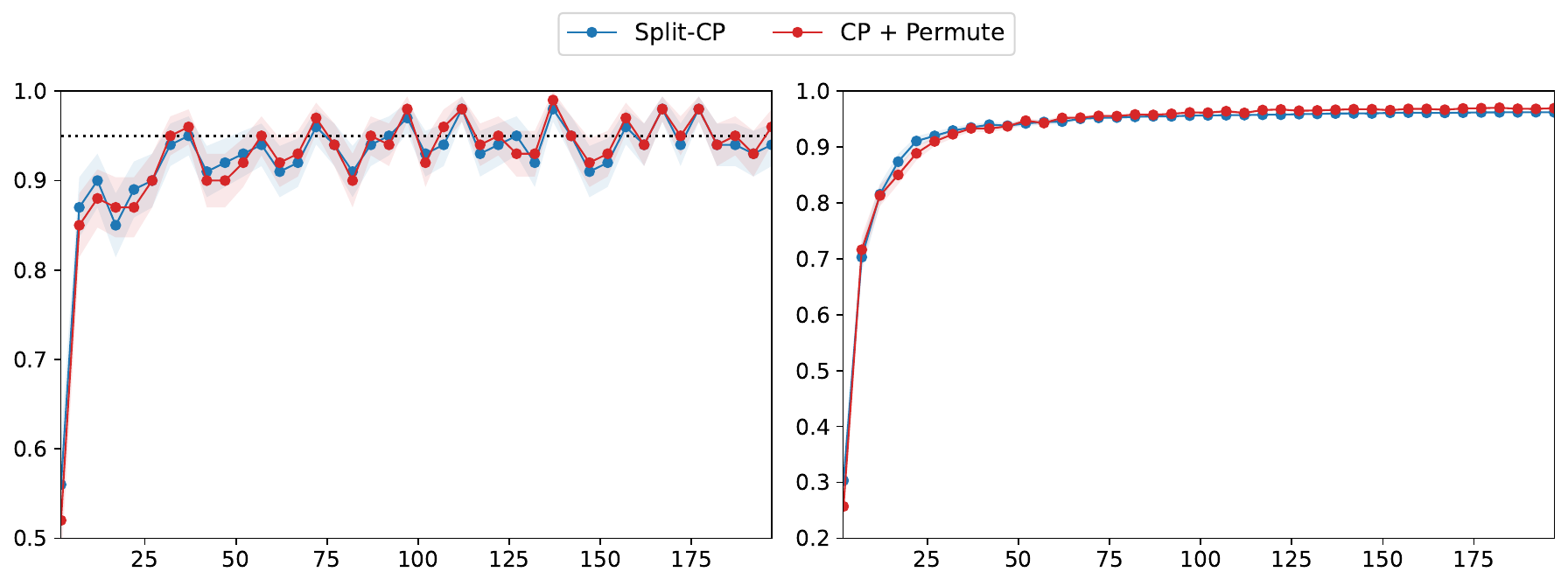}};
        \node[inner sep=0pt] at (-3.5,-2.8){\scriptsize $n$};
        \node[inner sep=0pt] at (4,-2.8){\scriptsize $n$};

        \node[inner sep=0pt] at (-3.5,2.2){\scriptsize coverage};

        \node[inner sep=0pt] at (4,2.2){\scriptsize CI length};

  \end{tikzpicture}
  \\[-.2em]
  \caption{DAB-variants of conformal prediction for predicting the answer of a multiple-choice question from the MMLU dataset \citep{hendrycks2020measuring} with a pretrained Qwen3-14B language model \citep{yang2025qwen3} and conformal score function from \cite{kumar2023conformal}. All experiments are over $100$ random seeds, $1-\alpha=95\%$, $k=1$ and varying $n$.}
  \label{fig:mmlu:k1}
\end{figure}

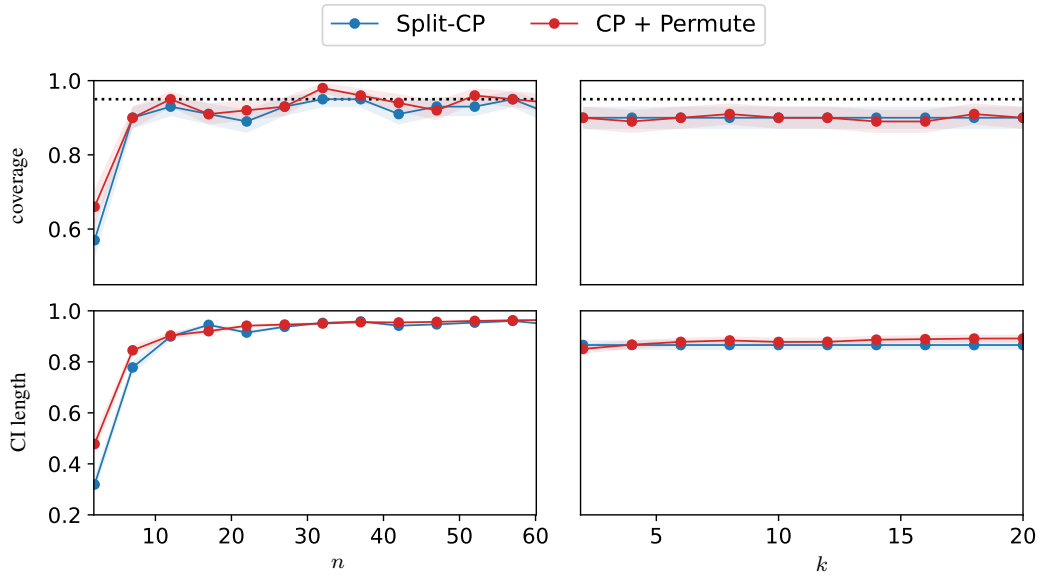
\begin{figure}[t]
    \centering 
    \begin{tikzpicture}
        \node[inner sep=0pt] at (0,0)
            {\includegraphics[width=.9\textwidth]{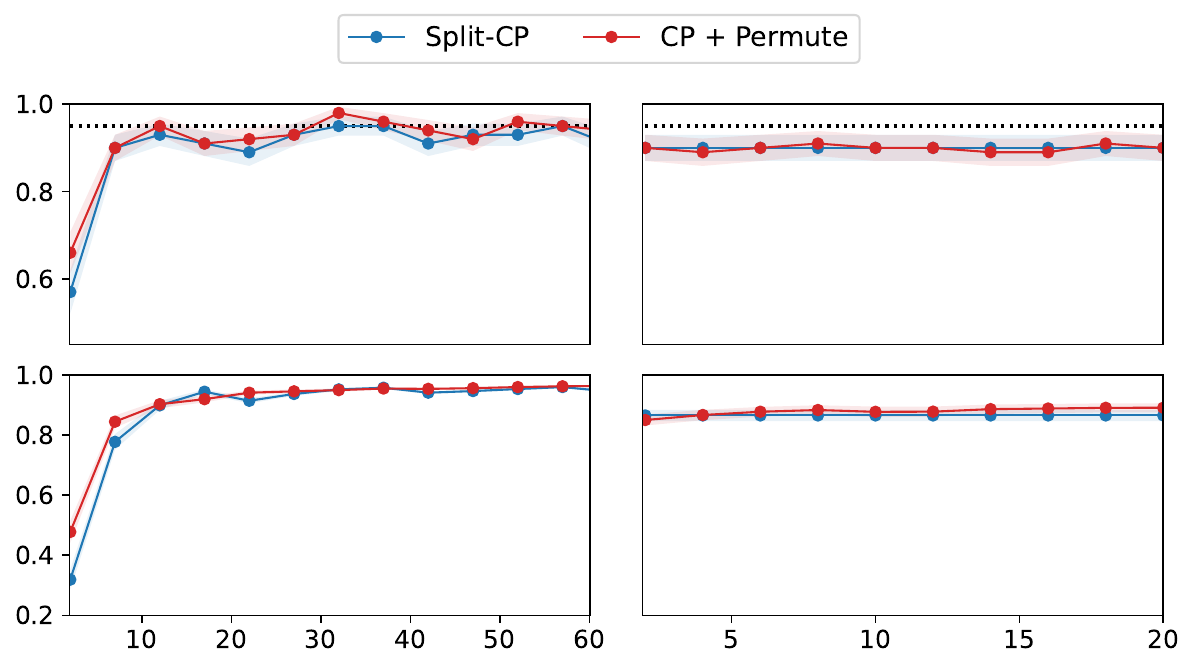}};
        \node[inner sep=0pt] at (-2.7,-3.8){\scriptsize $n$};
        \node[inner sep=0pt] at (3.7,-3.8){\scriptsize $k$};

        \node[inner sep=0pt,rotate=90] at (-6.9,1.2){\scriptsize coverage};

        \node[inner sep=0pt,rotate=90] at (-6.9,-1.8){\scriptsize CI length};

  \end{tikzpicture}
  \\[-.2em]
  \caption{Same setup as \Cref{fig:mmlu:k1}, except that the left plots are with $k=10$ and the right plots are with varying $k$.}
  \label{fig:mmlu}
\end{figure}

\end{document}